\documentclass[10pt]{article}

\usepackage{amsmath,amssymb,amsthm,amsfonts, mathrsfs,latexsym,bbm,xspace,graphicx,float,mathtools,braket,}
\usepackage[letterpaper,margin=0.94in]{geometry}

\usepackage{enumitem}

\usepackage{ytableau}

\usepackage[colorinlistoftodos]{todonotes}
\newtheorem{theorem}{Theorem}[section]
\newtheorem*{namedtheorem}{\theoremname}
\newcommand{\theoremname}{testing}

\newtheorem{lemma}[theorem]{Lemma}
\newtheorem{claim}[theorem]{Claim}
\newtheorem{proposition}[theorem]{Proposition}

\newtheorem{corollary}[theorem]{Corollary}

\newtheorem{problem}[theorem]{Problem}

\theoremstyle{definition}
\newtheorem{definition}[theorem]{Definition}
\newtheorem{construction}[theorem]{Construction}

\newtheorem{remark}[theorem]{Remark}

\newtheorem{example}[theorem]{Example}

\newtheorem{fact}[theorem]{Fact}

\renewcommand{\Pr}{\mathop{\bf Pr\/}}
\newcommand{\E}{\mathop{\bf E\/}}

\newcommand{\Var}{\mathop{\bf Var\/}}
\newcommand{\Cov}{\mathop{\bf Cov\/}}

\newcommand{\tr}{\mathrm{tr}} \newcommand{\Tr}{\tr} 
\newcommand{\poly}{\mathrm{poly}}

\newcommand{\sgn}{\mathrm{sgn}}

\newcommand{\diag}{\mathrm{diag}}

\newcommand{\R}{\mathbb R}
\newcommand{\C}{\mathbb C}
\newcommand{\N}{\mathbb N}

\newcommand{\F}{\mathbb F}

\newcommand{\p}{\mathsf{p}}

\newcommand{\eps}{\varepsilon}

\newcommand{\Hom}{\mathrm{Hom}}

\newcommand{\dens}{\mathfrak{D}}

\newcommand{\ignore}[1]{}

\usepackage[table]{xcolor}

\usepackage[T1]{fontenc}
\usepackage[utf8]{inputenc}
\usepackage{booktabs}
\definecolor{TSUYUKUSA}{RGB}{46, 169, 223}
\definecolor{KURENAI}{RGB}{203, 27, 69}
\definecolor{hkugreen}{HTML}{49B880}
\definecolor{hkublue}{HTML}{009CDE}
\definecolor{hkudarkgreen}{HTML}{1C584B}
\definecolor{hkudarkred}{HTML}{FF665E}
\definecolor{navyblue}{HTML}{000080}
\definecolor{quantumviolet}{HTML}{771C56}
\usepackage[pdfstartview=FitH,pdfpagemode=UseNone,colorlinks=true,citecolor=hkudarkgreen,linkcolor=quantumviolet,backref=page,linktocpage=true, urlcolor = navyblue]
{hyperref}

\usepackage[nameinlink,capitalize]{cleveref}

\usepackage{hyperref}
\usepackage{graphicx}
\usepackage{subcaption}
\usepackage{amssymb}
\usepackage{mathtools}
\usepackage{amsthm}
\usepackage{dsfont}
\usepackage{ytableau}
\usepackage{braket}
\usepackage{diagbox}
\usepackage{bbm}

\usepackage{pgfplots}

\pgfplotsset{compat=1.18}
\usepackage{tikz}
\usetikzlibrary{quantikz2}
\usepackage{aligned-overset}
\usetikzlibrary{decorations.pathmorphing, calc, arrows.meta}
\usepackage{thm-restate}

\usepackage{framed}
\usepackage{environ}

\definecolor{deepblue}{RGB}{0, 51, 153}
\definecolor{deepred}{RGB}{153, 0, 51}

\usepackage{enumitem}
\usepackage{algorithmic}
\usepackage{float}

\newcommand{\I}{\mathrm{i}}

\allowdisplaybreaks[1]

\newcommand{\rank}{\textnormal{rank}}


\renewcommand{\Ket}[1]{| #1 \rangle \! \rangle}
\renewcommand{\Bra}[1]{\langle \! \langle #1 |}
\renewcommand{\Braket}[1]{\langle \! \langle #1 \rangle \! \rangle}

\newcommand{\gap}{\mathrm{gap}}

\newcommand{\ybox}{%
  \tikz[baseline={([yshift=-.5ex]current bounding box.center)}, scale=0.25] {
    \draw (0,0) rectangle (1,1);
  }%
}
\newcommand{\yrow}{%
  \tikz[baseline={([yshift=-.5ex]current bounding box.center)}, scale=0.25] {
    \draw (0,0) rectangle (2,1);
    \draw (1,0) -- (1,1);
  }%
}
\newcommand{\ycol}{%
  \tikz[baseline={([yshift=-.5ex]current bounding box.center)}, scale=0.25] {
    \draw (0,0) rectangle (1,2);
    \draw (0,1) -- (1,1);
  }%
}

\newcommand{\sybox}{%
  \tikz[baseline={([yshift=-.5ex]current bounding box.center)}, scale=0.15] {
    \draw (0,0) rectangle (1,1);
  }%
}
\newcommand{\syrow}{%
  \tikz[baseline={([yshift=-.5ex]current bounding box.center)}, scale=0.15] {
    \draw (0,0) rectangle (2,1);
    \draw (1,0) -- (1,1);
  }%
}
\newcommand{\sycol}{%
  \tikz[baseline={([yshift=-.5ex]current bounding box.center)}, scale=0.15] {
    \draw (0,0) rectangle (1,2);
    \draw (0,1) -- (1,1);
  }%
}

\newcommand{\su}{\mathsf{SU}}
\usepackage{bm}

\usepackage{tabularx}
\usepackage{multirow}
\usepackage{tabularray}
\usepackage{array}
\usepackage{afterpage}
\usepackage{enumitem}
\usepackage{authblk}
\UseTblrLibrary{varwidth}

\newcommand{\etc}{etc.}
\newcommand{\X}{\mathcal{X}}
\newcommand{\obs}{\mathsf{Obs}}
\newcommand{\B}{\mathscr{B}}
\renewcommand{\P}{\mathscr{P}}
\newcommand{\ie}{\text{i.e.}\xspace}
\newcommand{\eg}{\text{e.g.}\xspace}
\newcommand{\etal}{\text{et al.}\xspace}

\newcommand{\Y}{\mathsf{Y}}
\newcommand{\dd}{\,\mathrm{d}}
\newcommand{\U}{\mathsf{U}}
\newcommand{\gl}{\mathrm{GL}}

\newcommand{\qnc}{\mathsf{QNC}}
\newcommand{\nc}{\mathsf{NC}}

\renewcommand{\b}{\mathsf{B}}

\newcommand{\Herm}{\mathrm{Herm}}
\newcommand{\bfp}{\mathbf{p}}
\newcommand{\bfq}{\mathbf{q}}
\renewcommand{\H}{\mathcal{H}}

\newcommand{\openone}{\mathbb{I}}
\newcommand{\caH}{\mathcal{H}}
\newcommand{\caU}{\mathcal{U}}
\newcommand{\bigO}{\mathcal{O}}
\newcommand{\bigo}[1]{\mathcal{O}\left(#1\right)}
\newcommand{\bigtheta}[1]{\Theta\left(#1\right)}

\newcommand{\covering}{\mathsf{C}}

\usepackage{mathtools,mleftright}
\mleftright
\def\eqref#1{\textup{(\ref{#1})}}
\newcommand{\eref}[1]{Eq.~\textup{(\ref{#1})}}
\newcommand{\lref}[1]{Lemma~\ref{#1}}
\newcommand{\tref}[1]{Theorem~\ref{#1}}
\newcommand{\pref}[1]{Proposition~\ref{#1}}
\newcommand{\corref}[1]{Corollary~\ref{#1}}

\newcommand{\dref}[1]{Definition~\ref{#1}}
\newcommand{\ham}{\mathsf{Ham}}
\newcommand{\haar}{\mathrm{Haar}}

\title{Optimal classical shadow estimation of unitary channels at Heisenberg limit
}

\author[1]{Entong He\thanks{\href{mailto:ethe@cs.hku.hk}{\texttt{ethe@cs.hku.hk}}}}
\author[1]{Zihao Li\thanks{\href{mailto:zihaoli@hku.hk}{\texttt{zihaoli@hku.hk}}}}
\author[2, 3]{Noam Scully}
\author[3, 2, 4]{Sisi Zhou}
\author[1]{Yuxiang Yang\thanks{\href{mailto:yuxiang@cs.hku.hk}{\texttt{yuxiang@cs.hku.hk}}}}

\affil[1]{\normalsize QICI Quantum Information and Computation Initiative, \protect\\ School of Computing and Data Science, The University of Hong Kong, Hong Kong SAR, China.}
\affil[2]{\normalsize Department of Physics and Astronomy, University of Waterloo, Ontario N2L 3G1, Canada.}
\affil[3]{\normalsize Perimeter Institute for Theoretical Physics, Waterloo, Ontario N2L 2Y5, Canada.}
\affil[4]{\normalsize Department of Applied Mathematics and Institute for Quantum Computing, \protect\\ University of Waterloo, Ontario N2L 3G1, Canada.}

\allowdisplaybreaks
\date{}

\begin{document}

\maketitle

\begin{abstract}
Full tomography of an unknown quantum evolution is resource-intensive and often unnecessary when the goal is only to predict selected properties. 
This motivates the study of classical shadow estimation of unitary channels (CSEU), a task in which one queries an unknown $d$-dimensional unitary $U$ and stores classical data that can later be used to predict expectation values $\mathrm{tr}[O  U\rho U^\dagger]$ for arbitrary input states $\rho$ and observables $O$. Here we propose a parallel and non-adaptive CSEU protocol that uses $\mathcal{O}(\min\{d\varepsilon^{-2},d^{3/2}\varepsilon^{-1}\})$ queries to $U$ when the observables have constant effective size.
This achieves Heisenberg scaling with respect to the additive error $\varepsilon$ and is query-optimal, as we prove a matching lower bound $\Omega(\min\{d\varepsilon^{-2},d^{3/2}\varepsilon^{-1}\})$ for the query complexity, which remains valid even when stronger access to the unknown unitary is available. 
Our CSEU protocol is a versatile primitive with applications to several quantum learning tasks. In particular, it yields a near-optimal algorithm for learning an arbitrary Hamiltonian from its real-time evolution in the high-precision regime and recovers the best-known bounds for several tasks in a unified way.
\end{abstract}

\hypersetup{linktoc=section}

\tableofcontents

\newpage

\section{Introduction}\label{sec:intro}

Learning the properties of an unknown quantum evolution is a fundamental task in quantum information processing \cite{4655455, PhysRevLett.97.170501, HKOT23, Gebhart_2023, Angrisani_2025}.
The most direct approach is full quantum process tomography, which reconstructs a complete classical description of the underlying quantum channel \cite{Kahn_2007, Yang_2020, HKOT23, mele2025optimal, grewal2026efficientlearningstructuredquantum, chen_Girardi2026}.
However, full tomography often provides much more information than necessary, as in many physical applications the goal is only to predict selected properties of the quantum process.
This is analogous to the task of predicting properties of quantum states \cite{buadescu2021improved,aasonson_shadow2018}, where classical shadows allow one to estimate the expectation values of many observables on an unknown state without reconstructing the full state \cite{Huang_2020,Grier_2024,ChenRobust2021,west2026classicalshadowsarbitrarygroup,PhysRevResearch.5.023027,PhysRevLett.131.240602,Zhou2023perform,PRXQuantum.5.010352,5khs-7dyz}. 
This classical shadow perspective has had a broad impact because it separates data acquisition from subsequent prediction, enabling a ``measure once, use repeatedly'' paradigm \cite{Huang_2020,elben2023randomized}.

Building on this perspective, we study the \emph{classical shadow estimation of unitary channels} (CSEU) problem: given access to an unknown $d$-dimensional unitary $U$, the goal is to produce classical data from which one can predict the expectation values $\Tr[OU\rho U^\dagger]$ for classically specified input states $\rho$ and observables $O$ \cite{Kunjummen_2023, Levy_2024, Li_2025}.
This formulation is particularly useful when queries to the unknown quantum dynamics are costly or available only during a limited experimental time window, while the quantities of interest may be specified only afterwards. CSEU therefore provides a natural framework for learning and
characterizing unknown quantum dynamics.

A line of prior work has gradually pieced together the optimal picture for CSEU, yet a fundamental gap remains. On the one hand, a na\"ive yet once-and-for-all approach to CSEU is to perform full tomography of the unknown unitary channel and then use the resulting estimate to answer all subsequent prediction requests without any further queries.
This approach can achieve a query complexity scaling as $\eps^{-1}$ in the target precision $\eps$ \cite{HKOT23, Yang_2020}, known as the Heisenberg-limited scaling, but it potentially incurs a large query cost in the system dimension $d$. 
On the other hand, earlier CSEU protocols achieved nearly optimal dependence on $d$, but did not attain the Heisenberg scaling \cite{Li_2025}. These results leave open the optimal joint dependence of CSEU on $d$ and $\eps$.
Another line of recent work improves query efficiency under additional restrictions on the channel or on the allowed input pairs of states and observables \cite{Huang_2022,Huang_learningtopredict2023,Huang_2024,caro2022learning,Zhao_2024,du2025efficient}. 
By contrast, in this work, we focus on the fully general CSEU setting, where the unknown process is an arbitrary unitary and the prediction requests may involve arbitrary combinations of input states and observables, ensuring that our protocol is maximally flexible for a broad range of applications.

In this fully general setting, our main result closes this gap: we propose a parallel and non-adaptive CSEU protocol that achieves the optimal scaling in both the dimension $d$ and the target precision $\eps$. For observables with constant effective size, the protocol uses $\mathcal{O}(\min\{d\eps^{-2},d^{3/2}\eps^{-1}\})$ queries to the unknown unitary channel, which attains the Heisenberg-limited scaling. 
Furthermore, we prove a matching lower bound of $\Omega(\min\{d\eps^{-2}, d^{3/2}\eps^{-1}\})$ on the query complexity, which holds even in a stronger access model where each query can be made by the learner to the unknown unitary, its inverse, or the controlled version of either. These two bounds jointly give a complete characterization of the query complexity of the CSEU problem.
At the technical level, our upper-bound protocol builds upon the covariant unitary-learning framework of \cite{Bisio_2010}.
We use strictly parallel queries to the unknown unitary channel and directly apply a covariant measurement to generate unitary-valued classical snapshots, from which we construct unbiased estimators of $\Tr[OU\rho U^\dagger]$ tailored to different precision regimes. For the lower bound, we construct the hard instance from a large packing of fractional reflections and show that with high probability, the predictions of a CSEU solver can be used to decode a list that contains the hidden reflection labels, a task commonly known as list decoding \cite{list_decoding_2008}. A dimension-counting argument based on the polynomial method then upper bounds the advantage of successful decoding, and by tuning the parameters we can establish the matching query lower bound.

Beyond being optimal for the CSEU problem itself, our protocol turns out to be a versatile primitive for a variety of downstream quantum learning tasks.
For learning completely general Hamiltonians from real-time dynamics without assuming locality, sparsity, or any ansatz structure, it achieves near-optimal total evolution time in the high-precision regime. For learning Pauli transfer matrices of unitary channels, it likewise achieves near-optimal query complexity in the high-precision regime.
We further invoke our protocol to recover the best-known query-complexity upper bounds for unitary channel tomography under parallel queries and for inverse-free amplitude estimation, and to improve the query efficiency of learning shallow quantum circuits. Together, these applications establish our CSEU protocol as a general-purpose tool for learning reversible quantum dynamics, especially for quantum learning tasks that require high precision.

\section{Optimal classical shadow estimation of unitary channels}\label{sec:main_results}

Throughout this paper, we will use standard asymptotic notations from \cite{concrete_math}: the symbols $\mathcal{O}(\cdot)$, $\Omega(\cdot)$, $\Theta(\cdot)$, $o(\cdot)$ hide universal constants independent of the relevant parameters. 
The notations $\widetilde{\mathcal{O}}(\cdot)$, $\widetilde{\Omega}(\cdot)$, and $\widetilde{\Theta}(\cdot)$ further suppress polylogarithmic factors in relevant parameters. 
When we say that a protocol is \emph{near-optimal}, we mean that its query or time complexity matches the corresponding lower bound up to such polylogarithmic factors. 
For a unitary operator $U$, we write the calligraphic $\mathcal U$ for the associated unitary channel,
$\mathcal U(\cdot):=U(\cdot)U^\dagger$.
For a pure state vector $\ket{\psi}$, we use the symbol $\psi$ to denote the rank-one projector $\psi:=\ket{\psi}\!\bra{\psi}$. 

We now formalize the classical shadow estimation task for unitary dynamics (CSEU). 
The definition separates $(a)$ the data-acquisition stage, where the unknown unitary is queried, from $(b)$ the later prediction stage, where arbitrary state-observable requests are answered using only the stored classical data.

\begin{problem}[CSEU]\label{prob:CSEU}
Consider a $d$-dimensional quantum system. Let $1\leq \B\leq d$, $0<\eps<1$, and 
\[
\obs(\B) = \left\{ O \in \mathbb{C}^{d \times d}: O = O^\dagger, \|O\|_{\infty} = 1, \tr(O_0^2)\leq \B \right\}
\]
be a set of bounded observables, where $O_0=O-\Tr(O) \cdot \openone/d$ denotes the traceless part of $O$. 
Let $U\in \U(d)$ be an unknown unitary accessible only through black-box oracle queries. 
The goal of CSEU is to output some classical data $\mathsf{CS}(U)\in\{0,1\}^*$ (called the \emph{classical shadow}) together with a deterministic prediction function
\[
f:\{0,1\}^*\times \C^{d \times d} \times \C^{d \times d} \rightarrow \mathbb{R}, 
\]
such that for any classically specified quantum state $\rho\in \C^{d \times d}$ and observable $O\in \obs(\B)$, 
\begin{align}\label{eq:desired_accuracy}
\Pr\left[
\bigl|f \bigl(\mathsf{CS}(U),\rho,O\bigr)-\Tr[O  U\rho U^\dag]\bigr|
\leq \eps
\right]
\geq \frac{2}{3}.
\end{align}
Here, the probability is taken over the randomness in the generation of $\mathsf{CS}(U)$.
\end{problem}

The condition $\|O\|_\infty=1$ in Problem~\ref{prob:CSEU} is imposed only for convenience; a general nonzero observable $O$ can be handled by applying the protocol to $O/\|O\|_\infty$ and rescaling the final estimate accordingly. 

\begin{remark}[Predicting many properties]
\label{remark:Mproperties}
The CSEU guarantee immediately extends from a single prediction request to a finite batch of requests. 
Let $0<\delta<1$ and let
$(\rho_1,O_1),\ldots,(\rho_M,O_M)$ be $M$ classically specified state-observable pairs with $O_\ell\in\obs(\B)$. 
If a protocol solves Problem~\ref{prob:CSEU} using $K$ oracle queries, then by repeating the protocol independently
$\mathcal{O}(\log(M/\delta))$ times and taking coordinate-wise medians, one can estimate all quantities
$\Tr[ O_{\ell} U \rho_{\ell} U^\dag]$
to additive error $\eps$ with overall failure probability at most $\delta$, using $\mathcal{O}\left(K\log(M/\delta)\right)$ queries (see Section~\ref{sec:mean_Amplification} for more details). 
Thus, as in standard shadow protocols for quantum states \cite{Huang_2020}, the query overhead for predicting many properties of a unitary channel scales logarithmically in the number of requested properties.
\end{remark}

Several approaches have been developed for the CSEU problem in Problem~\ref{prob:CSEU}. 
One natural approach is to first perform full tomography of the unknown unitary $U$ and then use the resulting estimate $\widehat{U}$ to answer prediction requests.
This corresponds to the case $\mathsf{CS}(U)= \widehat{U}$ and $f \bigl(\mathsf{CS}(U),\rho,O\bigr)= \tr\bigl( O\widehat{U}\rho\widehat{U}^\dag\bigr)$. 
For example, the query-optimal tomography protocol of \cite{HKOT23}  achieves diamond-norm error $\eps$
using $\bigO(d^2\eps^{-1})$ queries to $U$. 
This approach achieves Heisenberg scaling in $\eps$, but it incurs a quadratic dependence on the dimension $d$. 
The works that initiated the study of CSEU \cite{Kunjummen_2023, Levy_2024} reduce process shadow estimation to state shadow estimation through the conventional Choi-Jamiołkowski isomorphism \cite{CHOI1975285, Jamiokowski1972}, leading to a query complexity of $\bigO(d^2\B\eps^{-2})$. 
Subsequently, Grier, Pashayan, and Shaeffer \cite{Grier_2024} developed a sample-optimal protocol for shadow estimation of pure states, which yields a query complexity of $\bigO(d^2\eps^{-2})$ when adapted to CSEU \cite{Li_2025}. 
More recently, Li, Yi, Zhou, and Zhu \cite{Li_2025} achieved the best-known query complexity prior to this work, namely $\bigO(d\eps^{-2}+d\sqrt{\B}\eps^{-1})$. 
Their protocol reduces the dimension dependence to linear, but it still does not achieve Heisenberg scaling in the target precision $\eps$. Table~\ref{tab:compare} summarizes these bounds and highlights the improvement achieved in this work.

Our work substantially improves the query efficiency of CSEU by designing a two-stage learning protocol that builds upon the canonical framework of Bisio, Chiribella, D'Ariano, Facchini, and Perinotti for optimal average-case learning of unitary transformations \cite{Bisio_2010}. 
Within this framework, we make carefully tailored choices of probe states and covariant measurements for different parameter regimes. These choices enable efficient extraction of classical data from parallel queries to the unknown unitary, from which we construct classical estimators of the target expectation value $\Tr[O  U\rho U^\dagger]$. Our main results are summarized as follows.

\renewcommand{\thetable}{\arabic{table}}
\begin{table*}[t!]
	\centering
	\renewcommand{\arraystretch}{1.35}
	\setlength{\tabcolsep}{15pt}
	\begin{tabular}{c|l|l}
		\hline\hline
		& \multicolumn{1}{c|}{\textbf{Upper bounds}} & \multicolumn{1}{c}{\textbf{Lower bounds}} \\[0.3ex]
		\hline 
		
		\multirow{6.5}{*}{\textbf{Prior work}} 
		&\rule{-3.3pt}{4.3ex}  $\bigo{\dfrac{d^2 \B}{\eps^2}}$ \cite{Kunjummen_2023, Levy_2024} &\multirow{6.5}{*}{$\Omega\left(\dfrac{d}{\log d}\right)$ \cite{Li_2025}}  \\[2.2ex]
		
		& $\bigo{\dfrac{d^2}{\eps^2}}$ \cite{Grier_2024}
		& \\[2.2ex]

		& $\bigo{\dfrac{d^2}{\eps}}$ \cite{HKOT23}
		& \\[2.2ex]
		
		& $\bigo{\dfrac{d}{\eps^2} + \dfrac{d\sqrt{\B}}{\eps}}$ \cite{Li_2025}
		&  \\[2ex]
		
		\hline \rowcolor{hkugreen!15}
		\textbf{This work}
		&\rule{-3.3pt}{4.4ex} $\bigo{\min\left\{\dfrac{d}{\eps^2} + \dfrac{d\sqrt{\B}}{\eps}, \dfrac{d^{3/2}}{\eps} \right\}}$
		& $\Omega\left(\min\left\{ \dfrac{d}{\eps^2}, \dfrac{d^{3/2}}{\eps} \right\}\right)$ \\[2.1ex]
		\hline\hline
	\end{tabular}
    	\caption{\label{tab:compare}
Comparison of bounds on the query complexity of the CSEU problem.
Here we compare our results on the (general) CSEU problem with the best-known previous results (summarised from \cite[Table I]{Li_2025}).
Both our upper and lower bounds improve upon all previous results.
For the first time, we close the gap between the upper and lower bounds when $\B$ is a constant. 
	}
\end{table*}

\begin{theorem}[Upper bound]
\label{thm:cseu_upper_bound_informal}
There exists a protocol that solves Problem~\ref{prob:CSEU} using at most
$$
\mathcal{O}\left( \min\left\{ \frac{d}{\eps^2} + \frac{d \sqrt{\B}}{\eps} , \frac{d^{3/2}}{\eps} \right\}\right)
$$
parallel queries to the unknown unitary $U$. For this protocol, the resulting classical shadow data can be stored using $\mathcal{O}(\poly(d))$ complex numbers, and for any given CSEU request $(\rho,O)$, the prediction function $f$ can be evaluated in $\mathcal{O}(\poly(d))$ classical time, up to standard finite-precision overheads\footnote{Here and throughout, standard finite-precision overheads refer to the additional bit/time complexity required to represent real or complex numbers and to perform arithmetic operations with sufficient precision to preserve the stated error guarantees. Moreover, for simplicity, we will ignore the approximate implementation error of all quantum operations: as per the Solovay-Kitaev theorem \cite{Nielsen2012}, it can be made sufficiently small with merely polylogarithmic gate complexity overhead.}.
\end{theorem}

Theorem~\ref{thm:cseu_upper_bound_informal} gives a CSEU protocol that achieves the Heisenberg scaling $1/\eps$ in the target precision, while substantially improving the dependence on the dimension $d$ over tomography-based approaches.
Specifically, in the high-precision regime $\eps\leq d^{-1/2}$, the query complexity becomes $\bigo{d^{3/2}\eps^{-1}}$. Table~\ref{tab:compare} compares the query complexity of our protocol with that of existing CSEU protocols in the prior work. Moreover, we complement Theorem \ref{thm:cseu_upper_bound_informal} with the following query lower bound.

\begin{theorem}[Lower bound]
\label{thm:cseu_lower_bound}
Any protocol that solves Problem~\ref{prob:CSEU}
through black-box queries to 
$U$, $U^\dagger$, $\mathsf{c}U = |0\rangle\langle 0| \otimes \mathbb{I} + |1\rangle\langle 1| \otimes U$,
and $\mathsf{c}U^\dagger = |0\rangle\langle 0| \otimes \mathbb{I} + |1\rangle\langle 1| \otimes U^\dagger$ must use ${\Omega}\left( \min\{ d \eps^{-2}, d^{3/2} \eps^{-1} \}\right)$ queries.
\end{theorem}

Together, Theorems \ref{thm:cseu_upper_bound_informal} and \ref{thm:cseu_lower_bound} settle the query complexity of the CSEU problem: for low-rank (or equivalently, small effective-size) observables subject to $\B=\bigo{1}$, it reads
$$
\Theta\left(\min\left\{\frac{d}{\eps^2}, \frac{d^{3/2}}{\eps}\right\} \right).
$$
In particular, our protocol is optimal in both the dimension $d$ and the accuracy $\eps$. Furthermore, as the lower bound statement in Theorem \ref{thm:cseu_lower_bound} holds even when stronger coherent queries to the unknown unitary are allowed, this optimality persists against protocols with such stronger access models.

\section{Applications}
\label{sec:applications_cseu}

Our query-optimal CSEU protocol serves as a versatile subroutine in quantum learning. In this section, we invoke our CSEU protocol to solve the following representative tasks: learning general Hamiltonians, learning Pauli transfer matrices, inverse-free amplitude estimation, unitary channel tomography under parallel queries, and learning shallow quantum circuits. Remarkably, for Hamiltonian and PTM learning, we further establish lower bounds showing that our resulting protocols are near-optimal in the high-precision regime. For the latter three tasks, CSEU improves on or matches the best known query complexities. Taken together, these results position CSEU as a unifying and powerful primitive for query-efficient quantum learning.

\subsection{Learning general Hamiltonians}
\label{subsec:learning_general_hamiltonians}

Learning an unknown Hamiltonian from its real-time dynamics is a central task in quantum learning theory and many-body physics.
Consider an arbitrary unknown $n$-qubit Hamiltonian $H$ that is bounded and traceless:
\[
H=\sum_{\bfp\in\{0,1,2,3\}^n}\mu_H(\bfp)\sigma_{\bfp},
\]
where $\sigma_{\bfp}$ are $n$-qubit Pauli operators, $\mu_H(\bfp)=\frac1d\Tr(H\sigma_{\bfp})\in\mathbb R$ are the Pauli coefficients, and $d=2^n$. 
Given query access to the time evolutions
$U_t=e^{-\mathrm{i}Ht}$
for tunable times $t\geq 0$, for every $0 < \eps < 1$, we consider two natural learning tasks: 
\begin{enumerate}
\item \textbf{Pauli coefficient learning.} The goal is to output estimates $\widehat{\bm{\mu}}_H$ of all Pauli coefficients of $H$, such that $\left|\widehat{\mu}_H(\bfp) - \mu_H(\bfp) \right| \leq \eps \|H\|_{\infty}$ for all $\bfp \in \{0, 1, 2, 3\}^n$ with high success probability. 

\item \textbf{Hamiltonian learning in normalized-Frobenius-norm (NFN).} The goal is to output a classical description $\widehat H\in \mathbb{C}^{d\times d}$ of the Hamiltonian, such that
$\frac{1}{\sqrt{d}}\|\widehat H-H\|_F\leq  \eps \|H\|_{\infty}$
with high success probability. 
\end{enumerate}

\noindent 
These two tasks are equivalent to learning the Pauli coefficient vector
$\boldsymbol\mu_H=(\mu_H(\bfp))_{\bfp \in \{0, 1, 2, 3\}^n }$
in $\ell_\infty$ and $\ell_2$ norms, respectively. 
Throughout this subsection, we refer to $\eps$ as the \emph{normalized error}, since the estimation error in both tasks is rescaled by $\|H\|_\infty$\footnote{
Physically, scaling up the energies of the system speeds up its dynamics by the same factor, so the evolution time needed to acquire a given amount of information shrinks accordingly. We thus choose the normalized error to factor out this trivial dependence and focus only on the accuracy relative to the intrinsic energy scale of the system. As a sanity check, the total evolution time in both our bounds and those in Section~\ref{sec:comparison_with_prior_protocols} scales as $\|H\|_{\infty}^{-1}$.
}.
The NFN $d^{-1/2}\|\cdot\|_F$ is a natural global metric for Hamiltonian learning, as it controls the average prediction error of observables under the learned dynamics~\cite{ma2024learning}. Pauli coefficient learning instead provides uniform entrywise accuracy, enabling the identification of individual interaction terms.

Hamiltonian learning from dynamics has been studied extensively in recent years~\cite{
Bairey2019LearnLocalHamiltonian,
innocenti2020supervised,
caro2022learning,
Yu2023robustefficient,
huang_2023,
castaneda2023hamiltonian,
Dutkiewicz2024advantageofquantum,
gu2024practical,
stilck2024efficient,
ma2024learning,
bakshi2024structure1,
haah2024learning,
zhao2024learning,
HuAnsatzFree2025,
Bluhm_2026,
depradenne2026learninghamiltonianslongtimes,abbas2025nearly,sinha2025improved,arunachalam2026testing}.
Many existing efficient protocols, however, rely on structural assumptions such as geometric locality, sparsity, or low-intersection structure.
Such assumptions can enable scalable learning and, in several settings, Heisenberg-limited scaling in the target precision~\cite{huang_2023,ma2024learning,HuAnsatzFree2025,bakshi2024structure1,abbas2025nearly,sinha2025improved}.
Here we instead consider completely general Hamiltonians, which may be dense and nonlocal, without assuming locality, sparsity, or any specific ansatz. Our results are summarized in the following theorem, and we defer the proof to Appendix~\ref{app:Hamiltonian_learning}.

\begin{theorem}[Learning general Hamiltonians]
\label{thm:hamiltonian_learning_informal}
Let $H$ be any unknown traceless $n$-qubit Hamiltonian accessible through its real-time evolution $e^{-\I Ht}$ for any $t \in \mathbb{R}_{\geq 0}$. For any $0 < \eps < 1$, 

\begin{enumerate}
\item
There exists a protocol that estimates all Pauli
coefficients of $H$ up to normalized error $\eps$ with success probability
at least $2/3$, using total evolution time
\begin{align}\label{eq:tot_time_coeff_main}
T_{\mathrm{tot}}=
\begin{cases}
\widetilde{\bigO}\left(\dfrac{d}{\|H\|_\infty\eps}\right),
&
0< \eps \leq \dfrac{1}{d^{3/2}},
\\[2ex]
\widetilde{\bigO}\left(\dfrac{d^{7/4}}{\|H\|_\infty\sqrt{\eps}}\right),
&
\dfrac{1}{d^{3/2}}<\eps \leq \dfrac{1}{d^{3/4}},
\\[2ex]
\widetilde{\bigO}\left(\dfrac{d }{\|H\|_\infty\eps^{3/2}}\right),
&
\dfrac{1}{d^{3/4}}<\eps  < 1.
\end{cases}
\end{align}
Conversely, for $\eps<\frac{1}{4d}$, any protocol achieving the same guarantee for all traceless $n$-qubit Hamiltonians requires total evolution time $\widetilde{\Omega}(d \|H\|_{\infty}^{-1} \eps^{-1})$ in the worst case.

\item
There exists a protocol that learns $H$  in NFN up to normalized error $\eps$ with success probability at least $2/3$, using total evolution time 
\begin{align}\label{eq:tot_time_NFN_main}
T_{\mathrm{tot}}=
\begin{cases}
\widetilde{\mathcal{O}}\left(\dfrac{d^2}{\|H\|_{\infty}\eps}\right),
&
0 < \eps \leq \dfrac{1}{d^{1/2}},
\\[2ex]
\widetilde{\mathcal{O}}\left(\dfrac{d^{9/4}}{\|H\|_{\infty} \sqrt{\eps}}\right),
&
\dfrac{1}{d^{1/2}}< \eps < 1.
\end{cases}
\end{align}
Conversely, any protocol achieving the same guarantee requires total
evolution time $\widetilde{\Omega}(d^2 \|H\|_\infty^{-1}\eps^{-1})$. 
\end{enumerate}

\end{theorem}

Together, the upper and lower bounds in
\tref{thm:hamiltonian_learning_informal} provide the first nearly time-optimal learning guarantees for completely general Hamiltonians in the sufficiently high-precision regimes. The bounds match up to polylogarithmic factors when $\eps \leq \widetilde{\mathcal{O}}(d^{-3/2})$ for Pauli coefficient learning, and when $\eps \leq \widetilde{\mathcal{O}}(d^{-1/2})$ for Hamiltonian learning in NFN.
Notably, to our knowledge, the lower bound $\widetilde{\Omega}(d \|H\|_{\infty}^{-1}\eps^{-1})$ for Pauli coefficient learning is the first to go beyond the generic $\Omega(\|H\|_{\infty}^{-1}\eps^{-1})$ bound implied by the Heisenberg limit. Specifically, it gives a negative answer to a question raised in \cite{huang_2023}: the parameters of an unstructured $n$-qubit Hamiltonian with all-to-all interactions cannot be learned in total evolution time $\widetilde{\mathcal{O}}(\eps^{-1})$ independent of the dimension.

The main idea behind our Pauli coefficient learning protocol is to combine short-time polynomial interpolation with CSEU.
For each nontrivial Pauli label $\bfp$, the interpolation techniques of \cite{caro2022learning,stilck2024efficient,gu2024practical} provide an input state $\rho_{\bfp}$, an observable $O_{\bfp}$, and a collection of short evolution times $t_1,\dots,t_L$ such that the expectation values
$\Tr[O_{\bfp}U_{t_j}\rho_{\bfp}U_{t_j}^\dagger]$
determine the coefficient $\mu_H(\bfp)$ through classical post-processing.
For each fixed $t_j$, our CSEU protocol can simultaneously estimate these quantities for all Pauli labels. 
Combining the resulting estimates with the interpolation procedure yields the first set of upper bounds in \tref{thm:hamiltonian_learning_informal}.
The guarantee in NFN follows directly from the Pauli coefficient learning result.
Indeed, estimating every Pauli coefficient to normalized error $\eps /d$ and reconstructing
$\widehat H=\sum_{\bfp}\widehat\mu_H(\bfp)\sigma_{\bfp}$
ensures that $\frac1{\sqrt d}\|\widehat H-H\|_F\leq \eps \|H\|_{\infty}$.
Applying the first learning protocol with normalized error $\eps /d$ therefore yields the NFN Hamiltonian learning evolution time upper bounds stated in \tref{thm:hamiltonian_learning_informal}.

\subsubsection{Comparison with prior Hamiltonian learning protocols}
\label{sec:comparison_with_prior_protocols}
Here we compare our results with several of the most efficient existing protocols for learning Pauli coefficients of Hamiltonians and learning Hamiltonians in NFN. Notably, prior works bound absolute errors, so we will restate all results in terms of the normalized error $\eps$ for comparison on an equal footing.

\paragraph{Learning Pauli coefficients.} There have been several protocols that learn the Pauli coefficients of Hamiltonians with Heisenberg scaling. 
The protocol of \cite{HuAnsatzFree2025} uses total evolution time $\widetilde{\mathcal{O}}(m^2 \|H\|_{\infty}^{-1}\eps^{-1})$ for learning all Pauli coefficients of a sparse Hamiltonian with at most $m$ nonzero Pauli terms. More recently, two independent works~\cite{abbas2025nearly,sinha2025improved} improved this scaling to $\widetilde{\mathcal{O}}(m \|H\|_{\infty}^{-1}\eps^{-1})$.
These protocols are highly efficient when $m=\poly(n)$. However, for a completely general Hamiltonian, $m$ can be as large as $d^2$, in which case these improved results yield a total evolution time of $\widetilde{\mathcal{O}}(d^2 \|H\|_{\infty}^{-1}\eps^{-1})$.
A comparable upper bound can also be obtained from the protocol of \cite{OdakeHigher-order2024}, which estimates a single Pauli coefficient to normalized error $\eps$ using total evolution time $\widetilde{\mathcal{O}}(\|H\|_{\infty}^{-1}\eps^{-1})$. Applying this procedure independently to each of the $d^2$ Pauli coefficients and summing the evolution time over
all coefficients again yields a total evolution time of $\widetilde{\mathcal{O}}(d^2 \|H\|_{\infty}^{-1}\eps^{-1})$.

In comparison, our protocol applies to completely general Hamiltonians and achieves the total evolution time presented in \eref{eq:tot_time_coeff_main}. In particular, in the sufficiently high-precision regime $\eps \leq \widetilde{\mathcal{O}}(d^{-3/2})$, our protocol attains the Heisenberg scaling $\widetilde{\mathcal{O}}(d \|H\|_{\infty}^{-1}\eps^{-1})$, improving the dimension dependence of the above time complexity from quadratic to linear. More generally, throughout the regime $\eps< 1$, our bounds outperform the previous $\widetilde{\mathcal{O}}(d^2 \|H\|_{\infty}^{-1}\eps^{-1})$ bound.

Several other structure-free protocols achieve total evolution time with no explicit dependence on $d$, at the expense of a worse dependence on the target precision $\eps$. In particular, the protocol of \cite{arunachalam2026testing} estimates all Pauli coefficients to additive error $\eps$ using a total evolution time of $\widetilde{\mathcal{O}}(\|H\|_{\infty}^{-1}\eps^{-3})$, whereas the total evolution times of the protocols in \cite{caro2022learning,zhao2024learning} scale quartically with $\eps^{-1}$.
In comparison, our protocol has a more favorable dependence on the target precision, at the cost of an explicit dependence on $d$. For example, comparing with $\widetilde{\mathcal{O}}(\|H\|_{\infty}^{-1}\eps^{-3})$, our bound is asymptotically smaller, up to polylogarithmic factors, whenever $\eps \leq \mathcal{O}(  d^{-2/3})$.
In addition, our dependence on the Hamiltonian norm is also milder. This leads to a complementary regime of advantage: when the target error is sufficiently small, our protocol yields a more favorable total evolution time than those of \cite{caro2022learning,zhao2024learning,arunachalam2026testing}.

Our results also provide a partial affirmative answer to an open
question posed in \cite{caro2022learning}: whether the quartic
dependence $\eps^{-4}$ in structure-free Hamiltonian learning can be
improved toward the Heisenberg scaling $\eps^{-1}$.
Specifically, we show that the Heisenberg scaling can indeed be attained
for completely general Hamiltonians in the sufficiently high-precision
regime
$\eps \leq \widetilde{\mathcal{O}}(d^{-3/2})$.
Although this improvement comes at the cost of introducing an explicit
dependence on the system dimension $d=2^n$, the lower bound in
\tref{thm:hamiltonian_learning_informal} shows that a linear dependence
on $d$ is unavoidable for Pauli coefficient learning in the
high-precision regime.

\paragraph{Hamiltonian learning in NFN.}
Compared with Pauli coefficient learning, Hamiltonian learning in NFN has received less attention in existing literature. 
To learn a $k$-body Hamiltonian with $m$ nonzero Pauli terms to normalized error $\eps$ in NFN, the SPAM-robust protocol of \cite{ma2024learning} uses total evolution time
$\widetilde{\mathcal{O}}\left(9^k m \|H\|_{\infty}^{-1}\eps^{-1}\right)$.
This protocol is efficient when both $k$ and $m$ are small.
However, without structural assumptions, one may have up to $k=n$ and $m=d^2$, in which case the total evolution time becomes
\[
T
=
\widetilde{\mathcal{O}}\left(9^n d^2 \|H\|_{\infty}^{-1}\eps^{-1}\right)
=
\widetilde{\mathcal{O}}\left(d^{2+2\log_2 3} \|H\|_{\infty}^{-1}\eps^{-1}\right)
\approx
\widetilde{\mathcal{O}}\left(d^{5.17} \|H\|_{\infty}^{-1}\eps^{-1}\right),
\]
which is unfavorable for learning completely general Hamiltonians.

As discussed in the previous context, another direct approach to achieve NFN learning with normalized error $\eps$ is to run any Pauli coefficient learning protocol with target normalized error $\eps/d$.
Under this direct and handy conversion, the protocols proposed in \cite{abbas2025nearly,sinha2025improved,OdakeHigher-order2024}
yield a total evolution time of 
$\widetilde{\mathcal{O}}(d^3 \|H\|_{\infty}^{-1}\eps^{-1})$ in the worst-case.
Building on a different CSEU approach, \cite{Li_2025} uses total evolution time
$\widetilde{\mathcal{O}}(d^{5/2} \|H\|_{\infty}^{-1}\eps^{-2})$ for Hamiltonian learning in NFN.
Another recent protocol proposed in \cite{castaneda2023hamiltonian} uses near-optimal total evolution time $\widetilde{\mathcal{O}}(d^2 \|H\|_{\infty}^{-1}\eps^{-1})$, but requires additional access to the backward time evolution $e^{\I Ht}$, corresponding to a strictly stronger access model that is generally not physically realizable.

By contrast, our protocol requires only forward real-time evolution
$e^{-\I Ht}$ and uses a total evolution time presented in \eref{eq:tot_time_NFN_main}. 
In particular, in the high-precision regime $\eps \leq \widetilde{\mathcal{O}}(d^{-1/2})$, our protocol is near-optimal for learning completely general Hamiltonians in NFN.
Outside this regime, the second branch in \eref{eq:tot_time_NFN_main} still improves over the forward-only approaches that we discussed above.

\subsection{Learning Pauli transfer matrices}
\label{subsec:ptm_learning}

A natural representation of a quantum process is its \emph{Pauli transfer matrix} (PTM) \cite{caro2022learning}. 
For an $n$-qubit unitary $U$, the PTM is the real $d^2\times d^2$ matrix whose entries are
\[
    R_U(\bfp,\bfp' )
    :=
    \frac1d
    \Tr \left(\sigma_{\bfp}U\sigma_{\bfp' }U^\dagger\right),
    \qquad
    \bfp,\bfp' \in\{0,1,2,3\}^n ,
\]
where $d=2^n$ and $\sigma_{\bfp}$ are $n$-qubit Pauli operators. 
The PTM describes the action of the (unitary) channel in the Pauli basis. It serves as a standard representation for tasks like quantum process tomography \cite{ChenRobust2021}, noise characterization \cite{Chen2023Learnability}, randomized benchmarking \cite{Helsen2019}, and Hamiltonian learning \cite{caro2022learning}.

One concrete application of the PTM is to characterize operator spreading (see, \eg, \cite{Nahum_2018, Nadimpalli_2024}) under unknown quantum dynamics.
Indeed, each row of $R_U$ gives the Pauli expansion of a Heisenberg-evolved Pauli operator under the dynamics of $U$:
\[
U^\dagger\sigma_{\bfp'}U
=
\sum_{\bfp\in \{0,1,2,3\}^n}
R_U(\bfp',\bfp)\sigma_{\bfp}.
\]
Suppose, for instance, that $\sigma_{\bfp'}$ is initially a single-site Pauli operator.
After the evolution, $U^\dagger\sigma_{\bfp'}U$ may contain Pauli strings supported on many qubits.
The cumulative Pauli-weight profile \cite{Nadimpalli_2024}
\[
\mathbf{Wt}_{\leq k}(\bfp')
:=
\sum_{\bfp:\,|\bfp|\leq k}
|R_U(\bfp',\bfp)|^2
\]
measures how much of the evolved operator remains supported on at most $k$-body Pauli components.
For shallow or local dynamics, this quantity would be close to one for small $k$, reflecting limited operator growth. 
For scrambling dynamics or strong coherent error propagation, weight is transferred to higher-body components, causing $\mathbf{Wt}_{\leq k}(\bfp')$ to decrease for small $k$.
Thus, PTM learning provides a direct way to quantify how an initially local observable, error, or piece of information spreads into increasingly nonlocal degrees of freedom.

Our CSEU protocol gives a direct procedure for learning the PTM of an unknown unitary channel. 
The key observation is that each nontrivial PTM entry can be written as a CSEU prediction problem with a highly mixed input state. 
For any non-identity Pauli $\sigma_{\bfp' }$, define
\[
    \rho_{\bfp' }:=\frac{\openone+\sigma_{\bfp' }}{d},
    \qquad
    O_{\bfp}:=\sigma_{\bfp}.
\]
Then, whenever $\bfp\neq  \mathbf 0$,
\[
    \Tr \left(O_{\bfp}U\rho_{\bfp' }U^\dagger\right)
    =
    \frac1d \, \Tr \left(\sigma_{\bfp}U\sigma_{\bfp' }U^\dagger\right)
    =
    R_U(\bfp,\bfp' ).
\]
Although $O_{\bfp}$ is a full-rank Pauli observable, the corresponding input state $\rho_{\bfp' }$ is highly mixed. 
This low purity of $\rho_{\bfp' }$ compensates for the large effective size of the Pauli observable: as explained in Section~\ref{sec:technical_overview_UB}, the relevant query complexity depends on the product of the observable size and the input-state purity. 
Moreover, since the same shadow can be reused for many prediction requests (see Remark~\ref{remark:Mproperties}), all $d^4$ PTM entries of the unknown unitary channel can be estimated simultaneously with only a logarithmic overhead.

\begin{theorem}[Near-optimal PTM learning for unitary channels]
\label{thm:unitary_ptm_learning}
Suppose $d=2^n$ and  $U\in\mathsf U(d)$ is an unknown unitary. 
There exists a parallel and non-adaptive protocol that estimates all $d^4$ PTM entries of $\mathcal U$
to additive error $0<\eps<1$ with success probability at least $2/3$, using 
\[
K_{\mathrm{PTM}}
=
\begin{cases}
\widetilde{\mathcal{O}}\left(\dfrac{d}{\eps}\right),
&
0<\eps\leq \dfrac{1}{d^{3/2}},
\\[2ex]
\widetilde{\mathcal{O}}\left(\dfrac{d^{7/4}}{\sqrt{\eps}}\right),
&
\dfrac{1}{d^{3/2}}<\eps\leq \dfrac{1}{d^{3/4}},
\\[2ex]
\widetilde{\mathcal{O}}\left(\dfrac{d}{\eps^{3/2}}\right),
&
\dfrac{1}{d^{3/4}}<\eps<1.
\end{cases}
\]
black-box queries to $U$.
Conversely, there exists a universal constant $c > 0$ such that for $0<\eps<c/d$, any protocol that achieves the same performance guarantee must use $\widetilde{\Omega}(d \eps^{-1})$ queries to $U$. 
\end{theorem}

In the high-precision regime $0<\eps\leq d^{-3/2}$, the upper and lower bounds match up to polylogarithmic factors, establishing a near-optimal query complexity of
$\widetilde{\Theta}(d\eps^{-1})$ for PTM learning of unitary channels.
In the low-precision regimes, our protocol achieves the piecewise query complexities stated in
Theorem~\ref{thm:unitary_ptm_learning}, while the optimal scaling of the query complexity in these regimes remains open.

Our lower bound builds upon an information-theoretic packing argument, as does the proof technique of \cite[Theorem 1.2]{HKOT23}. 
We construct a family of size $\exp(\Theta(d^2))$ consisting of small rotations
$U_x=\exp(-\mathrm i\lambda R_x)$
around well-separated reflections $R_x$. 
By taking $\lambda=\Theta(d\eps)$, the corresponding PTM vectors are separated by $\Omega(\eps)$ in $\ell_\infty$ distance, so any protocol that learns all PTM entries can identify the hidden index $x$. 
We then show that this identification task requires $\widetilde{\Omega}(d\eps^{-1})$ queries, which implies the lower bound in Theorem~\ref{thm:unitary_ptm_learning}. The detailed proof is deferred to Appendix~\ref{appendix:ptm_learning_lower_bound}.

It is rather useful to compare our result with the PTM-learning protocol of
Caro~\cite{caro2022learning}, which learns the full PTM of an arbitrary
$n$-qubit quantum channel using $\mathcal{O}(n\eps^{-4})$ oracle queries.
Caro's protocol applies to general quantum channels and has only logarithmic dependence on the system dimension, making it preferable at fixed or moderate precision.
By contrast, our CSEU-based protocol is tailored to unitary channels and has a substantially better dependence on precision. In particular, it achieves the Heisenberg-limited scaling $\widetilde{\mathcal{O}}(d\eps^{-1})$ in the high-precision regime $\eps \leq \mathcal{O}(d^{-3/2})$.

\subsection{Inverse-free amplitude estimation}
\label{subsec:inverse_free_amplitude_estimation}

Amplitude estimation is a fundamental estimation task with broad applications in quantum information processing.
In its standard model, one is given a known ``good'' subspace $\caH_{\rm good}$ and query access to an unknown state-preparation unitary $U$. 
Let $\Pi$ be the projector onto $\caH_{\rm good}$, which defines a two-outcome measurement $\{\Pi,\openone -\Pi\}$. 
The goal is to estimate the probability that the state $U\ket{0^n}$ passes this measurement, namely
\[
a(U)
:=
\Tr \left[\Pi \, U\ket{0^n}\!\bra{0^n}U^\dagger\right].
\]
Equivalently, suppose the state prepared by $U$ admits the decomposition
\[
U\ket{0^n}
=
\sqrt{a(U)}\,\ket{\phi_{\rm good}}
+
\sqrt{1-a(U)}\,\ket{\phi_{\rm bad}},
\]
where $\ket{\phi_{\rm good}}$ is supported on $\caH_{\rm good}$ and $\ket{\phi_{\rm bad}}$ is supported on the orthogonal complement of $\caH_{\rm good}$. Then the task is to estimate the amplitude $a(U)$ within additive error $0<\eps<1$. 

To solve the amplitude estimation task, the direct sampling protocol achieves query complexity $\bigO(\eps^{-2})$, while standard quantum amplitude estimation attains the Heisenberg scaling $\bigO(\eps^{-1})$ by using coherent amplitude amplification. 
However, the latter requires access to both $U$ and $U^\dagger$, an assumption that can be difficult to satisfy in many scenarios. 
Here we consider amplitude estimation in the inverse-free setting, where the learner has query access only to the forward unitary $U$.

Our CSEU protocol gives a simple inverse-free amplitude-estimation protocol, 
since estimating the amplitude $a(U)=\Tr [\Pi \, U\ket{0^n}\!\bra{0^n}U^\dagger]$ reduces directly to a CSEU prediction problem with input state $\rho=\ket{0^n}\! \bra{0^n}$ and observable $O=\Pi$. 
Applying \tref{thm:cseu_upper_bound} gives the following result.

\begin{corollary}[Inverse-free amplitude estimation]
\label{cor:inverse_free_amplitude_estimation}
Let $U\in \U(d)$ be an unknown unitary accessible only through forward oracle queries, and let $\Pi$ be a known projector. 
There exists an inverse-free protocol that estimates the amplitude
$a(U)=\Tr \left[\Pi\,U\ket{0^n}\!\bra{0^n}U^\dagger\right]$
within additive error $\eps$ with probability at least $2/3$, using at most $\mathcal{O}(d^{3/2} \eps^{-1})$ queries to $U$. 
\end{corollary}
Combining this CSEU-based protocol with the direct sampling approach gives the query upper bound
$$
K_{\text{AE}}=
\bigo{\min\left\{ \dfrac{d^{3/2}}{\eps} , \frac{1}{\eps^2}\right\}
}.
$$
This recovers the best-known query-complexity upper bound for inverse-free amplitude estimation recently established by Chen~\cite{chen2025inverse}, which was obtained via quantum state estimation with queries to the state-preparation unitary. Specifically, Chen's approach effectively reduces the amplitude estimation to the inverse-free counterpart of the main problem considered in~\cite{doi:10.1137/1.9781611977554.ch47}.

\subsection{Unitary channel tomography with parallel queries}

Process tomography of a unitary channel is a fundamental primitive for learning, verification, and control of quantum dynamics~\cite{scott2008optimizing,PhysRevA.90.012110,Gutoski_2014,Yang_2020,Nielsen_2021,HKOT23}.
Given black-box access to an unknown $U\in\mathsf U(d)$, the goal is to output a classical description of a unitary $\widehat U$ such that
$\Pr\!\big[\|\widehat{\mathcal U}-\mathcal U\|_\diamond\leq  \eps
\big] \geq 2/3$, where $\|\cdot\|_\diamond$ denotes the diamond norm [see Definition \ref{def:diamond_norm}]. The unitary channel tomography protocols of
\cite{HKOT23,grewal2026efficientlearningstructuredquantum}
achieve the optimal query complexity $\Theta(d^2\eps^{-1})$, but rely on sequential oracle access and adaptive refinement to attain this $\eps$-dependence through the ``bootstrapping'' technique. By contrast, parallel protocols place all oracle queries in a single query layer, thereby avoiding both adaptive feedback and long coherent sequences of oracle uses.
This can effectively reduce oracle-query depth and the serial experimental runtime, but at the cost of a more restrictive query architecture. Under this architectural restriction, the best-known query-complexity for unitary channel tomography is $\mathcal{O}\left( \min\left\{ d^2 \eps^{-2}, d^{5/2}\eps^{-1}\right\}\right)$ \cite{Yang_2020, HKOT23}.
In this section, we will show that this scaling follows directly from our parallel CSEU protocol.

\begin{figure}[tbp!]
\centering
\includegraphics[width=0.95\linewidth]
{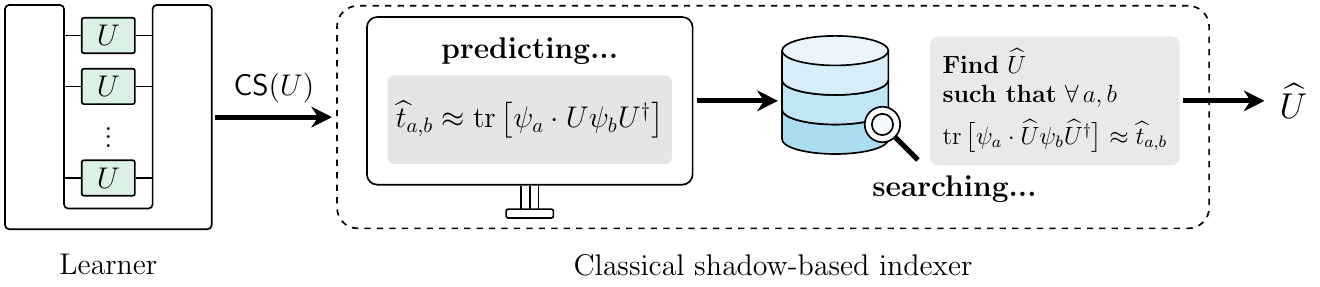}
\caption{Conversion of a CSEU protocol into a unitary channel tomography protocol. The indexing database consists of pairs of pure states $(\psi_a,\psi_b)$ chosen from a constant-precision covering net. The CSEU protocol estimates the transition probabilities $\Tr[\psi_a U\psi_b U^\dagger]$ to a pre-determined precision, and the reconstruction step outputs a unitary $\widehat U$ whose transition probabilities are (approximately) consistent with these estimates.}
\label{fig:cseu_to_tomography}
\end{figure}

The key to this result is a canonical reduction from full unitary channel tomography to CSEU that we introduce in this work. The idea is illustrated in \autoref{fig:cseu_to_tomography}.
We first choose a constant-precision covering net $\covering$ of pure states, that is, a finite set such that every pure state lies within a constant distance of at least one element of $\covering$.
We then invoke our CSEU protocol to estimate the transition probabilities $\Tr[\psi_a U\psi_bU^\dagger]$ for all $\psi_a,\psi_b\in\covering$.
These transition probabilities form a classical fingerprint of the unknown unitary $U$: any unitary whose predicted transition probabilities are consistent with all the estimates must be close to $U$ in diamond distance.
Since the reduction uses only independent repetitions of the CSEU protocol followed by classical post-processing, it preserves the parallel and non-adaptive query architecture.
The reduction can be summarized as follows.

\begin{proposition}[CSEU-to-tomography reduction, informal]
\label{prop:TomoEfficiency}
Suppose $\mathcal A$ is a CSEU protocol that solves
Problem~\ref{prob:CSEU} for $\B=1$ and accuracy $\eps$, using
$K(d,\eps)$ oracle queries. Then, by running $\mathcal A$ as a subroutine, one can construct a unitary channel tomography protocol using $\bigo{d}\,K(d,\eps)$ oracle queries.
Moreover, the conversion preserves the query architecture of $\mathcal A$.
\end{proposition}

The precise analysis of the conversion procedure and the proof of \pref{prop:TomoEfficiency} are deferred to Appendix~\ref{appendix:proof_LB}.
Applying Proposition \ref{prop:TomoEfficiency} to the performance of our CSEU protocol gives the following query complexity for unitary channel tomography: 
$$
K_{\rm tomo}^{\rm unitary}= \mathcal{O}(d)\cdot \mathcal{O}\left( \min\left\{\frac{d}{\eps^2}, \frac{d^{3/2}}{\eps}\right\} \right) = \mathcal{O}\left( \min\left\{\frac{d^2}{\eps^2}, \frac{d^{5/2}}{\eps} \right\} \right).
$$
Thus, through a conceptually different CSEU-based route, we are able to recover the best-known query-complexity upper bound for unitary channel tomography under parallel queries \cite{HKOT23,Yang_2020}.

The same result also provides an alternative derivation of the best-known upper bound for quantum channel tomography in the boundary regime. The task is defined as follows: consider an unknown quantum channel
$\mathcal E:\mathcal L(\mathbb C^{d_1})\to\mathcal L(\mathbb C^{d_2})$
with Kraus rank at most $r$, and its dilation rate is defined as
$\tau:=r d_2/d_1$.
The boundary regime corresponds to $\tau=1$. In this regime, Chen, Girardi, Oufkir, Yu, and Zhang~\cite{chen_Girardi2026} established the query complexity upper bound
$\mathcal{O} \big(
\min\{
r d_1d_2\eps^{-2},
r d_1^{3/2}d_2\eps^{-1}
\} \big)$
under diamond distance. 

In the boundary regime, the Stinespring dilation of $\mathcal E$ gives a $d_1$-dimensional unitary channel.
A recent ``dilation does not help'' lemma \cite{chen_Girardi2026,girardi2025random} shows that any parallel protocol using queries to such a dilation can be converted into a parallel protocol using the same number of queries directly to $\mathcal E$.
Since our CSEU-based unitary channel tomography protocol is fully parallel, it is directly compatible with the query architecture of this reduction.
Applying it to the $d_1$-dimensional dilation therefore gives
$$
K_{\rm tomo}^{(\tau=1)}
=
\mathcal{O} \left(
\min\left\{
\frac{d_1^2}{\eps^2},
\frac{d_1^{5/2}}{\eps}
\right\}\right)
=
\mathcal{O}\left(
\min\left\{
\frac{r d_1d_2}{\eps^2},
\frac{r d_1^{3/2}d_2}{\eps}
\right\}
\right)
$$
queries to $\mathcal E$, where we have used the boundary condition $r d_2=d_1$.
Thus, our CSEU-based approach also recovers the best-known
query-complexity upper bound for boundary-regime quantum channel tomography.

\subsection{Learning shallow quantum circuits}

Learning shallow, or constant-depth, quantum circuits allows us to conduct process tomography of near-term quantum devices \cite{preskill2018nisq}. Specifically, we consider learning a family of shallow circuits with strictly bounded fan-in gates: the $\qnc^0$ circuits, which are the quantum analogs of $\nc^0$ \cite{moore1999quantumcircuitsfanoutparity}. Huang, Liu, Broughton, Kim, Anshu, Landau, and McClean \cite{Huang_2024} initiated the study of learning $\mathsf{QNC}^0$ circuits.
Subsequent work has explored learning other circuit families containing long-range gate sets \cite{Nadimpalli_2024, qac0_barely_superlinear_2025, pmlr-v291-vasconcelos25a}.

Previous work on learning $\qnc^0$ circuits primarily considered incoherent access to shallow-circuit unitaries, using classical shadows of the output states produced by these circuits.
With CSEU, we can design a coherent protocol for learning $\qnc^0$ unitaries in diamond distance that reaches the Heisenberg scaling in the sufficiently high-precision regime. Specifically, we first learn the Heisenberg-evolved local operators under the action of the shallow-circuit unitary $U$, then invoke the sewing lemma of Huang \etal~\cite{Huang_2024} to efficiently recover the unitary via post-processing. The resulting performance guarantee is stated in the following corollary, and we defer its proof to Appendix~\ref{sec:proof_for_learning_shallow_circuit}.

\begin{corollary}[Learning $\qnc^0$ circuits]
\label{corollary:learn_shallow_circuit}
Given query access to an unknown $n$-qubit unitary channel $\mathcal{U}$ generated by a $\qnc^0$ quantum circuit with arbitrary ancillas, there exists a protocol that outputs a classical description of an $n$-qubit channel $\widehat{\mathcal S}_U$ satisfying
$\|\widehat{\mathcal{S}}_U - \mathcal{U}\|_{\diamond} \leq \eps$ with probability at least \(2/3\),
using at most
$$
K_{\qnc^0}
=
\begin{cases}
\bigo{\dfrac{dn\log n}{\eps}},
&
0<\eps \leq \dfrac{n}{d^{3/2}},
\\[2ex]
\bigo{\dfrac{d^{7/4} \sqrt{n} \log n}{\sqrt{\eps}}
},
&
\dfrac{n}{d^{3/2}} < \eps \leq \dfrac{n}{d^{3/4}},
\\[2.1ex]
\bigo{
\dfrac{d n^{3/2} \log n}{\eps^{3/2}}},
&
\dfrac{n}{d^{3/4}} <\eps<1,
\end{cases}
$$
black-box queries to $U$.
\end{corollary}

In the high-precision regime $\eps \leq nd^{-3/2}$, our protocol attains Heisenberg-limited scaling.
For comparison, full unitary channel tomography has optimal query complexity $\Theta(d^2\eps^{-1})$ for learning an arbitrary $d$-dimensional unitary channel in diamond distance \cite{HKOT23}. In this regime, our protocol improves upon the full unitary channel tomography by nearly a factor of $d$. For larger $\eps$, the piecewise query complexity in
Corollary~\ref{corollary:learn_shallow_circuit} remains lower than the $\bigo{d^2\eps^{-1}}$ complexity of full unitary channel tomography.
Moreover, when $\eps = o(nd^{-7/6})$, our bound improves upon the $\mathcal{O}(n^2\log n\cdot\eps^{-2})$-query algorithm due to Huang et al.~\cite[Theorem 5]{Huang_2024} for learning generic $\mathsf{QNC}^0$ circuits.

Meanwhile, our result can be compared with that of \cite[Theorem 9.11]{grewal2026efficientlearningstructuredquantum}, although their protocol relies on a stronger access model that allows queries to both $U$ and $U^\dagger$.
They define the Clifford nullity $t\in\{0,1,\dots,2n\}$ as a metric for the non-Cliffordness of the quantum circuit, and their protocol requires
$\bigo{2^t n^2\log n\cdot\eps^{-1}}$ queries. According to Corollary \ref{corollary:learn_shallow_circuit}, our protocol is more query-efficient whenever $d/n = o(2^t)$, despite using only forward access to $\mathcal{U}$.

\section{Technical overview} 

In this section, we provide an overview of the main technical ideas behind our query-complexity results for CSEU in Section~\ref{sec:main_results}, consisting of an upper bound showing that Heisenberg scaling can be achieved with parallel queries and a matching lower bound.

\subsection{Establishing the upper bound}\label{sec:technical_overview_UB}
We prove the upper bound through a more general version of CSEU. Unlike the standard CSEU problem (Problem~\ref{prob:CSEU}), where the prediction guarantee must hold for all state-observable pairs $(\rho,O)$, this general formulation only requires the protocol to work for a prescribed subclass of prediction requests. This allows us to track refined structural parameters of the requests, such as the purity of the input state and the rank of the state or observable. The standard CSEU upper bound in \tref{thm:cseu_upper_bound_informal} will then follow by taking the worst-case value of these parameters. We now formalize this general CSEU problem as follows.  

\begin{problem}[General CSEU problem]
\label{prob:restricted_CSEU}
Let $d\in\mathbb N$, $1\leq \B\leq d$, $0<\eps<1$, and 
$\mathcal R$ be a prescribed class of prediction requests $(\rho,O)$, where
$\rho\in\mathbb C^{d\times d}$ is a quantum state and $O\in\mathbb C^{d\times d}$ is an observable. 
Let $U\in\U(d)$ be an unknown unitary accessible only through black-box oracle queries.
The goal of general CSEU on the request class $\mathcal R$ is to output classical data
$\mathsf{CS}(U)\in\{0,1\}^*$, together with a deterministic prediction function
\[
f:\{0,1\}^*\times \C^{d\times d}\times \C^{d\times d}\to \mathbb R,
\]
such that for any $(\rho,O)\in\mathcal R$,
\begin{align}
\label{eq:restricted_desired_accuracy}
\Pr\left[
\bigl|
f\bigl(\mathsf{CS}(U),\rho,O\bigr)
-
\Tr[O U\rho U^\dagger]
\bigr|
\leq \eps
\right]
\geq \frac23 .
\end{align}
Here, the probability is taken over the randomness in the generation of $\mathsf{CS}(U)$.
\end{problem}

The standard CSEU problem in Problem~\ref{prob:CSEU} is recovered by taking $\mathcal R$ to be the class of all quantum states $\rho$ and all observables $O\in\obs(\B)$. The general formulation is useful because many applications only require predictions for a structured family of requests. For instance, the input states may have a small purity upper bound. The following theorem records the refined upper bound that we establish in this section.

\begin{theorem}[Strengthened version of \tref{thm:cseu_upper_bound_informal}]
\label{thm:cseu_upper_bound}
Let $1\leq \B\leq d$, $0<\eps<1$, and $d^{-1} \leq \P\leq 1$. 
Define the request class
\[
\mathcal R_{\P,\B}
:=
\left\{
(\rho,O):
\rho \text{ is a quantum state},\
\Tr[\rho^2]\leq \P,\
O\in\obs(\B),\ \B\P\geq 1
\right\}.
\]
Then the following two statements hold.

\begin{enumerate}
\item There exists a protocol that solves Problem~\ref{prob:restricted_CSEU} on $\mathcal R_{\P,\B}$ using
\[
K =
\bigO\left(
\min_{1\leq s\leq d} \left\{
\frac{1}{\eps^2}
    \left(
    s  +
        \frac{d^2}{s^3} \B\P + d \P
    \right)
+
\frac{d^2}{s\eps}\sqrt{\B\P}
\right\}
\right)
\]
parallel queries to the unknown unitary $U$. {For this choice of parameters, the resulting classical shadow data can be stored using $\poly(d)$ complex numbers, and for any given request $(\rho,O)$, the prediction function $f$ can be evaluated in $\poly(d)$ classical time, up to standard finite-precision overheads.}

\item Assume that $\eps\leq d^{-1/2} \max\{\sqrt{\P}, \sqrt{1 - \Tr[\rho^2]}\}$. 
In this regime, there exists a protocol that solves Problem~\ref{prob:restricted_CSEU} on
$\mathcal R_{\P,\B}$ using
\[
K=\mathcal{O} \left( \frac{d^{3/2}\sqrt{\P}}{\eps}
        +
        \frac{d^{7/4}}{\sqrt{\eps}}
        \left(1-\Tr(\rho^2)\right)^{1/4} \right)
\]
parallel queries to $U$. Specifically, the query complexity saturates its upper bound when taking pure states $\rho$. The resulting classical shadow data can be stored using $\poly(d)$ complex numbers, and for any given request $(\rho,O)$, the prediction function $f$ can be evaluated in $\poly(d)$ classical time, up to standard finite-precision overheads.
\end{enumerate}
\end{theorem}

Let us briefly interpret the theorem before discussing the protocol. The first item is a general-purpose upper bound for CSEU over the request class $\mathcal R_{\P,\B}$ with only a linear dependence on the system dimension $d$. The second item is where Heisenberg-limited scaling kicks in: once the target precision $\eps$ is sufficiently small, this bound directly yields a query complexity of $\mathcal{O}(d^{3/2}\eps^{-1})$ for several important classes of requests. 
In particular, the standard CSEU problem [cf. Problem~\ref{prob:CSEU}] is the special case that arbitrary input states are allowed, for which the only uniform purity bound yields $\P=1$, and the condition $\B\P\geq 1$ is satisfied automatically. Hence, Problem~\ref{prob:CSEU} is equivalent to Problem~\ref{prob:restricted_CSEU} on the request class $\mathcal R_{\P=1,\B}$, and \tref{thm:cseu_upper_bound_informal} follows immediately from \tref{thm:cseu_upper_bound}.

The rest of this subsection explains the main ingredients in the proof of \tref{thm:cseu_upper_bound_informal}. Our approach follows an estimator-based protocol. For the prescribed request class, the protocol queries the unknown unitary $U$ to produce classical shadow data $\mathsf{CS}(U)$. The prediction function $f$ is then designed to be unbiased, in the sense that for every request $(\rho, O)$, it holds that
$$
\E\left[f\bigl(\mathsf{CS}(U), \rho, O\bigr)\right] = \Tr\left[OU\rho U^\dagger\right],
$$
where the expectation is over the randomness in the generation of $\mathsf{CS}(U)$. Thus, $f(\mathsf{CS}(U),\rho,O)$ is an unbiased estimator of the desired prediction value. To meet the accuracy requirement in \pref{prob:restricted_CSEU}, it suffices to control the variance of this estimator. Indeed, with the guarantee
\[
\Var\left[f\bigl(\mathsf{CS}(U),\rho,O\bigr)\right]\leq \frac{\eps^2}{3},
\]
Chebyshev's inequality implies the desired accuracy guarantee in \eref{eq:restricted_desired_accuracy}. The upper-bound proof, therefore, reduces to constructing suitable unbiased estimators and proving the corresponding variance bounds.

Our protocol contains two layers. The first layer is data acquisition: we use a covariant unitary-learning experiment to convert parallel queries to the unknown unitary into classical unitary-valued data.
The second layer is classical post-processing: given a prediction request $(\rho,O)$, we use the collected data to construct an unbiased estimator of $\Tr[O U\rho U^\dagger]$. The concrete choice of the covariant learning strategy and the estimator depends on the target-accuracy regime. In the following, we describe the data-acquisition layer first, and then turn to the estimator design and variance analysis.

\subsubsection{From parallel unitary queries to unitary snapshots}
\label{subsec:bisio_etal_optimal_unitary_learning}

We now describe the data-acquisition phase of our protocol. Learning properties of quantum resources, including quantum states and quantum processes, often leverages covariant quantum operations when the objects to learn possess proper group structure by exploiting fundamental physical symmetries \cite{Hayashi_1998, Bu_ek_1999, Bisio_2010, Holevo2011, Hayashi2017, Yang_2020, Yang_2022, pelecanos2025mixedstatetomographyreduces, Yoshida_2026}. With the same spirit, we use a covariant learning procedure to extract information from the unknown unitary. In such an experiment, one queries the unknown unitary $U$ in parallel on a suitably chosen probe state. After the action of the parallel unitaries, the resulting memory state is measured by a covariant POVM\footnote{That is, a unitary-indexed measurement whose probability of getting outcome $\widehat{U}$ conditioned on the true unitary $U$ satisfies $p(\widehat{U} | U) = p(W \widehat{U} V^\dagger | W U V^\dagger)$ for any unitary $W, V$.
}, whose outcome is a unitary-valued classical variable $\widehat U$.
This unitary learning procedure is illustrated in Figure~\ref{fig:unitary_data_acquisition}.

\begin{figure}
    \centering
    \includegraphics[width=0.58\linewidth]{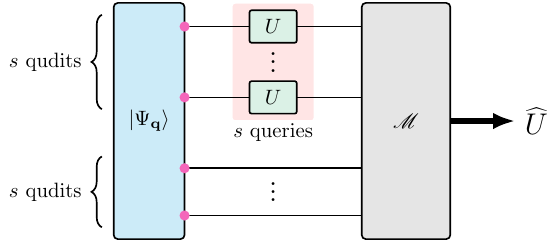}
    \caption{Covariant unitary learning protocol for generating one unitary snapshot.  Starting from the probe state $\ket{\Psi_{\bfq}}$, the protocol applies $s$ parallel uses of the unknown unitary $U$ and then measures the resulting memory state with the covariant POVM $\mathscr M$.  The measurement outcome $\widehat U$ is stored as part of the unitary-valued classical shadow data $\mathsf{CS}(U)$.}
    \label{fig:unitary_data_acquisition}
\end{figure}

This framework is closely related to the storage-and-retrieval problem for unitary channels \cite{Bisio_2010, Sedl_k_2019}. In that setting, the goal is to use several queries to an unknown unitary $U$ to store a quantum memory, from which one can later retrieve an approximation of the channel action on an arbitrary input state.
Regarding the average-case reconstruction in terms of the entanglement fidelity, Bisio \etal~\cite{Bisio_2010} proved, albeit purely existentially, that a family of covariant protocols must contain one that achieves optimality. By focusing solely on this covariant family, we can rule out inefficient protocols and restrict our attention to those operating in a ``measure-and-operate'' (MO) manner: after the parallel-query storage stage, one performs a covariant measurement on the memory state and, conditioned on the measurement outcome $\widehat{U}$, applies the channel $\widehat{U}(\cdot)\widehat{U}^\dagger$ to the input state. Although this covariant method has successfully demonstrated its power in estimation of $\U(1)$ \cite{Bu_ek_1999}, $\su(2)$ \cite{Chiribella2004EfficientUse}, $\su(3)$ \cite{yoshida2025asymptoticallyoptimalunitaryestimation}, $\su(d)$ \cite{Yang_2020}, \etc, designing a learning strategy that efficiently solves the corresponding estimation task remains highly nontrivial, as evidenced by both the aforementioned studies and other prior work \cite{Kahn_2007, Sedl_k_2019, Yang_2022}.

For our purposes, we use this framework differently. We do not use the MO protocol to implement or reconstruct the unknown channel directly. Instead, we drop the ``operate'' phase and keep only the covariant measurement outcome as classical data. Thus, each run of the covariant learning experiment produces a unitary snapshot $\widehat U$, which will later be post-processed to estimate our target expectation $\Tr[O  U\rho U^\dagger]$. The canonical form of the optimal covariant reconstruction protocol rules out numerous query-inefficient approaches and provides a structured family of probe states and measurements, from which we design our unitary snapshots.

We now present this canonical form. Let $\mathcal H^{\otimes s}$ be an $s$-qudit Hilbert space with $\dim\mathcal H=d$. By Schur-Weyl duality (see, e.g., \cite{goodman2009symmetry}), the tensor product space $\mathcal H^{\otimes s}$ and the $s$-fold unitary operator $U^{\otimes s}$ decompose as
\begin{equation}\label{eqn:Schur-Weyl_main}
    \mathcal H^{\otimes s}
    \cong
    \bigoplus_{\lambda\in\Y_s^d}
    \mathcal H_\lambda\otimes\mathcal M_\lambda,
    \qquad
    U^{\otimes s}
    \cong
    \bigoplus_{\lambda\in\Y_s^d}
    U_\lambda\otimes\openone_{\mathcal M_\lambda},
\end{equation}
where
$\Y_s^d    = \left\{\lambda\in\mathbb N^d: \sum_j\lambda_j=s,\, \lambda_i\ge\lambda_j\ \forall i<j \right\}$
is the set of \emph{Young diagrams}, $\mathcal H_\lambda$ is the carrier space of the irreducible representation $U_\lambda$ of $\U(d)$ with dimension $d_\lambda$, and $\mathcal M_\lambda$ is the corresponding multiplicity space with dimension $s_\lambda$ \cite{Yang_2020}.

As shown in \cite{Bisio_2010}, for the average-case storage-and-retrieval problem of an unknown $d$-dimensional unitary channel, it suffices to consider covariant protocols of the following canonical form. A unitary learning strategy is characterized by a subset $\Y\subseteq\Y_s^d$ of Young diagrams and a probability distribution
\[
    \bfq=(q_\lambda)_{\lambda\in\Y}\in\mathbb R_{\geq 0}^{|\Y|}
\]
over $\Y$. Given such a strategy $(\Y, \bfq)$, the probe state for parallel unitary storage is chosen as
\[
    \ket{\Psi_{\bfq}}
    =
    \bigoplus_{\lambda\in\Y}
    \sqrt{\frac{q_\lambda}{d_\lambda}} \, 
    \Ket{\openone_{\mathcal H_\lambda}}\otimes\ket{\eta_\lambda},
\]
where $\ket{\eta_\lambda}$ is an arbitrary state in the bipartite multiplicity space $\mathcal M_\lambda^{\otimes 2}$, and the double-ket notation is defined as $\Ket{X}:=\sum_{i,j}\braket{i|X|j}\ket{i}\ket{j}$.  
After applying $U^{\otimes s}$ to the probe state, the resulting memory state $U^{\otimes s}\ket{\Psi_{\bfq}}$ is measured using the covariant POVM
$\mathscr{M}=\big\{\!\ket{\Psi_{\widehat{U}}}\!\bra{\Psi_{\widehat{U}}}\dd {\widehat{U}}\big\}_{{\widehat{U}}\in\U(d)}$, where $\dd \widehat{U}$ is the Haar measure on $\U(d)$, and the reference frame vector
\[
    \ket{\Psi_{\widehat{U}}}
    =
    \bigoplus_{\lambda\in\Y}
    \sqrt{d_\lambda}\, 
    \Ket{\widehat{U}_\lambda}\otimes\ket{\eta_\lambda}.
\]  
Together, the probe state $\ket{\Psi_{\bfq}}$ and the measurement $\mathscr{M}$ give the covariant unitary learning procedure shown in Figure~\ref{fig:unitary_data_acquisition}. In our CSEU protocol, we keep the measurement outcome $\widehat U$ and store it in the classical shadow $\mathsf{CS}(U)$ for subsequent post-processing. The core remaining question is therefore how to properly choose $(\Y,\bfq)$ and recover the expectation $\Tr[O  U \rho U^\dagger]$ from the classical shadow in post-processing. To solve CSEU, we need another degree of freedom in addition to $(\Y, \bfq)$: the choice of a proper estimator $\widehat{E}$ of the target expectation. In our covariant framework, therefore, a CSEU protocol is characterized by $(\Y, \bfq, \widehat{E})$. The next subsection explains how these choices are made and how the corresponding variance bounds lead to the desired query complexity.

\subsubsection{From unitary snapshots to predictions}

We now explain how the unitary snapshots described above are used to construct estimators, and how the learning strategy $(\Y,\bfq)$ is chosen in different accuracy regimes.  The analysis has two components.  First, we identify unbiased estimators for the target prediction value.  Second, we choose the learning strategy and the parameters $s,L$ so that the corresponding variance is at most $\mathcal{O}(\eps^2)$, which gives the desired success probability by Chebyshev's inequality.

\paragraph{Reduction to traceless requests.}
Before constructing the estimators, we first isolate the part of the target quantity that depends nontrivially on the unknown unitary.  For any classically specified state $\rho$ and observable $O$, write
\[
    O = O_0+\frac{\Tr(O)}{d}\openone,
    \qquad
    \rho=\rho_0+\frac{\openone}{d},
\]
where $O_0$ and $\rho_0$ denote the traceless parts of $O$ and $\rho$, respectively.  Then
\[
    \Tr\left[O  U\rho U^\dagger\right]
    =
    \Tr\left[O_0  U\rho_0 U^\dagger\right]
    +
    \frac{\Tr(O)}{d}.
\]
The second term is known classically and is independent of $U$.  Hence, the only nontrivial part of the estimation problem is the $U$-dependent traceless component
$\Tr\left[O_0 U\rho_0U^\dagger\right]$.
Accordingly, in the analysis below, we focus on traceless observables; the scalar contribution $\Tr(O)/d$ is added back in the final prediction.  Unless otherwise specified, we therefore assume that the observable $O$ is traceless.

\paragraph{Unitary snapshots and query count.}
Following the data-acquisition procedure illustrated in Figure~\ref{fig:unitary_data_acquisition}, one run of the covariant learning experiment uses $s$ parallel queries to $U$ and produces a unitary-valued classical outcome $\widehat U$.  Repeating this experiment independently $L$ times gives the classical dataset
\[
    \mathsf{UData}(s,L)
    =
    \{\widehat U_j\}_{j\in[L]}.
\]
This dataset is the classical shadow $\mathsf{CS}(U)$ used for prediction with two adjustable parameters: $s$, the number of parallel queries in one covariant learning experiment, and $L$, the number of independent repetitions used for averaging.

\begin{fact}
\label{fact:total_query_count}
Creating a dataset $\mathsf{UData}(s,L)$ requires $K=sL$ queries to the unknown unitary $U$.
\end{fact}
Next, we will construct two unbiased estimators for $\Tr[O  U\rho U^\dagger]$. 

\paragraph{Averaged channel and the linear estimator.}
The first key property of the covariant learning outcome is that it acts as a depolarizing average channel around the true unitary channel.

\begin{lemma}[see, e.g., \cite{Yang_2020, he2025resourcequantificationprogramminglowdepth}]
\label{lemma:SAR_channel_average_result}
Let $\widehat U$ be distributed as one of the i.i.d. outcomes $\widehat U_j$.  Then, for any input operator $A$,
\[
    \E\left[\widehat U A\widehat U^\dagger\right]
    =
    \p_{\bfq} UA U^\dagger
    +
    (1-\p_{\bfq})\frac{\Tr(A)}{d}\openone,
\]
where $\p_{\bfq}\in[0,1]$ is the depolarizing parameter determined by the learning strategy $\bfq$. Its explicit form is given in the discussion in Appendix \ref{subsec:evaluate_choi_state}.
\end{lemma}

Since $O$ and $\rho_0$ are traceless, the depolarizing component in \lref{lemma:SAR_channel_average_result} does not contribute to the target expectation.  This immediately gives the linear estimator
\begin{equation}
\label{eqn:linear_estimator}
    \widehat X_j
    :=
    \frac{1}{\p_{\bfq}}
    \Tr\left[
        O  \widehat U_j\rho_0\widehat U_j^\dagger
    \right],
    \qquad \forall \, j\in[L].
\end{equation}

\paragraph{The quadratic estimator from Choi-state purity.}
Our second estimator exploits the channel-state duality and the fact that the Choi state of a unitary channel is pure\footnote{A similar trick has been used in \cite{Grier_2024, Li_2025}: if $\widehat\psi_i$ and $\widehat\psi_j$ are independent unbiased estimators of a pure state $\psi$, then their product $\widehat\psi_i\widehat\psi_j$ is still an unbiased estimator of $\psi$.}. Recall that \cite{Watrous_2018}
\begin{equation}
\label{eqn:CSEU_target_reformulation}
    \Tr\left[O  U\rho U^\dagger\right] 
    =
    \Tr\left[
        \left(O\otimes\rho^T\right)\Phi_{\mathcal U}
    \right],
\end{equation}
where $\Phi_{\mathcal U}:=\Ket{U}\Bra{U}$ denotes the Choi operator of the unitary channel $\mathcal U$.  Since $\Phi_{\mathcal U}$ is rank one and $\Braket{U|U}=d$, we have $\Phi_{\mathcal U}^2=d\,\Phi_{\mathcal U}$.  Therefore,
\[
    \Tr\left[
        \left(O\otimes\rho^T\right)\Phi_{\mathcal U}
    \right]
    =
    \frac{1}{d}
    \Tr\left[
        \left(O\otimes\rho^T\right)\Phi_{\mathcal U}^2
    \right].
\]
This motivates a quadratic estimator built from two independent unitary snapshots. Define 
\begin{equation}
\label{eqn:quadratic_estimator}
\widehat Y_j
    :=
    \frac{1}{\p_{\bfq}}\Ket{\widehat U_j}\Bra{\widehat U_j} \ \; \text{for} \ \; j\in[L], 
    \qquad \text{and} \qquad
    \widehat\Lambda_{i,j}
    :=
    \frac{1}{
        d+\frac{2(1-\p_{\bfq})}{d\p_{\bfq}}
    }
    \Tr\left[
        \left(O\otimes\rho_0^T\right)
        \widehat Y_i\,\widehat Y_j
    \right] \ \; \text{for} \ \; i\neq  j \in [L]. 
\end{equation}

\begin{fact}
\label{fact:unbiased_estimators}
For any $j\in[L]$ and $i\neq  j \in [L]$, both $\widehat X_j$ and $\widehat\Lambda_{i,j}$ are unbiased estimators of
$\Tr\left[O  U\rho U^\dagger\right]$.
\end{fact}

\begin{proof}
The unbiasedness of $\widehat X_j$ follows directly from \lref{lemma:SAR_channel_average_result}, because $O$ and $\rho_0$ are traceless:
\[
    \E[\widehat X_j]
    =
    \frac{1}{\p_{\bfq}}
    \Tr\left[
        O 
        \E\left(\widehat U_j\rho_0\widehat U_j^\dagger\right)
    \right]  
    =
    \frac{1}{\p_{\bfq}}
    \Tr\left[
        O
        \left(
        \p_{\bfq} U\rho_0U^\dagger
        +
        (1-\p_{\bfq})\frac{\Tr(\rho_0)}{d}\openone
        \right)
    \right]
    =
    \Tr\left[O  U\rho_0U^\dagger\right]=
    \Tr\left[O  U\rho U^\dagger\right].
\]
For the quadratic estimator, using \lref{lemma:SAR_channel_average_result} again, we have 
$$
\E\left[\widehat{Y}_j \right] = \frac{1}{\p_{\bfq}} \E\left[\Ket{\widehat{U}_j} \Bra{\widehat{U}_j} \right] = \frac{1}{\p_{\bfq}} \left( \p_{\bfq} \Ket{U} \Bra{U} + \frac{1 - \p_{\bfq}}{d } \openone \otimes \openone \right) = \Ket{U} \Bra{U} + \frac{1-\p_{\bfq}}{d\p_{\bfq}} \openone \otimes \openone, 
$$
    and therefore
    $$
    \begin{aligned}
        \E\left[ \widehat{\Lambda}_{i, j} \right] &= \frac{1}{d + \frac{2(1-\p_{\bfq})}{d \p_{\bfq}} } \Tr\left[ \left( O \otimes \rho_0^T \right) \E\left[ \widehat{Y}_i \,\widehat{Y}_j \right]\right] = \frac{1}{d + \frac{2(1-\p_{\bfq})}{d \p_{\bfq}} } \Tr\left[ \left( O \otimes \rho_0^T \right) \E\left[ \widehat{Y}_i \right] \E\left[ \widehat{Y}_j \right]\right] \\
        &= \frac{1}{d + \frac{2(1-\p_{\bfq})}{d \p_{\bfq}} } \Tr\left[ \left( O \otimes \rho_0^T \right) \left( \left(d + \frac{2(1-\p_{\bfq})}{d\p_{\bfq}} \right) \Ket{U} \Bra{U} + \left(\frac{1-\p_{\bfq}}{d\p_{\bfq}} \right)^2 \openone \otimes \openone \right) \right] \\
        &= \Tr\left[ O  U \rho_0 U^\dagger \right] =
    \Tr\left[O  U\rho U^\dagger\right]. 
    \end{aligned}
    $$
This completes the proof. 
\end{proof}

Let $\widehat{\mathsf{X}} = \{ \widehat{X}_j \}_{j \in [L]}$ and $\widehat{\mathsf{\Lambda}} =\{\widehat{\Lambda}_{i, j}\}_{i \neq j \in [L]}$. 
We use their batch averages
\begin{equation}\label{eqn:average_estimators}
    \widehat Z(\widehat{\mathsf X},L)
    :=
    \frac{1}{L}\sum_{j=1}^L \widehat X_j,
    \qquad
    \widehat Z(\widehat{\mathsf\Lambda},L)
    :=
    \frac{1}{L(L-1)}
    \sum_{i\neq  j\in[L]}\widehat\Lambda_{i,j}.
\end{equation}
The remaining task is to tailor the learning strategy $(\Y,\bfq)$, along with parameters $s$ and $L$, to these two averaged estimators so that their variances are bounded by $\mathcal{O}(\eps^2)$.

\paragraph{Variance bounds for two learning strategies.}
The two estimators above are useful in complementary regimes of the error $\eps$. 
For the quadratic estimator $\widehat Z(\widehat{\mathsf\Lambda},L)$, we use a learning strategy based on the Plancherel measure on Young diagrams \cite{Borodin2000}: we restrict to $s\leq d$ and choose
\begin{equation}
\label{eq:Plancherel_measure}
    \Y=\Y_s^d,
    \qquad
    \bfq_{\mathrm{Plan}} = (q_{\lambda})_{\lambda \in \Y}, \quad  q_\lambda:=\frac{s_\lambda^2}{s!}, 
\end{equation}
where $s_\lambda$ denotes the dimension of the multiplicity space $\mathcal M_\lambda$ in \eref{eqn:Schur-Weyl_main}.
The following lemma provides the variance bound needed for the moderate-precision regime; its proof is given in Appendix~\ref{appendix:Proof_small_s_regime_estimator_Lambda_variance}.

\begin{lemma}
\label{lemma:small_s_regime_estimator_Lambda_variance}
Suppose the quantum state $\rho$ and the traceless observable $O$ satisfy $\Tr(\rho^2)\leq\P$ and $O\in\obs(\B)$. When $s\leq d$,  the CSEU protocol $(\Y^{d}_s, \bfq_{\mathrm{Plan}}, \widehat Z(\widehat{\mathsf\Lambda},L))$ achieves the following performance
\[
\Var\left[\widehat Z(\widehat{\mathsf\Lambda},L)\right]
\leq
\bigO\left(
    \frac{1}{L} \left(  \min\{1, \B\P\} + \frac{d^2}{s^4} \B\P + \frac{d}{s} \P  \right)
    +
    \frac{1}{L^2}
    \left(
        \frac{d^4}{s^4}\B\P
    \right)
\right).
\]
\end{lemma}

For the linear estimator $\widehat Z(\widehat{\mathsf X},L)$, we introduce a different family of learning strategies, denoted as $(\Y_{\sin}, \bfq_{\sin})$, which we call the \emph{sine-power state strategies}. The full details are provided in Construction~\ref{construction:YRC_programming_scheme} of Appendix~\ref{appendix:Proof_large_s_regime_estimator_X_variance}. This family is tailored to the high-precision regime, where one uses a large number $s$ of parallel queries in a single covariant learning experiment. Its key feature is that it yields a sufficiently small variance for the linear estimator, enabling the Heisenberg-scaling query complexity in this regime.  The resulting variance bound is stated below and proved in Appendix~\ref{appendix:Proof_large_s_regime_estimator_X_variance}.

\begin{lemma}
\label{lemma:large_s_regime_estimator_X_variance}
Suppose the quantum state $\rho$ and the traceless observable $O$ satisfy $\Tr(\rho^2)\leq\P$ and $O\in\obs(\B)$. When $s\geq Cd^2$ for a sufficiently large universal constant $C$, for any $j \in [L]$, the CSEU protocol $(\Y_{\sin}, \bfq_{\sin}, \widehat{X}_j)$ [cf. Construction~\ref{construction:YRC_programming_scheme}] achieves the following performance

\begin{enumerate}
    \item  $\Var[\widehat X_j]
        \leq
        \bigO\left(
            \frac{d^3}{s^2}
        \right)$.

    \item $\Var[\widehat X_j]
        \leq
        \bigO\left(
            \frac{d^3}{s^2}
            \P
            +
            \frac{d^7}{s^4}
            \min\{1-\Tr(\rho^2),\B\P\}
        \right)$.
\end{enumerate}
\end{lemma}

Since the snapshots $\{\widehat U_j\}_{j\in[L]}$ are independent, the averaged linear estimator satisfies
\[
    \Var\left[\widehat Z(\widehat{\mathsf X},L)\right]
    =
    \frac{1}{L}\Var[\widehat X_j].
\]
For either averaged estimator, $\widehat Z(\widehat{\mathsf X},L)$ or $\widehat Z(\widehat{\mathsf\Lambda},L)$, Chebyshev's inequality gives  
\[
\Pr\left[
    \left|
    \mathsf{estimator}
    -
    \Tr\left[O  U\rho_0U^\dagger\right]
    \right|
    \geq \eps
\right]
\leq
\frac{\Var[\mathsf{estimator}]}{\eps^2}.
\]
Thus, it suffices to choose $s,L$ so that the relevant variance is $\mathcal{O}(\eps^2)$.

\paragraph{Parameter choice in the moderate-precision regime.}
We first use the quadratic estimator $\widehat Z(\widehat{\mathsf\Lambda},L)$ together with the Plancherel learning strategy.  By \lref{lemma:small_s_regime_estimator_Lambda_variance}, to ensure  $\Var\!\big[\widehat Z(\widehat{\mathsf\Lambda},L)\big]\leq \mathcal{O}(\eps^2)$, we take
\[
L
=
\bigO\left(
    \frac{1}{\eps^2}
    \left(
        \min\{1, \B\P\} + \frac{d^2}{s^4} \B\P + \frac{d}{s} \P
    \right)
    +
    \frac{1}{\eps}
    \frac{d^2}{s^2}\sqrt{\B\P}
\right).
\]
Using $K=sL$, the query complexity therefore becomes
\[
K
=
\bigO\left(
    \frac{1}{\eps^2}
    \left(
    s \min \{1, \B\P\} +
        \frac{d^2}{s^3} \B\P + d \P
    \right)
    +
    \frac{1}{\eps}
    \frac{d^2}{s}\sqrt{\B\P}
\right),
\]
which gives the query complexity in the first part of \tref{thm:cseu_upper_bound}. 
In particular, taking $s=d$ yields
\[
K
=
\bigO\left(
    \frac{d \left( \min\{1, \B\P\} + \P \right) }{\eps^2}
    +
    \frac{d\sqrt{\B\P}}{\eps} 
\right).
\]

With the single-shot batch size set to $s = d$ in the moderate-precision regime
$\eps\geq d^{-1/2}$, the storage cost of the classical shadow dataset and the classical post-processing time required to evaluate the prediction function are both polynomial in $d$.
Indeed, in this case, the required number of independent snapshots satisfies
$$
L = \frac{K}{s} = \bigO\left(\sqrt{\B\P}\,\eps^{-1} + \left( \min\{1, \B\P\} + \P \right) \eps^{-2}\right) \leq \bigO\left(d \left( \min\{1, \B\P\} + \P \right) + \sqrt{d \B\P} \right) \leq \bigO(d),
$$
where we have used $\B\leq d$ and $\P \leq 1$ in the last step. 
The classical shadow consists of $L$ unitary snapshots $\{\widehat{U}_j\}_{j \in [L]}$. Since each snapshot is a $d \times d$ complex matrix, the entire shadow dataset can be stored using $\bigo{Ld^2}$ complex numbers, which is polynomial in $d$ as $L = \bigo{d}$.
Moreover, for any fixed classical request $(\rho,O)$, the quadratic prediction function $\widehat Z(\widehat{\mathsf\Lambda},L)$ can be evaluated by summing over $L(L-1)$ snapshot pairs, with each term computable in $\poly(d)$ time. 
Thus, the prediction function $f$ can be evaluated with classical time  $\mathcal{O}(L^2\poly(d))= \bigo{\poly(d)}$,
up to standard finite-precision overheads. 
Together, these results establish the storage and classical prediction-time complexity stated in the first part of \tref{thm:cseu_upper_bound}.

\paragraph{Parameter choice in the high-precision regime.}
It now remains to establish the claimed query complexity in the high-precision regime $\eps \leq d^{-1/2}$.
Here we use the linear estimator $\widehat Z(\widehat{\mathsf X},L)$ together with the sine-power state learning strategy.  We use a single covariant learning experiment with $s$ parallel queries, namely $L=1$, so that the query complexity is $K=s$. Recall that our goal is to ensure  $\Var\!\big[\widehat Z(\widehat{\mathsf X},L)\big]\leq \mathcal{O}(\eps^2)$. \lref{lemma:large_s_regime_estimator_X_variance} shows that it suffices to choose
\[
    K=s
    =
    \Theta\left(
        \frac{d^{3/2}}{\eps}
    \right),
\]
with a sufficiently large leading constant so that $s\geq Cd^2$. 
The second bound in \lref{lemma:large_s_regime_estimator_X_variance} shows that in the case of mixed input states, it suffices to choose
\[
K=s
    =
    \Theta\left(
        \frac{d^{3/2}\sqrt{\P}}{\eps}
        +
        \frac{d^{7/4}}{\sqrt{\eps}}
        \left(1-\Tr(\rho^2)\right)^{1/4}
    \right),
\]
again provided that a sufficiently large leading constant so that $s\geq Cd^2$.  
In particular, the first term dominates whenever $\eps\sqrt{1-\Tr[\rho^2]} \leq d^{-1/2}\P$, which holds for all $\eps$ if $\rho$ is pure. In this case, we are left with
$$
    K=\bigO\left( \frac{d^{3/2} \sqrt{\P}}{\eps} \right).
$$
Together, these results establish the query complexity in the second part of \tref{thm:cseu_upper_bound}.

For the high-precision case considered here, the storage cost and classical prediction time are also polynomial in $d$. 
Indeed, since $L=1$, the classical shadow data consist of a single unitary snapshot $\widehat U$, which can be stored using $\mathcal{O}(d^2)$ complex numbers in a dense-matrix representation. 
Moreover, for any fixed request $(\rho,O)$, the linear prediction function only requires evaluating
$\widehat X=\p_{\bfq}^{-1}\Tr\big[O\widehat U\rho_0\widehat U^\dagger\big]$,
which can be done in $\poly(d)$ classical time using standard matrix operations. 
Together, these results establish the storage and classical prediction-time complexity stated in the second part of \tref{thm:cseu_upper_bound}.

\subsubsection{Amplification for multiple prediction requests}\label{sec:mean_Amplification}
Here, we briefly explain the standard median amplification argument \cite{Huang_2020, Grier_2024, Li_2025} used in Remark~\ref{remark:Mproperties}. 
Let $\mathcal A$ be a protocol that solves Problem~\ref{prob:CSEU} with success probability at least $2/3$ using $K$ oracle queries. 
Run $\mathcal A$ independently for
$R=\mathcal{O}(\log(M/\delta))$ times, obtaining independent classical shadows
\[
    S^{(1)}(U),S^{(2)}(U),\ldots,S^{(R)}(U).
\]
For each request $(\rho_\ell,O_\ell)$ and each repetition $t\in[R]$, define
\[
    \widehat y_{\ell}^{(t)}
    :=
    f\bigl(S^{(t)}(U),\rho_\ell,O_\ell\bigr).
\]
For every fixed $\ell$, each $\widehat y_{\ell}^{(t)}$ is $\eps$-accurate with probability at least $2/3$. 
By a Chernoff bound, with probability at least $1-\delta/M$, more than half of
$\widehat y_{\ell}^{(1)},\ldots,\widehat y_{\ell}^{(R)}$
are $\eps$-accurate. On this event, the median estimator
\[
    \widehat y_\ell
    :=
    \mathsf{median}
    \bigl\{
    \widehat y_{\ell}^{(1)},\ldots,\widehat y_{\ell}^{(R)}
    \bigr\}
\]
also has additive error at most $\eps$. 
A union bound over all $\ell\in[M]$ then ensures that with probability at least $1-\delta$, all $M$ estimates $\widehat y_1,\dots, \widehat y_M$ are simultaneously $\eps$-accurate. 
The total number of queries is thus 
\[
KR=\mathcal{O}(K\log(M/\delta)).
\]
This establishes the results stated in Remark~\ref{remark:Mproperties}.

\subsection{Establishing the lower bound}

We now explain the high-level proof idea behind the CSEU lower bound in \tref{thm:cseu_lower_bound}. The detailed proofs are deferred to Appendix~\ref{sec:detailed_proof_of_cseu_lower_bound}, in which the arguments are presented in a different order for technical clarity.

The starting point is the construction of a hard family of unitaries inherited from the reflection packing used in \cite{HHJ17, HKOT23}: For every sufficiently large choice of the block size $D$, there exists a family of reflections $\{R_x\}_{x \in [N]} \subseteq \U(D)$ with exponential size $N = e^{\Theta(D^2)}$ and sufficient separation $\| R_x - R_y \|_F = \Omega(\sqrt{D})$ for every $x \neq y \in [N]$. To further tune the separations, we use the fractional reflections $R_{x}^{\alpha} = \frac{\openone + R_x}{2} + e^{-\I \pi \alpha} \frac{\openone - R_x}{2}$ to introduce the dependence of the lower bound on $\eps$ \cite[Definition 4.4]{HKOT23}.

The first observation is that, given access to any $R_{x}^{\alpha}$ where $x \in [N]$ for constant-rank CSEU requests, one can decode $x$ once the CSEU problem [cf. Problem \ref{prob:CSEU}] can be solved. To see this, we take a proper choice of pure states $\rho, O$ and randomly rotate them with a Haar random unitary $U$ over $\U(D)$ to obtain $\rho(U) = U \rho U^\dagger, O(U) = U O U^\dagger$, and we set the Hermitian commutator $H(\rho, O) := -\I[\rho, O]$, it holds by the separation lower bound that for every $x \neq y$, $| \Tr[H(\rho(U), O(U))(Q_x - Q_y)] | = \Omega(D^{-1/2})$ with constant advantage, where we set $Q_x = \frac{1}{2}(\openone - R_x)$. Moreover, if we write the difference between two CSEU-type expectations $g_x(U) = \Tr\left[ O(U) R_{x}^{\alpha} \rho(U) {R_{x}^{\alpha}}^\dagger \right] - \Tr\left[ \rho(U) R_{x}^{\alpha} O(U) {R_{x}^{\alpha}}^\dagger \right]$, then the random rotation ensures that these $g_x$s with different indices can be separated with high probability: For $x \neq y$,
\begin{equation}
\label{eqn:separation_bound}
\Pr_{ U \in \haar(\U(D)) } \left[ \left| g_x(U) - g_y(U) \right| \geq \Omega\left(\frac{\sin(\pi \alpha)}{\sqrt{D}}\right)\right] \geq \frac{1}{64}.
\end{equation}
Therefore, since $\Omega(\sin(\pi \alpha) / \sqrt{D}) = \Omega(\alpha / \sqrt{D})$, choosing the fraction $\alpha = \Theta(\eps \sqrt{D}) = \bigo{1}$ makes the separation a sufficiently large constant multiple of $\eps$ with high probability, thus ensuring that solving CSEU $\eps$-accurately allows us to properly detect these separations. Moreover, we are to invoke this reflection packing as a subroutine for our $\U(d)$ lower bound, so that we should restrict $D = O(d)$. Putting these constraints together tells that setting $D = \Theta(\min\{d, \eps^{-2}\})$ is a proper choice.

Next, we expand the size of the hard family by concatenating $m = \lfloor d / D \rfloor$ such blocks in direct sum and use $\mathbf{x} \in [N]^m$ to index them. On each coordinate, we maintain the choice of the fractional reflection hard family independently, resulting in a family of size $N^m$. Therefore, the effective dimension yields $\log N^m = m \log N = \Theta(m D^2) = \Theta(dD)$. By a probabilistic argument \cite{Alon2008} one can show that there exists a family of $T = \Theta(D^2)$ CSEU requests similar to that of \eref{eqn:separation_bound} such that every two distinct labels within one block would differ by $\Omega(\eps)$ upon evaluating these CSEU requests. Embedding the same request evaluation into every block gives a total of $M = 2 mT = \Theta(k)$ requests, with the value of evaluating the requests supported on the $j$-th block where $j \in [m]$ depends only on $x_j$ for any $\mathbf{x} \in [N]^m$.

Now we are ready to utilize the reusability of CSEU: We repeat the data acquisition process only a universal constant number of times and take their median to ensure that a constant small fraction\footnote{This requirement differs from the usual one in classical shadows: rather than requiring per-instance accuracy, we only need a constant fraction of them to be accurate, and by Markov's inequality, this incurs only a constant overhead. For details, see \lref{lemma:concatenating_CSEU_hard_family}.} of all requested coordinates are inaccurate with high probability. The separation guarantee implies that a single block is decoded incorrectly only when a constant fraction of its own requested coordinates via CSEU are identified. Consequently, for all $m$ blocks, a small constant fraction (say, some $0 < \delta < 1$) of them have their label recovered incorrectly, so that $\ham(\widehat{\mathbf{x}}, \mathbf{x}) \leq \delta m$ with high probability. If we choose the Hamming ball around the recovered label $\widehat{\mathbf{x}}$ with an appropriate radius, i.e., $\mathcal{L} = \{\mathbf{y} \in [N]^m : \ham(\widehat{\mathbf{x}}, \mathbf{y}) \leq \lfloor \delta m \rfloor\}$, then this yields a set that is exponentially smaller than the full hard family while still containing the ground-truth label with high probability.
This is an instance of the list-decoding problem~\cite{list_decoding_2008}, which generalizes the indexing problem underlying the lower bound proof of \cite{HKOT23}.

It remains to be shown that such a list we decode should not be recovered with very few queries. The related lower bound argument follows from an affine combination version of the polynomial method utilized in \cite{BealsEtAl01, she_et_al:LIPIcs.ITCS.2023.96}. Note that our block-diagonal reflection family is described by merely $k = \Theta(mD^2)$ real parameters, and with queries to the unknown reflection oracle, their entries are affine linear in terms of these parameters. By expanding a $Q$-query quantum state by any algorithm, the amplitudes of the final state yield a polynomial of degree $\leq Q$ in terms of all these $k$ variables, and therefore the set of all possible final states lies in the subspace with bases being all $k$-variable, degree-$\leq Q$ monomials, and there are $\binom{k+Q}{Q}$ many by standard combinatorics. A linear algebra argument yields an upper bound for the average probability that a list of size at most $L$ contains the target label $\mathbf{x}$, which yields $\frac{L}{N^m} \binom{k+Q}{Q}$.

It suffices to put all the ingredients together: As the hard instances we create for CSEU are fractional reflections and control-free, we apply the Haah-Kothari-O'Donnell-Tang conversion \cite{HKOT23} together with a direct sum argument to adapt all allowed query types, so that any $q$-query algorithm to the varieties of the fractional reflections [cf. \tref{thm:cseu_lower_bound}] can be converted into an alternative one making $Q = \bigo{\alpha q}$ queries to a controlled reflection, while retaining the list decoding success probability. Comparing the ultimate success probability with the affine counting upper bound forces $\alpha q = \Omega(k)$. As the success probability amplification uses only extra constant repetitions, and since $q = \Theta(K(d, \eps))$, we have $K(d, \eps) = \Omega(k/\alpha)$. Finally, substituting $k = \Theta(dD)$, $\alpha = \Theta(\eps \sqrt{D})$ and $D = \Theta(\min\{d, \eps^{-2}\})$ yields
$$
K(d, \eps) = \Omega\left( \frac{k}{\alpha} \right) = \Omega\left( \frac{d}{\eps} \sqrt{\min \left\{d, \frac{1}{\eps^2} \right\} } \right) = \Omega\left( \min\left\{\frac{d^{3/2}}{\eps}, \frac{d}{\eps^2} \right\} \right),
$$
which proves the lower bound claimed in \tref{thm:cseu_lower_bound}.

\section{Discussion and outlook}
\label{sec:discussions}

There are several interesting directions for future work that would significantly strengthen our understanding of protocol design for tomographic tasks of quantum processes and extend our protocol to broader applicability.

\paragraph{Lower bound on $\B$.}
While the query complexity lower bound we derived has shown that the minimal dependence on $d$ and $\eps$ is $\Omega(\min\{ d \eps^{-2}, d^{3/2} \eps^{-1} \})$, the minimal dependence on $\B$ remains unknown. Grier, Pashayan, and Schaeffer \cite{Grier_2024} utilized techniques in the study of communication complexity to derive the minimal dependence on $\B$ for CSEPS. It is plausible that similar techniques can be adapted to establish a lower bound on the dependence of the query complexity of CSEU on $\B$.

\paragraph{Computational efficiency and noise-resilience.} Our design of the CSEU protocol relies on quantum operations conducted in the Schur-Weyl basis, including synthesizing highly entangled probe states and applying covariant measurements with high Kraus rank. According to \cite{harrow_schur_2005, Krovi2019efficienthigh, burchardt2025highdimensionalquantumschurtransforms}, the Schur transform that transforms the computational bases to the Schur-Weyl basis in the space $(\C^{\otimes d})^{\otimes s}$ requires a quantum circuit of size $\poly(s, \log d)$. At the retrieval phase, generating the Haar randomness (or an $(s + 2)$-design) requires a quantum circuit of size $\poly(d, s)$ \cite{Schuster2025, he2025resourcequantificationprogramminglowdepth}. To achieve query-optimal learning, an exponential circuit complexity of $\poly(d)$ with $\poly(d)$ ancillary qubits is required for a single-shot data acquisition, since we need $s = \Omega(d)$ in (almost) all cases. A candidate approach for a quantumly efficient CSEU protocol is to reduce the problem to pure-state classical shadows \cite{Grier_2024, Li_2025}, at the cost of restricting the query complexity to the standard quantum limit with a quadratic dependence on $d$. However, the gate complexity for coherent access to $s = \Omega(d)$ copies of the unknown unitary is inherently $\poly(d)$.

Moreover, since our probe state is highly entangled, it is fragile in the presence of certain noises, leading to the performance degradation of our CSEU protocol. To see hints of this phenomenon, one can look to \cite{Chen2023, DongMIP24}, which demonstrate how noise affects the performance of quantum devices that solve decision problems; for computational problems, the situation is predictably worse. Notably, noise-robust protocols for classical shadow estimation of quantum states have been considered \cite{ChenRobust2021, PhysRevLett.134.090801}. For CSEU protocols, however, it remains open how to improve both operational efficiency and noise resilience, while preserving the query-complexity optimality achieved in this work.

\section*{Acknowledgments}
E.H., Z.L., and Y.Y. conceived the project.
E.H. and Z.L. formulated the theoretical framework, developed the theoretical tools and results, and wrote the paper with Y.Y.. N.S. and S.Z. contributed to framing and contextualizing the full tomography application and discussions. Y.Y. supervised the project. 

E.H. thanks Yangjing Dong, Fengning Ou, and Penghui Yao for insightful discussions on learning shallow circuits.
E.H., Z.L., and Y.Y. are supported by the National Natural Science Foundation of China via the Excellent Young Scientists Fund (Hong Kong and Macau) Project 12322516, the National Natural Science Foundation of China (NSFC)/Research Grants Council (RGC) Joint Research Scheme via Project N\_HKU7107/24, the Guangdong Provincial Quantum Science Strategic Initiative via Project GDZX2503001 and the Hong Kong Research Grants Council (RGC) through the General Research Fund (GRF) Grant 17305625.
N.S. and S.Z. acknowledge the support of the Natural Sciences and Engineering Research Council of Canada (NSERC), [DGECR-2025-00505, RGPIN-2025-04054], National Research Council of Canada (NRC), [AQC-217-1], and Perimeter Institute for Theoretical Physics, a research institute supported in part by the Government of Canada through the Department of Innovation, Science and Economic Development Canada and by the Province of Ontario through the Ministry of Colleges and Universities.

The main idea behind the improved matching lower bound on the query complexity of CSEU in Theorem~\ref{thm:cseu_lower_bound} was suggested by ChatGPT 5.6. Our initial prompts pointed the model to the index decoding problem of \cite{HKOT23} and asked it to derive the $\min\{\eps, d^{-2}\}$-dependence. The model proposed the proof idea via list decoding; the authors subsequently verified,
reorganized, and extended it into the argument presented here.

\bibliographystyle{alphaurl}
\bibliography{Sample_Optimal_CSEU/cseu_refs_clean}

\newcommand{\etalchar}[1]{$^{#1}$}
\begin{thebibliography}{CVdPCH26}

\bibitem[Aar18]{aasonson_shadow2018}
Scott Aaronson.
\newblock Shadow tomography of quantum states.
\newblock In {\em Proceedings of the 50th Annual ACM SIGACT Symposium on Theory of Computing}, pages 325--338. Association for Computing Machinery, 2018.
\newblock \href {https://doi.org/10.1145/3188745.3188802} {\path{doi:10.1145/3188745.3188802}}.

\bibitem[ACEG{\etalchar{+}}25]{abbas2025nearly}
Amira Abbas, Nunzia Cerrato, Francisco Escudero~Guti{\'e}rrez, Dmitry Grinko, Francesco~Anna Mele, and Pulkit Sinha.
\newblock Nearly optimal algorithms to learn sparse quantum {Hamiltonians} in physically motivated distances, 2025.
\newblock \href {https://arxiv.org/abs/2509.09813} {\path{arXiv:2509.09813}}.

\bibitem[ADEG26]{arunachalam2026testing}
Srinivasan Arunachalam, Arkopal Dutt, and Francisco Escudero~Guti{\'e}rrez.
\newblock Testing and learning structured quantum {Hamiltonians}.
\newblock {\em Communications in Mathematical Physics}, 407(1):16, 2026.
\newblock \href {https://doi.org/10.1007/s00220-025-05517-w} {\path{doi:10.1007/s00220-025-05517-w}}.

\bibitem[ADOY25]{qac0_barely_superlinear_2025}
Anurag Anshu, Yangjing Dong, Fengning Ou, and Penghui Yao.
\newblock On the computational power of {$\mathsf{QAC}^0$} with barely superlinear ancillae.
\newblock In {\em Proceedings of the 57th Annual {ACM} Symposium on Theory of Computing}, pages 1476--1487. ACM, 2025.
\newblock \href {https://doi.org/10.1145/3717823.3718189} {\path{doi:10.1145/3717823.3718189}}.

\bibitem[Ang25]{Angrisani_2025}
Armando Angrisani.
\newblock Learning unitaries with quantum statistical queries.
\newblock {\em Quantum}, 9:1817, 2025.
\newblock \href {https://doi.org/10.22331/q-2025-07-30-1817} {\path{doi:10.22331/q-2025-07-30-1817}}.

\bibitem[AS08]{Alon2008}
Noga Alon and Joel~H. Spencer.
\newblock {\em The Probabilistic Method}.
\newblock Wiley, third edition, 2008.
\newblock \href {https://doi.org/10.1002/9780470277331} {\path{doi:10.1002/9780470277331}}.

\bibitem[AS17]{Aubrun2017}
Guillaume Aubrun and Stanis{\l}aw~J. Szarek.
\newblock {\em Alice and Bob Meet Banach: The Interface of Asymptotic Geometric Analysis and Quantum Information Theory}, volume 223 of {\em Mathematical Surveys and Monographs}.
\newblock American Mathematical Society, 2017.
\newblock \href {https://doi.org/10.1090/surv/223} {\path{doi:10.1090/surv/223}}.

\bibitem[BAL19]{Bairey2019LearnLocalHamiltonian}
Eyal Bairey, Itai Arad, and Netanel~H. Lindner.
\newblock Learning a local {H}amiltonian from local measurements.
\newblock {\em Physical Review Letters}, 122(2):020504, 2019.
\newblock \href {https://doi.org/10.1103/PhysRevLett.122.020504} {\path{doi:10.1103/PhysRevLett.122.020504}}.

\bibitem[BBC{\etalchar{+}}01]{BealsEtAl01}
Robert Beals, Harry Buhrman, Richard Cleve, Michele Mosca, and Ronald de~Wolf.
\newblock Quantum lower bounds by polynomials.
\newblock {\em Journal of the ACM}, 48(4):778--797, 2001.
\newblock \href {https://doi.org/10.1145/502090.502097} {\path{doi:10.1145/502090.502097}}.

\bibitem[BCC{\etalchar{+}}14]{BCC17}
Dominic~W. Berry, Andrew~M. Childs, Richard Cleve, Robin Kothari, and Rolando~D. Somma.
\newblock Exponential improvement in precision for simulating sparse {Hamiltonians}.
\newblock In {\em Proceedings of the Forty-Sixth Annual {ACM} Symposium on Theory of Computing}, {STOC} '14, pages 283--292. Association for Computing Machinery, 2014.
\newblock \href {https://doi.org/10.1145/2591796.2591854} {\path{doi:10.1145/2591796.2591854}}.

\bibitem[BCD{\etalchar{+}}10]{Bisio_2010}
Alessandro Bisio, Giulio Chiribella, Giacomo~Mauro D'Ariano, Stefano Facchini, and Paolo Perinotti.
\newblock Optimal quantum learning of a unitary transformation.
\newblock {\em Physical Review A}, 81(3):032324, 2010.
\newblock \href {https://doi.org/10.1103/PhysRevA.81.032324} {\path{doi:10.1103/PhysRevA.81.032324}}.

\bibitem[BCH{\etalchar{+}}94]{Benkart1994}
Georgia Benkart, Manish Chakrabarti, Tom Halverson, Robert Leduc, Chanyoung Lee, and Jeffrey Stroomer.
\newblock Tensor product representations of general linear groups and their connections with {Brauer} algebras.
\newblock {\em Journal of Algebra}, 166(3):529--567, 1994.
\newblock \href {https://doi.org/10.1006/jabr.1994.1166} {\path{doi:10.1006/jabr.1994.1166}}.

\bibitem[BCO26]{Bluhm_2026}
Andreas Bluhm, Matthias~C. Caro, and Aadil Oufkir.
\newblock {Hamiltonian} property testing.
\newblock {\em Quantum}, 10:1979, 2026.
\newblock \href {https://doi.org/10.22331/q-2026-01-21-1979} {\path{doi:10.22331/q-2026-01-21-1979}}.

\bibitem[BDM99]{Bu_ek_1999}
V.~Bu{\v{z}}ek, R.~Derka, and S.~Massar.
\newblock Optimal quantum clocks.
\newblock {\em Physical Review Letters}, 82(10):2207--2210, 1999.
\newblock \href {https://doi.org/10.1103/PhysRevLett.82.2207} {\path{doi:10.1103/PhysRevLett.82.2207}}.

\bibitem[BFG{\etalchar{+}}25]{burchardt2025highdimensionalquantumschurtransforms}
Adam Burchardt, Jiani Fei, Dmitry Grinko, Martin Larocca, Maris Ozols, Sydney Timmerman, and Vladyslav Visnevskyi.
\newblock High-dimensional quantum {Schur} transforms, 2025.
\newblock \href {https://arxiv.org/abs/2509.22640} {\path{arXiv:2509.22640}}.

\bibitem[BHRK25]{PhysRevLett.134.090801}
Raphael Brieger, Markus Heinrich, Ingo Roth, and Martin Kliesch.
\newblock Stability of classical shadows under gate-dependent noise.
\newblock {\em Physical Review Letters}, 134(9):090801, 2025.
\newblock \href {https://doi.org/10.1103/PhysRevLett.134.090801} {\path{doi:10.1103/PhysRevLett.134.090801}}.

\bibitem[BJY24]{Belovs2024tamingquantumtime}
Aleksandrs Belovs, Stacey Jeffery, and Duyal Yolcu.
\newblock Taming quantum time complexity.
\newblock {\em Quantum}, 8:1444, 2024.
\newblock \href {https://doi.org/10.22331/q-2024-08-23-1444} {\path{doi:10.22331/q-2024-08-23-1444}}.

\bibitem[BKD14]{PhysRevA.90.012110}
Charles~H. Baldwin, Amir Kalev, and Ivan~H. Deutsch.
\newblock Quantum process tomography of unitary and near-unitary maps.
\newblock {\em Physical Review A}, 90(1):012110, 2014.
\newblock \href {https://doi.org/10.1103/PhysRevA.90.012110} {\path{doi:10.1103/PhysRevA.90.012110}}.

\bibitem[BLMT24]{bakshi2024structure1}
Ainesh Bakshi, Allen Liu, Ankur Moitra, and Ewin Tang.
\newblock Structure learning of {H}amiltonians from real-time evolution.
\newblock In {\em 2024 {IEEE} 65th Annual Symposium on Foundations of Computer Science ({FOCS})}, pages 1037--1050. IEEE, 2024.
\newblock \href {https://doi.org/10.1109/FOCS61266.2024.00069} {\path{doi:10.1109/FOCS61266.2024.00069}}.

\bibitem[BO21]{buadescu2021improved}
Costin B{\u{a}}descu and Ryan O'Donnell.
\newblock Improved quantum data analysis.
\newblock In {\em Proceedings of the 53rd Annual ACM SIGACT Symposium on Theory of Computing}, pages 1398--1411. Association for Computing Machinery, 2021.
\newblock \href {https://doi.org/10.1145/3406325.3451109} {\path{doi:10.1145/3406325.3451109}}.

\bibitem[BOO00]{Borodin2000}
Alexei Borodin, Andrei Okounkov, and Grigori Olshanski.
\newblock Asymptotics of {Plancherel} measures for symmetric groups.
\newblock {\em Journal of the American Mathematical Society}, 13(3):481--515, 2000.
\newblock \href {https://doi.org/10.1090/S0894-0347-00-00337-4} {\path{doi:10.1090/S0894-0347-00-00337-4}}.

\bibitem[Bra07]{Bravyi2007}
Sergey Bravyi.
\newblock Upper bounds on entangling rates of bipartite {H}amiltonians.
\newblock {\em Physical Review A}, 76(5):052319, 2007.
\newblock \href {https://doi.org/10.1103/PhysRevA.76.052319} {\path{doi:10.1103/PhysRevA.76.052319}}.

\bibitem[Car24]{caro2022learning}
Matthias~C. Caro.
\newblock Learning quantum processes and {Hamiltonians} via the {Pauli} transfer matrix.
\newblock {\em ACM Transactions on Quantum Computing}, 5(2):1--53, 2024.
\newblock \href {https://doi.org/10.1145/3670418} {\path{doi:10.1145/3670418}}.

\bibitem[CCHL23]{Chen2023}
Sitan Chen, Jordan Cotler, Hsin-Yuan Huang, and Jerry Li.
\newblock The complexity of {NISQ}.
\newblock {\em Nature Communications}, 14(1):6001, 2023.
\newblock \href {https://doi.org/10.1038/s41467-023-41217-6} {\path{doi:10.1038/s41467-023-41217-6}}.

\bibitem[CDPS04]{Chiribella2004EfficientUse}
G.~Chiribella, G.~M. D'Ariano, P.~Perinotti, and M.~F. Sacchi.
\newblock Efficient use of quantum resources for the transmission of a reference frame.
\newblock {\em Physical Review Letters}, 93(18):180503, 2004.
\newblock \href {https://doi.org/10.1103/PhysRevLett.93.180503} {\path{doi:10.1103/PhysRevLett.93.180503}}.

\bibitem[CDPV13]{Chiribella_2013}
Giulio Chiribella, Giacomo~Mauro D'Ariano, Paolo Perinotti, and Benoit Valiron.
\newblock Quantum computations without definite causal structure.
\newblock {\em Physical Review A}, 88(2):022318, 2013.
\newblock \href {https://doi.org/10.1103/PhysRevA.88.022318} {\path{doi:10.1103/PhysRevA.88.022318}}.

\bibitem[CDVDM08]{cox2008blocks}
Anton Cox, Maud De~Visscher, Stephen Doty, and Paul Martin.
\newblock On the blocks of the walled {Brauer} algebra.
\newblock {\em Journal of Algebra}, 320(1):169--212, 2008.
\newblock \href {https://doi.org/10.1016/j.jalgebra.2008.01.026} {\path{doi:10.1016/j.jalgebra.2008.01.026}}.

\bibitem[CGO{\etalchar{+}}26]{chen_Girardi2026}
Kean Chen, Filippo Girardi, Aadil Oufkir, Nengkun Yu, and Zhicheng Zhang.
\newblock Quantum channel tomography: Optimal bounds and a {H}eisenberg-to-classical phase transition, 2026.
\newblock \href {https://arxiv.org/abs/2604.17369} {\path{arXiv:2604.17369}}.

\bibitem[Che25]{chen2025inverse}
Kean Chen.
\newblock Inverse-free quantum state estimation with {H}eisenberg scaling, 2025.
\newblock \href {https://arxiv.org/abs/2510.25750} {\path{arXiv:2510.25750}}.

\bibitem[Cho75]{CHOI1975285}
Man-Duen Choi.
\newblock Completely positive linear maps on complex matrices.
\newblock {\em Linear Algebra and its Applications}, 10(3):285--290, 1975.
\newblock \href {https://doi.org/10.1016/0024-3795(75)90075-0} {\path{doi:10.1016/0024-3795(75)90075-0}}.

\bibitem[CLO{\etalchar{+}}23]{Chen2023Learnability}
Senrui Chen, Yunchao Liu, Matthew Otten, Alireza Seif, Bill Fefferman, and Liang Jiang.
\newblock The learnability of {Pauli} noise.
\newblock {\em Nature Communications}, 14(1):52, 2023.
\newblock \href {https://doi.org/10.1038/s41467-022-35759-4} {\path{doi:10.1038/s41467-022-35759-4}}.

\bibitem[C{\'S}06]{CollinsSniady06}
Beno{\^i}t Collins and Piotr {\'S}niady.
\newblock Integration with respect to the {Haar} measure on unitary, orthogonal and symplectic group.
\newblock {\em Communications in Mathematical Physics}, 264(3):773--795, 2006.
\newblock \href {https://doi.org/10.1007/s00220-006-1554-3} {\path{doi:10.1007/s00220-006-1554-3}}.

\bibitem[CVdPCH26]{depradenne2026learninghamiltonianslongtimes}
Constantin Cedillo Vayson~de Pradenne, Jordan Cotler, and Hsin-Yuan Huang.
\newblock Learning {Hamiltonians} at long times, 2026.
\newblock \href {https://arxiv.org/abs/2606.05690} {\path{arXiv:2606.05690}}.

\bibitem[CW25]{castaneda2023hamiltonian}
Juan Castaneda and Nathan Wiebe.
\newblock {Hamiltonian} learning via shadow tomography of pseudo-{Choi} states.
\newblock {\em Quantum}, 9:1700, 2025.
\newblock \href {https://doi.org/10.22331/q-2025-04-09-1700} {\path{doi:10.22331/q-2025-04-09-1700}}.

\bibitem[CYZF21]{ChenRobust2021}
Senrui Chen, Wenjun Yu, Pei Zeng, and Steven~T. Flammia.
\newblock Robust shadow estimation.
\newblock {\em PRX Quantum}, 2(3):030348, 2021.
\newblock \href {https://doi.org/10.1103/PRXQuantum.2.030348} {\path{doi:10.1103/PRXQuantum.2.030348}}.

\bibitem[CZ26]{christensen2026learningfermioniclinearoptics}
Aria Christensen and Andrew Zhao.
\newblock Learning fermionic linear optics with {H}eisenberg scaling and physical operations, 2026.
\newblock \href {https://arxiv.org/abs/2602.05058} {\path{arXiv:2602.05058}}.

\bibitem[DFN{\etalchar{+}}24]{DongMIP24}
Yangjing Dong, Honghao Fu, Anand Natarajan, Minglong Qin, Haochen Xu, and Penghui Yao.
\newblock The computational advantage of {$\mathsf{MIP}^*$} vanishes in the presence of noise.
\newblock In Rahul Santhanam, editor, {\em 39th Computational Complexity Conference ({CCC} 2024)}, volume 300 of {\em Leibniz International Proceedings in Informatics ({LIPIcs})}, pages 30:1--30:71. Schloss Dagstuhl -- Leibniz-Zentrum f{\"u}r Informatik, 2024.
\newblock \href {https://doi.org/10.4230/LIPIcs.CCC.2024.30} {\path{doi:10.4230/LIPIcs.CCC.2024.30}}.

\bibitem[DHT25]{du2025efficient}
Yuxuan Du, Min-Hsiu Hsieh, and Dacheng Tao.
\newblock Efficient learning for linear properties of bounded-gate quantum circuits.
\newblock {\em Nature Communications}, 16(1):3790, 2025.
\newblock \href {https://doi.org/10.1038/s41467-025-59198-z} {\path{doi:10.1038/s41467-025-59198-z}}.

\bibitem[DOS24]{Dutkiewicz2024advantageofquantum}
Alicja Dutkiewicz, Thomas~E. O'Brien, and Thomas Schuster.
\newblock The advantage of quantum control in many-body {H}amiltonian learning.
\newblock {\em Quantum}, 8:1537, 2024.
\newblock \href {https://doi.org/10.22331/q-2024-11-26-1537} {\path{doi:10.22331/q-2024-11-26-1537}}.

\bibitem[EFH{\etalchar{+}}23]{elben2023randomized}
Andreas Elben, Steven~T. Flammia, Hsin-Yuan Huang, Richard Kueng, John Preskill, Beno{\^\i}t Vermersch, and Peter Zoller.
\newblock The randomized measurement toolbox.
\newblock {\em Nature Reviews Physics}, 5(1):9--24, 2023.
\newblock \href {https://doi.org/10.1038/s42254-022-00535-2} {\path{doi:10.1038/s42254-022-00535-2}}.

\bibitem[FH04]{Fulton2004}
William Fulton and Joe Harris.
\newblock {\em Representation Theory: A First Course}, volume 129 of {\em Graduate Texts in Mathematics}.
\newblock Springer New York, 2004.
\newblock \href {https://doi.org/10.1007/978-1-4612-0979-9} {\path{doi:10.1007/978-1-4612-0979-9}}.

\bibitem[GCC24]{gu2024practical}
Andi Gu, {\L}ukasz Cincio, and Patrick~J. Coles.
\newblock Practical {Hamiltonian} learning with unitary dynamics and {Gibbs} states.
\newblock {\em Nature Communications}, 15(1):312, 2024.
\newblock \href {https://doi.org/10.1038/s41467-023-44008-1} {\path{doi:10.1038/s41467-023-44008-1}}.

\bibitem[GJ14]{Gutoski_2014}
Gus Gutoski and Nathaniel Johnston.
\newblock Process tomography for unitary quantum channels.
\newblock {\em Journal of Mathematical Physics}, 55(3):032201, 2014.
\newblock \href {https://doi.org/10.1063/1.4867625} {\path{doi:10.1063/1.4867625}}.

\bibitem[GKP94]{concrete_math}
Ronald~L. Graham, Donald~E. Knuth, and Oren Patashnik.
\newblock {\em Concrete Mathematics: A Foundation for Computer Science}.
\newblock Addison-Wesley, 2nd edition, 1994.

\bibitem[GL25]{grewal2026efficientlearningstructuredquantum}
Sabee Grewal and Daniel Liang.
\newblock Efficient learning of structured quantum circuits via {P}auli dimensionality and sparsity, 2025.
\newblock \href {https://arxiv.org/abs/2510.00168} {\path{arXiv:2510.00168}}.

\bibitem[GMZ{\etalchar{+}}25]{girardi2025random}
Filippo Girardi, Francesco~Anna Mele, Haimeng Zhao, Marco Fanizza, and Ludovico Lami.
\newblock Random {Stinespring} superchannel: Converting channel queries into dilation isometry queries, 2025.
\newblock \href {https://arxiv.org/abs/2512.20599} {\path{arXiv:2512.20599}}.

\bibitem[GPS24]{Grier_2024}
Daniel Grier, Hakop Pashayan, and Luke Schaeffer.
\newblock Sample-optimal classical shadows for pure states.
\newblock {\em Quantum}, 8:1373, 2024.
\newblock \href {https://doi.org/10.22331/q-2024-06-17-1373} {\path{doi:10.22331/q-2024-06-17-1373}}.

\bibitem[GR08]{list_decoding_2008}
Venkatesan Guruswami and Atri Rudra.
\newblock Explicit codes achieving list decoding capacity: Error-correction with optimal redundancy.
\newblock {\em IEEE Transactions on Information Theory}, 54(1):135--150, 2008.
\newblock \href {https://doi.org/10.1109/TIT.2007.911222} {\path{doi:10.1109/TIT.2007.911222}}.

\bibitem[Gri25]{grinko2025mixed}
Dmitry Grinko.
\newblock {\em Mixed {Schur--Weyl} Duality in Quantum Information}.
\newblock PhD thesis, Universiteit van Amsterdam, 2025.
\newblock URL: \url{https://eprints.illc.uva.nl/id/eprint/2346/}.

\bibitem[GSG{\etalchar{+}}23]{Gebhart_2023}
Valentin Gebhart, Raffaele Santagati, Antonio~Andrea Gentile, Erik~M. Gauger, David Craig, Natalia Ares, Leonardo Banchi, Florian Marquardt, Luca Pezz{\`e}, and Cristian Bonato.
\newblock Learning quantum systems.
\newblock {\em Nature Reviews Physics}, 5(3):141--156, 2023.
\newblock \href {https://doi.org/10.1038/s42254-022-00552-1} {\path{doi:10.1038/s42254-022-00552-1}}.

\bibitem[GW09]{goodman2009symmetry}
Roe Goodman and Nolan~R. Wallach.
\newblock {\em Symmetry, Representations, and Invariants}, volume 255 of {\em Graduate Texts in Mathematics}.
\newblock Springer New York, 2009.
\newblock \href {https://doi.org/10.1007/978-0-387-79852-3} {\path{doi:10.1007/978-0-387-79852-3}}.

\bibitem[Har05]{harrow_schur_2005}
Aram~W. Harrow.
\newblock {\em Applications of Coherent Classical Communication and the {Schur} Transform to Quantum Information Theory}.
\newblock PhD thesis, Massachusetts Institute of Technology, 2005.
\newblock URL: \url{https://hdl.handle.net/1721.1/34973}.

\bibitem[Hay98]{Hayashi_1998}
Masahito Hayashi.
\newblock Asymptotic estimation theory for a finite-dimensional pure state model.
\newblock {\em Journal of Physics A: Mathematical and General}, 31(20):4633--4655, 1998.
\newblock \href {https://doi.org/10.1088/0305-4470/31/20/006} {\path{doi:10.1088/0305-4470/31/20/006}}.

\bibitem[Hay17]{Hayashi2017}
Masahito Hayashi.
\newblock {\em A Group Theoretic Approach to Quantum Information}.
\newblock Springer International Publishing, 2017.
\newblock \href {https://doi.org/10.1007/978-3-319-45241-8} {\path{doi:10.1007/978-3-319-45241-8}}.

\bibitem[HBC{\etalchar{+}}22]{Huang_2022}
Hsin-Yuan Huang, Michael Broughton, Jordan Cotler, Sitan Chen, Jerry Li, Masoud Mohseni, Hartmut Neven, Ryan Babbush, Richard Kueng, John Preskill, and Jarrod~R. McClean.
\newblock Quantum advantage in learning from experiments.
\newblock {\em Science}, 376(6598):1182--1186, 2022.
\newblock \href {https://doi.org/10.1126/science.abn7293} {\path{doi:10.1126/science.abn7293}}.

\bibitem[HCP23]{Huang_learningtopredict2023}
Hsin-Yuan Huang, Sitan Chen, and John Preskill.
\newblock Learning to predict arbitrary quantum processes.
\newblock {\em PRX Quantum}, 4(4):040337, 2023.
\newblock \href {https://doi.org/10.1103/PRXQuantum.4.040337} {\path{doi:10.1103/PRXQuantum.4.040337}}.

\bibitem[HCY23]{PhysRevResearch.5.023027}
Hong-Ye Hu, Soonwon Choi, and Yi-Zhuang You.
\newblock Classical shadow tomography with locally scrambled quantum dynamics.
\newblock {\em Physical Review Research}, 5(2):023027, 2023.
\newblock \href {https://doi.org/10.1103/PhysRevResearch.5.023027} {\path{doi:10.1103/PhysRevResearch.5.023027}}.

\bibitem[Hei19]{heil2019introduction}
Christopher Heil.
\newblock {\em Introduction to Real Analysis}, volume 280 of {\em Graduate Texts in Mathematics}.
\newblock Springer Cham, 2019.
\newblock \href {https://doi.org/10.1007/978-3-030-26903-6} {\path{doi:10.1007/978-3-030-26903-6}}.

\bibitem[HHJ{\etalchar{+}}17]{HHJ17}
Jeongwan Haah, Aram~W. Harrow, Zhengfeng Ji, Xiaodi Wu, and Nengkun Yu.
\newblock Sample-optimal tomography of quantum states.
\newblock {\em IEEE Transactions on Information Theory}, 63(9):5628--5641, 2017.
\newblock \href {https://doi.org/10.1109/TIT.2017.2719044} {\path{doi:10.1109/TIT.2017.2719044}}.

\bibitem[HKOT23]{HKOT23}
Jeongwan Haah, Robin Kothari, Ryan O'Donnell, and Ewin Tang.
\newblock Query-optimal estimation of unitary channels in diamond distance.
\newblock In {\em 2023 IEEE 64th Annual Symposium on Foundations of Computer Science (FOCS)}, pages 363--390. IEEE, 2023.
\newblock \href {https://doi.org/10.1109/FOCS57990.2023.00028} {\path{doi:10.1109/FOCS57990.2023.00028}}.

\bibitem[HKP20]{Huang_2020}
Hsin-Yuan Huang, Richard Kueng, and John Preskill.
\newblock Predicting many properties of a quantum system from very few measurements.
\newblock {\em Nature Physics}, 16(10):1050--1057, 2020.
\newblock \href {https://doi.org/10.1038/s41567-020-0932-7} {\path{doi:10.1038/s41567-020-0932-7}}.

\bibitem[HKT24]{haah2024learning}
Jeongwan Haah, Robin Kothari, and Ewin Tang.
\newblock Learning quantum {H}amiltonians from high-temperature {G}ibbs states and real-time evolutions.
\newblock {\em Nature Physics}, 20(6):1027--1031, 2024.
\newblock \href {https://doi.org/10.1038/s41567-023-02376-x} {\path{doi:10.1038/s41567-023-02376-x}}.

\bibitem[HLB{\etalchar{+}}24]{Huang_2024}
Hsin-Yuan Huang, Yunchao Liu, Michael Broughton, Isaac Kim, Anurag Anshu, Zeph Landau, and Jarrod~R. McClean.
\newblock Learning shallow quantum circuits.
\newblock In {\em Proceedings of the 56th Annual ACM Symposium on Theory of Computing}, pages 1343--1351. Association for Computing Machinery, 2024.
\newblock \href {https://doi.org/10.1145/3618260.3649722} {\path{doi:10.1145/3618260.3649722}}.

\bibitem[HMG{\etalchar{+}}25]{HuAnsatzFree2025}
Hong-Ye Hu, Muzhou Ma, Weiyuan Gong, Qi~Ye, Yu~Tong, Steven~T. Flammia, and Susanne~F. Yelin.
\newblock Ansatz-free {Hamiltonian} learning with {Heisenberg}-limited scaling.
\newblock {\em PRX Quantum}, 6(4):040315, 2025.
\newblock \href {https://doi.org/10.1103/j7b8-pb77} {\path{doi:10.1103/j7b8-pb77}}.

\bibitem[Hol11]{Holevo2011}
Alexander Holevo.
\newblock {\em Probabilistic and Statistical Aspects of Quantum Theory}, volume~1 of {\em Publications of the Scuola Normale Superiore}.
\newblock Edizioni della Normale, 2011.
\newblock \href {https://doi.org/10.1007/978-88-7642-378-9} {\path{doi:10.1007/978-88-7642-378-9}}.

\bibitem[HTFS23]{huang_2023}
Hsin-Yuan Huang, Yu~Tong, Di~Fang, and Yuan Su.
\newblock Learning many-body {H}amiltonians with {H}eisenberg-limited scaling.
\newblock {\em Physical Review Letters}, 130(20):200403, 2023.
\newblock \href {https://doi.org/10.1103/PhysRevLett.130.200403} {\path{doi:10.1103/PhysRevLett.130.200403}}.

\bibitem[HW23]{PhysRevLett.131.240602}
Jonas Helsen and Michael Walter.
\newblock Thrifty shadow estimation: Reusing quantum circuits and bounding tails.
\newblock {\em Physical Review Letters}, 131(24):240602, 2023.
\newblock \href {https://doi.org/10.1103/PhysRevLett.131.240602} {\path{doi:10.1103/PhysRevLett.131.240602}}.

\bibitem[HXVW19]{Helsen2019}
Jonas Helsen, Xiao Xue, Lieven M.~K. Vandersypen, and Stephanie Wehner.
\newblock A new class of efficient randomized benchmarking protocols.
\newblock {\em npj Quantum Information}, 5(1):71, 2019.
\newblock \href {https://doi.org/10.1038/s41534-019-0182-7} {\path{doi:10.1038/s41534-019-0182-7}}.

\bibitem[HY26]{he2025resourcequantificationprogramminglowdepth}
Entong He and Yuxiang Yang.
\newblock Resource quantification for programming low-depth quantum circuits.
\newblock {\em Quantum}, 10:2166, 2026.
\newblock \href {https://doi.org/10.22331/q-2026-07-20-2166} {\path{doi:10.22331/q-2026-07-20-2166}}.

\bibitem[IBF{\etalchar{+}}20]{innocenti2020supervised}
Luca Innocenti, Leonardo Banchi, Alessandro Ferraro, Sougato Bose, and Mauro Paternostro.
\newblock Supervised learning of time-independent {H}amiltonians for gate design.
\newblock {\em New Journal of Physics}, 22(6):065001, 2020.
\newblock \href {https://doi.org/10.1088/1367-2630/ab8aaf} {\path{doi:10.1088/1367-2630/ab8aaf}}.

\bibitem[Jam72]{Jamiokowski1972}
A.~Jamio{\l}kowski.
\newblock Linear transformations which preserve trace and positive semidefiniteness of operators.
\newblock {\em Reports on Mathematical Physics}, 3(4):275--278, 1972.
\newblock \href {https://doi.org/10.1016/0034-4877(72)90011-0} {\path{doi:10.1016/0034-4877(72)90011-0}}.

\bibitem[JWD{\etalchar{+}}08]{4655455}
Zhengfeng Ji, Guoming Wang, Runyao Duan, Yuan Feng, and Mingsheng Ying.
\newblock Parameter estimation of quantum channels.
\newblock {\em IEEE Transactions on Information Theory}, 54(11):5172--5185, 2008.
\newblock \href {https://doi.org/10.1109/TIT.2008.929940} {\path{doi:10.1109/TIT.2008.929940}}.

\bibitem[Kah07]{Kahn_2007}
Jonas Kahn.
\newblock Fast rate estimation of a unitary operation in $\mathrm{SU}(d)$.
\newblock {\em Physical Review A}, 75(2):022326, 2007.
\newblock \href {https://doi.org/10.1103/PhysRevA.75.022326} {\path{doi:10.1103/PhysRevA.75.022326}}.

\bibitem[Koi89]{koike1989decomposition}
Kazuhiko Koike.
\newblock On the decomposition of tensor products of the representations of the classical groups: By means of the universal characters.
\newblock {\em Advances in Mathematics}, 74(1):57--86, 1989.
\newblock \href {https://doi.org/10.1016/0001-8708(89)90004-2} {\path{doi:10.1016/0001-8708(89)90004-2}}.

\bibitem[Kro19]{Krovi2019efficienthigh}
Hari Krovi.
\newblock An efficient high dimensional quantum {Schur} transform.
\newblock {\em Quantum}, 3:122, 2019.
\newblock \href {https://doi.org/10.22331/q-2019-02-14-122} {\path{doi:10.22331/q-2019-02-14-122}}.

\bibitem[KTCT23]{Kunjummen_2023}
Jonathan Kunjummen, Minh~C. Tran, Daniel Carney, and Jacob~M. Taylor.
\newblock Shadow process tomography of quantum channels.
\newblock {\em Physical Review A}, 107(4):042403, 2023.
\newblock \href {https://doi.org/10.1103/PhysRevA.107.042403} {\path{doi:10.1103/PhysRevA.107.042403}}.

\bibitem[LLC24]{Levy_2024}
Ryan Levy, Di~Luo, and Bryan~K. Clark.
\newblock Classical shadows for quantum process tomography on near-term quantum computers.
\newblock {\em Physical Review Research}, 6(1):013029, 2024.
\newblock \href {https://doi.org/10.1103/PhysRevResearch.6.013029} {\path{doi:10.1103/PhysRevResearch.6.013029}}.

\bibitem[LLY{\etalchar{+}}26]{5khs-7dyz}
Qing Liu, Zihao Li, Xiao Yuan, Huangjun Zhu, and You Zhou.
\newblock Auxiliary-free replica shadows: Efficient estimation of multiple nonlinear quantum properties.
\newblock {\em Physical Review Letters}, 136(10):100602, 2026.
\newblock \href {https://doi.org/10.1103/5khs-7dyz} {\path{doi:10.1103/5khs-7dyz}}.

\bibitem[LYZZ25]{Li_2025}
Zihao Li, Changhao Yi, You Zhou, and Huangjun Zhu.
\newblock Nearly query-optimal classical shadow estimation of unitary channels.
\newblock {\em PRX Quantum}, 6(3):030366, 2025.
\newblock \href {https://doi.org/10.1103/wccm-zys6} {\path{doi:10.1103/wccm-zys6}}.

\bibitem[MAVAV16]{MarienAudenaertVanAcoleyenVerstraete2016}
Micha{\"e}l Mari{\"e}n, Koenraad M.~R. Audenaert, Karel Van~Acoleyen, and Frank Verstraete.
\newblock Entanglement rates and the stability of the area law for the entanglement entropy.
\newblock {\em Communications in Mathematical Physics}, 346(1):35--73, 2016.
\newblock \href {https://doi.org/10.1007/s00220-016-2709-5} {\path{doi:10.1007/s00220-016-2709-5}}.

\bibitem[MB25]{mele2025optimal}
Antonio~Anna Mele and Lennart Bittel.
\newblock Optimal learning of quantum channels in diamond distance, 2025.
\newblock \href {https://arxiv.org/abs/2512.10214} {\path{arXiv:2512.10214}}.

\bibitem[Mel24]{Mele2024introductiontohaar}
Antonio~Anna Mele.
\newblock Introduction to {Haar} measure tools in quantum information: A beginner's tutorial.
\newblock {\em Quantum}, 8:1340, 2024.
\newblock \href {https://doi.org/10.22331/q-2024-05-08-1340} {\path{doi:10.22331/q-2024-05-08-1340}}.

\bibitem[MFPT26]{ma2024learning}
Muzhou Ma, Steven~T. Flammia, John Preskill, and Yu~Tong.
\newblock Learning {$k$}-body {Hamiltonians} via compressed sensing.
\newblock {\em IEEE Transactions on Information Theory}, 72(10):7553--7577, 2026.
\newblock \href {https://doi.org/10.1109/TIT.2026.3720996} {\path{doi:10.1109/TIT.2026.3720996}}.

\bibitem[ML06]{PhysRevLett.97.170501}
M.~Mohseni and D.~A. Lidar.
\newblock Direct characterization of quantum dynamics.
\newblock {\em Physical Review Letters}, 97(17):170501, 2006.
\newblock \href {https://doi.org/10.1103/PhysRevLett.97.170501} {\path{doi:10.1103/PhysRevLett.97.170501}}.

\bibitem[Moo99]{moore1999quantumcircuitsfanoutparity}
Cristopher Moore.
\newblock Quantum circuits: Fanout, parity, and counting, 1999.
\newblock \href {https://arxiv.org/abs/quant-ph/9903046} {\path{arXiv:quant-ph/9903046}}.

\bibitem[NC12]{Nielsen2012}
Michael~A. Nielsen and Isaac~L. Chuang.
\newblock {\em Quantum Computation and Quantum Information}.
\newblock Cambridge University Press, 10th anniversary edition, 2012.
\newblock \href {https://doi.org/10.1017/CBO9780511976667} {\path{doi:10.1017/CBO9780511976667}}.

\bibitem[NGR{\etalchar{+}}21]{Nielsen_2021}
Erik Nielsen, John~King Gamble, Kenneth Rudinger, Travis Scholten, Kevin Young, and Robin Blume-Kohout.
\newblock Gate set tomography.
\newblock {\em Quantum}, 5:557, 2021.
\newblock \href {https://doi.org/10.22331/q-2021-10-05-557} {\path{doi:10.22331/q-2021-10-05-557}}.

\bibitem[NPVY24]{Nadimpalli_2024}
Shivam Nadimpalli, Natalie Parham, Francisca Vasconcelos, and Henry Yuen.
\newblock On the {P}auli spectrum of {$\mathsf{QAC}^0$}.
\newblock In {\em Proceedings of the 56th Annual {ACM} Symposium on Theory of Computing}, pages 1498--1506. ACM, 2024.
\newblock \href {https://doi.org/10.1145/3618260.3649662} {\path{doi:10.1145/3618260.3649662}}.

\bibitem[NVH18]{Nahum_2018}
Adam Nahum, Sagar Vijay, and Jeongwan Haah.
\newblock Operator spreading in random unitary circuits.
\newblock {\em Physical Review X}, 8(2):021014, 2018.
\newblock \href {https://doi.org/10.1103/PhysRevX.8.021014} {\path{doi:10.1103/PhysRevX.8.021014}}.

\bibitem[OKSM24]{OdakeHigher-order2024}
Tatsuki Odake, Hl{\'e}r Kristj{\'a}nsson, Akihito Soeda, and Mio Murao.
\newblock Higher-order quantum transformations of {Hamiltonian} dynamics.
\newblock {\em Physical Review Research}, 6(1):L012063, 2024.
\newblock \href {https://doi.org/10.1103/PhysRevResearch.6.L012063} {\path{doi:10.1103/PhysRevResearch.6.L012063}}.

\bibitem[OW16]{efficient_state_tomography}
Ryan O'Donnell and John Wright.
\newblock Efficient quantum tomography.
\newblock In {\em Proceedings of the Forty-Eighth Annual ACM Symposium on Theory of Computing}, pages 899--912. Association for Computing Machinery, 2016.
\newblock \href {https://doi.org/10.1145/2897518.2897544} {\path{doi:10.1145/2897518.2897544}}.

\bibitem[Pre18]{preskill2018nisq}
John Preskill.
\newblock Quantum computing in the {NISQ} era and beyond.
\newblock {\em Quantum}, 2:79, 2018.
\newblock \href {https://doi.org/10.22331/q-2018-08-06-79} {\path{doi:10.22331/q-2018-08-06-79}}.

\bibitem[PSTW25]{pelecanos2025mixedstatetomographyreduces}
Angelos Pelecanos, Jack Spilecki, Ewin Tang, and John Wright.
\newblock Mixed state tomography reduces to pure state tomography, 2025.
\newblock \href {https://arxiv.org/abs/2511.15806} {\path{arXiv:2511.15806}}.

\bibitem[PSW26]{pelecanos2025debiasedkeylsalgorithmnew}
Angelos Pelecanos, Jack Spilecki, and John Wright.
\newblock The debiased {Keyl}'s algorithm: A new unbiased estimator for full state tomography.
\newblock In {\em Proceedings of the 58th Annual {ACM} Symposium on Theory of Computing}, pages 1266--1277. ACM, 2026.
\newblock \href {https://doi.org/10.1145/3798129.3800837} {\path{doi:10.1145/3798129.3800837}}.

\bibitem[PW27]{Peter1927}
F.~Peter and H.~Weyl.
\newblock {Die Vollst{\"a}ndigkeit der primitiven Darstellungen einer geschlossenen kontinuierlichen Gruppe}.
\newblock {\em Mathematische Annalen}, 97(1):737--755, 1927.
\newblock \href {https://doi.org/10.1007/BF01447892} {\path{doi:10.1007/BF01447892}}.

\bibitem[PZ32]{PaleyZygmund32}
R.~E. A.~C. Paley and A.~Zygmund.
\newblock A note on analytic functions in the unit circle.
\newblock {\em Mathematical Proceedings of the Cambridge Philosophical Society}, 28(3):266--272, 1932.
\newblock \href {https://doi.org/10.1017/S0305004100010112} {\path{doi:10.1017/S0305004100010112}}.

\bibitem[Rot06]{Roth2006}
Ron~M. Roth.
\newblock {\em Introduction to Coding Theory}.
\newblock Cambridge University Press, 2006.
\newblock \href {https://doi.org/10.1017/CBO9780511808968} {\path{doi:10.1017/CBO9780511808968}}.

\bibitem[SBZ19]{Sedl_k_2019}
Michal Sedl{\'a}k, Alessandro Bisio, and M{\'a}rio Ziman.
\newblock Optimal probabilistic storage and retrieval of unitary channels.
\newblock {\em Physical Review Letters}, 122(17):170502, 2019.
\newblock \href {https://doi.org/10.1103/PhysRevLett.122.170502} {\path{doi:10.1103/PhysRevLett.122.170502}}.

\bibitem[Sco08]{scott2008optimizing}
Andrew~J. Scott.
\newblock Optimizing quantum process tomography with unitary 2-designs.
\newblock {\em Journal of Physics A: Mathematical and Theoretical}, 41(5):055308, 2008.
\newblock \href {https://doi.org/10.1088/1751-8113/41/5/055308} {\path{doi:10.1088/1751-8113/41/5/055308}}.

\bibitem[SFMD{\etalchar{+}}24]{stilck2024efficient}
Daniel Stilck~Fran{\c{c}}a, Liubov~A. Markovich, V.~V. Dobrovitski, Albert~H. Werner, and Johannes Borregaard.
\newblock Efficient and robust estimation of many-qubit {Hamiltonians}.
\newblock {\em Nature Communications}, 15(1):311, 2024.
\newblock \href {https://doi.org/10.1038/s41467-023-44012-5} {\path{doi:10.1038/s41467-023-44012-5}}.

\bibitem[SHH25]{Schuster2025}
Thomas Schuster, Jonas Haferkamp, and Hsin-Yuan Huang.
\newblock Random unitaries in extremely low depth.
\newblock {\em Science}, 389(6755):92--96, 2025.
\newblock \href {https://doi.org/10.1126/science.adv8590} {\path{doi:10.1126/science.adv8590}}.

\bibitem[SSW25]{scharnhorst2025optimallowerboundsquantum}
Thilo Scharnhorst, Jack Spilecki, and John Wright.
\newblock Optimal lower bounds for quantum state tomography, 2025.
\newblock \href {https://arxiv.org/abs/2510.07699} {\path{arXiv:2510.07699}}.

\bibitem[ST25]{sinha2025improved}
Savar~D. Sinha and Yu~Tong.
\newblock Improved {Hamiltonian} learning and sparsity testing through {Bell} sampling, 2025.
\newblock \href {https://arxiv.org/abs/2509.07937} {\path{arXiv:2509.07937}}.

\bibitem[SY23]{she_et_al:LIPIcs.ITCS.2023.96}
Adrian She and Henry Yuen.
\newblock Unitary property testing lower bounds by polynomials.
\newblock In Yael Tauman~Kalai, editor, {\em 14th Innovations in Theoretical Computer Science Conference ({ITCS} 2023)}, volume 251 of {\em Leibniz International Proceedings in Informatics ({LIPIcs})}, pages 96:1--96:17. Schloss Dagstuhl -- Leibniz-Zentrum f{\"u}r Informatik, 2023.
\newblock \href {https://doi.org/10.4230/LIPIcs.ITCS.2023.96} {\path{doi:10.4230/LIPIcs.ITCS.2023.96}}.

\bibitem[vACGN23]{doi:10.1137/1.9781611977554.ch47}
Joran van Apeldoorn, Arjan Cornelissen, Andr{\'a}s Gily{\'e}n, and Giacomo Nannicini.
\newblock Quantum tomography using state-preparation unitaries.
\newblock In {\em Proceedings of the 2023 Annual {ACM-SIAM} Symposium on Discrete Algorithms ({SODA})}, pages 1265--1318. Society for Industrial and Applied Mathematics, 2023.
\newblock \href {https://doi.org/10.1137/1.9781611977554.ch47} {\path{doi:10.1137/1.9781611977554.ch47}}.

\bibitem[VAMV13]{VanAcoleyenMarienVerstraete2013}
Karel Van~Acoleyen, Micha{\"e}l Mari{\"e}n, and Frank Verstraete.
\newblock Entanglement rates and area laws.
\newblock {\em Physical Review Letters}, 111(17):170501, 2013.
\newblock \href {https://doi.org/10.1103/PhysRevLett.111.170501} {\path{doi:10.1103/PhysRevLett.111.170501}}.

\bibitem[VH25]{pmlr-v291-vasconcelos25a}
Francisca Vasconcelos and Hsin-Yuan Huang.
\newblock Learning shallow quantum circuits with many-qubit gates.
\newblock In Nika Haghtalab and Ankur Moitra, editors, {\em Proceedings of the Thirty-Eighth Conference on Learning Theory}, volume 291 of {\em Proceedings of Machine Learning Research}, pages 5553--5604. PMLR, 2025.
\newblock URL: \url{https://proceedings.mlr.press/v291/vasconcelos25a.html}.

\bibitem[VRS{\etalchar{+}}24]{PRXQuantum.5.010352}
Beno{\^\i}t Vermersch, Aniket Rath, Bharathan Sundar, Cyril Branciard, John Preskill, and Andreas Elben.
\newblock Enhanced estimation of quantum properties with common randomized measurements.
\newblock {\em PRX Quantum}, 5(1):010352, 2024.
\newblock \href {https://doi.org/10.1103/PRXQuantum.5.010352} {\path{doi:10.1103/PRXQuantum.5.010352}}.

\bibitem[Wat18]{Watrous_2018}
John Watrous.
\newblock {\em The Theory of Quantum Information}.
\newblock Cambridge University Press, 2018.
\newblock \href {https://doi.org/10.1017/9781316848142} {\path{doi:10.1017/9781316848142}}.

\bibitem[Wri16]{wright2016learn}
John Wright.
\newblock {\em How to Learn a Quantum State}.
\newblock PhD thesis, Carnegie Mellon University, 2016.
\newblock URL: \url{https://csd.cs.cmu.edu/academics/doctoral/degrees-conferred/john-wright}.

\bibitem[WSS{\etalchar{+}}26]{west2026classicalshadowsarbitrarygroup}
Maxwell West, Frederic Sauvage, Aniruddha Sen, Roy Forestano, David Wierichs, Nathan Killoran, Dmitry Grinko, M.~Cerezo, and Martin Larocca.
\newblock Classical shadows with arbitrary group representations, 2026.
\newblock \href {https://arxiv.org/abs/2604.01429} {\path{arXiv:2604.01429}}.

\bibitem[YKS{\etalchar{+}}26]{Yoshida_2026}
Satoshi Yoshida, Yuki Koizumi, Micha{\l} Studzi{\'n}ski, Marco~T{\'u}lio Quintino, and Mio Murao.
\newblock One-to-one correspondence between deterministic port-based teleportation and unitary estimation.
\newblock {\em IEEE Transactions on Information Theory}, 72(4):2358--2377, 2026.
\newblock \href {https://doi.org/10.1109/TIT.2026.3658543} {\path{doi:10.1109/TIT.2026.3658543}}.

\bibitem[YMR{\etalchar{+}}22]{Yang_2022}
Yuxiang Yang, Yin Mo, Joseph~M. Renes, Giulio Chiribella, and Mischa~P. Woods.
\newblock Optimal universal quantum error correction via bounded reference frames.
\newblock {\em Physical Review Research}, 4(2):023107, 2022.
\newblock \href {https://doi.org/10.1103/PhysRevResearch.4.023107} {\path{doi:10.1103/PhysRevResearch.4.023107}}.

\bibitem[YRC20]{Yang_2020}
Yuxiang Yang, Renato Renner, and Giulio Chiribella.
\newblock Optimal universal programming of unitary gates.
\newblock {\em Physical Review Letters}, 125(21):210501, 2020.
\newblock \href {https://doi.org/10.1103/PhysRevLett.125.210501} {\path{doi:10.1103/PhysRevLett.125.210501}}.

\bibitem[YSHY23]{Yu2023robustefficient}
Wenjun Yu, Jinzhao Sun, Zeyao Han, and Xiao Yuan.
\newblock Robust and efficient {H}amiltonian learning.
\newblock {\em Quantum}, 7:1045, 2023.
\newblock \href {https://doi.org/10.22331/q-2023-06-29-1045} {\path{doi:10.22331/q-2023-06-29-1045}}.

\bibitem[YYM25]{yoshida2025asymptoticallyoptimalunitaryestimation}
Satoshi Yoshida, Hironobu Yoshida, and Mio Murao.
\newblock Asymptotically optimal unitary estimation in {$\mathrm{SU}(3)$} by the analysis of graph {Laplacian}, 2025.
\newblock \href {https://arxiv.org/abs/2509.20608} {\path{arXiv:2509.20608}}.

\bibitem[Zha25]{zhao2024learning}
Andrew Zhao.
\newblock Learning the structure of any {Hamiltonian} from minimal assumptions.
\newblock In {\em Proceedings of the 57th Annual ACM Symposium on Theory of Computing}, pages 1201--1211. Association for Computing Machinery, 2025.
\newblock \href {https://doi.org/10.1145/3717823.3718115} {\path{doi:10.1145/3717823.3718115}}.

\bibitem[ZL23]{Zhou2023perform}
You Zhou and Qing Liu.
\newblock Performance analysis of multi-shot shadow estimation.
\newblock {\em Quantum}, 7:1044, 2023.
\newblock \href {https://doi.org/10.22331/q-2023-06-29-1044} {\path{doi:10.22331/q-2023-06-29-1044}}.

\bibitem[ZLK{\etalchar{+}}24]{Zhao_2024}
Haimeng Zhao, Laura Lewis, Ishaan Kannan, Yihui Quek, Hsin-Yuan Huang, and Matthias~C. Caro.
\newblock Learning quantum states and unitaries of bounded gate complexity.
\newblock {\em PRX Quantum}, 5(4):040306, 2024.
\newblock \href {https://doi.org/10.1103/PRXQuantum.5.040306} {\path{doi:10.1103/PRXQuantum.5.040306}}.

\end{thebibliography}

\appendix

\section{Preliminaries}
\subsection{Notations and miscellaneous}
We use the regular abbreviation r.v. for ``random variable''. Besides, we will use several notational conventions.

We use asymptotic notations $\bigO, \widetilde{\bigO}, \Omega, \widetilde{\Omega}, o, \omega, \Theta, \widetilde{\Theta}, \simeq$ from \cite{concrete_math}. $\log$ denotes base-$e$ logarithm, and $\poly(x_1, \dots, x_k) = \bigcup_{c_1, \dots, c_k \in \N} \bigo{x_1^{c_1} \cdots x_k^{c_k}}$. The imaginary unit $\I := \sqrt{-1}$. For any complex number $c = \Re c + \I \Im c \in \C$, its complex conjugate $\overline{c} = \Re c - \I \Im c$. By default, $d = 2^n$ when appearing concurrently.

For a linear operator $X$, we use $\Ket{X} = \sum_{i, j} X_{i, j} \ket{i} \ket{j}$ to denote its vectorization. We use $\Omega_{AB} = \Ket{\openone} \Bra{\openone}_{AB} \otimes \openone_{\overline{AB}}$ and $\F_{AB}$ to denote the unnormalized maximally entangled state and the (local) swap operator on the bipartite subspace $AB$, embedded in the full Hilbert space. Specifically, $\Omega_{AB}$ satisfies the ricochet property $(X_A \otimes I_B) \Omega_{AB} = (I_A \otimes X_B^T) \Omega_{AB}$ for equi-dimensional spaces $A$ and $B$. $\Phi_{\mathcal{E}} = (\mathcal{E} \otimes \mathcal{I})(\Ket{\openone} \Bra{\openone})$ denotes the Choi operator of a channel $\mathcal{E}$. $\Tr_A[X_{AB}]$ denotes the partial trace over subspace $A$ for a bipartite operator $X_{AB}$. 
We denote by $\dens(\mathcal{H})$ the set of quantum states on space $\mathcal{H}$, and $\mathfrak{P}(\mathcal{H}) \subset \dens(\mathcal{H})$ the set of pure states on $\mathcal{H}$. 
For any pure state $\ket{\psi} \in \mathfrak{P}(\mathcal{H})$, we use $\psi$ to represent its density matrix $\ket{\psi} \bra{\psi}$. We use $\mathcal{L}(\mathcal{A}, \mathcal{B})$ to represent the set of linear operators from the space $\mathcal{A}$ to space $\mathcal{B}$, and it is abbreviated to $\mathcal{L}(\mathcal{A})$ when $\mathcal{A} = \mathcal{B}$. For an operator $X$, we use $X^T$, $X^*$, and $X^\dagger$ to denote its transpose, conjugate, and conjugate transpose, respectively. When $X$ is unitary, we use the calligraphic $\mathcal{X}(\cdot)$ to denote the corresponding unitary channel $X(\cdot)X^\dagger$. The boldface symbols ($\mathbf{p}$, $\mathbf{q}$, \etc) denote vectors, and $e_i$ denotes the $i$-th standard basis vector. $\|\cdot\|_2$ and $\|\cdot\|_{\infty}$ stands for the $\ell_2$- and $\ell_{\infty}$-norm of vectors, respectively. For $n \in \N$, we use $[n]$ to denote the index set $\{1, 2, \dots, n\}$. For two objects $ x$ and $ y$, $\delta_{x, y}$ denotes the Kronecker delta. We write $X \succeq 0$ for a Hermitian $X$ if $X$ is positive semi-definite. For $X, Y \in \mathcal{L}(\mathcal{H})$, the commutator/anti-commutator $[X, Y] = XY-YX$/$\{X, Y\} = XY + YX$.

We use $\U(d)$ and $\mathfrak{S}_n$ to denote the $d$-dimensional unitary group and the symmetric group on $[n]$. We write irrep(s) as shorthand for irreducible representation(s), which is defined later in Definition \ref{def:irreps}. Since we are mainly working with the representation theory of the unitary group and the symmetric group, for a (mixed) Young diagram $\lambda$, we use $d_{\lambda} = \dim W_{\lambda}$ to denote the dimension of the Weyl module, and $s_{\lambda} = \dim(\lambda) = \dim S_{\lambda}$ to denote the dimension of the Specht module. A minor caveat is that the representation theory we discuss is originally dedicated to the general linear group $\gl(d, \mathbb{C})$, but adapted to $\U(d)$ via Weyl's unitarian trick \cite[Section 3.3.4]{goodman2009symmetry}. For a vector space $V$, we write $V^*$ for its dual space, and $\vee^k V$, $\wedge^k V$ for the symmetric and anti-symmetric subspace of $k$-fold product space $V^{\otimes k}$. For a group $G$, the complex group algebra is the collection of linear combinations of group elements over the complex number field $\C[G] = \{ \sum_{g \in G} c_g g: c_g \in \C \}$.

\begin{fact}
\label{fact:taylor_series}
    According to Taylor series, $\cos(x) = 1 - \frac{1}{2} x^2 + \frac{1}{24} x^4 - \bigo{x^6}$, $\sin(x) = x - \frac{1}{6} x^3 + \bigo{x^5}$, $2/\sin(x) - 1/\sin(2x) = \frac{3}{2x} - \frac{7x^3}{60} - \frac{31x^5}{504} - \mathcal{O}(x^7) $, $\log(1 + x) = x - \frac{x^2}{2} + \bigo{x^3}$ and $e^{x} = 1 + x + \frac{1}{2}x^2 + \bigo{x^3}$ for small $x$.
\end{fact}

\begin{fact}[Folklore]
\label{fact:triangular_inequality_application}
    For a finite collection of real numbers $\{ a_{i, j} \}_{i, j}$, it holds that $\prod_{j}\left( \sum_i a_{i, j} \right) \leq \prod_{j} \left( \sum_i |a_{i, j}| \right)$. Moreover, for real coefficients $\{x_i\}_i$, $\{y_j\}_j$, it holds that $\sum_{i, j} a_{i, j} x_i y_j \leq (\sum_i |x_i|)(\sum_j |y_j|) \cdot \max_{i, j} |a_{i, j}|$.
\end{fact}

\subsection{Basics of quantum information}
\begin{definition}[Schatten-$p$ norm and H\"older's inequality, \cite{Watrous_2018}]
\label{def:schatten_norm_and_holder}
    For any bounded operator $X \in \mathcal{L}(\mathcal{H})$, and $p \in [1, \infty]$, its Schatten-$p$ norm is given by
    $
    \left\| X \right\|_p = \left( \Tr\left[ |X|^{p} \right] \right)^{1/p} = \left(  \sum_{j} \mathfrak{s}_j^{p}\right)^{1/p}
    $,
    where $\{\mathfrak{s}_j\}_j$ are the singular values of $X$. In particular, $p = 1$ yields the trace norm, $p = 2$ yields the Frobenius norm $\|\cdot\|_2 =: \|\cdot\|_F$, and $p = \infty$ yields the operator norm. For $1 \leq p < q \leq \infty$, $\|X\|_{q} \leq \|X\|_{p} \leq (\rank X )^{1/p - 1/q} \|X\|_q$. The Schatten-$p$ norm induces a valid distance metric $\|A - B\|_{p}$ for two operators $A, B \in \mathcal{L}(\mathcal{H})$.
    
    The (tracial) H\"older's inequality says that for any $A, B \in \mathcal{L}(\mathcal{H})$ and $p, q, \in [1, \infty]$ subject to $p^{-1} + q^{-1} = 1$, it holds that 
    $$
    \left| \Tr\left[A^\dagger B \right] \right| \leq \|A\|_{p} \|B\|_{q}.
    $$
    Specifically, when $p = q = 2$, this gives the Cauchy-Schwarz inequality.
\end{definition}

\begin{corollary}[Adapted from {\cite[Lemma D4]{Li_2025}} for detraced states]
\label{corollary:inner_product_terms_upper_bound}
    For $O \in \mathcal{L}(\mathcal{H})$ and a state $\rho \in \dens(\mathcal{H})$ where $d = \dim \mathcal{H}$. Denote $\rho_0 = \rho  - \openone / d$ as the detraced state. If $O \in \obs(\B)$, $\Tr[\rho^2] \leq \P$, the H\"older's inequality [cf. Definition \ref{def:schatten_norm_and_holder}] gives
    $$
    \begin{aligned}
        &\Tr[O \rho_0]^2 = \Tr[O \rho]^2 \leq \min\left\{ \|O\|_{\infty}^2 \|\rho \|_1^2, \|O\|_F^2 \|\rho \|_F^2 \right\} \leq \min\left\{ 1, \B\P \right\}, \\
        &\Tr[O^2 \rho_0^2] \leq \|O^2\|_{\infty} \|\rho_0^2\|_1 = \|\rho_0\|_F^2 \leq \P, \quad \Tr[O \rho_0 O \rho_0] \leq \|\rho_0 O\|_F \|O \rho_0\|_F \leq \|\rho_0\|_F^2 \leq \P, \\
        &\Tr[O^2 \rho_0] \leq \Tr[O^2 \rho] \leq \min\left\{ \left\|O^2 \right\|_{\infty} \left\| \rho \right\|_1, \left\| O^2 \right\|_F \left\|\rho  \right\|_F \right\} \leq \min\left\{1, \sqrt{\B\P}  \right\}.
    \end{aligned}
    $$
\end{corollary}

\begin{lemma}[{Holevo-Helstrom, \cite[Theorem 3.4]{Watrous_2018}}]
\label{lemma:holevo_helstrom}
    For two density operators $\rho_0, \rho_1 \in \dens(\mathcal{H})$, and $\lambda \in [0, 1]$, for all dichotomous POVM $\mathscr{M} = \{\Pi_0, \Pi_1\}$, it holds that
    $$
    \lambda \cdot \Tr[\Pi_0 \rho_0] + (1 - \lambda) \cdot \Tr[\Pi_1 \rho_1] \leq \frac{1}{2} + \frac{1}{2} \left\|\lambda \rho_0 - (1 - \lambda) \rho_1 \right\|_1.
    $$
\end{lemma}

\begin{lemma}
\label{lemma:Helstrom2}
For arbitrary pure states $\psi_1, \psi_2 \in \dens(\mathcal{H})$ with $\psi_1\neq  \psi_2$, we have 
\[
\left\|\psi_1-\psi_2 \right\|_1
=
2\,\Tr\left[\varphi_+ \left(\psi_1-\psi_2 \right)\right],
\]
where $\varphi_+ = \ket{\varphi_+} \bra{\varphi_+}$ is the rank-one projector onto the positive eigenspace of $\psi_1-\psi_2$, equivalently, the maximization in the variational definition of the trace distance between pure states is saturated at a pure state.
\end{lemma}

\begin{proof}[Proof of \lref{lemma:Helstrom2}]
Since $\psi_1-\psi_2$ is Hermitian, has trace zero, and satisfies
$\rank(\psi_1-\psi_2)\leq 2$, 
its nonzero eigenvalues must be of the form $+\lambda$ and $-\lambda$ for some $\lambda> 0$. Therefore,
\[
\psi_1-\psi_2=\lambda \varphi_+ - \lambda \varphi_-,
\]
where $\varphi_+$ and $\varphi_-$ are the projectors onto the positive and negative eigenspaces, respectively. Hence
\[
\|\psi_1-\psi_2\|_1 = |\lambda| + |-\lambda| = 2\lambda,
\]
and
\[
\Tr\left[\varphi_+ \left(\psi_1-\psi_2 \right) \right]
=
\lambda.
\]
Combining these two equalities confirms the statement.
\end{proof}

\begin{definition}[\cite{Nielsen2012,Watrous_2018}]
\label{def:diamond_norm}
    The diamond norm of a supermap $\mathcal{E} \in \mathcal{L}(\mathcal{L}(\mathcal{H}_A), \mathcal{L}(\mathcal{H}_B))$ is given by
    $$
    \left\|\mathcal{E} \right\|_{\diamond} = \max_{\mathcal{H}_{R}} \max_{X \in \mathcal{L}(\mathcal{H}_A \otimes \mathcal{H}_R) } \frac{\left\|(\mathcal{E} \otimes \mathcal{I}_R)(X)  \right\|_1}{\left\|X\right\|_1}.
    $$
\end{definition}

\begin{lemma}[Adapted from {\cite[Theorem 3.55]{Watrous_2018}}]
\label{lemma:saturability_diamond_norm_Krank_1}
    Let $\mathcal{H}_0, \mathcal{H}_1$ be two spaces with $\dim \mathcal{H}_0 \leq \dim \mathcal{H}_1$, and define channels $\mathcal{V}_0, \mathcal{V}_1$ as $\forall \rho \in \dens(\mathcal{H}),~\mathcal{V}_0(\rho) = V_0 \rho V_0^\dagger, \mathcal{V}_1(\rho) = V_1 \rho V_1^\dagger$ for two isometries $V_0, V_1 \in \mathcal{L}(\mathcal{H}_0, \mathcal{H}_1)$. Then for every $\lambda \in [0, 1]$, there exists a pure state $\ket{\psi} \in \mathcal{H}_0$ such that
    $$
    \left\| \lambda \mathcal{V}_0(\psi) - (1 - \lambda) \mathcal{V}_1(\psi) \right\|_1 = \left\| \lambda \mathcal{V}_0 - (1 - \lambda) \mathcal{V}_1  \right\|_{\diamond}.
    $$
\end{lemma}

\begin{lemma}[{\cite[Section 3.3]{Watrous_2018}}]
\label{fact:diamond_norm_contractivity}
    For any linear super-operator $\mathcal{E} \in \mathcal{L}(\mathcal{L}(\mathcal{H}))$ and any linear operator $X \in \mathcal{L}(\mathcal{H})$, it holds that $\|\mathcal{E}(X)\|_{1} \leq \left\|\mathcal{E}\right\|_{\diamond} \|X\|_1$.
\end{lemma}

\begin{definition}[\cite{Nielsen2012}]
    The $n$-fold tensors of Pauli operators, or the set of $n$-qubit Pauli operators $\mathsf{P}_n = \left\{ \bigotimes_{j=1}^n \sigma_{\mathrm{p}_j}: \mathrm{p}_j\in \{0,1,2,3\} \right\}$, where
    $$
   \sigma_0 = \openone_2 = \begin{pmatrix}
        1 & 0 \\ 0 & 1
    \end{pmatrix},~~\sigma_1 = \sigma_x = \begin{pmatrix}
        0 & 1 \\ 1 & 0
    \end{pmatrix},~~\sigma_2 = \sigma_y = \begin{pmatrix}
        0 & -\I \\ \I & 0
    \end{pmatrix},~~\sigma_3 = \sigma_z = \begin{pmatrix}
        1 & 0 \\ 0 & -1
    \end{pmatrix}
    $$
    forms an orthonormal basis for $\mathcal{L}(\C^{2^n})$ under the normalized Hilbert-Schmidt innder product $\left<X, Y\right>_{\mathrm{HS}} = \frac{1}{2^n} \Tr[X^\dagger Y]$. For notational brevity, we use a string $\bfp \in \{0, 1, 2, 3\}^n$ to identify the operators in $\mathsf{P}_n$, via
    $$
    \forall\,\bfp \in \{0, 1, 2, 3\}^n, \quad \sigma_{\bfp} := \bigotimes_{j=1}^n \sigma_{\mathrm{p}_j}.
    $$
\end{definition}

\subsection{Representation theory}

\begin{definition}[Basics of group representation \cite{goodman2009symmetry, efficient_state_tomography}]
    Let $G$ be a group. A (complex, unitary, finite-dimensional) representation of $G$ is a tuple $(\lambda, V)$ where $V$ is a finite-dimensional complex vector space $V$, and $\lambda: G \to \U(V)$ is a group homomorphism. We denote the dimension of this representation by $d_{\lambda} = \dim V_{\lambda}$. We write $\chi_{\lambda}(g) = \Tr[\lambda(g)]$ for every $g \in G$ as the character of representation $\lambda$. For two representations $(\lambda, V_{\lambda}), (\mu, V_{\mu})$ of $G$, an intertwining map $T$ is a map $T: V_{\lambda} \to V_{\mu}$ such that $T \circ \lambda(g) = \mu(g) \circ T$ for all $g \in G$. Moreover, if $T$ is invertible, \ie, $\mu(g) = T \circ \lambda(g) \circ T^{-1}$, then we say $\lambda$ and $\mu$ are isomorphic, or notationally $\lambda \cong \mu$. We can also define the dual $(\lambda^*, V_{\lambda}^*)$ of the representation $(\lambda, V_{\lambda})$, where $\lambda^*(g) = \lambda(g^{-1})^T$ for all $g \in G$, and $\chi_{\lambda^*}(g) = \overline{\chi_{\lambda}(g)}$.
\end{definition}

\begin{lemma}[\cite{heil2019introduction}]
\label{lemma:haar_measure}
    For every compact Lie group $G$, there exists a unique left- and right-invariant probability measure $\dd g$ on it, subject to $\dd(hg) = \dd(gh) = \dd g$ for any $h \in G$.
\end{lemma}

\begin{definition}[Irreducible representations \cite{goodman2009symmetry, efficient_state_tomography}]
\label{def:irreps}
    Let $(\lambda, V_{\lambda})$ be the representation of a group $G$, a subspace $W$ is called an invariant subspace of $V_{\lambda}$ if it is invariant under the action induced by $\lambda$, \ie, $\forall g \in G$, $\lambda(g) W \subseteq W$. Such an invariant subspace $W$ is trivial if $W = \{0\}$ or $W = V_{\lambda}$. If $V_{\lambda}$ has a non-trivial invariant subspace, we say that it is reducible, and irreducible otherwise. Correspondingly, the representation is called reducible/irreducible. We denote $\widehat{G}$ as the set of equivalence classes of all irreps of $G$.
\end{definition}

\begin{definition}
\label{def:class_function}
    A class function on a group $G$ is a function $f$ that is invariant under conjugacy over $G$, that is, $\forall g, h \in G$, $f(g) = f(h g h^{-1})$. Moreover, if $f$ is a scalar function and $G$ is a compact Lie group, we denote $\ell^2(G)$ as the set of $f$ subject to $\|f\|_2 := ( \int_G |f(g)|^{2} \dd g )^{1/2} < \infty$. Note that $\ell^2(G)$ is a Hilbert space equipped with an inner product $\left< f_1, f_2 \right> = \int_G f_1(g) \overline{f_2(g)} \dd g$.
\end{definition}

\begin{lemma}[Peter-Weyl, \cite{Peter1927}]
\label{lemma:peter_weyl}
    Let $G$ be a compact Lie group, then the irreducible characters of $G$ generate a dense subspace of the space of continuous class functions on $G$, and for any class function $f \in \ell^2(G)$,
    $$
    f = \sum_{\lambda \in \widehat{G} } \left<f, \chi_{\lambda}\right> \chi_{\lambda}.
    $$
\end{lemma}

\begin{lemma}[Schur's lemma, generalized \cite{goodman2009symmetry}] \label{lemma:schur}
    For a compact Lie group $G$, and a map $T: V_{\lambda} \to V_{\mu}$ where $(\lambda, V_{\lambda})$ and $(\mu, V_{\mu})$ are two representations of $G$. Suppose $T$ is $G$-equivariant, \ie, $\forall g \in G$, $T \circ \lambda(g) = \mu(g) \circ T$. Then $T = c \cdot \openone_{V_{\lambda}}$ for some $c \in \mathbb{C}$ if $\lambda \cong \mu$, and $T = 0$ otherwise. Let $\dd g$ be the normalized Haar measure on $G$, and denote the right translation on $G$ as $\mathcal{R}(g)$, then we can formulate the isotypic projector onto the $\lambda$-subspace for any $\lambda \in \widehat{G}$ by a $G$-equivariant operator:
    \begin{equation}
    \label{eqn:canonical_projection}
    \Pi_\lambda = \dim V_{\lambda} \int_G \overline{\chi_{\lambda}(g)} \mathcal{R}(g) \dd g.
    \end{equation}
    The operation $\Pi_{\lambda}$ being a projection is a direct consequence of \lref{lemma:haar_measure} and the Schur orthogonality relation $\int_G [\lambda(g)]_{i, j} \overline{[\mu(g)]_{k, \ell}} \dd g = \frac{1}{\dim V_{\lambda}} \bm{1}_{\lambda \cong \mu} \delta_{i, k} \delta_{j, \ell}$.
    Furthermore, the Wedderburn-Artin theorem \cite[Theorem 2.5.15]{grinko2025mixed} tells us that $\dim \mathbb{C}[G] = \sum_{ \lambda \in \widehat{\mathbb{C}[G]}  } (\dim V_{\lambda})^2$. 
\end{lemma}

\begin{corollary}
\label{corollary:covariance_translation_reformulation}
    For a compact Lie group $G$, and a class function $f \in \ell_2(G)$, the weighted right translation $\mathcal{E}_f  = \int_{G} f(g) \mathcal{R}(g) \dd g$ can be decomposed by
    $$
    \mathcal{E}_f = \sum_{\lambda \in \widehat{G}} c_{\lambda} \Pi_{\lambda}, \quad c_{\lambda} = \frac{1}{\dim V_{\lambda}} \int_{G} f(h) \chi_{\lambda}(h) \dd h,
    $$
    where the projection operator $\Pi_{\lambda}$ is specified in \eref{eqn:canonical_projection}.
\end{corollary}

\begin{proof}[Sketch proof of Corollary \ref{corollary:covariance_translation_reformulation}]
    The reformulation is a natural consequence of \lref{lemma:peter_weyl}: summing over $\lambda \in \widehat{G}$ is equivalent to summing over its dual $\lambda^* \in \widehat{G}$, we can reformulate $\mathcal{E}_f$ as
    \begin{align*}
    \mathcal{E}_f &= \int_{G} f(g) \mathcal{R}(g) \dd g  = \int_{G} \sum_{\lambda \in \widehat{G}} \left< f, \chi_{\lambda^*} \right> \chi_{\lambda^*}(g) \mathcal{R}(g) \dd g = \sum_{\lambda \in \widehat{G}} \left<f, \chi_{\lambda^*}\right> \int_{G} \overline{\chi_{\lambda}(g)} \mathcal{R}(g) \dd g \\
    &= \sum_{\lambda \in \widehat{G}} \left( \frac{1}{\dim V_{\lambda}} \int_{G} f(h) \chi_{\lambda}(h) \dd h  \right) \left( \dim V_{\lambda} \int_{G} \overline{\chi_{\lambda}(g)} \mathcal{R}(g) \dd g \right) = \sum_{\lambda \in \widehat{G}} c_{\lambda} \Pi_{\lambda}.
    \end{align*}
    This completes the proof.
\end{proof}
Tomographical protocols for quantum states and processes often rely on representation theory to fully utilize the symmetry of i.i.d. quantum resources \cite{Bisio_2010,efficient_state_tomography, Yang_2020, scharnhorst2025optimallowerboundsquantum, pelecanos2025debiasedkeylsalgorithmnew, pelecanos2025mixedstatetomographyreduces, west2026classicalshadowsarbitrarygroup}. We review several standard facts and parlance from representation theory that are either necessary for subsequent derivations (\eg, Appendix \ref{appendix:deferred_proofs}) or are included for completeness. Several essential components required for the proofs of subsequent lemmas are also derived in this section. These preliminaries are largely adapted from \cite{harrow_schur_2005, wright2016learn, grinko2025mixed}, to which we direct the reader for a more detailed exposition.

\begin{definition}[Permutation matrix algebra, \cite{grinko2025mixed}]
\label{def:perm_matrix_algebra}
    The natural action of the symmetric group $\mathfrak{S}_n$ on $(\mathbb{C}^d)^{\otimes n}$ is the tensor representation of $\mathfrak{S}_n$, given by a map $\psi_n^d: \mathbb{C}[\mathfrak{S}_n] \to \mathcal{L}((\mathbb{C}^d)^{\otimes n})$, acting as
    $$
    \psi_n^d(\pi) \ket{x_1} \ket{x_2} \cdots \ket{x_n} = \ket{x_{\pi^{-1}(1)}} \ket{x_{\pi^{-1}(2)}} \cdots \ket{x_{\pi^{-1}(n)}}, \quad \forall \pi \in \mathfrak{S}_n.
    $$
    The image of $\mathbb{C}[\mathfrak{S}_n]$ under the action of $\psi_n^d$ is the permutation matrix algebra $\mathcal{A}_n^d := \psi_n^d(\mathbb{C}[\mathfrak{S}_n])$.
\end{definition}

\begin{definition}[Cartan subalgebra and weight multiplicity, {\cite[Chapter 3]{goodman2009symmetry}}]
\label{def:weight_space}
    Let $\mathfrak{g} =\mathfrak{u}(d)$ (resp. $\mathfrak{su}(d)$) be the Lie algebra of $\U(d)$ (resp. $\su(d)$), the Cartan subalgebra $\mathfrak{h}  \subset \mathfrak{u}(d)$ (resp. $\mathfrak{su}(d)$) is the maximal abelian subalgebra of all diagonal matrices $h = \diag\{H_j\}_{j=1}^d$. Let $(\pi, W_{\pi})$ be a finite-dimensional representation of $\mathfrak{g}$, for a linear functional $w \in  \mathbb{Z}^d$, the weight space $W_{\pi}(w) = \{ \ket{v} \in W_{\pi}: \pi(h) \ket{v} = \left< w, h \right> \mathbf{v},~\forall h \in \mathfrak{h} \}$. When $W_{\pi}(w) \neq \{0\}$, $w$ is a weight, and the collection of weights is denoted as $\mathfrak{X}(W_{\pi})$. The $W_{\pi}$ is decomposable with weights $W_{\pi} = \bigoplus_{w} W_{\pi}(w)$. We denote the weight multiplicity $m_{\pi}(w) = \dim W_{\pi}(w)$. The Weyl group that captures the symmetries of $\mathfrak{h}$ for $\U(d)$ elements is $\mathfrak{S}_d$ \cite[Proposition 3.1.20]{goodman2009symmetry}. Any element $t \in \mathfrak{S}_d$ acts on $w \in \mathbb{Z}^d$ as $t(w) := (w_{t^{-1}(j)})_{j=1}^d$, and $\sgn(t)$ is the standard sign of the permutation $t$. This action partitions the collection of weights $\mathfrak{X}(W_\pi)$ into Weyl group orbits $\mathsf{O}(w) = \{t(w) : t \in \mathfrak{S}_d\}$. By symmetry, weight multiplicity is invariant across an orbit, meaning $m_\pi(t(w)) = m_\pi(w)$ for any $t \in \mathfrak{S}_d$.
\end{definition}

\begin{definition}[\cite{goodman2009symmetry}]
    A vector $\lambda = (\lambda_1, \lambda_2, \dots, \lambda_d) \in \mathbb{Z}^d$ satisfying $\lambda_1 \geq \lambda_2 \geq \cdots \geq \lambda_d$ is called the highest weight [cf. Definition \ref{def:weight_space}], and $\widehat{\U(d)} = \{\lambda \in \mathbb{Z}^d: \lambda_1 \geq \lambda_2 \geq \cdots \geq \lambda_d\}$, and without loss of generality, when the global phase does not matter, so that we obtain irreps for $\su(d)$ by tracking the heighest weight $(\lambda_1 - \lambda_2, \lambda_2 - \lambda_3, \dots, \lambda_{d-1} - \lambda_d)$. For a general highest weight $\lambda \in \widehat{\U(d)}$, its dual representation is induced by the highest weight $\lambda^* = (-\lambda_d, -\lambda_{d-1}, \dots, -\lambda_1)$. The sum of entries of $\lambda$, denoted as $|\lambda| = \sum_{j=1}^{k} \lambda_j$, indicates the global phase, which is assumed trivial in $\su(d)$.
\end{definition}

\begin{definition}[Partitions and Young diagrams \cite{goodman2009symmetry}] \label{def:partitions_and_young_diagrams}
A partition $\lambda \vdash n$ is a restriction of the highest weight $\lambda$ such that $\lambda_1 \geq \lambda_2 \geq \cdots \geq \lambda_{k} \geq 0$ and $|\lambda| = n$. The length of a non-empty partition $\lambda$ is given by $\ell(\lambda) = \max\{j : \lambda_j > 0\}$ and if $\ell(\lambda) \leq d$, we write $\lambda \vdash_d n$. Any partition $\lambda$ is graphically (the shape of) a Young diagram on $n$ cells arranged in $\ell(\lambda)$ rows with $\lambda_j$ cells on the $j$-th row. For instance, $\lambda = (2, 0) \in \Y_2^2$ corresponds to the diagram $\left(\syrow\right)$. 
We write $\Y_n^d$ for the collection of Young diagrams $\lambda \vdash_d n$, and it uniquely identifies the equivalence classes of irreps of the permutation matrix algebra $\mathcal{A}_{n}^d$ defined in Definition \ref{def:perm_matrix_algebra}.
\end{definition}

\begin{definition}[Representation of the symmetric group and general linear group]
    Each partition $\lambda \in \Y_n^d$ indexes an irreducible representation of $\mathfrak{S}_n$ and, when the global phase does not matter, an irreducible polynomial representation of $\U(d)$. For each $\lambda \in \Y_n^d$, suppose $(\pi_{\lambda}, S_{\lambda}) \in \widehat{\mathfrak{S}_n}$ and $(\nu_{\lambda}, W_{\lambda}) \in \widehat{\U(d)}$, the free modules $S_{\lambda}$ and $W_{\lambda}$ are known as the Specht module and the (Schur-)Weyl module, respectively. 
\end{definition}

\begin{lemma}[Littlewood-Richardson, \cite{goodman2009symmetry}]
\label{lemma:littlewood_richardson}
    For two Weyl modules $W_{\lambda}$ and $W_{\mu}$ on the highest weight $\lambda, \mu$, their tensor product can generally be decomposed into irreps of $\U(d)$. Mathematically,
    $$
    W_{\lambda} \otimes W_{\mu} \cong \bigoplus_{\gamma \in \widehat{\U(d)} } W_{\gamma}^{\oplus C_{\lambda, \mu}^{\gamma}} = \bigoplus_{\gamma \vdash |\lambda| + |\mu| } W_{\gamma}^{\oplus C_{\lambda, \mu}^{\gamma}}, \quad C_{\lambda, \mu}^{\gamma} = \dim \Hom_{\U(d)} \left( W_{\gamma},  W_{\lambda} \otimes W_{\mu} \right),
    $$
    where $C_{\lambda, \mu}^{\nu} \geq 0$ are the Littlewood-Richardson coefficients that counts the multiplicity. When $\dd g$ is normalized, the orthogonality allows us to rewrite $C_{\lambda, \mu}^{\gamma} = \int_G \chi_{\lambda}(g) \chi_{\mu}(g) \overline{\chi_{\gamma}(g)} \dd g$. Notably, for those illegitimate weights $\gamma$, this coefficient is automatically zero. Moreover, it holds that
    $
   \binom{|\lambda| + |\mu|}{|\mu|} s_{\lambda} s_{\mu} = \sum_{\gamma \vdash |\lambda| + |\mu|} s_{\gamma} C_{\lambda, \mu}^{\gamma}
    $.
\end{lemma}

\begin{lemma}[Schur-Weyl duality \cite{goodman2009symmetry, grinko2025mixed}]
\label{lemma:schur_weyl}
    For any $n, d \in \mathbb{N}$, the vector space $(\mathbb{C}^d)^{\otimes n}$ and the $n$-fold tensor of each $U \in \U(d)$ can be decomposed according to the irreps induced by the Young diagrams:
    $$
    (\mathbb{C}^d)^{\otimes n} \cong \bigoplus_{\lambda \in \widehat{\mathcal{A}_{n}^d} } W_{\lambda} \otimes S_{\lambda} = \bigoplus_{\lambda \in \Y_n^d } W_{\lambda} \otimes S_{\lambda}, \quad U^{\otimes n} \cong \bigoplus_{\lambda \in \Y_n^d } U_{\lambda} \otimes \openone_{S_{\lambda}}.
    $$
    The canonical bases for the Weyl and Specht modules are collectively identified as the Schur-Weyl basis \cite{harrow_schur_2005}. Recall the canonical learning protocol in \ref{subsec:bisio_etal_optimal_unitary_learning}, the optimal probe state and POVM are expressed in these respective bases, under the interpretation that $W_{\lambda} \cong \mathcal{H}_{\lambda}$, $S_{\lambda} \cong \mathcal{M}_{\lambda}$ under the isomorphism $\mathcal{H} \cong \C^d$.
\end{lemma}

\begin{figure}[tbp!]
    \centering
    
    \subcaptionbox{\label{fig:walled_brauer_example}}{
        \includegraphics[height=2.4cm]{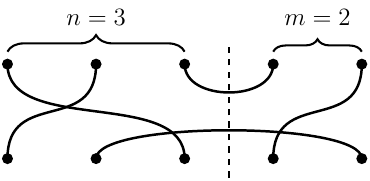}
    }
    \hspace{0.08\linewidth}
    \subcaptionbox{\label{fig:A22_generators}}{
        \includegraphics[height=2.4cm]{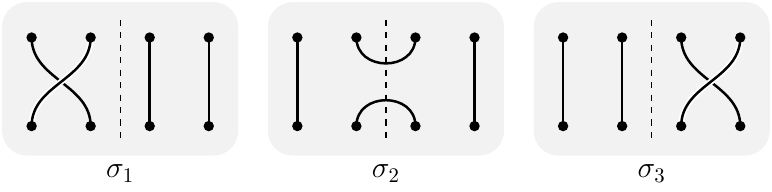}
    }
    
    \caption{(a) The diagram of an element in $\mathcal{B}_{3, 2}^d$. (b) The diagrams of the generators of $\mathcal{A}_{2, 2}^d$.}
    \label{fig:joint_example}
\end{figure}

\begin{definition}[Walled Brauer algebra and the associated matrix algebra \cite{grinko2025mixed}]
    The walled Brauer algebra $\mathcal{B}_{n, m}^d$ is a finite-dimensional associative algebra. Graphically, each basis element is represented as a diagram on two rows, each containing two groups of nodes (the first with $n$ nodes and the second with $m$ nodes) separated by a vertical wall. See \autoref{fig:walled_brauer_example} for an illustrative example. All nodes are connected in pairs, and all paired nodes are either on the same side of the wall or in different rows. 
    
    The generators of $\mathcal{B}_{n, m}^d$ are given by the SWAP operation between node $i$ and node $i + 1$ for $i \in [n-2]$, and a cross-wall contraction that connects node $n$ and node $n + 1$.
    Every basis element in $\mathcal{B}_{n, m}^d$ is formed by a product of this collection of generators. 
    \par Extending Defintion \ref{def:perm_matrix_algebra} on $\mathcal{B}_{n, m}^d$ results in a mixed tensor representation $\kappa_{n, m}^d: \mathcal{B}_{n, m}^d \to \mathcal{L}((\mathbb{C}^{d})^{\otimes n} \otimes (\mathbb{C}^d)^{\otimes m})$, and we refer readers to \cite[Section 3.3]{grinko2025mixed} for its defining rules. Notably, $[\mathcal{A}_{n, m}^d, U^{\otimes n} \otimes {U^*}^{\otimes m}] = 0$ for all $U \in \U(d)$. The matrix algebra $\mathcal{A}_{n, m}^d := \kappa_{n, m}^d(\mathcal{B}_{n, m}^d)$ serves as a natural extension of $\mathcal{A}_{n}^d$ defined in Definition \ref{def:perm_matrix_algebra}. When $d \geq n + m$, $\mathcal{A}_{n, m}^d \cong \mathcal{B}_{n, m}^d$.
\end{definition}

\begin{definition}[Mixed Young diagrams, and mixed Schur-Weyl duality \cite{Benkart1994}] \label{def:mixed_shur_weyl}
    The irreps of $\mathcal{A}_{n, m}^d$ are indexed by the mixed Young diagrams. Specifically, every such diagram is a tuple $\lambda = (\lambda_{\ell}, \lambda_r)$, where
    $$
    \widehat{\mathcal{A}_{n, m}^d} = \{ (\lambda_{\ell}, \lambda_r):~0 \leq k \leq \min\{n, m\}, \lambda_{\ell}\vdash n-k, \lambda_{r} \vdash m - k, \ell(\lambda_{\ell}) + \ell(\lambda_r) \leq d \} =: \Y_{n, m}^d.
    $$ 
    Here, $k$ is called the number of contractions, or graphically, the number of wires connecting the nodes on the side. An analogous duality that extends \lref{lemma:schur_weyl}, or the mixed Schur-Weyl duality, reads
    $$
    (\mathbb{C}^{d})^{\otimes n} \otimes (\mathbb{C}^{d})^{\otimes m} \cong \bigoplus_{\lambda \in \widehat{\mathcal{A}_{n, m}^d} } W_{\lambda} \otimes S_{\lambda} = \bigoplus_{\lambda \in \Y_{n, m}^d}  W_{\lambda} \otimes S_{\lambda}, \quad \forall U \in \U(d),~U^{\otimes n} \otimes {U^*}^{\otimes m} \cong \bigoplus_{\lambda \in \Y_{n, m}^d}  U_{\lambda} \otimes \openone_{S_{\lambda}}.
    $$
    Additionally, there is an isomorphism between each mixed diagram $\lambda = (\lambda_{\ell}, \lambda_r)$ and a highest weight $\lambda'$, defined by $\lambda' = (\lambda_{\ell, j} - \lambda_{r, d - j + 1})_{j=1}^d$. They introduce the same irrep for $\U(d)$. 
\end{definition}

\begin{corollary} \label{corollary:expression_reformulation_adjoint}
    If we replace the right translation $\mathcal{R}(g)$ in Lemma \ref{lemma:schur} with the adjoint action $\mathcal{T}_{m, n}(g): x \mapsto \mathcal{R}(g)^{\otimes m} x \mathcal{R}(g^{-1})^{\otimes n}$, the isotropic projection takes a similar form, and we can obtain a similar decomposition result to that in Corollary \ref{corollary:covariance_translation_reformulation}: For a compact Lie group $G$ and a class function $f$,
    $$
    \mathcal{E}_f = \int_G f(g) \mathcal{T}_{m, n}(g) \dd g = \sum_{\lambda \in \widehat{G}} \left( \frac{1}{\dim V_{\lambda}} \int_{G} f(h) \chi_{\lambda}(h) \dd h  \right) \left( \dim V_{\lambda} \int_{G} \overline{\chi_{\lambda}(g)} \mathcal{T}_{m, n}(g) \dd g \right) = \sum_{\lambda \in \widehat{G}} c_{\lambda} \Pi_{\lambda},
    $$
    where the support of irreps $\lambda$ differs. When $G = \U(d)$, the irreps in Corollary \ref{corollary:covariance_translation_reformulation} are indexed by Young diagrams [cf. Definition \ref{def:partitions_and_young_diagrams}], while for the case with mixed tensor action $\mathcal{T}_{m, n}$, they are indexed by mixed Young diagrams [cf. Definition \ref{def:mixed_shur_weyl}].
\end{corollary}

\begin{lemma}[Dimension formula for Weyl and Specht modules, \cite{goodman2009symmetry}]
\label{lemma:dimension_formulas}
    For a highest weight $\lambda \in \mathbb{Z}^d$, the dimension of the Weyl module $W_{\lambda}$ is given by
    $
   d_{\lambda} = \dim W_{\lambda} = \prod_{1 \leq i < j \leq d} \frac{\lambda_i - \lambda_j + j - i}{j - i}
    $, and when $\lambda \vdash n$ is a partition, the dimension of the Specht module $S_{\lambda}$ is given by $s_{\lambda} = \dim S_{\lambda} = \frac{n!}{\prod_{(i, j) \in \lambda} h_{\lambda}(i, j)}$, where the hook length $h_{\lambda}(i, j) = (\text{number of cells to the right of $\lambda(i, j)$}) + (\text{number of cells below $\lambda(i, j)$}) + 1$.
\end{lemma}

\begin{lemma}[{\cite[Corollary 2.3.1]{koike1989decomposition}}, reformulated]
\label{lemma:mixed_littlewood_richardson}
    For a mixed Young diagram $\nu = (\nu_{\ell}, \nu_{r})$ and two highest weights $\lambda, \mu$, we have $C_{\lambda, \nu}^{\mu} := \dim \Hom_{\U(d)}(W_{\mu}, W_{\lambda} \otimes W_{\nu}) = \int_G \chi_{\lambda}(g) \chi_{\nu}(g) \overline{\chi_{\mu}(g)} \dd g = \sum_{\gamma} C_{\nu_{\ell}, \gamma}^{\lambda} C_{\nu_r, \gamma}^{\mu}$.
\end{lemma}

\begin{lemma}[{\cite[Corollary 7.1.7]{goodman2009symmetry}}]
\label{lemma:character_formula}
    Let $\varrho = ( \frac{d + 1 - 2j }{2} )_{j=1}^d$ be the Weyl vector of $\U(d)$ [cf. \cite[Lemma 3.1.21]{goodman2009symmetry}], for weights $\lambda, \mu$ and $\gamma$, then
    \begin{equation}
    \label{eqn:LR_coefficient_and_weight_multiplicity}
    C_{\lambda, \mu}^{\gamma} = \sum_{t \in \mathfrak{S}_d} \sgn(t) \cdot m_{\mu}\left( \gamma + \varrho - t(\lambda + \varrho) \right).
    \end{equation}
\end{lemma}

\begin{corollary}[{\cite[Exercise 25.33]{Fulton2004}}]
    For highest weights $\lambda, \mu \in \widehat{\U(d)}$, if $\lambda + w \in \widehat{\U(d)}$ for every $w \in \mathbb{Z}^d$ where $m_{\mu}(w) \neq 0$, then for all $\gamma \in \widehat{\U(d)}$, $C_{\lambda, \mu}^{\gamma} = m_{\mu}(\gamma - \lambda)$.
\end{corollary}

\begin{corollary}
\label{corollary:no_multiplicity}
    For highest weights $\lambda, \mu, \nu \in \mathbb{Z}^d$, define the minimum adjacent gaps $\gap(\lambda) = \min_{j \in [d-1]} (\lambda_{j} - \lambda_{j+1})$ and $\gap(\mu) = \min_{j \in [d-1]} (\mu_{j} - \mu_{j+1})$. Suppose $(\nu)_{1} - (\nu)_{d} \leq \gap(\lambda) + \gap(\mu) + 1$, then $C_{\lambda, \nu}^{\mu} = m_{\nu}(\mu - \lambda)$.
\end{corollary}

\begin{proof}[Proof of Corollary \ref{corollary:no_multiplicity}]
    For the equality to hold, it suffices to show that in \eref{eqn:LR_coefficient_and_weight_multiplicity} of \lref{lemma:character_formula}, every contributing $t \in \mathfrak{S}_d$ must satisfy 
    $$
   w_t := \mu + \varrho - t(\lambda + \varrho) \in \mathfrak{X}(W_{\nu}).
    $$
    Equivalently, we need $t^{-1}(\mu + \varrho) = \lambda + \varrho + t^{-1}(w_t)$. For a non-identity permutation $t$, there exists at least one local inversion that $t(j) \geq t(j +1) + 1$, so that
    \begin{align*}
    \left( t^{-1}(\mu + \varrho) \right)_j - \left( t^{-1}(\mu + \varrho) \right)_{j+1} &= \mu_{t(j)} + \varrho_{t(j)} - \mu_{t(j+1)} - \varrho_{t(j+1)} \\
    &= (\mu_{t(j)} - \mu_{t(j+1)}) + (\varrho_{t(j)} - \varrho_{t(j+1)}) \leq -\gap(\mu) - 1.
    \end{align*}
    Similarly, for the right-hand side, since $w_t \in W_{\nu}$, 
    \begin{align*}
    ( \lambda + \varrho + t^{-1}(w_t) )_j - ( \lambda + \varrho + t^{-1}(w_t) )_{j+1} &= (\lambda_{j} - \lambda_{j+1}) + (\varrho_{j} - \varrho_{j+1}) + \left((w_t)_{t(j)} - (w_t)_{t(j+1)}\right) \\
    &\geq \gap(\lambda) + 1 - ((\nu)_1 - (\nu)_d).
    \end{align*}
    Combining these inequalities, we obtain
    $$
    -\gap(\mu) - 1 \geq \gap(\lambda) + 1 - ((\nu)_1 - (\nu)_d) \implies (\nu)_1 - (\nu)_d \geq \gap(\lambda) + \gap(\mu) + 2,
    $$
    a contradiction. Therefore, the only surviving term in \eref{eqn:LR_coefficient_and_weight_multiplicity} is $t = (1)$, concluding the proof.
\end{proof}

\begin{figure}[tbp!]
    \centering
    \includegraphics[width=0.75\linewidth]{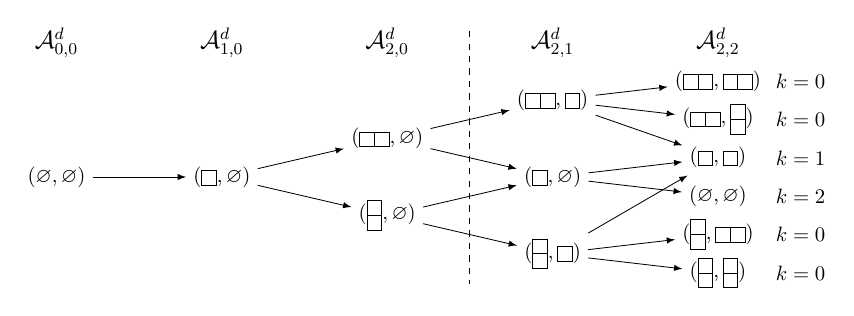}
    \caption{Constructing the diagrams in $\Y_{2, 2}^d$ determining the irreps of $\mathcal{A}_{2, 2}^d$.}
    \label{fig:Bratteli_2_2}
\end{figure}

\begin{example}[Mixed Young diagrams and generators of $\mathcal{A}_{2, 2}^d$]
\label{example:symmetric_subalgebra_A_2_2_d}
    The irreps of the matrix algebra $\mathcal{A}_{2, 2}^d$ are indexed by six mixed diagrams, as shown in \autoref{fig:Bratteli_2_2}. As remarked in \cite[Equation (3.10)]{grinko2025mixed}, the presentation of $\mathcal{A}_{2, 2}^d$ corresponds to an explicit matrix representation over the bipartite space $\mathcal{H}_A \otimes \mathcal{H}_B =\mathcal{H}_{A_1} \otimes \mathcal{H}_{A_2} \otimes \mathcal{H}_{B_1} \otimes \mathcal{H}_{B_2} \cong (\mathbb{C}^{d})^{\otimes 2} \otimes (\mathbb{C}^{d})^{\otimes 2}$ with local SWAPs $\sigma_1, \sigma_3$ and cross-wall contraction $\sigma_2$:
    $$
    \sigma_1 = \mathbb{F}_A, \quad \sigma_2 = \Omega_{A_2 B_1}, \quad \sigma_3 = \mathbb{F}_B.
    $$
    Their actions are graphically shown in \autoref{fig:A22_generators}.
    For the subalgebra of $\mathcal{A}_{2, 2}^d$ that is invariant under the action $x \mapsto \sigma_1 \sigma_3 x \sigma_3 \sigma_1$\footnote{
   Informally speaking, it is spanned by the standard $24$ basis elements of $\mathcal{A}_{2, 2}^d$ modulo simultaneous local swaps. 
    }, we can write its basis elements in a canonical form $P_{\pm} X Q_{\pm}$ where $P_{\pm} = \frac{\openone \pm \F_A}{2} \otimes \openone$ and $Q_{\pm} =\openone \otimes \frac{\openone \pm \F_B}{2}$ are projections onto the two sides of the wall, $X \in \{ \openone, \Delta, \Delta_{\times}, \Lambda\}$ where $\Delta = \Omega_{A_1 B_1} + \Omega_{A_2 B_2}$, $\Delta_{\times} = \Omega_{A_1 B_2} + \Omega_{A_2 B_1}$ and $\Lambda = \Omega_{A_1 B_1} \Omega_{A_2 B_2} = \Omega_{AB}$ capture the symmetric contractions of different types (singly/doubly, respectively).
\end{example}

\begin{example}[Multiplicity of weights associated with irreps of $\mathcal{A}_{2, 2}^d$]
\label{example_weight_multiplicities}
    For the sake of future clarity, we evaluate the multiplicity $m_{\nu}(w)$ and the counting of weight vectors $w$ for different orbits of $\mathfrak{X}(W_{\nu})$, where $\nu = \nu_1, \nu_2, \nu_3, \nu_4 \in \widehat{\mathcal{A}_{2, 2}^d}$. Here $\nu_1 = (\sybox, \sybox) \cong (1, 0, \dots, 0, -1)$, $\nu_2 = (\syrow, \syrow) \cong (2, 0, \dots, 0, -2)$, $\nu_3 = (\syrow, \sycol) \cong (2, 0, \dots, 0, -1, -1)$ (similar for $(\sycol, \syrow) \cong (1, 1, 0, \dots, 0, -2)$) [cf. \autoref{fig:Bratteli_2_2}], and $\nu_4 = (\sycol, \sycol) \cong (1, 1, 0, \dots, 0, -1, -1)$. By permutation invariance of multiplicity [cf. Definition \ref{def:weight_space}], it suffices to track the multiplicity for each orbit, and how many weights it contains.
    \begin{itemize}
        \item The non-zero weights of $\nu_1$ are roots of $\U(d)$, taking the form $e_i - e_j$ for $i \neq j \in [d]$. There are $d(d-1)$ many by choosing any ordered $i, j \in [d]$.

        \item The non-zero weights of $\nu_2$ composed of weights $e_i + e_j$ (from action on $\vee^2 {\mathbb{C}^d}$) for $i, j \in [d]$, and $-e_k - e_{\ell}$ (from action on $\vee^2 {\mathbb{C}^d}^*$) for $k, \ell \in [d]$. The following types of orbits are possible:
        \begin{enumerate}
            \item The orbit $2e_i - 2e_k$ for distinct $i , k$ with multiplicity $1$, formed only by $i = j$ and $k = \ell$. By choosing ordered $i \neq k$, there are $d(d-1)$ many weights.
            \item The orbit $\pm(2e_i - e_k - e_{\ell})$ for distinct $i, k, \ell$ with multiplicity $1$, formed by taking $i = j$ and unordered $k \neq \ell$ or $k = \ell$ and unordered $i \neq j$. There are thus $2 \times d \times \binom{d-1}{2} = d(d-1)(d-2)$ many weights.
            \item The orbit $e_i + e_j - e_{k} - e_{\ell}$ for distinct $i, j, k, \ell$ with multiplicity $1$, formed by taking unordered $i \neq j$ and unordered $k \neq \ell$. There are thus $\binom{d}{2} \times \binom{d-2}{2} = \frac{1}{4}d(d-1)(d-2)(d-3)$ many weights.
            \item The orbit $e_i - e_k$ for distinct $i, k$ with multiplicity $d - 1$, since $2e_i + (- e_i - e_k)$, $e_i + e_k + (- 2e_k)$ and $e_i + e_j + (- e_j - e_k)$ all give this weight, in a total of $1 + 1 + d - 2 = d$ ways, subtracting the multiplicity contributed by $\nu_1$, which is $1$. The total weights are $d(d-1)$ many by choosing any ordered pairs $i \neq k$.
        \end{enumerate}

        \item The non-zero weights of $\nu_3$ composed of weights $e_i + e_j$ (from action on $\vee^2 \C^d)$ for $i , j \in [d]$, and $-e_k - e_{\ell}$ (from action on $\wedge^2 {\C^d}^*$) for $k \neq \ell \in [d]$. The following types of orbits are possible:
        \begin{enumerate}
            \item The orbit $2e_i - e_k - e_\ell$ for distinct $i, k, \ell$ with multiplicity $1$, formed by taking $i = j$ and $k \neq \ell$. Choosing $i$ and unordered $k \neq \ell$, there are $d \times \binom{d-1}{2} = \frac{1}{2} d(d-1)(d-2)$ many weights.

            \item The orbit $e_i + e_j - e_k - e_{\ell}$ for distinct $i, j, k, \ell$ with multiplicity $1$, formed by taking unordered $i \neq j$ and unordered $k \neq \ell$, similarly there are $\binom{d}{2} \times \binom{d-2}{2} = \frac{1}{4} d(d-1)(d-2)(d-3)$ many.

            \item The orbit $e_i - e_k$ for distinct $i, k$ with multiplicity $d - 2$, since $2e_i + (-e_i - e_k)$ and $e_i + e_j + (-e_j - e_k)$ give this type of weight, with $1 + d - 2 = d - 1$ free choices, subtracting multiplicity $1$ contributed by $\nu_1$ yields multiplicity $d - 2$. The total counting is $d(d-1)$ by choosing unordered $i \neq k$.
        \end{enumerate}
        
        \item The non-zero weights of $\nu_4$ composed of weights $e_i + e_j$ (from action on $\wedge^2 {\mathbb{C}^d}$) for $i \neq j \in [d]$, and $-e_k - e_{\ell}$ (from action on $\wedge^2 {\mathbb{C}^d}^*$) for $k \neq \ell \in [d]$. The following types of orbits are possible:
        \begin{enumerate}
            \item The orbit $e_i + e_j - e_k - e_\ell$ for distinct $i, j, k, \ell$ with multiplicity $1$, formed by taking unordered $i \neq j$ and $k \neq \ell$. There are $\binom{d}{2} \times \binom{d-2}{2} = \frac{1}{4}d(d-1)(d-2)(d-3)$ many weights.
            \item The orbit $e_i - e_k$ for distinct $i, k$ with multiplicity $d - 3$, since $e_i + e_j + (-e_k - e_j)$ gives rise to $e_i - e_k$, and there are $d - 2$ many due to the free choice of $j \in [d] \setminus \{i, k\}$, subtracted by the multiplicity $1$ contributed by $\nu_1$. The total counting of weights is $d(d-1)$ as well, by choosing ordered $i \neq k$.
        \end{enumerate}
    \end{itemize}
\end{example}

\section{Deferred proofs of lemmata for the query-optimal protocol}
\label{appendix:deferred_proofs}
In this section, we present the clunky proofs of Lemmas \ref{lemma:small_s_regime_estimator_Lambda_variance} and \ref{lemma:large_s_regime_estimator_X_variance} that help evaluate our query-optimal CSEU protocol. The key ingredient is the full expression of the variances of the estimators $\hat{Z}(\hat{\mathsf{X}}, L)$ and $\hat{Z}(\hat{\mathsf{\Lambda}}, L)$. Recall that we have defined the detraced operator $\rho_0 = \rho - \openone/d$ obtained from the quantum state of interest throughout the section, and that $O$ is traceless.

\subsection{Reformulating the variances}
\label{appendix:variance_reformulation_for_both_estimators}
 We begin by defining the second-moment learning channel.
\begin{definition}
    The second-moment learning channel via the canonical optimal protocol [cf. Section \ref{subsec:bisio_etal_optimal_unitary_learning}] with learning strategy $(\Y, \bfq)$ is defined by 
    $$
    \begin{aligned}
    \forall \psi \in (\mathbb{C}^{d})^{\otimes 2}, \quad
    \mathcal{M}_{\bfq, s}^{(2)}(\psi) &= \int \braket{ \Psi_{V} | U^{\otimes s} \ket{\Psi_{\bfq}} \bra{\Psi_{\bfq}} {U^\dagger}^{\otimes s} |\Psi_V } \cdot V^{\otimes 2} \psi {V^\dagger}^{\otimes 2} \dd V \\
    &=  \int \left| \sum_{\lambda \in \Y} \sqrt{q_{\lambda}} \Tr[U^{\dagger}_{\lambda} V_{\lambda}] \right|^2 V^{\otimes 2} \psi {V^\dagger}^{\otimes 2} \dd V.
    \end{aligned}
    $$
\end{definition}
The variance of the estimators $\hat{X}_j$, $\hat{Z}(\hat{\mathsf{\Lambda}}, L)$ [cf. Eqs.~\eqref{eqn:linear_estimator} and \eqref{eqn:average_estimators}] can be fully expanded in terms of the second-moment channel $\mathcal{M}_{\bfq, s}^{(2)}$:
    \begin{equation}
\label{eqn:variance_expression_naive_estimator}
    \begin{aligned}
    \Var\left[\hat{X}_j\right] &= \E\left[\hat{X}_j^2\right] - \E\left[ \hat{X}_j \right]^2 \\
    &= \frac{1}{\p_{\bfq}^2} \E\left[\Tr\left[ O \otimes O  \hat{U}_j^{\otimes 2} \rho_0^{\otimes 2} { (\hat{U}_{j}^\dagger) }^{\otimes 2}  \right]  \right] - \left( \Tr\left[ O  U \rho_0 U^\dagger \right]  \right)^2 \\
    &= \frac{1}{\p_{\bfq}^2} \Tr\left[ O \otimes O  \mathbb{E}\left[ \hat{U}_j^{\otimes 2} \rho_0^{\otimes 2} { (\hat{U}_{j}^\dagger) }^{\otimes 2} \right] \right]   - \left( \Tr\left[ O  U \rho_0 U^\dagger \right]  \right)^2 \\
    &= \frac{1}{\p_{\bfq}^2} \Tr\left[ O \otimes O  \mathcal{M}_{\bfq, s}^{(2)}(\rho_0^{\otimes 2}) \right]   - \left( \Tr\left[ O  U \rho_0 U^\dagger \right]  \right)^2.
    \end{aligned}
    \end{equation}
    As for the variance of the estimator $\hat{Z}(\hat{\mathsf{\Lambda}}, L)$, we derive an upper bound instead. For conciseness, we use $\propto$ to hide the leading coefficients.
    \begin{align*}
    \Var\left[ \hat{Z}(\hat{\mathsf{\Lambda}}, L) \right] &=\frac{1}{L^2(L-1)^2\left(d + \frac{2(1-\p_{\bfq})}{d \p_{\bfq}}\right)^2} \Var\left[ \sum_{i \neq j \in [L]} \Tr\left[ \left( O \otimes \rho_0^T \right) \hat{Y}_i \hat{Y}_j \right] \right] \\
    &\propto \Var\left[ \sum_{i < j \in [L]} W_{i, j} \right]  \tag{$W_{i, j} = \Tr\left[ \left( O \otimes \rho_0^T \right) \{\hat{Y}_i,  \hat{Y}_j \} \right]$} \\
    &= \sum_{i < j \in [L]} \Var\left[W_{i, j} \right] + \sum_{\substack{i < j, k < \ell \in [L] \\ \{i, j\} \neq \{k, \ell\} } } \Cov\left( W_{i, j}, W_{k, \ell} \right) \\
    &= \sum_{i < j \in [L]} \Var\left[W_{i, j} \right] + \sum_{\substack{i < j; i, j \neq \ell \in [L]} } \Cov\left( W_{i, j}, W_{i, \ell} \right) \\
    &= \sum_{i < j \in [L]} \Var\left[W_{i, j} \right] + \sum_{\substack{i < j; i, j \neq \ell \in [L]} } \left( \E_{\hat{Y}_i} \left[\E\left[W_{i, j} | \hat{Y}_i \right] \E \left[W_{i, \ell} | \hat{Y}_i \right]  \right] - \E[W_{i, j}] \E[W_{i, \ell}] \right) \\
    &= \sum_{i < j \in [L]} \Var\left[W_{i, j} \right] + \sum_{\substack{i < j; i, j \neq \ell \in [L]} } \Var_{\hat{Y}_i} \left[ \E\left[ W_{i, j} | \hat{Y}_i \right]  \right],
    \end{align*}
    where
    $$
    \begin{aligned}
        \E\left[ W_{i, j} | \hat{Y}_i \right] &= \E_{\hat{Y}_j} \left[ \Tr\left[ \left(O \otimes \rho_0^T \right) \left(\hat{Y}_i \hat{Y}_j + \hat{Y}_j \hat{Y}_i \right)  \right]  \right] \\
        &=  \Tr\left[ \left(O \otimes \rho_0^T \right) \left(\hat{Y}_i \cdot \E[\hat{Y}_j] + \E[\hat{Y}_j] \cdot \hat{Y}_i \right)  \right] \\
        &=  \Tr\left[ \left\{\E[\hat{Y}_j], O \otimes \rho_0^T \right\} \hat{Y}_i 
        \right].
    \end{aligned}
    $$
Therefore, using the inequality $\Var[X] \leq \E[X^2]$,
    \begin{align*}
        & \Var\left[ \hat{Z}(\hat{\mathsf{\Lambda}}, L) \right] \propto \sum_{i < j \in [L]} \Var\left[ \Tr\left[\left(O \otimes \rho_0^T \right) \left\{\hat{Y}_i, \hat{Y}_j\right\}  \right]  \right] + \sum_{\substack{i < j; i, j \neq \ell \in [L]} } \Var\left[ \Tr\left[ \left\{\E[\hat{Y}_j], O \otimes \rho_0^T \right\} \hat{Y}_i 
        \right]  \right] \\
        &\quad= \frac{L(L-1)}{2} \cdot \Var\left[ \Tr\left[\left(O \otimes \rho_0^T \right) \left\{\hat{Y}_i, \hat{Y}_j\right\}  \right]  \right] + L(L-1)(L-2) \cdot \Var\left[ \Tr\left[ \left\{\E[\hat{Y}_j], O \otimes \rho_0^T \right\} \hat{Y}_i 
        \right]  \right] \\
        &\quad \leq \frac{L(L-1)}{2} \cdot \E\left[ \Tr\left[\left(O \otimes \rho_0^T \right) \left\{\hat{Y}_i, \hat{Y}_j\right\}  \right]^2 \right] + L(L-1)(L-2) \cdot \E\left[ \Tr\left[ \left\{\E[\hat{Y}_j], O \otimes \rho_0^T \right\} \hat{Y}_i 
        \right]^2 \right] \\
        &\quad = \frac{L(L-1)}{2} \cdot \E\left[ \Tr\left[\left(O \otimes \rho_0^T \right)^{\otimes 2} \left\{\hat{Y}_i, \hat{Y}_j\right\}^{\otimes 2}  \right] \right] + L(L-1)(L-2) \cdot \E\left[ \Tr\left[ \left\{\E[\hat{Y}_j], O \otimes \rho_0^T \right\}^{\otimes 2} \hat{Y}_i^{\otimes 2} 
        \right] \right] \\
        &\quad = \frac{L(L-1)}{2} \cdot  \E\left[ \Tr\left[\left(O \otimes \rho_0^T \right)^{\otimes 2} \left\{\hat{Y}_i, \hat{Y}_j\right\}^{\otimes 2}  \right] \right]  + L(L-1)(L-2) \cdot  \Tr\left[ \left\{\E[\hat{Y}_j], O \otimes \rho_0^T \right\}^{\otimes 2} \E\left[\hat{Y}_i^{\otimes 2}\right]
        \right].
    \end{align*}
    Finally, we have
    \begin{align*}
      \E\left[ \Tr\left[\left(O \otimes \rho_0^T \right)^{\otimes 2} \left\{\hat{Y}_i, \hat{Y}_j\right\}^{\otimes 2}  \right] \right] &= \Tr\left[\left(O \otimes \rho_0^T \right)^{\otimes 2}\E \left[(\hat{Y}_i \hat{Y}_j)^{\otimes 2} + (\hat{Y}_j \hat{Y}_i)^{\otimes 2} \right] \right] \\
      &\quad + \E\left[\Tr\left[\left(O \otimes \rho_0^T \right)^{\otimes 2}  \left( \hat{Y}_i \hat{Y}_j \otimes \hat{Y}_j \hat{Y}_i + \hat{Y}_j \hat{Y}_i \otimes \hat{Y}_i \hat{Y}_j \right) \right] \right], \\
      \Tr\left[\left(O \otimes \rho_0^T \right)^{\otimes 2}\E \left[(\hat{Y}_i \hat{Y}_j)^{\otimes 2} + (\hat{Y}_j \hat{Y}_i)^{\otimes 2} \right] \right] &= 2\Tr\left[\left(O \otimes \rho_0^T \right)^{\otimes 2} \E\left[ \hat{Y}_i^{\otimes 2} \hat{Y}_j^{\otimes 2} \right] \right] \\
      &= 2 \Tr\left[\left(O \otimes \rho_0^T \right)^{\otimes 2} \E \left[ \hat{Y}_i^{\otimes 2} \right]^2 \right], \\
      \E\left[\Tr\left[\left(O \otimes \rho_0^T \right)^{\otimes 2}  \left( \hat{Y}_i \hat{Y}_j \otimes \hat{Y}_j \hat{Y}_i + \hat{Y}_j \hat{Y}_i \otimes \hat{Y}_i \hat{Y}_j \right) \right] \right] &=2 \E\left[ \Tr\left[ \left( O \otimes \rho_0^T \right)\hat{Y}_i \hat{Y}_j \otimes   \left( O \otimes \rho_0^T \right) \hat{Y}_j \hat{Y}_i  \right]\right] \\
      &= 2 \E\left[ \Tr\left[ \left( O \otimes \rho_0^T \right)\hat{Y}_i \hat{Y}_j \otimes  \hat{Y}_i \left( O \otimes \rho_0^T \right) \hat{Y}_j  \right]\right] \\
    &= 2  \Tr\left[ \E\left[\left(  O \otimes \rho_0^T  \otimes \openone^{\otimes 2}\right)\left( \hat{Y}_i \otimes \hat{Y}_i \right)\left(\openone^{\otimes 2} \otimes  O \otimes \rho_0^T \right)\left(\hat{Y}_j \otimes \hat{Y}_j  \right) \right] \right] \\
    &= 2  \Tr\left[ \left(  O \otimes \rho_0^T  \otimes \openone^{\otimes 2}\right)\E\left[ \hat{Y}_i^{\otimes 2} \right]\left(\openone^{\otimes 2} \otimes  O \otimes \rho_0^T \right) \E\left[\hat{Y}_j^{\otimes 2}  \right]  \right].
    \end{align*}

\begin{figure}
    \centering
    \includegraphics[width=0.7\linewidth]{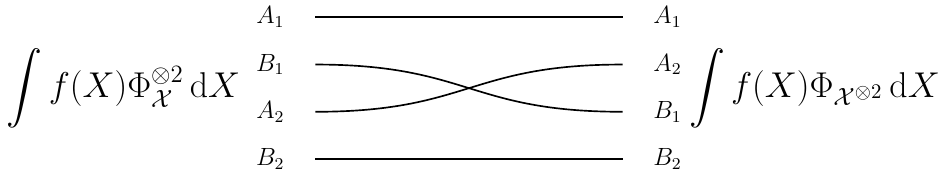}
    \caption{The local permutation that rearranges the bipartite space.}
    \label{fig:local_permutation}
\end{figure}

    It suffices to evaluate the expectation $\E[\hat{Y}_j^{\otimes 2}]$ for any $j \in [L]$. Recall by our construction in \eref{eqn:quadratic_estimator},
    $$
    \E\left[\hat{Y}_j^{\otimes 2}\right] = \frac{1}{\p_{\bfq}^2} \E\left[ \Phi_{\hat{U}_j}^{\otimes 2} \right].
    $$
    By linearity of the Choi isomorphism \cite{Watrous_2018}, the expectation $\E[ \Ket{\hat{U}_j} \Bra{\hat{U}_j}]$ evaluates to the Choi operator of the channel $\mathcal{M}_{\bfq, s}^{(2)}$ up to a local permutation, as per \autoref{fig:local_permutation}. We denote it as $\widetilde{\Phi}_{\bfq, s} = \F_{A_2 B_1} \Phi_{\mathcal{M}_{\bfq, s}^{(2)}} \F_{A_2 B_1}$. Finally, the variance of the estimator $\hat{Z}(\hat{\mathsf{\Lambda}}, L)$ is upper bounded by
    \begin{equation}
    \label{eqn:variance_upper_bound_quadratic_estimator}
    \begin{aligned}
    \Var\left[ \hat{Z}(\hat{\mathsf{\Lambda}}, L) \right] &\leq \frac{1}{L^2(L-1)^2\left(d + \frac{2(1-\p_{\bfq})}{d \p_{\bfq}}\right)^2} \Bigg( \frac{L(L-1)}{2} \cdot  \E\left[ \Tr\left[\left(O \otimes \rho_0^T \right)^{\otimes 2} \left\{\hat{Y}_i, \hat{Y}_j\right\}^{\otimes 2}  \right] \right] \\
    &\quad + L(L-1)(L-2) \cdot  \Tr\left[ \left\{\E\left[\hat{Y}_j\right], O \otimes \rho_0^T \right\}^{\otimes 2} \E\left[\hat{Y}_i^{\otimes 2}\right]
        \right] \Bigg) \\
        &= \frac{1}{L^2(L-1)^2\left(d + \frac{2(1-\p_{\bfq})}{d \p_{\bfq}}\right)^2}  \Bigg( \frac{1}{\p_{\bfq}^{4}} L(L-1)  \bigg( \Tr\left[\left(O \otimes \rho_0^T \right)^{\otimes 2} \widetilde{\Phi}_{\bfq, s}^2 \right] \\
        &\quad + \Tr\left[ \left(  O \otimes \rho_0^T  \otimes \openone^{\otimes 2}\right) \widetilde{\Phi}_{\bfq, s} \left(\openone^{\otimes 2} \otimes  O \otimes \rho_0^T \right) \widetilde{\Phi}_{\bfq, s}  \right]  \bigg) \\
        &\quad + \frac{1}{\p_{\bfq}^{2}} L(L-1)(L-2)  \cdot  \Tr\left[ \left\{\Ket{U} \Bra{U} + \frac{1-\p_{\bfq}}{d\p_{\bfq}} \openone \otimes \openone, O \otimes \rho_0^T \right\}^{\otimes 2} \widetilde{\Phi}_{\bfq, s}
        \right] \Bigg).
        \end{aligned}
    \end{equation}
    In the evaluations, we can distribute the local swaps that constitute the operator $\widetilde{\Phi}_{\bfq, s}$ to the observables due to the cyclic property of the trace. Essential to estimating the final scaling of $\Var[ \hat{Z}(\hat{\mathsf{\Lambda}}, L) ]$ is the expression of the Choi operator of the second-moment channel $\Phi_{\mathcal{M}_{\bfq, s}^{(2)}} = (\mathcal{M}_{\bfq, s}^{(2)} \otimes \mathcal{I})(\Ket{\openone} \Bra{\openone})$.

\subsection{Evaluating the Choi operator of the second-moment channel}
\label{subsec:evaluate_choi_state}

We start this section by reformulating the second-moment channel $\mathcal{M}_{\bfq, s}^{(2)}$. By the (two-sided) invariance of the Haar measure [cf. \lref{lemma:haar_measure}], we have 
\begin{equation}
\label{eqn:second_moment_channel_isolating_unitary}
\begin{aligned}
\mathcal{M}_{\bfq, s}^{(2)} &= \int \left| \sum_{\lambda \in \Y} \sqrt{q_{\lambda}} \Tr[U^{\dagger}_{\lambda} V_{\lambda}] \right|^2 \mathcal{V}^{\otimes 2} \dd V \\
&= \int \left| \sum_{\lambda \in \Y} \sqrt{q_{\lambda}} \Tr[V_{\lambda}] \right|^2 \left(\mathcal{U} \circ \mathcal{V}\right)^{\otimes 2} \dd V = \mathcal{U}^{\otimes 2} \circ \mathcal{E}_{\bfq, s}, \quad \mathcal{E}_{\bfq, s} = \int \left| \sum_{\lambda \in \Y} \sqrt{q_{\lambda}} \chi_{\lambda}(V) \right|^2 \mathcal{F}^{\otimes 2} \dd V.
\end{aligned}
\end{equation}
It suffices to explicitly evaluate $\mathcal{E}_{\bfq, s}$. Remarkably, the channel $\mathcal{E}_{\bfq, s}$ is collectively unitary-equivariant \cite[Appendix B.2]{Yang_2020}, \ie, it commutes with all 2-fold tensor products of unitary channels $\mathcal{W}^{\otimes 2}$ where $W \in \U(d)$. Consequently, its Choi operator commutes with the mixed tensor $W^{\otimes 2} \otimes {W^*}^{\otimes 2}$, and thus the action of $\mathcal{E}_{\bfq, s}$ can be decomposed according to irreps of $\mathcal{A}_{2, 2}^d$. Recall Definition \ref{def:class_function}, the function $f(V) = \left| \sum_{\lambda \in \Y} \sqrt{q_{\lambda}} \chi_{\lambda}(V) \right|^2$ is a class function, taking $G = \U(d)$ in Corollary \ref{corollary:expression_reformulation_adjoint} and \lref{lemma:mixed_littlewood_richardson}, $\mathcal{E}_{\bfq, s}$ can be reformulated as
\begin{equation}
\label{eqn:channel_reformulation}
\begin{aligned}
\mathcal{E}_{\bfq, s} = \sum_{\nu \in \Y_{2, 2}^d} c_{\nu}^{\bfq} \Pi_{\nu}, \quad c_{\nu}^{\bfq} &= \frac{1}{d_{\nu}} \int \left| \sum_{\lambda \in \Y} \sqrt{q_{\lambda}} \chi_{\lambda}(V) \right|^2 \chi_{\nu}(V) \dd V \\
&= \frac{1}{d_{\nu}} \sum_{\lambda, \mu \in \Y} \sqrt{q_{\lambda} q_{\mu}} \int \chi_{\lambda}(V) \overline{\chi_{\mu}(V)} \chi_{\nu}(V) \dd V =\frac{1}{d_{\nu}} \sum_{\lambda, \mu \in \Y} \sqrt{q_{\lambda} q_{\mu}} C_{\lambda, \nu}^{\mu} \geq 0,
\end{aligned}
\end{equation}
where the action $\mathcal{T}_{2, 2}(U): X \mapsto \mathcal{U}^{\otimes 2}(X)$.
There are six irreps associated with $\mathcal{A}_{2, 2}^d$ [cf. \autoref{fig:Bratteli_2_2}]. For notational clarity, we will index them by
\begin{equation}
\label{eqn:irreps_mixed_diagrams}
\nu_0 = \left(\varnothing, \varnothing\right),\quad \nu_1 = \left(\ybox, \ybox\right), \quad \nu_2 = \left(\yrow, \yrow\right), \quad \nu_{3, 1} = \left(\yrow, \ycol\right), \quad \nu_{3, 2} = \left(\ycol, \yrow\right), \quad \nu_4 = \left(\ycol, \ycol\right).
\end{equation}

\begin{remark}
\label{remark:factor_p_reformulation}
    The recovery coefficient $\p_{\bfq}$ [cf. \lref{lemma:SAR_channel_average_result}] in the previous discussion can be directly associated with a representation-theoretic coefficient. Its analytical expression \cite{Yang_2020, he2025resourcequantificationprogramminglowdepth} is given by
    $$
    \p_{\bfq} = \frac{1}{d^2 - 1} \left( \int \left| \sum_{\lambda \in \Y} \sqrt{q_{\lambda}} \sum_{\mu \in \lambda \otimes \sybox } \Tr\left[U_{\mu}^{\dagger} V_{\mu} \right] \right|^2 \dd V - 1 \right),
    $$
    where $\lambda \otimes \square$ is the collection of Young diagrams obtained by setting $\lambda_j$ to $\lambda_j + 1$ for all feasible $j \in [d]$, promised that the resulting vector is a legitimate Young diagram in $\Y_{s + 1}^d$ \cite{Yang_2020, Yang_2022}.
    By the Schur orthogonality relation and the definition in \eref{eqn:irreps_mixed_diagrams},
    $$
    \begin{aligned}
    \int \left| \sum_{\lambda \in \Y} \sqrt{q_{\lambda}} \sum_{\mu \in \lambda \otimes \sybox } \Tr\left[U_{\mu}^{\dagger} V_{\mu} \right] \right|^2 \dd V &=  \sum_{\lambda, \mu \in \Y} \sqrt{q_{\lambda} q_{\mu}} \sum_{\substack{\gamma \in \lambda \otimes \sybox, \xi \in \mu \otimes \sybox}} \int \chi_{\gamma}(Z) \overline{\chi_{\xi}(Z)} \dd Z \\
    &= \sum_{\lambda, \mu \in \Y} \sqrt{q_{\lambda} q_{\mu}} \left| \left(\lambda \otimes \square \right) \cap \left(\mu \otimes \square \right) \right| \\
    &= \sum_{\lambda, \mu \in \Y} \sqrt{q_{\lambda} q_{\mu}} \left( C_{\lambda, (\sybox, \sybox)}^{\mu} + \delta_{\lambda, \mu} \right) = \dim W_{\nu_1}^{d} \cdot c_{\nu_1}^{\bfq} + 1 = (d^2 - 1) c_{\nu_1}^{\bfq} + 1,
    \end{aligned}
    $$
    where $\delta_{\lambda, \mu}$ stems from the coefficient $C_{\lambda, (\varnothing, \varnothing)}^{\mu} = C_{\lambda, \nu_0}^{\mu}$. Therefore,
    $$
    \p_{\bfq} = \frac{1}{d^2 - 1} \left( \int \left| \sum_{\lambda \in \Y} \sqrt{q_{\lambda}} \sum_{\mu \in \lambda \otimes \sybox } \Tr\left[U_{\mu}^{\dagger} V_{\mu} \right] \right|^2 \dd V - 1 \right) = \frac{\left((d^2 - 1) c_{\nu_1}^{\bfq} + 1\right) - 1}{d^2 - 1} = c_{\nu_1}^{\bfq}.
    $$
\end{remark}

\begin{remark}
    When the ensemble $(\Y, \bfq)$ is clear in the context, we write $c_{\nu}$ for $c_{\nu}^{\bfq}$ for notational brevity.
\end{remark}

As described in Appendix \ref{appendix:variance_reformulation_for_both_estimators}, we require the Choi operator of $\mathcal{M}_{\bfq, s}^{(2)}$ to evaluate the variance of the estimator $\hat{Z}(\hat{\mathsf{\Lambda}}, L)$. Suppose the Hilbert space is indexed by $\mathcal{H}_{A_1} \otimes \mathcal{H}_{A_2} \otimes \mathcal{H}_{B_1} \otimes \mathcal{H}_{B_2}$ [cf. \autoref{fig:local_permutation}]. To evaluate the Choi operator of $\mathcal{M}_{\bfq, s}^{(2)}$, it suffices to evaluate the Choi operator of $\mathcal{E}_{\bfq, s}$, since they are readily connected via
$$
\Phi_{\mathcal{M}_{\bfq, s}^{(2)}} = \left( \mathcal{M}_{\bfq, s}^{(2)} \otimes \mathcal{I} \right)(\Lambda) = \left( \mathcal{U}^{\otimes 2} \circ \mathcal{E}_{\bfq, s}  \otimes \mathcal{I} \right)(\Lambda) = \left( \mathcal{U}^{\otimes 2} \otimes \mathcal{I} \right)(\Phi_{\mathcal{E}_{\bfq, s}}).
$$
Using \eref{eqn:channel_reformulation}, for a fixed ensemble $(\Y, \bfq)$, 
$$
\Phi_{\mathcal{E}_{\bfq, s}} = \left( \mathcal{E}_{\bfq, s} \otimes \mathcal{I} \right)(\Lambda) = \sum_{\nu \in \Y_{2, 2}^d} c_{\nu} \left( \Pi_{\nu} \otimes \mathcal{I} \right)(\Lambda) =: \sum_{\nu \in \Y_{2, 2}^d} c_{\nu} \Xi_{\nu}.
$$
Due to the collective unitary-equivariance of $\mathcal{E}_{\bfq, s}$, and note that the Choi operator $\Phi_{\mathcal{E}_{\bfq, s}}$ is also invariant under the joint local swap $\F_A\F_B(\cdot)\F_A\F_B$ due to $\F_A\F_B (\mathcal{W}^{\otimes 2} \otimes \mathcal{I}^{\otimes 2})(\Ket{\openone} \Bra{\openone}) \F_A\F_B = (\mathcal{W}^{\otimes 2} \otimes \mathcal{I}^{\otimes 2})(\Ket{\openone} \Bra{\openone})$ for all unitary $W \in \U(d)$ and the linearity of integration, $\Phi_{\mathcal{E}_{\bfq, s}}$ is naturally a member of the joint swap-invariant subalgebra of the walled Brauer algebra $\mathcal{A}_{2, 2}^d$, as remarked in Example \ref{example:symmetric_subalgebra_A_2_2_d}. Each operator $\Xi_{\nu}$ can be expressed as a linear combination of the canonical basis elements presented in Example \ref{example:symmetric_subalgebra_A_2_2_d}.

Specifically, the vector space $(\mathbb{C}^d)^{\otimes 2} \otimes ({\mathbb{C}^d}^*)^{\otimes 2}$ can be decomposed as 
$$
(\mathbb{C}^d)^{\otimes 2} \otimes ({\mathbb{C}^d}^*)^{\otimes 2} =\left( \vee^2 \mathbb{C}^d \otimes \vee^2 {\mathbb{C}^d}^* \right) \oplus \left(\wedge^2 \mathbb{C}^d \otimes \wedge^2 {\mathbb{C}^d}^* \right) \oplus \left( \vee^2 \mathbb{C}^d \otimes \wedge^2 {\mathbb{C}^d}^* \right) \oplus \left( \wedge^2 \mathbb{C}^d \otimes \vee^2 {\mathbb{C}^d}^* \right).
$$
The actions on $\left( \vee^2 \mathbb{C}^d \otimes \vee^2 {\mathbb{C}^d}^* \right)$ is contributed by irreps $\nu_0, \nu_1, \nu_2$, $\left(\wedge^2 \mathbb{C}^d \otimes \wedge^2 {\mathbb{C}^d}^* \right)$ contributed by $\nu_0, \nu_1, \nu_4$, $\left( \vee^2 \mathbb{C}^d \otimes \wedge^2 {\mathbb{C}^d}^* \right)$ contributed by $\nu_1, \nu_{3, 1}$, and $\left( \wedge^2 \mathbb{C}^d \otimes \vee^2 {\mathbb{C}^d}^* \right)$ contributed by $\nu_1, \nu_{3, 2}$ [cf. \eref{eqn:irreps_mixed_diagrams}]. The projections onto the half spaces $\vee^2 \mathbb{C}^d$, $\wedge^2 \mathbb{C}^d$ and their duals are given by $\frac{\openone \otimes \openone + \F}{2}$ and $\frac{\openone \otimes \openone - \F}{2}$, respectively, which, embedded into the full space, give rise to the local projections $P_{\pm}$, $Q_{\pm}$ [cf. Example \ref{example:symmetric_subalgebra_A_2_2_d}]. Each component $\Xi_{\nu}$ can, therefore, be written as a linear combination of the canonical basis operators, and sandwiched by $P_{\pm}$ and $Q_{\pm}$ accordingly to be projected onto different subspaces.

With an informal notation, we decompose $\nu_0$ and $\nu_1$ into their multiplicity-resolved components: $\nu_0 =\nu_0^+ \oplus \nu_0^-$, $\nu_1 = \nu_1^{++} \oplus \nu_1^{--} \oplus \nu_1^{+-} \oplus \nu_1^{-+}$, according to their contribution to the actions on different subspaces
\begin{equation}
\label{eqn:components_of_Choi_matrices}
\begin{aligned}
&P_{+} \Lambda Q_{+} = \Xi_{\nu_0^+} + \Xi_{\nu_1^{++}} + \Xi_{\nu_2}, \quad  P_{-} \Lambda Q_{-} = \Xi_{\nu_0^-} + \Xi_{\nu_1^{--}} + \Xi_{\nu_4}, \\
&P_{+} \Lambda Q_{-} = \Xi_{\nu_1^{+-}} + \Xi_{\nu_{3, 1}}, \quad P_{-} \Lambda Q_{+} = \Xi_{\nu_1^{-+}} + \Xi_{\nu_{3, 2}}.
\end{aligned}
\end{equation}
According to the contractive actions of each irrep, they apply to $\Lambda$ differently. We can figure out the expressions for each component: Firstly, under the symmetric and antisymmetric projectors,
\begin{align*}
\Xi_{\nu_0^{+}} &= \frac{1}{\dim \vee^2 \C^d} P_{+}Q_{+} = \frac{2}{d(d+1)} P_{+} Q_{+}, \quad \Xi_{\nu_0^-} = \frac{1}{\dim \wedge^2 \C^d} P_{-}Q_{-} = \frac{2}{d(d-1)} P_{-} Q_{-}.
\end{align*}
For the action of $\nu_1$, recall that $\operatorname{End}(\mathbb C^d)\simeq \mathbf 1\oplus\nu_1$, with $\nu_1$ identifying the traceless operators. For any $a,b\in\{\pm1\}$, define the equivariant map [cf. \cite[Section 9.3]{goodman2009symmetry}] such that
$$
\forall X \in \mathbb{C}^{d \times d}, \quad \mathcal T_{ab}(X):=P_a(X\otimes \openone)P_b,
\quad
\mathcal{U}^{\otimes2} \circ \mathcal T_{ab}(X)
=
\mathcal T_{ab} \circ \mathcal{U}(X).
$$
Hence $\mathcal T_{ab}$ restricted to $\Tr[X]=0$ realizes the action of component $\nu_1^{ab}$.  For any traceless $X,Y$, we expand by terms with $P_a=(\openone+a\F)/2$ and $P_b=(\openone+b\F)/2$ to obtain
$$
\Tr\left[\mathcal T_{ab}(X)^\dagger\mathcal T_{ab}(Y)\right]
= \Tr\left[ P_{b}(X \otimes \openone) P_a P_a (Y \otimes \openone) P_b  \right] = \Tr\left[ (X \otimes \openone) P_a (Y \otimes \openone) P_b \right] =
\frac{d+a+b}{4}\Tr\left[X^\dagger Y\right].
$$
Therefore, pick any orthonormal traceless basis $\{T_\alpha\}_{\alpha=1}^{d^2-1}$, the Choi matrix of $\nu_1^{ab}$ gives
$$
\Xi_{\nu_1^{ab}}
=
\frac{4}{d+a+b}
\sum_{\alpha=1}^{d^2-1}
\mathcal T_{ab}(T_\alpha)\otimes\mathcal T_{ab}(T_\alpha)^*.
$$
Completing the basis by the identity component $T_0=\openone/\sqrt d$, the matrix-unit completeness relation gives
$$
\begin{aligned}
\sum_{\alpha=0}^{d^2-1}
\mathcal T_{ab}(T_\alpha)\otimes\mathcal T_{ab}(T_\alpha)^*
&= \sum_{\alpha = 0}^{d^2 - 1} P_{a} (T_{\alpha} \otimes \openone) P_b \otimes P_a (T_{\alpha}^* \otimes \openone) P_b \\
&= \frac{1}{4} \sum_{\alpha= 0}^{d^2 - 1} P_a\left(  T_{\alpha} \otimes \openone \otimes T_{\alpha} \otimes \openone + \openone \otimes T_{\alpha} \otimes \openone \otimes T_{\alpha} \right) Q_b \\
& \quad + \frac{ab}{4} \sum_{\alpha = 0}^{d^2 - 1} P_{a} \left( \openone \otimes T_{\alpha} \otimes T_{\alpha} \otimes \openone + T_{\alpha} \otimes \openone \otimes \openone \otimes T \right) Q_{b} \\
&= \frac{1}{4} P_a(\Delta+ab\,\Delta_\times)Q_b,
\end{aligned}
$$
and the scalar component is given by
$$
\mathcal T_{ab}(T_0)\otimes\mathcal T_{ab}(T_0)^*
 = \frac{1}{d} P_a P_b \otimes P_b P_a = \frac{\delta_{a, b}}{d} P_a Q_b =
\frac{1+ab}{2d}P_aQ_b.
$$
Subtracting the trivial component yields
$$
\Xi_{\nu_1^{ab}} = \frac{1}{d + a + b} P_{a} \left( \Delta + ab \cdot \Delta_{\times} - \frac{2(1 + ab)}{d} \openone \otimes \openone \right) Q_{b}, \quad \forall a, b \in \{+, -\}.
$$
Correspondingly, via \eref{eqn:components_of_Choi_matrices} we have
\begin{align*}
\Xi_{\nu_2} &= P_{+} \Lambda Q_{+} - \Xi_{\nu_0^+} - \Xi_{\nu_1^{++}} \\
& = P_{+} \left( \Lambda - \frac{2}{d(d+1)} \openone \otimes \openone - \left( \frac{1}{d+2}\left( \Delta + \Delta_{\times}\right) - \frac{4}{d(d+2)} \openone \otimes \openone  \right) \right)Q_{+} \\
&= P_{+} \left(\Lambda - \frac{1}{d+2}\left( \Delta + \Delta_{\times}\right) + \frac{2}{(d+1)(d + 2)} \openone \otimes \openone   \right)Q_{+}, \\
\Xi_{\nu_4} &= P_{-} \Lambda Q_{-} - \Xi_{\nu_0^-} - \Xi_{\nu_1^{--}} \\
& = P_{-} \left( \Lambda - \frac{2}{d(d - 1)} \openone \otimes \openone - \left( \frac{1}{d-2} \left( \Delta + \Delta_{\times}\right)- \frac{4}{d(d-2)} \openone \otimes \openone  \right) \right)Q_{-} \\
&= P_{-} \left( \Lambda - \frac{1}{d - 2}\left( \Delta + \Delta_{\times}\right) + \frac{2}{(d - 1)(d-2)} \openone \otimes \openone \right)Q_{-}, \\
\Xi_{\nu_{3, 1}} &= P_{+} \Lambda Q_{-} - \Xi_{\nu_1^{+-}} = P_{+}\left( \Lambda - \frac{1}{d} \left(\Delta - \Delta_{\times}\right) \right) Q_{-}, \\
\Xi_{\nu_{3, 2}} &= P_{-} \Lambda Q_{+} - \Xi_{\nu_1^{-+}} = P_{-}\left( \Lambda - \frac{1}{d}\left(\Delta - \Delta_{\times}\right) \right) Q_{+}.
\end{align*}
Note that $\nu_{3, 1} = \nu_{3, 2}^*$ by examining their highest weight; these two diagrams introduce the same probe scheme-relevant coefficient $c_{\nu_{3}}$ due to the symmetry of the quadratic form that defines $c_{\nu}$ [cf. \eref{eqn:channel_reformulation}]. Therefore, we can collect the terms and fully expand $\Phi_{\mathcal{E}_{\bfq, s}}$ in the following canonical form:
\begin{equation}
\label{eqn:choi_state_full_expression}
\begin{aligned}
    \Phi_{\mathcal{E}_{\bfq, s}} &= \sum_{\nu \in \Y_{2, 2}^d} c_{\nu} \Xi_{\nu} \\
    &= c_{\nu_0} \left(\frac{2}{d(d+1)} P_{+} Q_{+} + \frac{2}{d(d-1)} P_{-} Q_{-} \right) \\
    &\quad + c_{\nu_1}\bigg( \frac{1}{d+2} P_{+}\left( \Delta + \Delta_{\times} - \frac{4}{d} \openone \otimes \openone \right) Q_{+} + \frac{1}{d -2} P_{-}\left( \Delta + \Delta_{\times} - \frac{4}{d} \openone \otimes \openone \right) Q_{-} \\
    &\qquad + \frac{1}{d} P_{+}\left(\Delta - \Delta_{\times} \right)Q_{-} + \frac{1}{d} P_{-}\left(\Delta - \Delta_{\times} \right) Q_{+}
    \bigg)
    \\
    &\quad + c_{\nu_2} \cdot P_{+} \left(\Lambda - \frac{1}{d+2}\left( \Delta + \Delta_{\times}\right) + \frac{2}{(d+1)(d + 2)} \openone \otimes \openone   \right)Q_{+} \\
    &\quad + c_{\nu_3} \left( P_{+}\left( \Lambda - \frac{1}{d} \left(\Delta - \Delta_{\times}\right) \right) Q_{-} + P_{-}\left( \Lambda - \frac{1}{d}\left(\Delta - \Delta_{\times}\right) \right) Q_{+} \right)
    \\
    &\quad + c_{\nu_4} \cdot P_{-} \left( \Lambda - \frac{1}{d -2}\left( \Delta + \Delta_{\times}\right) + \frac{2}{(d - 1)(d-2)} \openone \otimes \openone \right)Q_{-}
    \\
    &= P_{+}\left( \left(\frac{2(1 - c_{\nu_2})}{d(d+1)} - \frac{4(c_{\nu_1} - c_{\nu_2})}{d(d+2)} \right) \openone \otimes \openone + \frac{c_{\nu_1} - c_{\nu_2}}{d+2} \left( \Delta + \Delta_{\times}\right) + c_{\nu_2} \Lambda \right)Q_{+} \\
    &\quad + P_{-}\left( \left( \frac{2(1-c_{\nu_4})}{d(d-1)} - \frac{4(c_{\nu_1} - c_{\nu_4})}{d(d-2)} \right) \openone \otimes \openone + \frac{c_{\nu_1} - c_{\nu_4}}{d-2} \left( \Delta + \Delta_{\times}\right) + c_{\nu_4} \Lambda \right) Q_{-} \\
    &\quad + P_{+} \left(  \left( \frac{1}{d} c_{\nu_1} - \frac{1}{d} c_{\nu_3} \right) \left(\Delta - \Delta_{\times}\right) + c_{\nu_3} \Lambda \right) Q_{-} + P_{-}\left(  \left( \frac{1}{d} c_{\nu_1} - \frac{1}{d} c_{\nu_3} \right) \left(\Delta - \Delta_{\times}\right) + c_{\nu_3} \Lambda  \right) Q_{+} \\
    &= P_{+} \left( u_{+} \openone \otimes \openone + v_{+} \left(\Delta + \Delta_{\times}\right) + c_{\nu_2} \Lambda \right) Q_{+}  + P_{-} \left( u_{-} \openone \otimes \openone + v_{-} \left(\Delta + \Delta_{\times}\right) + c_{\nu_4} \Lambda \right) Q_{-} \\
    &\quad + P_{+}\left(  r \left(\Delta - \Delta_{\times}\right) + c_{\nu_3} \Lambda \right) Q_{-} + P_{-}\left( r \left(\Delta - \Delta_{\times}\right) + c_{\nu_3} \Lambda \right) Q_{+},
\end{aligned}
\end{equation}
where 
\begin{equation}
\label{eqn:auxiliary_coefficients}
\begin{aligned}
    &u_{+} = \frac{2}{d(d+1)} - \frac{4}{d(d+2)} c_{\nu_1} + \frac{2}{(d + 1)(d+2)} c_{\nu_2}, \quad u_{-} = \frac{2}{d(d-1)} - \frac{4}{d(d-2)} c_{\nu_1} + \frac{2}{(d-1)(d-2)} c_{\nu_4}, \\
    &v_{+} = \frac{c_{\nu_1} - c_{\nu_2}}{d+2}, \quad v_{-} = \frac{c_{\nu_1} - c_{\nu_4}}{d-2}, \quad r = \frac{c_{\nu_1} - c_{\nu_3}}{d}.
\end{aligned}
\end{equation}

\subsection{Evaluating the second-moment channel on tensor product inputs}
According to Appendix \ref{appendix:variance_reformulation_for_both_estimators}, the variance of the na\"ive linear estimator $\hat{X}_j$ can be readily evaluated once we obtain $\mathcal{M}_{\bfq, s}^{(2)}(\rho_0^{\otimes 2})$. When the input $\rho_0 \otimes \rho_0$, the covariance $\mathcal{M}_{\bfq, s}^{(2)} = \mathcal{U}^{\otimes 2} \circ \mathcal{E}_{\bfq, s} = \mathcal{E}_{\bfq, s} \circ \mathcal{U}^{\otimes 2}$ [cf. Appendix \ref{subsec:evaluate_choi_state}] ensures that the input to the channel $\mathcal{E}_{\bfq, s}$ is simply the operator $\mathcal{U}(\rho_0) \otimes \mathcal{U}(\rho_0)$. 

\begin{remark}
    For the rest of the paper, we denote the evolved quantum state $\sigma = \mathcal{U}(\rho)$, and its detraced counterpart $\sigma_0 = \mathcal{U}(\rho_0)$. By the unitarily invariance of the Forbenius norm, $\|\sigma\|_F^2 = \|\rho\|_F^2$ and $\|\sigma_0\|_F^2 = \|\rho_0\|_F^2$.
\end{remark}

It suffices to project the operator $\sigma_0 \otimes \sigma_0$ onto the isotypic subspace of each irrep. Inspired by the expressions in \eref{eqn:choi_state_full_expression} and the conventional state-channel duality $X \mapsto \Tr_2[(\openone \otimes X) \Xi_{\nu}]$ \cite{CHOI1975285}, the projection reads:
\begin{itemize}
    \item For the trivial representation $\nu_0$, note that $c_{\nu_0} = \frac{1}{d_{\nu_0}} \sum_{\lambda, \mu \in \Y} \sqrt{q_{\lambda} q_{\mu}} \delta_{\lambda, \mu} = \sum_{\lambda \in \Y} q_{\lambda} = 1$. By noting that $\chi_{\nu_0}(V) \equiv 1$ for any $V \in \U(d)$, the the action $\Pi_{\nu_0}(\sigma_0 \otimes \sigma_0)$ coincides with the 2-moment Haar twirling operation over $\U(d)$. By \cite[Equation (94)]{Mele2024introductiontohaar},
    $$
    \Pi_{\nu_0}(\sigma_0 \otimes \sigma_0) =  - \frac{\Tr[\sigma_0^2]}{d(d^2 - 1)}  \openone \otimes \openone + \frac{\Tr[\sigma_0^2]}{d^2 - 1} \F.
    $$
    A sanity check can be given by
    \begin{align*}&\Tr\left[(\sigma_0 \otimes \sigma_0) P_{+} \right] \cdot \frac{2}{d(d+1)} P_{+} + \Tr\left[(\sigma_0 \otimes \sigma_0) P_{-} \right] \cdot \frac{2}{d(d-1)} P_{-} \\
    &= \frac{\openone \otimes \openone}{2}\left( \frac{\Tr[\sigma_0^2]}{d(d+1)} - \frac{\Tr[\sigma_0^2]}{d(d-1)}  \right) + \frac{\F}{2} \left( \frac{\Tr[\sigma_0^2]}{d(d+1)} + \frac{\Tr[\sigma_0^2]}{d(d-1)} \right) = - \frac{\Tr[\sigma_0^2]}{d(d^2 - 1)}  \openone \otimes \openone + \frac{\Tr[\sigma_0^2]}{d^2 - 1} \F.
    \end{align*}
    \item For the representation $\nu_1$, we have
    \begin{align*}
        \Pi_{\nu_1^{++}}(\sigma_0 \otimes \sigma_0) &= \frac{1}{d+2} P_{+}\left[ \sigma_0^2 \otimes \openone + \openone \otimes \sigma_0^2 - \frac{2 \Tr[\sigma_0^2]}{d} \openone \otimes \openone \right] \\
        \Pi_{\nu_1^{--}}(\sigma_0 \otimes \sigma_0) &= -\frac{1}{d - 2} P_{-}\left[ \sigma_0^2 \otimes \openone + \openone \otimes \sigma_0^2 - \frac{2 \Tr[\sigma_0^2]}{d} \openone \otimes \openone \right] \\
        \Pi_{\nu_1^{+-}}(\sigma_0 \otimes \sigma_0) &= \Pi_{\nu_1^{-+}}(\sigma_0 \otimes \sigma_0) = 0.
    \end{align*}
    \item For the mixed representation $\nu_3$, since $\sigma_0 \otimes \sigma_0$ is symmetric, the projection $P_{\pm}(\sigma_0 \otimes \sigma_0) P_{\mp}$ vanishes. Meanwhile, the actions corresponding to $P_{\pm}(\Delta - \Delta_{\times}) Q_{\mp}$ vanishes via the same argument as that of $\nu_{1}^{\pm \mp}$.
    
    \item Similarly, for the fully symmetric/anti-symmetric representation $\nu_2$/$\nu_4$, subtracting the correction terms from $\nu_0$ and $\nu_1$ gives
    $$
    \begin{aligned}
    &\Pi_{\nu_2}(\sigma_0 \otimes \sigma_0) = 
    P_{+} (\sigma_0 \otimes \sigma_0) - \frac{1}{d+2} P_{+} ( \sigma_0^2 \otimes \openone + \openone \otimes \sigma_0^2) + \frac{2\Tr[\sigma_0^2]}{(d+1)(d+2)} P_{+} \\
    &\Pi_{\nu_4}(\sigma_0 \otimes \sigma_0) = 
    P_{-} (\sigma_0 \otimes \sigma_0) + \frac{1}{d-2} P_{-} ( \sigma_0^2 \otimes \openone + \openone \otimes \sigma_0^2) - \frac{2\Tr[\sigma_0^2]}{(d-1)(d-2)} P_{-}.
    \end{aligned}
    $$
\end{itemize}
Combining everything, the action of $\mathcal{E}_{\bfq, s}$ on input $\sigma_0 \otimes \sigma_0$ reads
\begin{equation}
\label{eqn:second_moment_channel_on_symmetric_input}
\begin{aligned}
&\mathcal{E}_{\bfq, s}(\sigma_0 \otimes \sigma_0) =  - \frac{\Tr[\sigma_0^2]}{d(d^2 - 1)}  \openone \otimes \openone + \frac{\Tr[\sigma_0^2]}{d^2 - 1} \F  \\
&\qquad + c_{\nu_1} \left[ \frac{1}{d+2} P_{+}\left[ \sigma_0^2 \otimes \openone + \openone \otimes \sigma_0^2 - \frac{2 \Tr[\sigma_0^2]}{d} \openone \otimes \openone \right] -\frac{1}{d - 2} P_{-}\left[ \sigma_0^2 \otimes \openone + \openone \otimes \sigma_0^2 - \frac{2 \Tr[\sigma_0^2]}{d} \openone \otimes \openone \right] \right]
\\
&\qquad + c_{\nu_2} \left[ P_{+} (\sigma_0 \otimes \sigma_0) - \frac{1}{d+2} P_{+} ( \sigma_0^2 \otimes \openone + \openone \otimes \sigma_0^2) + \frac{2\Tr[\sigma_0^2]}{(d+1)(d+2)} P_{+} \right] \\
&\qquad + c_{\nu_4} \left[ P_{-} (\sigma_0 \otimes \sigma_0) + \frac{1}{d-2} P_{-} ( \sigma_0^2 \otimes \openone + \openone \otimes \sigma_0^2) - \frac{2\Tr[\sigma_0^2]}{(d-1)(d-2)} P_{-} \right].
\end{aligned}
\end{equation}

\subsection{Bounding the variance of the quadratic estimator}
There are three main components of the upper bound of the variance of the quadratic estimator $\hat{Z}(\hat{\mathsf{\Lambda}}, L)$, as per \eref{eqn:variance_upper_bound_quadratic_estimator}. The section is dedicated to evaluating these terms, which are listed as follows:
$$
\begin{aligned}
&\text{\hypertarget{variance_term_1}{\textbf{(I)}}}: ~\Tr\left[ \left\{\Ket{U} \Bra{U} + \frac{1-\p_{\bfq}}{d\p_{\bfq}} \openone \otimes \openone, O \otimes \rho_0^T \right\}^{\otimes 2} \widetilde{\Phi}_{\bfq, s}
        \right], \\
        &\text{\hypertarget{variance_term_2}{\textbf{(II)}}}:~ \Tr\left[\left(O \otimes \rho_0^T \right)^{\otimes 2} \widetilde{\Phi}_{\bfq, s}^2 \right], \quad \text{\hypertarget{variance_term_3}{\textbf{(III)}}}:~\Tr\left[ \left(  O \otimes \rho_0^T  \otimes \openone^{\otimes 2}\right) \widetilde{\Phi}_{\bfq, s} \left(\openone^{\otimes 2} \otimes  O \otimes \rho_0^T \right) \widetilde{\Phi}_{\bfq, s}  \right].
\end{aligned}
$$
For notational brevity, we would denote the scalar $\Gamma = \frac{1-\p_{\bfq}}{d\p_{\bfq}}$ when evaluating term \hyperlink{variance_term_1}{(I)}. The evaluation of term \hyperlink{variance_term_2}{(II)} and \hyperlink{variance_term_3}{(III)} will make extensive use of tensor network diagrams, which we refer readers to \cite[Section 6]{Mele2024introductiontohaar} for a comprehensive introduction.

\subsubsection{Bounding variance term \protect\hyperlink{variance_term_1}{(I)}}
\label{appendix:bound_term_1}
Recall that $\widetilde{\Phi}_{\bfq, s} = \mathbb{F}_{A_2 B_1} \Phi_{\mathcal{M}_{\bfq, s}^{(2)}} \mathbb{F}_{A_2 B_1} = \mathbb{F}_{A_2 B_1} \Phi_{\mathcal{U}^{\otimes 2} \circ \mathcal{E}_{\bfq, s}} \mathbb{F}_{A_2 B_1}$, and that
$$
\begin{aligned}
    \left\{\Ket{U} \Bra{U} + \Gamma \openone \otimes \openone, O \otimes \rho_0^T \right\}^{\otimes 2} &= \left(\left\{ \Ket{U} \Bra{U}, O \otimes \rho_0^T \right\} + 2 \Gamma O \otimes \rho_0^T \right)^{\otimes 2}.
\end{aligned}
$$
In the Heisenberg picture, if we denote $\left\{ \Ket{I} \Bra{I}, U^\dagger O U \otimes \rho_0^T \right\} + 2 \Gamma U^\dagger O U \otimes \rho_0^T =: W$, then
$$
\begin{aligned}
    \Tr\left[ \left\{\Ket{U} \Bra{U} + \frac{1-\p_{\bfq}}{d\p_{\bfq}} \openone \otimes \openone, O \otimes \rho_0^T \right\}^{\otimes 2} \widetilde{\Phi}_{\bfq, s}
        \right] &= \Tr\left[\mathbb{F}_{A_2 B_1}\left(W_{A_1 B_1} \otimes W_{A_2 B_2} \right)\mathbb{F}_{A_2 B_1} \Phi_{\mathcal{E}_{\bfq, s}}  \right].
\end{aligned}
$$

By the linearity of the trace, it suffices to evaluate the inner product of the effective observable with the basis elements [cf. \eref{eqn:choi_state_full_expression}]. For any $X_A^{x} = \frac{\openone + x \F_A}{2}$ and $Y_B^{y} = \frac{\openone + y \F_B}{2}$, and $M \in \{\openone, \Delta \pm \Delta_{\times}, \Lambda\}$, then
$$
\begin{aligned}
    \Tr\left[\mathbb{F}_{A_2 B_1}\left(W_{A_1 B_1} \otimes W_{A_2 B_2} \right)\mathbb{F}_{A_2 B_1} X_A^{x} M Y_{B}^y  \right] &= \frac{1}{4} \left[  T_{M}^{++} + x T_{M}^{-+} + y T_{M}^{+-} + xy T_{M}^{--} \right]. 
\end{aligned}
$$
Denote $\sigma_0 = U \rho_0 U^\dagger$, and evaluating the terms gives us
\begin{align*}
    T^{++}_{\openone} &= 4 \Tr[O \sigma_0]^2, \quad T_{\openone}^{-+} = T_{\openone}^{+-} = 2 \left( \Tr[O^2 \sigma_0^2] + \Tr[O \sigma_0 O \sigma_0] \right), \\
    T_{\openone}^{--} &= 2 \Tr[O \sigma_0]^2 + 2(d + 4 \Gamma) \Tr[O^2 \sigma_0^2] + 4 \Gamma^2 \Tr[O^2] \Tr[\sigma_0^2], \\
    T_{\Delta + \Delta_{\times}}^{++} &= 8(d + \Gamma) \Tr[O \sigma_0]^2 + 4 \left( \Tr[O^2 \sigma_0^2] + \Tr[O \sigma_0 O \sigma_0] \right), \\ T_{\Delta + \Delta_{\times}}^{-+} &= T_{\Delta + \Delta_{\times}}^{+-} = 8 \Tr[O \sigma_0]^2 + 4(d + 2 \Gamma) \left( \Tr[O^2 \sigma_0^2] + \Tr[O \sigma_0 O \sigma_0] \right), \\
    T_{\Delta + \Delta_{\times}}^{--} &= (6d + 8\Gamma) \Tr[O \sigma_0]^2 + (2d^2 + 4 + 8d \Gamma + 16\Gamma^2) \Tr[O^2 \sigma_0^2] + 4 \Tr[O \sigma_0 O \sigma_0], \\
    T_{\Delta - \Delta_{\times}}^{++} &= 8 (d + \Gamma) \Tr[O \sigma_0]^2 - 4(\Tr[O^2 \sigma_0^2] + \Tr[O \sigma_0 O \sigma_0]), \quad T_{\Delta - \Delta _{\times}}^{-+} = T_{\Delta - \Delta_{\times}}^{+-} = 0, \\
    T_{\Delta - \Delta_{\times}}^{--} &= -(6d + 8\Gamma) \Tr[O \sigma_0]^2 + (-2d^2 - 8d \Gamma + 4) \Tr[O^2 \sigma_0^2] + 4 \Tr[O \sigma_0 O \sigma_0], \\
    T_{\Lambda}^{++} &= 4(d + \Gamma)^2 \Tr[O \sigma_0]^2, \quad T_{\Lambda}^{-+} = T_{\Lambda}^{+-} = (3d + 4\Gamma) \Tr[O \sigma_0]^2 + (d + 2\Gamma)^2 \Tr[O \sigma_0 O \sigma_0], \\
    T_{\Lambda}^{--} &= (4\Gamma^2 + 2) \Tr[O \sigma_0]^2 + (2d + 8\Gamma) \Tr[O \sigma_0 O \sigma_0].
\end{align*}
Therefore, using Fact \ref{fact:triangular_inequality_application} and Corollary \ref{corollary:inner_product_terms_upper_bound} and keeping only the leading terms, we have
$$
\begin{aligned}
    \max_{x, y \in \{-1, 1\} } \left|\Tr\left[\mathbb{F}_{A_2 B_1}\left(W_{A_1 B_1} \otimes W_{A_2 B_2} \right)\mathbb{F}_{A_2 B_1} X_A^{x} Y_{B}^y  \right] \right| &\leq \mathcal{O}\big(\Tr[O \sigma_0]^2 + d \Tr[O^2 \sigma_0^2] \\
    &\qquad + \left|\Tr[O \sigma_0 O \sigma_0] \right| + \Gamma^2 \Tr[O^2] \Tr[\sigma_0^2]  \big) \\
    &\leq \bigo{ d \P + \Gamma^2 \B\P }, \\
    \max_{x, y \in \{-1, 1\} } \left|\Tr\left[\mathbb{F}_{A_2 B_1}\left(W_{A_1 B_1} \otimes W_{A_2 B_2} \right)\mathbb{F}_{A_2 B_1} X_A^{x} (\Delta \pm \Delta_{\times}) Y_{B}^y  \right] \right| &\leq \bigo{d \Tr[O \sigma_0]^2 + d^2 \Tr[O^2 \sigma_0^2] + d \Tr[O^2] \Tr[\sigma_0^2]} \\
    &\leq \bigo{d^2 \P}, \\
    \max_{x, y \in \{-1, 1\} } \left|\Tr\left[\mathbb{F}_{A_2 B_1}\left(W_{A_1 B_1} \otimes W_{A_2 B_2} \right)\mathbb{F}_{A_2 B_1} X_A^{x} \Lambda Y_{B}^y  \right] \right| &\leq \bigo{d^2 \Tr[O \sigma_0]^2 + d^2 \left|\Tr[O \sigma_0 O \sigma_0] \right|} \\
    &\leq \bigo{d^2 \min\{1, \B\P\} }.
\end{aligned}
$$

Combining the above inequalities, \eref{eqn:choi_state_full_expression}, Fact \ref{fact:triangular_inequality_application}, the variance term \hyperlink{variance_term_1}{(I)} can be bounded from above by
\begin{equation}
\label{eqn:term_1_upper_bound}
\begin{aligned}
&\Tr\left[ \left\{\Ket{U} \Bra{U} + \frac{1-\p_{\bfq}}{d\p_{\bfq}} \openone \otimes \openone, O \otimes \rho_0^T \right\}^{\otimes 2} \widetilde{\Phi}_{\bfq, s}
        \right] \\
        & \quad \leq \left( |u_+| + |u_-| \right) \cdot \max_{x, y \in \{-1, 1\} } \left|\Tr\left[\mathbb{F}_{A_2 B_1}\left(W_{A_1 B_1} \otimes W_{A_2 B_2} \right)\mathbb{F}_{A_2 B_1} X_A^{x} Y_{B}^y  \right] \right| \\
        &\qquad + \left( |v_{+}| + |v_{-}| + 2|r| \right) \cdot
        \max_{x, y \in \{-1, 1\} } \left|\Tr\left[\mathbb{F}_{A_2 B_1}\left(W_{A_1 B_1} \otimes W_{A_2 B_2} \right)\mathbb{F}_{A_2 B_1} X_A^{x} (\Delta \pm \Delta_{\times}) Y_{B}^y  \right] \right|
        \\
        &\qquad \left( c_{\nu_2} + c_{\nu_4} + 2 c_{\nu_3} \right) \cdot\max_{x, y \in \{-1, 1\} } \left|\Tr\left[\mathbb{F}_{A_2 B_1}\left(W_{A_1 B_1} \otimes W_{A_2 B_2} \right)\mathbb{F}_{A_2 B_1} X_A^{x} \Lambda Y_{B}^y  \right] \right| \\
        &\quad\leq \left( |u_+| + |u_-| \right) \cdot \bigo{ d \P + \Gamma^2 \B\P } + \left( |v_{+}| + |v_{-}| + 2|r| \right) \cdot \bigo{d^2 \P} \\
        &\qquad + \left( c_{\nu_2} + c_{\nu_4} + 2 c_{\nu_3} \right) \cdot \bigo{d^2 \min\{1, \B\P\}}.
\end{aligned}
\end{equation}

\subsubsection{Bounding variance term \protect\hyperlink{variance_term_2}{(II)}}
\label{appendix:bound_term_2}
Still, we start by swapping the space and reformulating the expression using the property of $\mathcal{E}_{\bfq, s}$:
\begin{align*}
\Tr\left[\left(O \otimes \rho_0^T \right)^{\otimes 2} \widetilde{\Phi}_{\bfq, s}^2 \right] &= \Tr\left[ \left( O \otimes \rho_0^T \otimes O \otimes \rho_0^T \right) \mathbb{F}_{A_2 B_1} \Phi_{\mathcal{U}^{\otimes 2} \circ \mathcal{E}_{\bfq, s}}^2 \mathbb{F}_{A_2 B_1}\right] \\
&= \Tr\left[ \mathbb{F}_{A_2 B_1}\left( O \otimes \rho_0^T \otimes O \otimes \rho_0^T \right) \mathbb{F}_{A_2 B_1} \Phi_{\mathcal{U}^{\otimes 2} \circ \mathcal{E}_{\bfq, s}}^2 \right] \\
&= \Tr\left[ \left( O  \otimes O \otimes \rho_0^T \otimes \rho_0^T \right)  \Phi_{\mathcal{U}^{\otimes 2} \circ \mathcal{E}_{\bfq, s}}^2 \right] = \Tr\left[ \left( O  \otimes O \otimes \rho_0^T \otimes \rho_0^T \right)  \Phi_{ \mathcal{E}_{\bfq, s} \circ \mathcal{U}^{\otimes 2} }^2 \right] \\
&= \Tr\left[ \left( O  \otimes O \otimes (U\rho_0 U^\dagger)^T \otimes (U\rho_0 U^\dagger)^T \right)  \Phi_{\mathcal{E}_{\bfq, s} }^2 \right] =:\Tr\left[ \left( O  \otimes O \otimes \sigma_0^T \otimes \sigma_0^T \right)  \Phi_{ \mathcal{E}_{\bfq, s}}^2 \right].
\end{align*}

\begin{figure}[tbp!]
    \centering
\includegraphics[width=0.9\linewidth]{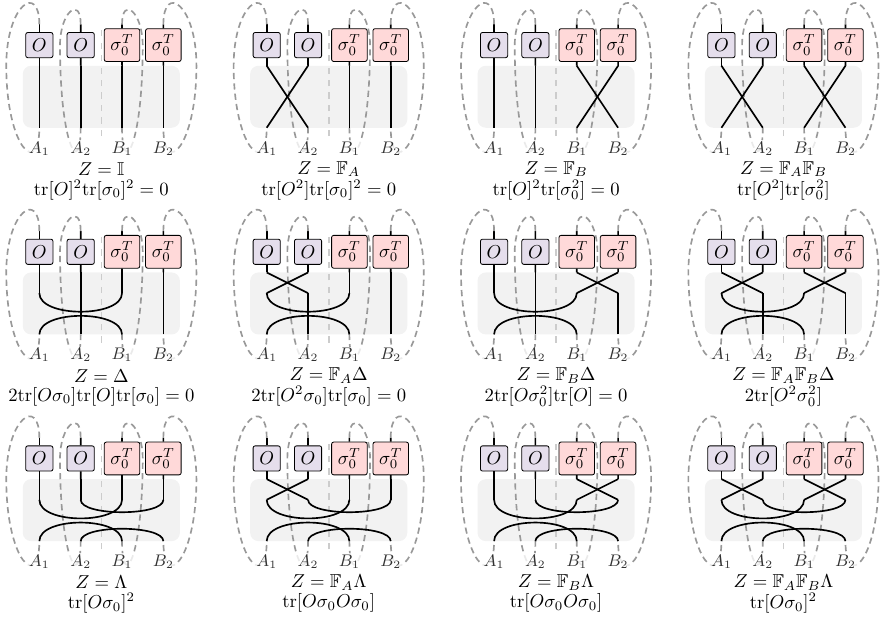}
    \caption{Rules of evaluating $\Tr[(O \otimes O \otimes \sigma_0^T \otimes \sigma_0^T) Z]$ for different choices of $Z$.}
    \label{fig:contraction_rule_2}
\end{figure}

It suffices to evaluate this expression on the product of every pair of canonical bases, according to the contraction rules shown in \autoref{fig:contraction_rule_2}. Moreover, since the observable $O \otimes O \otimes \sigma_0^T \otimes \sigma_0^T$ is invariant under local swaps on the spaces $A$ and $B$, the same results hold after permuting $Z$ with these operations. For conciseness, we define the absolute value of the inner product $\mathsf{E}(Z) = \Tr[(O \otimes O \otimes \rho_0^T \otimes \rho_0^T) Z] $ and, with slight notational overlap, write $\mathsf{E}(X, Y) := \mathsf{E}(XY)$. 
We write $X = P_{a_x} \{\openone, \Delta, \Delta_{\times}, \Lambda\} Q_{b_x}$ and $Y = P_{a_y} \{\openone, \Delta, \Delta_{\times}, \Lambda\} Q_{b_y}$ for $a_x, a_y, b_x, b_y \in \{+, -\}$ [cf. Example \ref{example:symmetric_subalgebra_A_2_2_d}]. Using the joint local permutation invariance and that $\Delta$ and $\Delta_{\times}$ are connected by a single local SWAP, we can evaluate $\mathsf{E}$ on necessary pairs of $X, Y$:

\begin{align*}
    \mathsf{E}\left(P_{a_x} \openone Q_{b_x} ,  P_{a_y} \openone Q_{b_y} \right) &= \mathsf{E}(P_{a_x} P_{a_y} Q_{b_x} Q_{b_y}) = \mathsf{E}(\delta_{a_x, a_y} \delta_{b_x, b_y} P_{a_x} Q_{b_x}),  \\
    &= \frac{\delta_{a_x, a_y} \delta_{b_x, b_y}}{4} \mathsf{E}\left( \openone + a_x \F_A + b_x \F_B + a_x b_x \F_A \F_B \right) \\
    &= \frac{\delta_{a_x, a_y} \delta_{b_x, b_y}}{4}\left(  a_x b_x \Tr[O^2] \Tr[\sigma_0^2] \right),
    \\
    \mathsf{E}\left(P_{a_x} \Delta Q_{b_x},  P_{a_y} \openone Q_{b_y} \right) &= \mathsf{E}\left(P_{a_x} \Delta P_{a_y} Q_{b_x} Q_{b_y}\right) = \mathsf{E}\left(\delta_{b_x, b_y} P_{a_x} \Delta P_{a_y} Q_{b_x}\right)  \\
    &= \frac{\delta_{b_x, b_y}}{8} \mathsf{E}\big( \Delta + a_x \mathbb{F}_A \Delta + a_y \Delta \mathbb{F}_A + a_x a_y \mathbb{F}_A \Delta \mathbb{F}_A \\
    &\quad + b_x \Delta \F_B + a_x b_x \F_A \Delta \F_B + a_y b_x \Delta \F_A \F_B + a_x a_y b_x \F_A \Delta \F_A \F_B \big) \\
    &= \frac{\delta_{b_x, b_y}}{4} \left(a_x + a_y  \right)  \Tr[O^2 \sigma_0^2], \\
    \mathsf{E}\left(P_{a_x} \Delta_{\times} Q_{b_x},  P_{a_y} \openone Q_{b_y} \right) &= a_x a_y \mathsf{E}\left(P_{a_x} \Delta Q_{b_x},  P_{a_y} \openone Q_{b_y} \right),
    \\
    \mathsf{E}\left(P_{a_x} \Lambda Q_{b_x},  P_{a_y} \openone Q_{b_y} \right) &= \mathsf{E}\left(P_{a_x} \Lambda P_{a_y} Q_{b_x} Q_{b_y}\right) = \mathsf{E}\left(\delta_{b_x, b_y} P_{a_x} \Lambda P_{a_y} Q_{b_x}\right) \\
    &= \frac{\delta_{b_x, b_y}}{8} \mathsf{E}\big( \Lambda + a_x \mathbb{F}_A \Lambda + a_y \Lambda \mathbb{F}_A + a_x a_y \mathbb{F}_A \Lambda \mathbb{F}_A  \\
    &\quad + b_x \Lambda \F_B + a_x b_x \F_A \Lambda \F_B + a_y b_x \Lambda \F_A \F_B + a_x a_y b_x \F_A \Lambda \F_A \F_B
    \big)
    \\
    &= \frac{\delta_{b_x, b_y}}{8} \left( ( 1 + a_x a_y + a_x b_y + a_y b_x ) \Tr[O \sigma_0]^2 + (a_x + a_y + b_x + a_x a_y b_x) \Tr[O \sigma_0 O \sigma_0]  \right),  \\
    \mathsf{E}\left(P_{a_x} \openone Q_{b_x}, P_{a_y} \Delta Q_{b_y}\right) &= \mathsf{E}\left( P_{a_x} P_{a_y} Q_{b_x} \Delta Q_{b_y}  \right) = \mathsf{E}\left( \delta_{a_x, a_y} P_{a_x} Q_{b_x} \Delta Q_{b_y} \right) \\
    &= \frac{\delta_{a_x, a_y}}{8}\big( 
    \Delta + b_x \F_B \Delta + b_y \Delta \F_B + b_x b_y \F_B \Delta \F_B
    \\ 
    &\quad + a_x \F_A \Delta + a_x b_x \F_A \F_B \Delta + a_x b_y \F_A \Delta \F_B + a_x b_x b_y \F_A \F_B \Delta \F_B \big) \\
    &= \frac{\delta_{a_x, a_y}}{4} (a_x b_x + a_x b_y) \Tr[O \sigma_0]^2   , \\
    \mathsf{E}\left(P_{a_x} \Delta Q_{b_x}, P_{a_y} \Delta Q_{b_y}\right) &= \mathsf{E}\left( P_{a_x} \Delta P_{a_y} Q_{b_x} \Delta Q_{b_y} \right) \\
    &= \frac{1}{4} \mathsf{E}\left( P_{a_x} \left( d \Delta + 2 \Lambda + (a_y + b_x)(\Delta +  \Delta \F_A \F_B) + a_y b_x (d \Delta + 2 \Lambda) \F_A \F_B  \right) Q_{b_y} \right) \\
    &= \frac{1}{4} \big( d \mathsf{E}(P_{a_x} \Delta Q_{b_y}) + 2 \mathsf{E}(P_{a_x} \Lambda Q_{b_y}) + (a_y + b_x) \mathsf{E}(P_{a_x} \Delta Q_{b_y}) \\
    &\quad + (a_y + b_x) \mathsf{E}(P_{a_x} \Delta \F_A \F_B Q_{b_y}) + d a_y b_x \mathsf{E}(P_{a_x} \Delta \F_A \F_B Q_{b_y}) + 2 a_y b_x  \mathsf{E}(P_{a_x} \Lambda \F_A \F_B Q_{b_y})\big) \\
    &= \frac{1}{8} \big( (d a_x b_y + a_x a_y b_y + a_x b_x b_y + d_{a_y} b_x + a_y + b_x) \Tr[O^2 \sigma_0^2] \\
    &\quad + (1 + a_x b_y + a_y b_x + a_x a_y b_x b_y) \Tr[O \sigma_0]^2 + (a_x + b_y + a_x a_y b_x + a_y b_x b_y) \Tr[O \sigma_0 O \sigma_0] \big), \\
    \mathsf{E}\left(P_{a_x} \Delta_{\times} Q_{b_x}, P_{a_y} \Delta Q_{b_y}\right) &= a_x a_y\mathsf{E}\left(P_{a_x} \Delta Q_{b_x}, P_{a_y} \Delta Q_{b_y}\right) \\
    \mathsf{E}( P_{a_x} \Lambda Q_{b_x}, P_{a_y} \Delta Q_{b_y} ) &= \mathsf{E}( P_{a_x} \Lambda P_{a_y} Q_{b_x} \Delta Q_{b_y} ) \\
    &= \frac{1}{4} \mathsf{E}\left(  P_{a_x} \left(2d \Lambda (1 + a_y b_x) + 2(a_y + b_x) \Lambda   \right) Q_{b_y} \right) = \frac{1}{2} \left( d (1 + a_y b_x) + a_y + b_x \right) \mathsf{E}(P_{a_x} \Lambda Q_{b_y}) \\
    \\
    &= \frac{1}{8} \left( d (1 + a_y b_x) + a_y + b_x \right) \left( (1 + a_x b_y) \Tr[O \sigma_0]^2 + (a_x + b_y) \Tr[O \sigma_0 O \sigma_0] \right),\\
\mathsf{E}( P_{a_x} \openone Q_{b_x}, P_{a_y} \Delta_{\times} Q_{b_y} ) &= b_x b_y \mathsf{E}( P_{a_x} \openone Q_{b_x}, P_{a_y} \Delta Q_{b_y} ), \\
\mathsf{E}( P_{a_x} \Delta Q_{b_x}, P_{a_y} \Delta_{\times} Q_{b_y} ) &= b_x b_y \mathsf{E}( P_{a_x} \Delta Q_{b_x}, P_{a_y} \Delta Q_{b_y} ), \\
    \mathsf{E}( P_{a_x} \Delta_{\times} Q_{b_x}, P_{a_y} \Delta_{\times} Q_{b_y} ) &= a_x a_y b_x b_y \mathsf{E}( P_{a_x} \Delta Q_{b_x}, P_{a_y} \Delta Q_{b_y} ), \\
    \mathsf{E}( P_{a_x} \Lambda Q_{b_x}, P_{a_y} \Delta_{\times} Q_{b_y} ) &= b_x b_y \mathsf{E}( P_{a_x} \Lambda Q_{b_x}, P_{a_y} \Delta Q_{b_y} ), \\
    \mathsf{E}( P_{a_x} \openone Q_{b_x}, P_{a_y} \Lambda Q_{b_y} ) &=  \mathsf{E}( P_{a_x} P_{a_y} Q_{b_x}\Lambda Q_{b_y} ) = \mathsf{E}(\delta_{a_x, a_y} P_{a_x} Q_{b_x}\Lambda Q_{b_y} ) \\
    &= \frac{\delta_{a_x, a_y}}{8} \mathsf{E}\big(  
    \Lambda + b_x \F_{B} \Lambda  + b_y \Lambda \F_B + b_x b_y \F_B \Lambda \F_B \\
    &\quad + a_x \F_A \Lambda + a_x b_x \F_A \F_B \Lambda + a_x b_y \F_A \Lambda \F_B + a_x b_x b_y \F_A \F_B \Lambda \F_B
    \big) \\
    &= \frac{\delta_{a_x, a_y}}{8} \left( (1 + a_x b_x + a_x b_y + b_x b_y) \Tr[O \sigma_0]^2 + (a_x + b_x + b_y + a_x b_x b_y) \Tr[O \sigma_0 O \sigma_0]  \right), \\
    \mathsf{E}( P_{a_x} \Delta Q_{b_x}, P_{a_y} \Lambda Q_{b_y} ) &= \mathsf{E}( P_{a_x} \Delta P_{a_y} Q_{b_x} \Lambda Q_{b_y} ) \\
        &= \frac{1}{4} \mathsf{E}\left(  P_{a_x} \left(  2(1 + a_y b_x) d \Lambda + 2(a_y + b_x) \Lambda 
    \right) Q_{b_y} \right) \\
    &=\frac{1}{8} \left( d (1 + a_y b_x) + a_y + b_x \right) \left( (1 + a_x b_y) \Tr[O \sigma_0]^2 + (a_x + b_y) \Tr[O \sigma_0 O \sigma_0] \right),  \\
        \mathsf{E}( P_{a_x} \Delta_{\times} Q_{b_x}, P_{a_y} \Lambda Q_{b_y} ) &= a_x a_y  \mathsf{E}( P_{a_x} \Delta Q_{b_x}, P_{a_y} \Lambda Q_{b_y} ), \\
    \mathsf{E}( P_{a_x} \Lambda Q_{b_x}, P_{a_y} \Lambda Q_{b_y} ) &= \mathsf{E}( P_{a_x} \Lambda P_{a_y} Q_{b_x}\Lambda Q_{b_y} ) \\
    &= \frac{1}{4} \mathsf{E}\left( P_{a_x} \left( d^2 (1 + a_y b_x) \Lambda + d(a_y + b_x) \Lambda  \right)  Q_{b_y} \right) \\
    &= \frac{1 }{4}\left(d^2 (1 + a_x b_y) + d(a_y + b_x)\right) \mathsf{E}(P_{a_x} \Lambda Q_{b_y}) \\
    &= \frac{1}{8} \left(d^2 (1 + a_x b_y) + d(a_y + b_x)\right) \left( (1 + a_x b_y) \Tr[O \sigma_0]^2 + (a_x + b_y) \Tr[O \sigma_0 O \sigma_0] \right).
\end{align*}

If we denote the upper bound $\mathsf{T}(W, V)  = \max_{a_x, b_x, a_y, b_y} \left|\mathsf{E}( P_{a_x} W Q_{b_x}, P_{a_y} V Q_{b_y} ) \right|$ for $W, V \in \{\openone, \Delta, \Delta_{\times}, \Lambda\}$, by the triangular inequality, it follows immediately that
\begin{equation}
\label{eqn:bilinear_form_1_upper_bound}
\begin{aligned}
    &\mathsf{T}(\openone, \openone) \leq \frac{1}{4} \Tr[O^2] \Tr[\sigma_0^2], \quad \mathsf{T}(\Delta, \openone), \mathsf{T}(\Delta_{\times}, \openone) \leq \frac{1}{2}\Tr[O^2 \sigma_0^2], \quad \mathsf{T}(\Lambda, \openone) \leq \frac{1}{2} \left( \Tr[O \sigma_0]^2 + \left| \Tr[O \sigma_0 O \sigma_0] \right| \right), \\
    &\mathsf{T}(\openone, \Delta), \mathsf{T}(\openone, \Delta_{\times}) \leq \frac{1}{2} \Tr[O \sigma_0]^2 , \\
    &\mathsf{T}(\Delta, \Delta), \mathsf{T}(\Delta_{\times}, \Delta), \mathsf{T}(\Delta, \Delta_{\times}), \mathsf{T}(\Delta_{\times}, \Delta_{\times}) \leq \frac{d+2}{4} \Tr[O^2 \sigma_0^2] + \frac{1}{2} \left( \Tr[O \sigma_0]^2 + \left| \Tr[O \sigma_0 O \sigma_0] \right| \right), \\
    &\mathsf{T}(\Lambda, \Delta), \mathsf{T}(\Lambda, \Delta_{\times}) \leq \frac{d + 1}{2} \left( \Tr[O \sigma_0]^2 + \left| \Tr[O \sigma_0 O \sigma_0] \right| \right), \quad \mathsf{T}(\openone, \Lambda) \leq \frac{1}{2} \left( \Tr[O \sigma_0]^2 + \left| \Tr[O \sigma_0 O \sigma_0] \right| \right), \\
    &\mathsf{T}(\Delta, \Lambda), \mathsf{T}(\Delta_{\times}, \Lambda) \leq \frac{d + 1}{2} \left( \Tr[O \sigma_0]^2 + \left| \Tr[O \sigma_0 O \sigma_0] \right| \right), \quad \mathsf{T}(\Lambda, \Lambda) \leq \frac{d^2 + d}{2}\left( \Tr[O \sigma_0]^2 + \left| \Tr[O \sigma_0 O \sigma_0] \right| \right).
\end{aligned}
\end{equation}
Finally, recall \eref{eqn:choi_state_full_expression} and Fact \ref{fact:triangular_inequality_application}, we use the bilinearity of $\mathsf{E}(\cdot)$ and replace the coefficients by their absolute value so as to bound the values of $\mathsf{E}(\cdot)$ from above by their respective maximum values $\mathsf{T}(\cdot)$:
\begin{equation}
\label{eqn:bilinear_form_2_upper_bound}
\begin{aligned}
   \Tr\left[\left(O \otimes \rho_0^T \right)^{\otimes 2} \widetilde{\Phi}_{\bfq, s}^2 \right] &= \mathsf{E}\left(\Phi_{\mathcal{E}_{\bfq, s}}^2 \right) \\
   &= \mathsf{E}\bigg(\big(P_{+} \left( u_{+} \openone \otimes \openone + v_{+} \left(\Delta + \Delta_{\times} \right) + c_{\nu_2} \Lambda \right) Q_{+}  + P_{-} \left( u_{-} \openone \otimes \openone + v_{-} \left(\Delta + \Delta_{\times} \right) + c_{\nu_4} \Lambda \right) Q_{-} \\
    &\quad \qquad \qquad \qquad + P_{+}\left(  r \left(\Delta - \Delta_{\times} \right) + c_{\nu_3} \Lambda \right) Q_{-} + P_{-}\left(r \left(\Delta - \Delta_{\times} \right) + c_{\nu_3} \Lambda \right) Q_{+} \big)^2 \bigg) \\
    &\leq (|u_{+}| + |u_{-}|)^2 \mathsf{T}(\openone, \openone) + (|v_{+}| + |v_{-}| + 2 |r|)^2 \left(\mathsf{T}(\Delta, \Delta) + \mathsf{T}(\Delta_{\times}, \Delta) + \mathsf{T}(\Delta, \Delta_{\times}) + \mathsf{T}(\Delta_{\times}, \Delta_{\times}) \right) \\
    &\quad + (c_{\nu_2} + c_{\nu_4} + 2 c_{\nu_3})^2 \mathsf{T}(\Lambda, \Lambda) \\
    &\quad + \left( |u_{+}| + |u_{-}|\right)(|v_{+}| + |v_{-}| + 2 |r|)\left[ \mathsf{T}(\openone, \Delta) + \mathsf{T}(\Delta, \openone) + \mathsf{T}(\openone, \Delta_{\times}) + \mathsf{T}(\Delta_{\times}, \openone) \right] \\
    &\quad + \left( |u_{+}| + |u_{-}| \right)(c_{\nu_2} + c_{\nu_4} + 2 c_{\nu_3}) \left[ \mathsf{T}(\openone, \Lambda) + \mathsf{T}(\Lambda, \openone) \right] \\
    &\quad + (|v_{+}| + |v_{-}| + 2 |r|)(c_{\nu_2} + c_{\nu_4} + 2 c_{\nu_3}) \left[ \mathsf{T}(\Delta, \Lambda) + \mathsf{T}(\Lambda, \Delta) + \mathsf{T}(\Delta_{\times}, \Lambda) + \mathsf{T}(\Lambda, \Delta_{\times}) \right].
\end{aligned}
\end{equation}
Combining \eref{eqn:bilinear_form_1_upper_bound} and \eref{eqn:bilinear_form_2_upper_bound} grants us an upper bound for the variance term \hyperlink{variance_term_2}{(II)}.

\subsubsection{Bounding variance term \protect\hyperlink{variance_term_3}{(III)}}
\label{appendix:bound_term_3}
Following the same procedure as Appendix \ref{appendix:bound_term_2}, we reformulate the variance in terms of $\Phi_{\mathcal{E}_{\bfq, s}}$
$$
\begin{aligned}
    &\Tr\left[ \left(  O \otimes \rho_0^T  \otimes \openone^{\otimes 2}\right) \widetilde{\Phi}_{\bfq, s} \left(\openone^{\otimes 2} \otimes  O \otimes \rho_0^T \right) \widetilde{\Phi}_{\bfq, s}  \right] \\
    =~&\Tr\left[ \left(  O \otimes \rho_0^T  \otimes \openone^{\otimes 2}\right) \mathbb{F}_{A_2 B_1} \Phi_{\mathcal{U}^{\otimes 2} \circ \mathcal{E}_{\bfq, s}} \mathbb{F}_{A_2 B_1} \left(\openone^{\otimes 2} \otimes  O \otimes \rho_0^T \right) \mathbb{F}_{A_2 B_1} \Phi_{\mathcal{U}^{\otimes 2} \circ \mathcal{E}_{\bfq, s}} \mathbb{F}_{A_2 B_1}  \right] \\
    =~&\Tr\left[ \mathbb{F}_{A_2 B_1}\left(  O \otimes \rho_0^T  \otimes \openone^{\otimes 2}\right) \mathbb{F}_{A_2 B_1} \Phi_{\mathcal{U}^{\otimes 2} \circ \mathcal{E}_{\bfq, s}} \mathbb{F}_{A_2 B_1} \left(\openone^{\otimes 2} \otimes  O \otimes \rho_0^T \right) \mathbb{F}_{A_2 B_1} \Phi_{\mathcal{U}^{\otimes 2} \circ \mathcal{E}_{\bfq, s}}   \right] \\
    =~&\Tr\left[ \left( O \otimes \openone \otimes \rho_0^T \otimes \openone \right) \Phi_{\mathcal{U}^{\otimes 2} \circ \mathcal{E}_{\bfq, s}} \left( \openone \otimes O \otimes \openone \otimes \rho_0^T \right) \Phi_{\mathcal{U}^{\otimes 2} \circ \mathcal{E}_{\bfq, s}}   \right] \\
    =~&\Tr\left[ \left( O \otimes \openone \otimes \rho_0^T \otimes \openone \right) \Phi_{\mathcal{E}_{\bfq, s} \circ \mathcal{U}^{\otimes 2}} \left( \openone \otimes O \otimes \openone \otimes \rho_0^T \right) \Phi_{\mathcal{E}_{\bfq, s} \circ \mathcal{U}^{\otimes 2}}   \right] \\
    =~&\Tr\left[ \left( O \otimes \openone \otimes (U \rho_0 U^\dagger)^T \otimes \openone \right) \Phi_{\mathcal{E}_{\bfq, s}} \left( \openone \otimes O \otimes \openone \otimes (U \rho_0 U^\dagger)^T \right) \Phi_{\mathcal{E}_{\bfq, s}} \right] \\
    =~&\Tr\left[ \left( O \otimes \openone \otimes \sigma_0^T \otimes \openone \right) \Phi_{\mathcal{E}_{\bfq, s}} \left( \openone \otimes O \otimes \openone \otimes \sigma_0^T \right) \Phi_{\mathcal{E}_{\bfq, s}} \right].
\end{aligned}
$$

\begin{figure}[tbp!]
    \centering
    \includegraphics[width=0.4\linewidth]{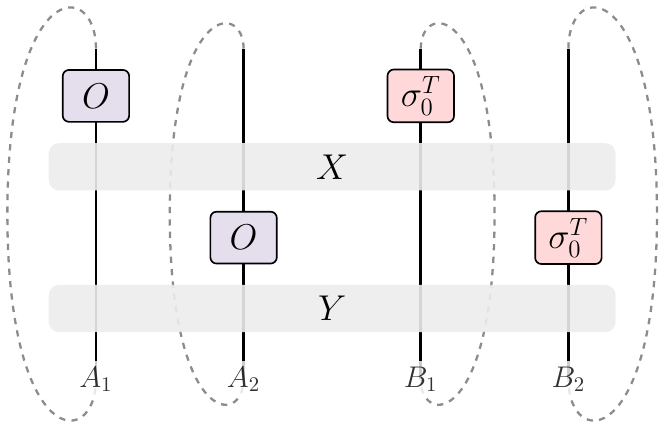}
    \caption{Diagrammatic illustration of evaluating $\Tr[\left( O \otimes \openone \otimes \sigma_0^T \otimes \openone \right) X\left( \openone \otimes O \otimes \openone \otimes \sigma_0^T \right) Y ]$.}
    \label{fig:contraction_rule_3}
\end{figure}

Define the bilinear form $\mathsf{B}(X, Y) = \left|\Tr[\left( O \otimes \openone \otimes \sigma_0^T \otimes \openone \right) X\left( \openone \otimes O \otimes \openone \otimes \sigma_0^T \right) Y ] \right|$ for $X = P_{a_x} \{\openone, \Delta, \Delta_{\times}, \Lambda\} Q_{b_x}$ and $Y = P_{a_y} \{\openone, \Delta, \Delta_{\times}, \Lambda\} Q_{b_y}$, and define $\mathsf{S}(W, V) = \max_{a_x, b_x, a_y, b_y} \mathsf{B}(P_{a_x} W Q_{b_x}, P_{a_y} V Q_{b_y} )$. The rule for evaluating $\mathsf{B}(X, Y)$ for each $X, Y$ is illustrated in \autoref{fig:contraction_rule_3}. Here we hide the internal wire-connecting structure of the basis elements obtained via expanding each $X, Y$. By the positivity of the terms in the evaluation outcomes, we directly derive upper bounds for $\mathsf{S}$ similar to \eref{eqn:bilinear_form_1_upper_bound}.

\begin{align*}
    \mathsf{S}(\openone, \openone) &\leq \frac{1}{16} \big( \b(\openone, \openone) + \b(\openone, \F_A) + \b(\openone, \F_B) + \b(\openone, \F_A \F_B) \\
    &\quad + \b(\F_A, \openone) + \b(\F_A, \F_A) + \b(\F_A, \F_B) + \b(\F_A, \F_A \F_B) \\
    &\quad + \b(\F_B, \openone) + \b(\F_B, \F_A) + \b(\F_B, \F_B) + \b(\F_B, \F_A \F_B) \\
    &\quad + \b(\F_A\F_B, \openone) + \b(\F_A\F_B, \F_A) + \b(\F_A\F_B, \F_B) + \b(\F_A\F_B, \F_A \F_B)
    \big) \\
    &\leq \frac{1}{16} (d + 2)^2 \Tr[O^2]  \Tr[\sigma_0^2] , \\
    \mathsf{S}(\openone, \Delta) &\leq \frac{1}{16} \big( \b(\openone, \Delta) + \b(\openone, \F_A \Delta) + \b(\openone,  \Delta\F_B) + \b(\openone, \F_A \Delta \F_B) \\
    &\quad + \b(\F_A, \Delta) + \b(\F_A,  \F_A \Delta) + \b(\F_A, \Delta\F_B) + \b(\F_A, \F_A \Delta \F_B) \\
    &\quad + \b(\F_B, \Delta) + \b(\F_B, \F_A \Delta) + \b(\F_B, \Delta\F_B) + \b(\F_B, \F_A \Delta \F_B) \\
    &\quad + \b(\F_A\F_B, \Delta) + \b(\F_A\F_B, \F_A \Delta) + \b(\F_A\F_B, \Delta\F_B) + \b(\F_A\F_B, \F_A \Delta \F_B)
    \big) \\
    &\leq \frac{1}{16} \left((4d+8) \Tr[O^2 \sigma_0^2] + (2d + 4) \Tr[O^2] \Tr[\sigma_0^2]  \right), \\
    \mathsf{S}(\openone, \Delta_{\times}) &\leq \frac{1}{16} \big( \b(\openone, \Delta_{\times}) + \b(\openone, \F_A \Delta_{\times}) + \b(\openone, \Delta_{\times}\F_B) + \b(\openone, \F_A \Delta_{\times} \F_B) \\
    &\quad + \b(\F_A, \Delta_{\times}) + \b(\F_A, \F_A \Delta_{\times}) + \b(\F_A, \Delta_{\times}\F_B) + \b(\F_A, \F_A \Delta_{\times} \F_B) \\
    &\quad + \b(\F_B, \Delta_{\times}) + \b(\F_B, \F_A \Delta_{\times}) + \b(\F_B, \Delta_{\times}\F_B) + \b(\F_B, \F_A \Delta_{\times} \F_B) \\
    &\quad + \b(\F_A\F_B, \Delta_{\times}) + \b(\F_A\F_B, \F_A \Delta_{\times}) + \b(\F_A\F_B, \Delta_{\times}\F_B) + \b(\F_A\F_B, \F_A \Delta_{\times} \F_B)
    \big) \\
    &\leq \frac{1}{16}\left[ (d^2 + 4d + 8) \Tr[O^2 \sigma_0^2] + (d + 4) \Tr[O^2] \Tr[\sigma_0^2] \right],
    \\
\mathsf{S}(\openone, \Lambda) &\leq \frac{1}{16} \big( \b(\openone, \Lambda) + \b(\openone, \F_A \Lambda) + \b(\openone,  \Lambda\F_B) + \b(\openone, \F_A \Lambda \F_B) \\
    &\quad + \b(\F_A, \Lambda) + \b(\F_A,  \F_A \Lambda) + \b(\F_A, \Lambda\F_B) + \b(\F_A, \F_A \Lambda \F_B) \\
    &\quad + \b(\F_B, \Lambda) + \b(\F_B, \F_A \Lambda) + \b(\F_B, \Lambda\F_B) + \b(\F_B, \F_A \Lambda \F_B) \\
    &\quad + \b(\F_A\F_B, \Lambda) + \b(\F_A\F_B, \F_A \Lambda) + \b(\F_A\F_B, \Lambda\F_B) + \b(\F_A\F_B, \F_A \Lambda \F_B)
    \big) \\
    &\leq \frac{1}{16}\left( 2 \Tr[O \sigma_0]^2 + 2\left| \Tr[O \sigma_0 O \sigma_0] \right| + (d + 6) \Tr[O^2 \sigma_0^2] + \Tr[O^2] \Tr[\sigma_0^2]  \right), \\
\mathsf{S}(\Delta, \openone) &= \max_{a_x, b_x, a_y, b_y} \mathsf{B}(P_{a_x} \Delta Q_{b_x}, P_{a_y} \openone Q_{b_y} ) = \max_{a_x, b_x, a_y, b_y} \mathsf{B}(\F_A \F_B P_{a_y} \openone Q_{b_y} \F_A \F_B,  \F_A \F_B P_{a_x} \Delta Q_{b_x} \F_A \F_B) \\
&= \max_{a_x, b_x, a_y, b_y} \mathsf{B}(P_{a_y} \openone Q_{b_y},   P_{a_x} \Delta Q_{b_x} ) = \mathsf{S}(\openone, \Delta) \\
&\leq \frac{1}{16} \left((4d+8) \Tr[O^2 \sigma_0^2] + (2d + 4) \Tr[O^2] \Tr[\sigma_0^2]  \right), \\
\mathsf{S}(\Delta, \Delta) &\leq \frac{1}{16} \big( \b(\Delta, \Delta) + \b(\Delta, \F_A \Delta) + \b(\Delta, \Delta \F_B) + \b(\Delta, \F_A \Delta  \F_B)  \\
&\quad + \b(\F_A \Delta, \Delta) + \b(\F_A \Delta, \F_A \Delta) + \b(\F_A \Delta, \Delta \F_B) + \b(\F_A \Delta, \F_A \Delta \F_B) \\
&\quad + \b(\Delta \F_B, \Delta) + \b(\Delta \F_B, \F_A \Delta) + \b(\Delta \F_B, \Delta \F_B) + \b(\Delta \F_B, \F_A \Delta  \F_B) \\
&\quad + \b(\F_A\Delta \F_B, \Delta) + \b(\F_A\Delta \F_B, \F_A \Delta) + \b(\F_A\Delta \F_B, \Delta \F_B) + \b(\F_A\Delta \F_B, \F_A \Delta  \F_B)
\big) \\
&\leq \frac{1}{16}\big( (2d + 4)\left| \Tr[O \sigma_0 O \sigma_0] \right| + (2d + 12)\Tr[O^2 \sigma_0^2] + 4 \Tr[O^2] \Tr[\sigma_0^2] + 4\Tr[O \sigma_0]^2   \big),
\\
\mathsf{S}(\Delta, \Delta_{\times}) &\leq \frac{1}{16} \big( \b(\Delta, \Delta_{\times}) + \b(\Delta, \F_A \Delta_{\times}) + \b(\Delta, \Delta_{\times} \F_B) + \b(\Delta, \F_A \Delta_{\times} \F_B) \\
&\quad + \b(\F_A \Delta, \Delta_{\times}) + \b(\F_A \Delta, \F_A \Delta_{\times}) + \b(\F_A \Delta, \Delta_{\times} \F_B) + \b(\F_A \Delta, \F_A \Delta_{\times} \F_B) \\
&\quad + \b(\Delta \F_B, \Delta_{\times}) + \b(\Delta \F_B, \F_A \Delta_{\times}) + \b(\Delta \F_B, \Delta_{\times} \F_B) + \b(\Delta \F_B, \F_A \Delta_{\times} \F_B) \\
&\quad + \b(\F_A\Delta \F_B, \Delta_{\times}) + \b(\F_A\Delta \F_B, \F_A \Delta_{\times}) + \b(\F_A\Delta \F_B, \Delta_{\times} \F_B) + \b(\F_A\Delta \F_B, \F_A \Delta_{\times} \F_B)
\big) \\
&\leq \frac{1}{16}\left[ (4d+12) \Tr[O^2 \sigma_0^2] + (2d + 4) \Tr[O^2] \Tr[\sigma_0^2] + (2d+4) \Tr[O \sigma_0]^2 + 4 \left| \Tr[O \sigma_0 O \sigma_0] \right| \right],
\\
\mathsf{S}(\Delta, \Lambda) &\leq \frac{1}{16} \big( \b(\Delta, \Lambda) + \b(\Delta, \F_A \Lambda) + \b(\Delta, \Lambda \F_B) + \b(\Delta, \F_A \Lambda  \F_B)  \\
&\quad + \b(\F_A \Delta, \Lambda) + \b(\F_A \Delta, \F_A \Lambda) + \b(\F_A \Delta, \Lambda \F_B) + \b(\F_A \Delta, \F_A \Lambda \F_B) \\
&\quad + \b(\Delta \F_B, \Lambda) + \b(\Delta \F_B, \F_A \Lambda) + \b(\Delta \F_B, \Lambda \F_B) + \b(\Delta \F_B, \F_A \Lambda  \F_B) \\
&\quad + \b(\F_A\Delta \F_B, \Lambda) + \b(\F_A\Delta \F_B, \F_A \Lambda) + \b(\F_A\Delta \F_B, \Lambda \F_B) + \b(\F_A\Delta \F_B, \F_A \Lambda \F_B)
\big) \\
&\leq \frac{1}{16}\left( (2d + 4) \left(\Tr[O^2 \sigma_0^2] + \Tr[O \sigma_0]^2 + \left| \Tr[O \sigma_0 O \sigma_0] \right| \right)  \right),
\\
\mathsf{S}(\Delta_{\times}, \openone) &= \max_{a_x, b_x, a_y, b_y} \mathsf{B}(P_{a_x} \Delta_{\times} Q_{b_x}, P_{a_y} \openone Q_{b_y} ) = \max_{a_x, b_x, a_y, b_y} \mathsf{B}(\F_A \F_B P_{a_y} \openone Q_{b_y} \F_A \F_B,  \F_A \F_B P_{a_x} \Delta_{\times} Q_{b_x} \F_A \F_B) \\
&= \max_{a_x, b_x, a_y, b_y} \mathsf{B}(P_{a_y} \openone Q_{b_y},   P_{a_x} \Delta_{\times} Q_{b_x} ) = \mathsf{S}(\openone, \Delta_{\times}) \\
&\leq \frac{1}{16}\left[ (d^2 + 4d + 8) \Tr[O^2 \sigma_0^2] + (d + 4) \Tr[O^2] \Tr[\sigma_0^2] \right],
\\
\mathsf{S}(\Delta_{\times}, \Delta) &= \max_{a_x, b_x, a_y, b_y} \mathsf{B}(P_{a_x} \Delta_{\times} Q_{b_x}, P_{a_y} \Delta Q_{b_y}) = \max_{a_x, b_x, a_y, b_y} \mathsf{B}(\F_A \F_B P_{a_y} \Delta Q_{b_y} \F_A \F_B, \F_A \F_B P_{a_x} \Delta_{\times} Q_{b_x} \F_A \F_B) \\
&= \max_{a_x, b_x, a_y, b_y} \mathsf{B}( P_{a_y} \Delta Q_{b_y} ,  P_{a_x} \Delta_{\times} Q_{b_x} ) = \mathsf{S}(\Delta, \Delta_{\times}  ) \\
&\leq \frac{1}{16}\left[ (4d+12) \Tr[O^2 \sigma_0^2] + (2d + 4) \Tr[O^2] \Tr[\sigma_0^2] + (2d+4) \Tr[O \sigma_0]^2 + 4 \left| \Tr[O \sigma_0 O \sigma_0] \right| \right], \\
\mathsf{S}(\Delta_{\times}, \Delta_{\times}) &\leq \frac{1}{16} \big( \b(\Delta_{\times}, \Delta_{\times}) + \b(\Delta_{\times}, \F_A \Delta_{\times}) + \b(\Delta_{\times}, \Delta_{\times} \F_B) + \b(\Delta_{\times}, \F_A \Delta_{\times} \F_B) \\
&\quad + \b(\F_A \Delta_{\times}, \Delta_{\times}) + \b(\F_A \Delta_{\times}, \F_A \Delta_{\times}) + \b(\F_A \Delta_{\times}, \Delta_{\times} \F_B) + \b(\F_A \Delta_{\times}, \F_A \Delta_{\times} \F_B) \\
&\quad + \b(\Delta_{\times} \F_B, \Delta_{\times}) + \b(\Delta_{\times} \F_B, \F_A \Delta_{\times}) + \b(\Delta_{\times} \F_B, \Delta_{\times} \F_B) + \b(\Delta_{\times} \F_B, \F_A \Delta_{\times} \F_B) \\
&\quad + \b(\F_A\Delta_{\times} \F_B, \Delta_{\times}) + \b(\F_A\Delta_{\times} \F_B, \F_A \Delta_{\times}) + \b(\F_A\Delta_{\times} \F_B, \Delta_{\times} \F_B) + \b(\F_A\Delta_{\times} \F_B, \F_A \Delta_{\times} \F_B)
\big) \\
&\leq \frac{1}{16} \left[(4d + 12) \Tr[O^2 \sigma_0^2] + (d^2 + 4d + 6) \Tr[O^2] \Tr[\sigma_0^2] + (d^2 + 4d + 8) \Tr[O \sigma_0]^2 + (2d + 4) \left| \Tr[O \sigma_0 O \sigma_0] \right|  \right],
\\
\mathsf{S}(\Delta_{\times}, \Lambda) &\leq \frac{1}{16} \big( \b(\Delta_{\times}, \Lambda) + \b(\Delta_{\times}, \F_A \Lambda) + \b(\Delta_{\times}, \Lambda \F_B) + \b(\Delta_{\times}, \F_A \Lambda \F_B) \\
&\quad + \b(\F_A \Delta_{\times}, \Lambda) + \b(\F_A \Delta_{\times}, \F_A \Lambda) + \b(\F_A \Delta_{\times}, \Lambda \F_B) + \b(\F_A \Delta_{\times}, \F_A \Lambda \F_B) \\
&\quad + \b(\Delta_{\times} \F_B, \Lambda) + \b(\Delta_{\times} \F_B, \F_A \Lambda) + \b(\Delta_{\times} \F_B, \Lambda \F_B) + \b(\Delta_{\times} \F_B, \F_A \Lambda \F_B) \\
&\quad + \b(\F_A\Delta_{\times} \F_B, \Lambda) + \b(\F_A\Delta_{\times} \F_B, \F_A \Lambda) + \b(\F_A\Delta_{\times} \F_B, \Lambda \F_B) + \b(\F_A\Delta_{\times} \F_B, \F_A \Lambda \F_B)
\big) \\
&\leq \frac{1}{16} \left[(d^2 + 4d + 6) \Tr[O^2 \sigma_0^2] + (d + 4) \Tr[O \sigma_0]^2 + 2 \left| \Tr[O \sigma_0 O \sigma_0] \right| \right],
\\
\mathsf{S}(\Lambda, \openone) &= \max_{a_x, b_x, a_y, b_y} \mathsf{B}(P_{a_x} \Lambda Q_{b_x},   P_{a_y} \openone Q_{b_y} ) = \max_{a_x, b_x, a_y, b_y} \mathsf{B}(\F_A \F_B  P_{a_y} \openone Q_{b_y} \F_A \F_B, \F_A \F_B P_{a_x} \Lambda Q_{b_x} \F_A \F_B ) \\
&= \max_{a_x, b_x, a_y, b_y} \mathsf{B}( P_{a_y} \openone Q_{b_y} , P_{a_x} \Lambda Q_{b_x} ) = \mathsf{S}(\openone, \Lambda) \\
&\leq \frac{1}{16}\left( 2 \Tr[O \sigma_0]^2 + 2\left| \Tr[O \sigma_0 O \sigma_0] \right| + (d + 6) \Tr[O^2 \sigma_0^2] + \Tr[O^2] \Tr[\sigma_0^2]  \right), \\
\mathsf{S}(\Lambda, \Delta) &= \max_{a_x, b_x, a_y, b_y} \mathsf{B}(P_{a_x} \Lambda Q_{b_x},   P_{a_y} \Delta Q_{b_y} ) = \max_{a_x, b_x, a_y, b_y} \mathsf{B}(\F_A \F_B  P_{a_y} \Delta Q_{b_y} \F_A \F_B, \F_A \F_B P_{a_x} \Lambda Q_{b_x} \F_A \F_B ) \\
&= \max_{a_x, b_x, a_y, b_y} \mathsf{B}( P_{a_y} \Delta Q_{b_y} , P_{a_x} \Lambda Q_{b_x} ) = \mathsf{S}(\Delta, \Lambda) \\
&\leq \frac{1}{16}\left((2d + 4) \left(\Tr[O^2 \sigma_0^2] + \Tr[O \sigma_0]^2 + \left| \Tr[O \sigma_0 O \sigma_0] \right| \right)\right), \\
\mathsf{S}(\Lambda, \Delta_{\times}) &= \max_{a_x, b_x, a_y, b_y} \mathsf{B}(P_{a_x} \Lambda Q_{b_x},   P_{a_y} \Delta_{\times} Q_{b_y} ) = \max_{a_x, b_x, a_y, b_y} \mathsf{B}(\F_A \F_B  P_{a_y} \Delta_{\times} Q_{b_y} \F_A \F_B, \F_A \F_B P_{a_x} \Lambda Q_{b_x} \F_A \F_B ) \\
&= \max_{a_x, b_x, a_y, b_y} \mathsf{B}( P_{a_y} \Delta_{\times} Q_{b_y} , P_{a_x} \Lambda Q_{b_x} ) = \mathsf{S}(\Delta_{\times}, \Lambda) \\
&\leq \frac{1}{16} \left[(d^2 + 4d + 6) \Tr[O^2 \sigma_0^2] + (d + 4) \Tr[O \sigma_0]^2 + 2 \left| \Tr[O \sigma_0 O \sigma_0] \right| \right],
\\
\mathsf{S}(\Lambda, \Lambda) &\leq \frac{1}{16} \big( \b(\Lambda, \Lambda) + \b(\Lambda, \F_A \Lambda) + \b(\Lambda, \Lambda \F_B) + \b(\Lambda, \F_A \Lambda  \F_B)  \\
&\quad + \b(\F_A \Lambda, \Lambda) + \b(\F_A \Lambda, \F_A \Lambda) + \b(\F_A \Lambda, \Lambda \F_B) + \b(\F_A \Lambda, \F_A \Lambda \F_B) \\
&\quad + \b(\Lambda \F_B, \Lambda) + \b(\Lambda \F_B, \F_A \Lambda) + \b(\Lambda \F_B, \Lambda \F_B) + \b(\Lambda \F_B, \F_A \Lambda  \F_B) \\
&\quad + \b(\F_A\Lambda \F_B, \Lambda) + \b(\F_A\Lambda \F_B, \F_A \Lambda) + \b(\F_A\Lambda \F_B, \Lambda \F_B) + \b(\F_A\Lambda \F_B, \F_A \Lambda \F_B)
\big) \\
&\leq \frac{1}{16}(d + 2)^2 \Tr[O \sigma_0]^2.
\end{align*}
Analogously, by the bilinearity of $\mathsf{B}(\cdot)$ and Fact \ref{fact:triangular_inequality_application}, we bound each $\mathsf{B}(X, Y)$ from above by the corresponding maximum values $\mathsf{S}(\cdot)$ and replace the coefficients by their absolute values. An upper bound is readily obtained
\begin{equation}
\label{eqn:bilinear_form_3_upper_bound}
\begin{aligned}
    &\Tr\left[ \left( O \otimes \openone \otimes \sigma_0^T \otimes \openone \right) \Phi_{\mathcal{E}_{\bfq, s}} \left( \openone \otimes O \otimes \openone \otimes \sigma_0^T \right) \Phi_{\mathcal{E}_{\bfq, s}} \right] = \b\left( \Phi_{\mathcal{E}_{\bfq, s}}, \Phi_{\mathcal{E}_{\bfq, s}} \right) \\
    &\leq (|u_{+}| + |u_{-}|)^2 \mathsf{S}(\openone, \openone) + (|v_{+}| + |v_{-}| + 2|r|)^2 \left[\mathsf{S}(\Delta, \Delta) + \mathsf{S}(\Delta_{\times}, \Delta_{\times}) + \mathsf{S}(\Delta_{\times}, \Delta) + \mathsf{S}(\Delta, \Delta_{\times}) \right] \\
    &\quad +(c_{\nu_2} + c_{\nu_4} + 2 c_{\nu_3})^2 \mathsf{S}(\Lambda, \Lambda) \\
    &\quad + (|u_{+}| + |u_{-}|)(|v_{+}| + |v_{-}| + 2|r|) \left[ \mathsf{S}(\openone, \Delta) + \mathsf{S}(\Delta, \openone) + \mathsf{S}(\openone, \Delta_{\times}) + \mathsf{S}(\Delta_{\times}, \openone) \right] \\
    &\quad + (|u_{+}| + |u_{-}|)(c_{\nu_2} + c_{\nu_4} + 2c_{\nu_3})\left[ \mathsf{S}(\openone, \Lambda) + \mathsf{S}(\Lambda, \openone) \right] \\
    &\quad +(|v_{+}| + |v_{-}| + 2|r|)(c_{\nu_2} + c_{\nu_4} + 2 c_{\nu_3}) \left[ \mathsf{S}(\Delta, \Lambda) + \mathsf{S}(\Lambda, \Delta) + \mathsf{S}(\Delta_{\times}, \Lambda) + \mathsf{S}(\Lambda, \Delta_{\times}) \right].
\end{aligned}
\end{equation}
Combining the above inequalities yields an upper bound for the variance term \hyperlink{variance_term_3}{(III)}. 

\subsection{Proof of \lref{lemma:small_s_regime_estimator_Lambda_variance}}\label{appendix:Proof_small_s_regime_estimator_Lambda_variance}
We are now ready to prove \lref{lemma:small_s_regime_estimator_Lambda_variance}. We restrict $s \leq d$ and choose the ensemble $(\Y, \bfq)$ as the Plancherel measure \cite{Borodin2000} supported on all legitimate partitions $\lambda \vdash_d s$:
$$
\Y = \Y_{s}^d, \quad q_{\lambda} := \frac{s_{\lambda}^2}{s!}.
$$
Then it suffices to evaluate the coefficients $c_{\nu_1}, c_{\nu_2}, c_{\nu_3}$ and $c_{\nu_4}$. We will use the following lemma:

\begin{lemma}
\label{lemma:plancherel_coefficient_evaluation_rule}
    For a mixed Young diagram $\nu = (\nu_{\ell}, \nu_r)$ where $|\nu_{\ell}| = |\nu_r| = t$, for $d \geq n \geq t$, it holds that
    $$
    \sum_{\lambda, \mu \vdash_d n } s_{\lambda} s_{\mu} C_{\lambda, \nu}^{\mu} = \binom{n}{t}^2 (n - t)! s_{\nu_{\ell}} s_{\nu_r}.
    $$
\end{lemma}
\begin{proof}[Proof of \lref{lemma:plancherel_coefficient_evaluation_rule}]
    Using Lemmas \ref{lemma:littlewood_richardson} and \ref{lemma:mixed_littlewood_richardson} and rearranging, it follows that
    \begin{align*}
        \sum_{\lambda, \mu \vdash_d n } s_{\lambda} s_{\mu} C_{\lambda, \nu}^{\mu} &= \sum_{\lambda, \mu \vdash_d n} s_{\lambda } s_{\mu} \sum_{\gamma \vdash_d n - t} C_{\nu_{\ell}, \gamma}^{\lambda} C_{\nu_r, \gamma}^{\mu} \\
        &= \sum_{\gamma \vdash_d n - t} \left( \sum_{\lambda \vdash_d n} s_{\lambda} C_{\nu_{\ell}, \gamma}^{\lambda} \right) \left( \sum_{\mu \vdash_d n} s_{\mu} C_{\nu_r, \gamma}^{\mu} \right) \\
        &= \binom{n}{t}^2 \sum_{\gamma \vdash_d n - t} s_{\gamma}^2 s_{\nu_{\ell}} s_{\nu_r} = \binom{n}{t}^2 (n - t)! s_{\nu_{\ell}} s_{\nu_r},
    \end{align*}
    concluding the proof. Graphically, using the language of \cite{cox2008blocks}, the leading coefficient of $s_{\nu_{\ell}} s_{\nu_r}$ corresponds to counting the number of partial one-row $(n, n, n - t)$ diagrams, where we have $n - t$ northern arcs due to contraction, freely paired up by some permutation in $\mathfrak{S}_{n - t}$. On both sides of the wall, there are $\binom{n}{t}$ choices of non-contracting propagating wires. A simple combinatorial argument yields the coefficient $\binom{n}{t}^2 (n-t)!$.
\end{proof}

Using \lref{lemma:plancherel_coefficient_evaluation_rule}, since each $\nu \in \{\nu_1, \nu_2, \nu_3, \nu_4\}$ satisfies $\nu = (\nu_{\ell}, \nu_{r})$ with $|\nu_{\ell}| = |\nu_r| \leq s \leq d$, we have
$$
\begin{aligned}
c_{\nu} &= \frac{1}{d_{\nu}} \sum_{\lambda, \mu \in \Y } \sqrt{q_{\lambda} q_{\mu}} C_{\lambda, \nu}^{\mu} \\
&= \frac{1}{d_{\nu} s!} \sum_{\lambda, \mu \vdash_d s} s_{\lambda} s_{\mu} C_{\lambda, \nu}^{\mu} = \frac{1}{d_{\nu} s!} \binom{s}{|\nu_{\ell}|}^2 (s - |\nu_{\ell}|)! s_{\nu_{\ell}} s_{\nu_{r}}. 
\end{aligned}
$$
Then it is a standard procedure to compute $c_{\nu}$ using the dimension formulas in \lref{lemma:dimension_formulas}:

\begin{align*}
   c_{\nu_1} &= \frac{1}{s! d_{\left(\sybox, \sybox\right)}} \binom{s}{1}^2 (s-1)! s_{\sybox}^2 =  \frac{s}{d^2 - 1}, \\
    c_{\nu_2} &= \frac{1}{s! d_{\left(\syrow, \syrow\right)}} \binom{s}{2}^2 (s-2)! s_{\syrow}^2 = \frac{s(s-1)}{d^2(d+3)(d-1)}, \\
    c_{\nu_3} &= \frac{1}{s! d_{\left(\syrow, \sycol\right)}} \binom{s}{2}^2 (s-2)! s_{\syrow} s_{\sycol} = \frac{1}{s! d_{\left(\sycol, \syrow\right)}} \binom{s}{2}^2 (s-2)! s_{\sycol} s_{\syrow} = \frac{s(s-1)}{(d^2 - 4)(d^2 - 1)}, \\
    c_{\nu_4} &= \frac{1}{s! d_{\left(\sycol, \syrow\right)}} \binom{s}{2}^2 (s-2)! s_{\sycol} s_{\syrow} = \frac{s(s - 1)}{d^2(d-3)(d+1)}.
\end{align*}
Notably, the expression of $c_{\nu_1}$, or equivalently $\mathsf{p}_{\bfq}$, matches the optimality result in \cite[Corollary 4]{Yoshida_2026} via the conversion from $\mathsf{p}_{\bfq}$ to the entanglement infidelity \cite{Yang_2020}, confirming its optimality. Since $c_{\nu} \geq 0$ for $\nu = \nu_1, \nu_2, \nu_3, \nu_4$ and $s \leq d$, the scaling of the coefficients in \eref{eqn:auxiliary_coefficients} and the correction coefficient $\Gamma$ reads
\begin{align*}
    &|u_{+}| \leq \frac{2}{d(d+1)} + \frac{2}{(d + 1)(d+2)} c_{\nu_2} = \bigo{ \frac{1}{d^2} + \frac{s^2}{d^6} } = \bigo{\frac{1}{d^2}}, \\
    &|u_{-}| \leq \frac{2}{d(d-1)} + \frac{2}{(d-1)(d-2)} c_{\nu_4} = \bigo{ \frac{1}{d^2} + \frac{s^2}{d^6} } = \bigo{\frac{1}{d^2}}, \\
    &|v_{+}| \leq \frac{1}{d} c_{\nu_1} = \bigo{ \frac{s}{d^3} }, \quad |v_{-}| \leq \frac{1}{d} c_{\nu_1} = \bigo{ \frac{s}{d^3} }, \quad |r| = \frac{1}{d}(c_{\nu_1} - c_{\nu_3}) =\frac{s(d^2 - s - 3)}{d(d^2-1)(d^2 - 4)} = \bigo{ \frac{s}{d^3} },  \\
    &\Gamma = \frac{1-\p_{\bfq}}{d\p_{\bfq}} = \frac{1 - c_{\nu_1}}{d c_{\nu_1}} = \bigtheta{ \frac{d}{s}\left(1 - \frac{s}{d^2} \right) } = \bigtheta{ \frac{d}{s} }.
\end{align*}
Hence, we have that
$$
\begin{aligned}
    &|u_{+}| + |u_{-}| \leq \bigo{ \frac{1}{d^2} }, \quad |v_{+}| + |v_{-}| + 2|r| \leq \bigo{ \frac{s}{d^3} }, \quad c_{\nu_2} + c_{\nu_4} + 2 c_{\nu_3} \leq \bigo{ \frac{s^2}{d^4} }.
\end{aligned}
$$
Using Corollary \ref{corollary:inner_product_terms_upper_bound} and that $O \in \obs(\B)$, $\Tr[\sigma_0^2] = \Tr[\rho_0^2] < \Tr[\rho^2] \leq \P$ [cf. Problem \ref{prob:CSEU}], we can estimate the scaling of variance term \hyperlink{variance_term_1}{(I)}: Recall \eref{eqn:term_1_upper_bound}, and for notational elegance we use the handy fact that $\min\{1, x\} \leq \min\{2, x\} \leq 2 \min\{1, x\}$ for any $x \geq 0$. The scaling upper bound thus reads
\begin{align*}
    \Tr\left[ \left\{\Ket{U} \Bra{U} + \frac{1-\p_{\bfq}}{d\p_{\bfq}} \openone \otimes \openone, O \otimes \rho_0^T \right\}^{\otimes 2} \widetilde{\Phi}_{\bfq, s}
        \right] &\leq \bigo{ \frac{1}{d^2} } \cdot \bigo{ \frac{d^2}{s^2} \B\P + d \P } + \bigo{ \frac{s}{d^3} } \cdot \bigo{ d^2 \P } \\
        &\qquad + \bigo{ \frac{s^2}{d^4} } \cdot \bigo{ d^2 \min\{1, \B\P\} } \\
        &\leq \bigo{ \frac{s^2}{d^2} \min\{1, \B\P\}  + \frac{1}{s^2} \B\P + \frac{s}{d} \P }.
\end{align*}
Therefore, in $\Var[ \hat{Z}(\hat{\mathsf{\Lambda}}, L) ]$ [cf. \eref{eqn:variance_upper_bound_quadratic_estimator}] the full variance expression associated with term \hyperlink{variance_term_1}{(I)} reads,
\begin{equation}
\label{eqn:quadratic_final_upperbound_scaling_term_1}
\begin{aligned}
    &\frac{L(L-1)(L-2)}{L^2(L-1)^2 {\p_{\bfq}^2}\left(d + 2\Gamma\right)^2}   \Tr\left[ \left\{\Ket{U} \Bra{U} + \Gamma \openone \otimes \openone, O \otimes \rho_0^T \right\}^{\otimes 2} \widetilde{\Phi}_{\bfq, s}
        \right] \\
        &\qquad \leq \bigtheta{ \frac{1}{L} } \cdot \bigtheta{\frac{d^2}{s^2}  } \cdot \bigo{ \frac{s^2}{d^2} \min\{1, \B\P\}  + \frac{1}{s^2} \B\P + \frac{s}{d} \P} \\
        &\qquad = \bigo{\frac{1}{L} \left(  \min\{1, \B\P\} + \frac{d^2}{s^4} \B\P + \frac{d}{s} \P  \right) },
\end{aligned}
\end{equation}
recovering the first half of the variance expression shown in \lref{lemma:small_s_regime_estimator_Lambda_variance}.

Using the main results of Appendix \ref{appendix:bound_term_2} and \ref{appendix:bound_term_3}, Corollary \ref{corollary:inner_product_terms_upper_bound} indicates
\begin{align*}
    &\mathsf{T}(\openone, \openone) \leq \bigo{\B\P}, \quad \mathsf{T}(\Delta, \openone), \mathsf{T}(\Delta_{\times}, \openone) \leq \bigo{ \P }, \quad \mathsf{T}(\Lambda, \openone) \leq \bigo{ \min\{1, \B \P\} + \P }, \\
    &\mathsf{T}(\openone, \Delta), \mathsf{T}(\openone, \Delta_{\times}) \leq \bigo{ \min\{1, \B\P\} }, \\
    &\mathsf{T}(\Delta, \Delta), \mathsf{T}(\Delta, \Delta_{\times}), \mathsf{T}(\Delta_{\times}, \Delta), \mathsf{T}(\Delta_{\times}, \Delta_{\times}) \leq \bigo{ d\min\{1, \B\P\}  + \P }, \\
    &\mathsf{T}(\Lambda, \Delta), \mathsf{T}(\Lambda, \Delta_{\times}) \leq \bigo{ d\left( \min\{1, \B \P \} + \P \right) }, \quad \mathsf{T}(\openone, \Lambda) \leq \bigo{ \min\{1, \B\P\} + \P }, \\
    &\mathsf{T}(\Delta, \Lambda), \mathsf{T}(\Delta_{\times}, \Lambda) \leq \bigo{ d \left( \min\{1, \B \P \} + \P \right) }, \quad \mathsf{T}(\Lambda, \Lambda) \leq \bigo{ d^2\left( \min\{1, \B\P\} + \P \right) },
\end{align*}
and 
\begin{align*}
    &\mathsf{S}(\openone, \openone) \leq \bigo{ d^2 \B\P }, \quad \mathsf{S}(\openone, \Delta), \mathsf{S}(\Delta, \openone) \leq \bigo{ d \left( \min\{1, \B\P\} +  \B \P \right) }, \\
        &\mathsf{S}(\openone, \Delta_{\times}), \mathsf{S}(\Delta_{\times}, \openone) = \bigo{d^2 \P}, \quad \mathsf{S}(\Delta, \Delta_{\times}), \mathsf{S}(\Delta_{\times}, \Delta) = \bigo{d(\B\P + \P)},
        \\
    &\mathsf{S}(\openone, \Lambda), \mathsf{S}(\Lambda, \openone) \leq \bigo{\min\{1, \B\P\} + \B\P + d \P }, \quad \mathsf{S}(\Delta, \Delta) \leq \bigo{ d\P + \min\{1, \B\P\} + \B\P }, 
    \\
    &\mathsf{S}(\Delta_{\times}, \Delta_{\times}) = \bigo{d^2 \B\P}, \quad \mathsf{S}(\Delta, \Lambda), \mathsf{S}(\Lambda, \Delta) \leq \bigo{d \left( \min\{1, \B\P\} + \P \right)}, \\
    &\mathsf{S}(\Delta_{\times}, \Lambda), \mathsf{S}(\Lambda, \Delta_{\times}) = \bigo{d^2 \P}, \quad \mathsf{S}(\Lambda, \Lambda) \leq \bigo{ d^2 \min\{1, \B\P\} }.
\end{align*}
Plugging these scalings back into \eref{eqn:bilinear_form_2_upper_bound} gives
\begin{align*}
    \Tr\left[\left(O \otimes \rho_0^T \right)^{\otimes 2} \widetilde{\Phi}_{\bfq, s}^2 \right] &\leq \bigo{\frac{1}{d^4}} \cdot \bigo{\B\P} + \bigo{\frac{s^2}{d^6}} \cdot \bigo{ \min\{1, \B\P\} + \P } \\
    &\quad + \bigo{ \frac{s^4}{d^8} } \cdot \bigo{ d^2\left( \min\{1, \B\P\} + \P \right) } + \bigo{ \frac{s}{d^5} } \cdot \bigo{ \min\{1, \B\P\} + \P }  \\
    &\quad + \bigo{ \frac{s^2}{d^6} } \cdot \bigo{ d\left(\min\{1, \B\P\} + \P \right)} + \bigo{ \frac{s^3}{d^7} } \cdot \bigo{ d \left( \min\{1, \B \P \} + \P \right) } \\
    &\leq \bigo{ \frac{\B\P}{d^4} + \left( \frac{s^2}{d^5} + \frac{s^4}{d^6} \right) \min\{1, \B\P\} + \left( \frac{s^2}{d^5} + \frac{s^4}{d^6} \right) \P }.
\end{align*}
Repeating the same for \eref{eqn:bilinear_form_3_upper_bound}, we have
\begin{align*}
    &\Tr\left[ \left(  O \otimes \rho_0^T  \otimes \openone^{\otimes 2}\right) \widetilde{\Phi}_{\bfq, s} \left(\openone^{\otimes 2} \otimes  O \otimes \rho_0^T \right) \widetilde{\Phi}_{\bfq, s}  \right] \\
    &\quad \leq \bigo{\frac{1}{d^4}} \cdot \bigo{d^2  \B \P} + \bigo{ \frac{s^2}{d^6} } \cdot \bigo{ d  \P + \min\{1, \B\P\}  + d^2\B\P } \\
    &\quad \quad  + \bigo{ \frac{s^4}{d^8} } \cdot \bigo{d^2 \min\{1, \B\P\} } + \bigo{\frac{s}{d^5}} \cdot \bigo{ d \left( \min\{1, \B \P\} +  \B \P \right) + d^2 \P } \\
    &\quad \quad +\bigo{\frac{s^2}{d^6}} \cdot \bigo{d\P + \B \P + \min\{1, \B\P\} + d(\B\P + \P) } + \bigo{\frac{s^3}{d^7}} \cdot \bigo{d \left(\min\{ 1, \B\P\} +  \P \right) + d^2\P } \\
    &\quad = \bigo{ \left( \frac{s}{d^4}  + \frac{s^4}{d^6} \right) \min\{1, \B\P\} + \left(\frac{1}{d^2} + \frac{s^2}{d^4} \right) \B\P   +  \left(\frac{s^3}{d^5} + \frac{s}{d^3} \right) \P + \frac{s^2}{d^4} \B\P }
\end{align*}
Note that we have assumed that $1 \leq \B \leq d$ in Problem \ref{prob:CSEU} and used the basic fact that $d^{-1} \leq \P \leq 1$. The full variance expression associated with the terms \hyperlink{variance_term_2}{(II)} and \hyperlink{variance_term_3}{(III)} can be simply bounded by
\begin{equation}
\label{eqn:quadratic_final_upperbound_scaling_term_2}
\begin{aligned}
&\frac{L(L-1)}{L^2(L-1)^2 \p_{\bfq}^{4}\left(d + 2\Gamma\right)^2}  \left(     \Tr\left[\left(O \otimes \rho_0^T \right)^{\otimes 2} \widetilde{\Phi}_{\bfq, s}^2 \right]  + \Tr\left[ \left(  O \otimes \rho_0^T  \otimes \openone^{\otimes 2}\right) \widetilde{\Phi}_{\bfq, s} \left(\openone^{\otimes 2} \otimes  O \otimes \rho_0^T \right) \widetilde{\Phi}_{\bfq, s}  \right]  \right) \\
&\quad \leq \bigtheta{ \frac{1}{L^2} } \cdot \bigtheta{ \frac{d^6}{s^4} } \cdot \bigo{\frac{\B\P}{d^2}} \leq \bigo{ \frac{1}{L^2} \left(\frac{d^4}{s^4} \B\P \right) },
\end{aligned}
\end{equation}
yielding the second half of the variance expression. Eq. \eqref{eqn:quadratic_final_upperbound_scaling_term_1} and \eqref{eqn:quadratic_final_upperbound_scaling_term_2} jointly conclude the proof.

\subsection{Proof of \lref{lemma:large_s_regime_estimator_X_variance}}
\label{appendix:Proof_large_s_regime_estimator_X_variance}
In this section, we analyze the variance of the linear estimator $\hat{X}_j$ for arbitrary $j \in [L]$. Recall \eref{eqn:variance_expression_naive_estimator} and Remark \ref{remark:factor_p_reformulation}, the variance reads
$$
\begin{aligned}
    \Var\left[\hat{X}_j\right] &= \frac{1}{\p_{\bfq}^2} \Tr\left[ O \otimes O  \mathcal{M}_{\bfq, s}^{(2)}(\rho_0^{\otimes 2}) \right]   - \left( \Tr\left[ O  U \rho_0 U^\dagger \right]  \right)^2 \\
    &= \frac{1}{c_{\nu_1}^2} \Tr\left[ O \otimes O  \mathcal{E}_{\bfq, s} \circ \mathcal{U}^{\otimes 2}(\rho_0^{\otimes 2}) \right]   - \Tr\left[ O \sigma_0 \right]^2 \\
    &=\frac{1}{c_{\nu_1}^2} \Tr\left[ O \otimes O  \mathcal{E}_{\bfq, s} (\sigma_0^{\otimes 2}) \right]   -  \Tr\left[ O \sigma_0 \right]^2.
\end{aligned}
$$

Using \eref{eqn:second_moment_channel_on_symmetric_input}, we have
\begin{align*}
    \Tr\left[ O \otimes O  \mathcal{E}_{\bfq, s}(\sigma_0 \otimes \sigma_0)  \right] &= \Tr\left[ \left(O \otimes O\right) \left( - \frac{\Tr[\sigma_0^2]}{d(d^2 - 1)}  \openone \otimes \openone + \frac{\Tr[\sigma_0^2]}{d^2 - 1} \F \right) \right] \\
    &\quad + c_{\nu_1} \cdot \Tr\bigg[\left(O \otimes O\right) \bigg( \frac{1}{d+2} P_{+}\left[ \sigma_0^2 \otimes \openone + \openone \otimes \sigma_0^2 - \frac{2 \Tr[\sigma_0^2]}{d} \openone \otimes \openone \right] \\
    &\qquad -\frac{1}{d - 2} P_{-}\left[ \sigma_0^2 \otimes \openone + \openone \otimes \sigma_0^2 - \frac{2 \Tr[\sigma_0^2]}{d} \openone \otimes \openone \right] \bigg)  \bigg] \\
    &\quad + c_{\nu_2} \cdot \Tr\left[\left(O \otimes O\right) \left(P_{+} (\sigma_0 \otimes \sigma_0) - \frac{1}{d+2} P_{+} ( \sigma_0^2 \otimes \openone + \openone \otimes \sigma_0^2) + \frac{2\Tr[\sigma_0^2]}{(d+1)(d+2)} P_{+}\right)  \right] \\
    &\quad + c_{\nu_4} \cdot \Tr\left[\left(O \otimes O\right) \left(P_{-} (\sigma_0 \otimes \sigma_0) + \frac{1}{d-2} P_{-} ( \sigma_0^2 \otimes \openone + \openone \otimes \sigma_0^2) - \frac{2\Tr[\sigma_0^2]}{(d-1)(d-2)} P_{-} \right) \right].
\end{align*}
The variance $\Var[\hat{X}_j]$ can then be written as
\begin{equation}
\label{eqn:estimator_X_variance_full_expression}
\begin{aligned}
    \Var\left[\hat{X}_j\right] &= \frac{1}{c_{\nu_1}^2} \bigg(\frac{\Tr[\sigma_0^2] \Tr[O^2]}{d^2 - 1} + c_{\nu_1} \cdot \frac{2d}{d^2 - 4} \Tr\left[O^2  \left( \sigma_0^2 - \frac{\Tr[\sigma_0^2]}{d} \openone \right)  \right] \\
    &\quad + c_{\nu_2} \cdot \left[ \frac{1}{2} \Tr[O \sigma_0]^2 + \frac{1}{2} \Tr[O \sigma_0 O \sigma_0] - \frac{1}{d+2} \Tr[O^2 \sigma_0^2] + \frac{1}{2(d+1)(d+2)} \Tr[\sigma_0^2] \Tr[O^2]  \right] \\
    &\quad + \frac{1}{c_{\nu_1}^2} \left( \frac{c_{\nu_1} - c_{\nu_2}}{d+2} + \frac{c_{\nu_1} - c_{\nu_4}}{d-2}  \right) \Tr[O^2 \sigma_0^2] + \frac{c_{\nu_2} - c_{\nu_4}}{2 c_{\nu_1}^2} \Tr[O \sigma_0 O \sigma_0] + \left( \frac{c_{\nu_2} + c_{\nu_4}}{2c_{\nu_1}^2} - 1  \right) \Tr[O \sigma_0]^2 \\
    &= \frac{1}{c_{\nu_1}^2} \left[ \frac{1}{d^2 - 1} - \frac{2c_{\nu_1}}{d^2 - 4} + \frac{c_{\nu_2}}{2(d+1)(d+2)} + \frac{c_{\nu_4}}{2(d-1)(d-2)} \right] \Tr[\sigma_0^2] \Tr[O^2] \\
    &\quad + \frac{1}{c_{\nu_1}^2} \left( \frac{c_{\nu_1} - c_{\nu_2}}{d+2} + \frac{c_{\nu_1} - c_{\nu_4}}{d-2}  \right) \Tr[O^2 \sigma_0^2] + \left( \frac{c_{\nu_2}}{c_{\nu_1}^2} - 1  \right)  \Tr[O \sigma_0 O \sigma_0] \\
    &\quad + \left( \frac{c_{\nu_2} + c_{\nu_4}}{2c_{\nu_1}^2} - 1  \right) \left( \Tr[O \sigma_0]^2 - \Tr[O \sigma_0 O \sigma_0] \right).
\end{aligned}
\end{equation}
It suffices to choose an appropriate learning strategy $(\Y, \bfq)$ and evaluate the variance. Specifically, we choose $(\Y, \bfq)$ from a family of sine-power states that resembles the optimal quantum clocks in quantum metrology \cite{Bu_ek_1999, Holevo2011}, while what we're recording is not the time flow, but the shadows of the quantum system's evolution. The construction is detailed in the following.

\begin{construction}[{Sine-power state family for unitary learning, \cite[Appendix B.3, adapted]{Yang_2020}}]
\label{construction:YRC_programming_scheme}
    The family of sine-power state learning strategies $(\Y, \bfq)$ is constructed as follows: First, define $N = \lfloor \frac{1}{3d-2} \left(\frac{2s}{d-1} + d - 2  \right) \rfloor $ and the residual $N_0 = s - \frac{1}{2} \left((3d-2) n - d + 2 \right)(d-1) $. Then, define the base Young diagram $\mu_0 \in \Y_{N_0}^d$ that yields a minimum row-wise growth rate:
    $$
    \mu_0 = (\mu_{0, j})_{1 \leq j \leq d}: \quad \sum_{j=1}^d |\mu_{0, j}| = N_0;~\forall j > i \in [d],~\mu_{0, j} \leq \mu_{0, i} \leq \mu_{0, j} + 1.
    $$
    The collection of Young diagrams is then supported over all free vectors $\tilde{\lambda} \in \{0, 1, \dots, N -1\}^{d-1}$ (that do not necessarily correspond to a legitimate weight vector from a Young diagram):
    $$
    \Y = \left\{ \lambda \in \Y_{s}^d: \exists \tilde{\lambda} \in \{0, 1, \dots, N -1\}^{d - 1},\,\forall j \in [d-1],\,\lambda_j = \mu_{0, j} + N(2d-3) + 1 - (N+1)(j-1) + \tilde{\lambda}_j \right\}.
    $$
    The distribution over $\Y$ is directly indexed by the free vector $\tilde{\lambda}$ due to the bijection $\lambda \leftrightarrow \tilde{\lambda}$, which acts as a family of sinusoidal amplitude distribution characterized by a tunable sharpness parameter $t$:
    $$
    \forall k \in \{0, 1, \dots, N -1\}, \quad g_{k}^{(t)} := \frac{1} {N} \frac{2^t}{\binom{t}{t/2}} \sin^{t} \left( \frac{2k+1}{2N} \pi \right),~t \in 2 \mathbb{N};\qquad
   \forall \lambda \in \Y, \quad  q_{\lambda} = q_{\tilde{\lambda}} := \prod_{j=1}^{d-1} g_{\tilde{\lambda}_j}^{(t)}.
    $$
\end{construction}

Firstly, we note that for each $\lambda \in \Y$ and $j \in [d - 1]$, inheriting the notation in Corollary \ref{corollary:no_multiplicity}, its minimum gap between adjacent rows satisfies
\begin{align*}
\gap(\lambda) &= \min_{j \in [d-1]} (\lambda_{j} - \lambda_{j+1}) \\
&=\min_{j \in [d-1]} \left( (\mu_{0, j} - \mu_{0, j + 1}) + (N+1) + (\tilde{\lambda}_j - \tilde{\lambda}_{j+1}) \right) \\
&\geq N + 1 + \min_{j \in [d-1]}(\mu_{0, j} - \mu_{0, j + 1}) + \min_{j \in [d-1]} (\tilde{\lambda}_j - \tilde{\lambda}_{j+1}) \\
&\geq N + 1 - (N - 1) = 2.
\end{align*}
Therefore, take any $\lambda, \mu \in \Y$, we have
$
\gap(\lambda) + \gap(\mu) + 1 \geq 5
$. Since for any $\nu \in \{\nu_1, \nu_2, \nu_3, \nu_4\}$, we have $(\nu)_1 - (\nu)_d \leq 4 < 5$, Corollary \ref{corollary:no_multiplicity} allows us to reformulate the coefficient $c_{\nu}$ as 
\begin{align*}
c_{\nu} &= \frac{1}{\dim W_{\nu}} \sum_{\lambda, \mu \in \Y  } \sqrt{q_{\lambda} q_{\mu}} C_{\lambda, \nu}^{\mu} \\
&= \frac{1}{\dim W_{\nu}} \sum_{\lambda, \mu \in \Y  } \sqrt{q_{\lambda} q_{\mu}} \cdot m_{\nu}(\mu - \lambda) \\
&= \frac{1}{\dim W_{\nu}} \sum_{\lambda\in \Y  } \sum_{w \in \mathfrak{X}(W_{\nu}) } \sqrt{q_{\lambda} q_{\lambda + w}} \cdot m_{\nu}(w) \\
&= \frac{1}{\dim W_{\nu}} \sum_{w \in \mathfrak{X}(W_{\nu}) } m_{\nu}(w) \sum_{\lambda\in \Y  }  \sqrt{q_{\lambda} q_{\lambda + w}}.
\end{align*}

Plugging in the expression of $\bfq$ [cf. Construction \ref{construction:YRC_programming_scheme}] and set the parameter $t = 4$. Since appending $w$ to $\lambda$ is equivalent to shifting the free vector $\tilde{\lambda}$, for any legitimate weight vector $w$, we have
\begin{align*}
    \sum_{\lambda\in \Y  }  \sqrt{q_{\lambda} q_{\lambda + w}} &  = \sum_{\substack{
    \forall j \in[d-1], \\ 0 \leq \tilde{\lambda}_j, \tilde{\lambda}_j + w_j \leq N-1 }}  \sqrt{q_{\tilde{\lambda}} q_{\tilde{\lambda} + w}} \\
    & = \sum_{\substack{
    \forall j \in[d-1], \\ 0 \leq \tilde{\lambda}_j, \tilde{\lambda}_j + w_j \leq N-1 }} \prod_{j=1}^{d-1} \frac{8}{3N} \sin^2\left( \frac{2 \tilde{\lambda}_j + 1}{2N} \pi \right) \sin^2 \left( \frac{2 (\tilde{\lambda}_j + w_j) + 1}{2N} \pi \right) \\
    & = \prod_{j=1}^{d-1} \sum_{\substack{0 \leq \tilde{\lambda}_j, \tilde{\lambda}_j + w_j \leq N-1 }} \frac{2}{3N} \bigg[ 1 + \frac{1}{2}\cos \left( \frac{2w_j}{N} \pi \right) \\
    &  \qquad \qquad \qquad  - \cos \left( \frac{2\tilde{\lambda}_j + 1}{N} \pi \right) - \cos \left( \frac{2(\tilde{\lambda}_j + w_j) + 1}{N} \pi \right) + \frac{1}{2} \cos \left( \frac{2(2 \tilde{\lambda}_j + w_j + 1)}{N} \pi \right) \bigg] \\
    & = \prod_{j=1}^{d-1} \underbrace{\left[ \frac{2}{3} \left( 1 - \frac{|w_j|}{N} \right) \left( \cos^2 \left( \frac{|w_j|}{N} \pi \right) + \frac{1}{2}  \right)  + \frac{1}{3N} \sin\left( \frac{2|w_j|}{N} \pi \right) \left( \frac{2}{\sin \frac{\pi}{N} } - \frac{1}{\sin \frac{2\pi}{N} } \right) \right]}_{=:\, h(|w_j|)}.
\end{align*}
For our choice of $\nu$, the number of nonzero entries in any weight vector $w \in \mathfrak{X}(W_{\nu})$ is upper-bounded by $4$, and the magnitude of $|w_j|$ is upper-bounded by $2$. For sufficiently large constant $N \geq \Omega(1)$, we can invoke Fact \ref{fact:taylor_series} to Taylor-expand each factor in the product:
$$
h(|w_j|) = 1 - \frac{2\pi^2}{3N^2}|w_j|^2 + \frac{2\pi^4}{9N^4} |w_j|^4 - \frac{\pi^4}{90N^5} \left( 8|w_j|^5 + 7|w_j| \right) + \bigo{ \frac{1}{N^6} }.
$$
Therefore, taking the logarithm on both sides,
$$
\log h(|w_j|) = - \frac{2\pi^2}{3N^2} |w_j|^2 - \frac{\pi^4}{90N^5}\left( 8|w_j|^5 + 7|w_j| \right) + \bigo{ \frac{1}{N^6} }.
$$
Since $h(|w_j|) < 1$\footnote{
To grasp this, the Cauchy-Schwarz inequality says $h(|w_j|) = \sum_{\tilde{\lambda}_j} \sqrt{ g_{\tilde{\lambda}_j}^{(t)} g_{\tilde{\lambda}_j + w_j}^{(t)} } \leq \left(\sum_{\tilde{\lambda}_j}  g_{\tilde{\lambda}_j}^{(t)} \right)^{1/2} \left(\sum_{\tilde{\lambda}_j} g_{\tilde{\lambda}_j + w_j}^{(t)} \right)^{1/2} < 1$.
}, applying Fact \ref{fact:taylor_series} again, the original product expands to
\begin{equation}
\label{eqn:taylor_expansion}
\begin{aligned}
\prod_{j=1}^{d-1} h(|w_j|) &= \exp\left( \sum_{j=1}^{d-1} \log h(|w_j|)  \right) \\
&= 1 - \frac{2 \pi^2}{3N^2} \sum_{j=1}^{d-1} |w_j|^2 + \frac{2 \pi^4}{9N^4} \left( \sum_{j=1}^{d-1} |w_j|^2 \right)^2 - \frac{\pi^4}{90 N^5} \left( 8\sum_{j=1}^{d-1} |w_j|^5 + 7 \sum_{j=1}^{d-1} |w_j|  \right) + \bigo{ \frac{1}{N^6} }.
\end{aligned}
\end{equation}
If we denote the factors
$$
\mathscr{L}_{k}(\nu) = \frac{1}{\dim W_{\nu}} \sum_{w \in \mathfrak{X}(W_{\nu}) } m_{\nu}(w) \sum_{j=1}^{d-1} |w_j|^k, \quad \mathscr{Q}(\nu) = \frac{1}{\dim W_{\nu}} \sum_{w \in \mathfrak{X}(W_{\nu}) } m_{\nu}(w) \left(\sum_{j=1}^{d-1} |w_j|^2 \right)^2.
$$
Then using \eref{eqn:taylor_expansion} and $\sum_{w \in \mathfrak{X}(W_{\nu})} m_{\nu}(w) = \dim W_{\nu}$ [cf. Definition \ref{def:weight_space} and \lref{lemma:schur}], the coefficient $c_{\nu}$ can be expressed as
$$
c_{\nu} = 1 - \frac{2\pi^2}{3N^2} \mathscr{L}_2(\nu) + \frac{2\pi^4}{9N^4} \mathscr{Q}(\nu) - \frac{\pi^4}{90N^5} \left( 8 \mathscr{L}_{5}(\nu) + 7 \mathscr{L}_1(\nu) \right) + \bigo{ \frac{1}{N^6} }.
$$

\begin{lemma}
\label{lemma:expression_of_factors}
    For $\nu = \nu_1, \nu_2, \nu_3, \nu_4$, the expressions of the factors read
    $$
    \begin{aligned}
    &\mathscr{L}_{k}(\nu_1) = \frac{2(d - 1)}{d+1}, \quad \mathscr{L}_{k}(\nu_2) = \frac{4(d + 2^k - 1)(d-1)}{d(d+3)}, \quad \mathscr{L}_{k}(\nu_4) = \frac{4(d-1)^2}{d(d+1)}, \quad \mathscr{L}_k(\nu_3) = \frac{2(d-1)(2d + 2^k)}{(d+1)(d+2)}; \\
    &\mathscr{Q}(\nu_1) = \frac{2(2d-3)}{d+1}, \quad \mathscr{Q}(\nu_2) = \frac{4(4d^3 + 13d^2 - 13d - 24)}{d^2(d+3)}, \\
    &\mathscr{Q}(\nu_4) = \frac{4(d-1)(4d^2 - 11d + 8)}{d^2(d+1)}, \quad \mathscr{Q}(\nu_3) = \frac{4(4d^2 + 3d - 12)}{(d+1)(d+2)}.
    \end{aligned}
    $$
\end{lemma}

\begin{proof}[Proof of \lref{lemma:expression_of_factors}]
We prove the lemma by brute-force enumeration. Recall that we have evaluated the multiplicity and counting result related to all the orbits of (non-zero) weight vectors in $\mathfrak{X}(W_{\nu})$ for $\nu = \nu_1, \nu_2, \nu_3, \nu_4$ in Example \ref{example_weight_multiplicities}. It suffices to plug them into our target quantities. By the strong symmetry among the coordinates of the weight $w$ [cf. Definition \ref{def:weight_space}], consider the full coordinate sum
    $$
    \widetilde{\mathscr{L}}_{k}(\nu) = \frac{1}{\dim W_{\nu}} \sum_{w \in \mathfrak{X}(W_{\nu}) } m_{\nu}(w) \sum_{j=1}^{d} |w_j|^k,
    $$
    then ${\mathscr{L}}_{k}(\nu) = \frac{d-1}{d} \widetilde{\mathscr{L}}_{k}(\nu)$. It suffices to track the sum $\sum_{w \in \mathfrak{X}(W_{\nu}) } m_{\nu}(w) \sum_{j=1}^{d} |w_j|^k$. For conciseness, we only analyze non-trivial weights $w \neq \mathbf{0}$ for each irrep.
    
    For $\nu_1$, orbit $e_i - e_j$ for $i \neq j$ has multiplicity $1$ and there are $d(d-1)$ such weights of distinct $i, j$. Since $\sum_{j=1}^d |w_j|^k \equiv 2$ for non-zero weights $w \in \mathfrak{X}(W_{\nu_1})$ always,
    $$
    \widetilde{\mathscr{L}}_k(\nu_1) = \frac{2d(d-1)}{d^2 - 1} = \frac{2d}{d+1}.
    $$

    For $\nu_2$, it has (1) Orbit $2e_i - 2e_j$ on distinct $i, j$ with multiplicity $1$ and there are $d(d-1)$ such weights, contributing $d(d-1)(2^k + 2^k) = 2d(d-1) \cdot 2^k$; (2) Orbit $\pm(2e_i - e_j - e_{\ell})$ on distinct $i, j, \ell$ with multiplicity $1$ and $d(d-1)(d-2)$ such weights, contributing $d(d-1)(d-2)(2^k + 1 + 1) = d(d-1)(d-2)(2^k + 2)$; (3) Orbit $e_i + e_j - e_{\ell} - e_r$ on distinct $i, j, \ell, r$ with multiplicity $1$ and $\frac{1}{4}d(d-1)(d-2)(d-3)$ such weights, contributing $\frac{1}{4}d(d-1)(d-2)(d-3)(1+1+1+1) = d(d-1)(d-2)(d-3)$; (4) Orbit $e_i - e_j$ on distinct $i, j$ with multiplicity $d-1$ and $d(d-1)$ such weights, contributing $d(d-1)^2(1+1) = 2d(d-1)^2$. Collecting all the terms gives
    $$
    \begin{aligned}
    \widetilde{\mathscr{L}}_k(\nu_2) &= \frac{2d(d-1) \cdot 2^k + d(d-1)(d-2)(2^k + 2) + d(d-1)(d-2)(d-3) + 2d(d-1)^2}{\frac{d^2(d-1)(d+3)}{4}} \\
    &= \frac{d^2(d-1)(d + 2^k - 1)}{\frac{d^2(d-1)(d+3)}{4}} = \frac{4(d + 2^k - 1)}{d + 3}.
    \end{aligned}
    $$

    For $\nu_4$, it contains (1) Orbit $e_i + e_j - e_{\ell} - e_r$ on distinct $i, j, \ell, r$ with multiplicity $1$ and there are $\frac{1}{4} d(d-1)(d-2)(d-3)$ such weights, contributing $d(d-1)(d-2)(d-3)$; (2) Orbit $e_i - e_j$ on distinct $i, j$ with multiplicity $d - 3$ and there are $d(d-1)$ such weights, contributing $2d(d-1)(d-3)$. Therefore, 
    $$
    \widetilde{\mathscr{L}}_k(\nu_4) = \frac{d(d-1)(d-2)(d-3) + 2d(d-1)(d-3)}{ \frac{d^2(d-3)(d+1)}{4}  } = \frac{d^2(d-1)(d-3)}{\frac{d^2(d-3)(d+1)}{4}} = \frac{4(d-1)}{d+1}.
    $$

    As for $\nu_3$, it contains (1) Orbit $2e_i - e_k - e_\ell$ on distinct $i, k, \ell$ with multiplicity $1$ and there are $\frac{1}{2}d(d-1)(d-2)$ many. For each these weight, $\sum_{j=1}^d |w_j|^k = 2^k + 1 + 1 = 2^k + 2$, contributing $\frac{1}{2}d(d-1)(d-2)(2^k + 2)$; (2) Orbit $e_i + e_j - e_k - e_\ell$ on distinct $i, j, k, \ell$ with multiplicity $1$ and there are $\frac{1}{4}d(d-1)(d-2)(d-3)$ such weights, each contributing $\sum_{j=1}^d |w_j|^k = 1 + 1 + 1 + 1 = 4$. The total contribution is $d(d-1)(d-2)(d-3)$; (3) Orbit $e_i - e_k$ on distinct $i, k$ with multiplicity $d - 2$, and there are $d(d-1)$ many, each contributing $\sum_{j=1}^d |w_j|^k = 1 + 1 = 2$. Total contribution is $2 \times d(d-1) \times (d-2) = 2d(d-1)(d-2)$. Hence,
    $$
    \widetilde{\mathscr{L}}_k(\nu_3) = \frac{\frac{1}{2}d(d-1)(d-2)(2^k + 2) + d(d-1)(d-2)(d-3) + 2d(d-1)(d-2)}{\frac{(d^2-1)(d^2 - 4)}{4}} = \frac{2d(2d + 2^k)}{(d+1)(d+2)}.
    $$
    
    The evaluation of $\mathscr{L}_k(\nu)$ can be reused to evaluate $\mathscr{Q}(\nu)$, subject to a restriction in the summation limit. For conciseness, we denote the sum $S(w) = \sum_{j=1}^{d-1} |w_j|^2$.

    For $\nu_1$, the only contributing terms stems from the non-trivial weight $w = e_i - e_j$ on distinct $i, j$ with multiplicity $1$; When $w_d = 0$, there are $(d-1)(d-2)$ such weights with $S(w)^2 = (1 + 1)^2 = 4$; When $w_d = \pm 1$, there are $2(d-1)$ such weights with $S(w)^2 = 1^2 = 1$. The sum reads $4(d-1)(d-2) + 2(d-1) = 2(d-1)(2d-3)$. Thus, we have
    $$
    \mathscr{Q}(\nu_1) = \frac{2(d-1)(2d-3)}{d^2 - 1} = \frac{2(2d - 3)}{d + 1}.
    $$

    As for $\nu_2$: (1) For orbit $2e_i - 2e_j$ on distinct $i, j$ of multiplicity $1$, when $w_d = 0$, there are $(d-1)(d-2)$ such weights with $S(w)^2 = (4 + 4)^2 = 64$; When $w_d = \pm 2$, there are $2(d-1)$ such weights with $S(w)^2 = 4^2 = 16$. The sum is given by $64(d-1)(d-2) + 32(d - 1) = 32(d-1)(2d-3)$. (2) For orbit $\pm(2e_i - e_j - e_\ell)$ on distinct $i, j, \ell$ with multiplicity $1$; When $w_d = 0$, there are $(d-1)(d-2)(d-3)$ such weights with $S(w)^2 = (4 + 1 + 1)^2 = 36$; When $w_d = \pm 2$, there are $(d-1)(d-2)$ such weights, with $S(w)^2 = (1 + 1)^2 = 4$; When $w_d = \pm 1$, there are $2(d-1)(d-2)$ such weights with $S(w)^2 = (4 + 1)^2 = 25$. They sum up to $36(d-1)(d-2)(d-3) + 4(d-1)(d-2) + 50(d-1)(d-2) = 18(d-1)(d-2)(2d-3)$. (3) For the orbit $e_i + e_j - e_\ell - e_r$ on distinct $i, j, \ell, r$ with multiplicity $1$; When $w_d = 0$, there are $\frac{1}{4} (d - 1)(d-2)(d-3)(d-4)$ such weights with $S(w)^2 = (1 +1 + 1 + 1)^2 = 16$; When $w_d = \pm 1$, there are $(d - 1)(d - 2)(d - 3)$ such weights. They sum up to $4(d-1)(d-2)(d-3)(d-4) + 9(d-1)(d-2)(d-3) = (d-1)(d-2)(d-3)(4d-7)$. (4) For the orbit $e_i - e_j$ with multiplicity $d - 1$, it behaves analogously to $2e_i - 2e_j$ but with different $S(w)^2$, the sum reads $(d - 1) \cdot \left[(1 + 1)^2 (d - 1)(d - 2) + 2(d-1) \right] = 2(d-1)^2(2d-3)$. Hence,
    $$
    \begin{aligned}
    \mathscr{Q}(\nu_2) &= \frac{32(d-1)(2d-3) + 18(d-1)(d-2)(2d-3) + (d-1)(d-2)(d-3)(4d-7) + 2(d-1)^2(2d-3) }{ \frac{d^2(d-1)(d+3)}{4} } \\
    &= \frac{(d - 1)(4d^3 + 12d^2 - 13d - 24)}{\frac{d^2(d-1)(d+3)}{4}} = \frac{4 (4d^3 + 13d^2 - 13d - 24)}{d^2(d+3)}.
    \end{aligned}
    $$
    For $\nu_4$: (1) For orbit $e_i + e_j - e_\ell - e_r$ on distinct $i, j, \ell, r$ with multiplicity $1$, it is similar to case (3) of representation $\nu_2$, giving rise to the sum $(d-1)(d-2)(d-3)(4d-7)$; (2) For orbit $e_i - e_j$ on distinct $i, j$ with multiplicity $d - 3$, it is similar to case (4) of $\nu_2$, but with a different multiplicity, contributing $(d-3) \cdot \left[(1 + 1)^2 (d - 1)(d - 2) + 2(d-1) \right] = 2(d-1)(d-3)(2d-3)$. Therefore, we have 
    $$
    \begin{aligned}
    \mathscr{Q}(\nu_4) &= \frac{(d-1)(d-2)(d-3)(4d-7) + 2(d-1)(d-3)(2d-3)}{ \frac{d^2(d-3)(d+1)}{4} } \\
    &= \frac{(d-1)(d-3)(4d^2 - 11d + 8)}{ \frac{d^2(d-3)(d+1)}{4} } = \frac{4(d-1)(4d^2 - 11d + 8)}{d^2(d+1)}.
    \end{aligned}
    $$
    Finally, for $\nu_3$: (1) For orbit $2e_i - e_k - e_\ell$ on distinct $i, k, \ell$ with multiplicity $1$, when $w_d = 0$, there are $\frac{1}{2}(d-1)(d-2)(d-3)$ weights with $S(w)^2 = (2^2 + 1^2 + 1^2)^2 = 36$. When $w_d = 2$, there are $\frac{1}{2}(d-1)(d-2)$ weights with $S(w)^2 = (1^2 + 1^2)^2 = 4$. When $w_d = -1$, there are $(d-1)(d-2)$ weights with $S(w)^2 = (2^2 + 1^2)^2 = 25$. Total contribution yields $36 \times \frac{1}{2}(d-1)(d-2)(d-3) + 4 \times \frac{1}{2}(d-1)(d-2) + 25 (d-1)(d-2) = 9(d-1)(d-2)(2d-3)$; (2) For orbit $e_i + e_j - e_k - e_{\ell}$ with multiplicity $1$, when $w_d = 0$, there are $\binom{d-1}{2} \times \binom{d-3}{2} = \frac{1}{4} (d-1)(d-2)(d-3)(d-4)$ many with $S(w)^2 = (1+1+1+1)^2 = 16$. When $w_d = \pm 1$, there are $2 \times (d-1) \times \binom{d-2}{2} = (d-1)(d-2)(d-3)$ many, with $S(w)^2 = (1^2 + 1^2 + 1^2)^2 = 9$. They sum up to $16 \times \frac{1}{4} (d-1)(d-2)(d-3)(d-4) + 9(d-1)(d-2)(d-3) = (d-1)(d-2)(d-3)(4d-7)$; (3) For orbit $e_i - e_k$ on distinct $i, k$ with multiplicity $d - 2$, when $w_d = 0$, there are $(d-1)(d-2)$ weights with $S(w)^2 = (1^2 + 1^2)^2 = 4$; When $w_d = \pm 1$, there are $2(d-1)$ weights with $S(w)^2 = 1$. Total contribution is $(d-2) \times (4(d-1)(d-2) + 2(d-1)) = 2(d-1)(d-2)(2d-3)$. Collecting the terms gives
    \begin{align*}
    \mathscr{Q}(\nu_3) &= \frac{(d-1)(d-2)\left(9(2d-3) + (d-3)(4d-7) + 2(2d-3) \right)}{ \frac{(d^2-1)(d^2-4)}{4} } = \frac{4(4d^2 + 3d - 12)}{(d+1)(d+2)}.
    \end{align*}
    Rescaling the expressions of $\widetilde{\mathscr{L}}_k(\nu)$ for ${\mathscr{L}}_k(\nu)$ completes the proof.
\end{proof}

Leveraging \lref{lemma:expression_of_factors}, we can formulate the leading coefficients that appear in the variance [cf. \eref{eqn:estimator_X_variance_full_expression}] by directly plugging in the analytical expressions of $\mathscr{L}_k(\nu)$ and $\mathscr{Q}(\nu)$. For sufficiently large $N$, we have
\begin{equation}
\label{eqn:scaling_of_leading_term}
\begin{aligned}
    &\frac{1}{d^2 - 1} - \frac{2c_{\nu_1}}{d^2 - 4} + \frac{c_{\nu_2}}{2(d+1)(d+2)} + \frac{c_{\nu_4}}{2(d-1)(d-2)} \\
    &\qquad = \frac{16\pi^4(d+4)(d^2-3)}
{9d^2(d+1)(d+2)(d+3)N^4} - \frac{14\pi^4(d-1)}{3d(d+1)(d+2)(d+3)N^5}
+ \bigo{\frac{1}{d^3 N^6}} = \bigtheta{\frac{1}{d^2N^4}}, \\
&\frac{c_{\nu_1} - c_{\nu_2}}{d+2} + \frac{c_{\nu_1} - c_{\nu_4}}{d-2} = \frac{8\pi^2(d-1)}
{3d(d+1)N^2}
-
\frac{8\pi^4(6d^3+23d^2-11d-62)}
{9d(d+1)(d+2)(d+3)N^4} + \bigo{\frac{1}{dN^5}} = \bigtheta{\frac{1}{dN^2}},
    \\
        &c_{\nu_2} - c_{\nu_1}^2 = - \frac{8\pi^2(d-1)}{3d(d+1)N^2} + \frac{8\pi^4(14d^4 + 33d^3 - 34d^2 - 61d - 24)}{9d^2(d+1)^2(d+3)N^4} - \bigo{\frac{1}{dN^5}} = - \bigtheta{ \frac{1}{dN^2}}, \\
    &c_{\nu_2} + c_{\nu_4} - 2 c_{\nu_1}^2 = \frac{32\pi^4(d^2+d-4)(2d^2+3d+3)}{9d^2(d+1)^2(d+3)N^4} - \bigo{ \frac{1}{dN^5} } = \bigtheta{ \frac{1}{dN^4}}.
\end{aligned}
\end{equation}
If we further reformulate $\Tr[O \sigma_0 O \sigma_0]$ in terms of $\sigma$:
$$
\Tr[O \sigma_0 O \sigma_0] = \Tr\left[ O\sigma O\sigma  \right] - \frac{2}{d} \Tr\left[O^2 \sigma  \right] + \frac{1}{d^2} \Tr[O^2].
$$
Plugging these expressions into \eref{eqn:estimator_X_variance_full_expression}, and note that in \eref{eqn:scaling_of_leading_term}, both $\frac{c_{\nu_1} - c_{\nu_2}}{d+2} + \frac{c_{\nu_1} - c_{\nu_4}}{d-2}$ and $c_{\nu_2} - c_{\nu_1}^2$ have identical leading term coefficients. Since $c_{\nu_1}^{-2} = 1 + o(1) = \bigo{1}$, the scaling of the variance of estimator $\hat{X}_j$ is bounded by
\begin{equation}
\label{eqn:variance_ultimate_upper_bound}
    \begin{aligned}
\Var\left[\hat{X}_j\right] &\leq \bigtheta{\frac{1}{d^2 N^4}} \cdot \Tr[\sigma_0^2] \Tr[O^2] + \bigtheta{ \frac{1}{dN^2}} \cdot \left(\Tr[O^2 \sigma_0^2] - \Tr[O \sigma_0 O \sigma_0] \right) \\
&\quad + \bigtheta{ \frac{1}{dN^4}} \cdot \left( \Tr[O \sigma_0]^2 - \Tr[O \sigma_0 O \sigma_0] \right) \\
&= \bigtheta{\frac{1}{d^2 N^4}} \cdot \Tr[\sigma_0^2] \Tr[O^2] + \bigtheta{ \frac{1}{dN^2}} \cdot \left\| [O, \sigma] \right\|_F^2 \\
&\quad + \bigtheta{ \frac{1}{dN^4}} \cdot \left( \Tr[O \sigma]^2 - \Tr[O \sigma O \sigma] + \frac{2}{d} \Tr[O^2 \sigma] - \frac{1}{d^2} \Tr[O^2] \right).
    \end{aligned}
\end{equation}
To proceed to further analysis of this upper bound, we will use the following lemma:

\begin{lemma}
\label{lemma:two_term_diff}
    For an observable $O \in \obs(\B)$ and a legitimate quantum state $\psi  \in \dens(\mathcal{H})$, it holds that
    $$
    \Tr[O \psi ]^2 - \Tr[O \psi  O \psi ] \leq \min\left\{1 - \Tr[\psi ^2] , \B\P \right\}.
    $$
\end{lemma}

\begin{proof}[Proof of \lref{lemma:two_term_diff}]
    One can easily see that $\Tr[O \psi ]^2 - \Tr[O \psi  O \psi ] \leq \Tr[O \psi ]^2 \leq \min\{1, \B\P\}$ using Corollary \ref{corollary:inner_product_terms_upper_bound}. To show the upper bound $1 - \Tr[\psi ^2]$, suppose $O$ has eigen decomposition $O = \sum_j \lambda_j \ket{\phi_j} \bra{\phi_j}$, we write $\psi _{i,j} = \braket{\phi_i | \psi  | \phi_j}$ and reformulate the expression
    $$
    \begin{aligned}
        \Tr[O \psi ]^2 - \Tr[O \psi  O \psi ] &= \left( \sum_{i} \psi _{i, i} \lambda_i \right) \left( \sum_{j} \psi _{j, j} \lambda_j \right) - \sum_{i, j} \left|\psi _{i,j}\right|^2 \lambda_i \lambda_j \\
        &= \sum_{i, j} \left( \psi _{i, i}\psi _{j, j} - \left|\psi _{i,j}\right|^2 \right) \lambda_{i} \lambda_j \\
        &\leq \max_{i, j} \left|\lambda_{i} \lambda_j\right| \cdot \sum_{i, j} \left( \psi _{i, i}\psi _{j, j} - \left|\psi _{i,j}\right|^2 \right) \\
        &= \max_{i, j} \left|\lambda_{i} \lambda_j\right| \cdot \left( \Tr[\psi ]^2 - \Tr[\psi ^2] \right) =  \max_{i, j} \left|\lambda_{i} \lambda_j\right| \cdot \left( 1 - \Tr[\psi ^2] \right),
    \end{aligned}
    $$
    where the first inequality is guaranteed by the principal submatrix $[\psi _{a, b}]_{a, b \in \{i, j\} } \succeq 0$ for any density matrix $\psi $. Since $O \in \obs(\B)$, $\max_{i, j} \left|\lambda_{i} \lambda_j\right| \leq \|O\|_{\infty}^2 \leq 1$. Combining these two upper bounds and noting that $\min\{\min\{1, \B\P\}, 1 - \Tr[\psi ^2]\} = \min\{1-\Tr[\psi ^2], \B\P\}$ conclude the proof.
\end{proof}

Then we can proceed to analyze the upper bound on the variance case by case.
\paragraph{When the state $\rho$ is pure.} The purity $\Tr[\sigma] =\P = 1$, and using \lref{lemma:two_term_diff}, the term $\Tr[O \sigma ]^2 - \Tr\left[ O\sigma O\sigma  \right]$ vanishes. Therefore, from Corollary \ref{corollary:inner_product_terms_upper_bound} we can readily establish an upper bound
\begin{equation}
    \label{eqn:pure_state_variance_upper_bound}
    \begin{aligned}
        \Var\left[\hat{X}_j\right] &\leq \bigo{ \frac{\B}{d^2N^4} } + \bigtheta{\frac{1}{d N^2}} = \bigo{ \frac{1}{d N^2} }
        = \bigo{ \frac{d^3}{s^2} } ,
    \end{aligned}
\end{equation}
here we have used the condition that $\B \geq 1$ [cf. Problem \ref{prob:CSEU}].

\paragraph{When the state $\rho$ is mixed.} Indeed, one can readily invoke Corollary \ref{corollary:inner_product_terms_upper_bound} and Lemma \ref{lemma:two_term_diff} again to obtain
$$
\begin{aligned}
\Var\left[\hat{X}_j\right] &\leq \bigo{ \frac{\B\P}{d^2N^4} } + \bigtheta{\frac{1}{d N^2}} \cdot \P + \bigtheta{\frac{1}{dN^4}} \cdot \min\left\{1 - \Tr[\rho^2], \B\P \right\} \\
&\leq \bigo{ \frac{\P}{d N^2} + \frac{1}{dN^4} \min\left\{1 - \Tr[\rho^2], \B\P \right\} }.
\end{aligned}
$$

\begin{remark}[Scaling of the leading coefficient in $d$ for higher-order terms in \eref{eqn:scaling_of_leading_term}]
    It is beneficial to put a minor caveat that in our evaluation of the coefficients presented in \eref{eqn:scaling_of_leading_term}, the leading terms we keep in terms of $1/N^k$ strictly dominate for a sufficiently large yet constant $N$ regardless of the value of $d$. We provide a coarse-grained analysis that suffices to unveil this fact.

    For a fixed weight $w \in \mathfrak{X}(W_{\nu})$, we denote $N_j(w)$ as the number of entries in the first $d - 1$ coordinate of $w$ with absolute value $j \in \{0, 1, 2\}$. If we write $X = h(1)$ and $Y = h(2)$, then $c_{\nu}$ can be written as
    $$
    c_{\nu} = \frac{1}{\dim W_{\nu}} \sum_{w \in \mathfrak{X}(W_{\nu}) } m_{\nu}(w) \prod_{j=1}^{d-1} h(|w_j|) = \frac{1}{\dim W_{\nu}} \sum_{w \in \mathfrak{X}(W_{\nu}) } m_{\nu}(w) X^{N_1(w)} Y^{N_2(w)} =: \mathfrak{M}_{\nu}(X, Y).
    $$
    The expression $\mathfrak{M}_{\nu}(X, Y)$ can be expressed as a linear combination of the monomials $X^{N_1(w)} Y^{N_2(w)}$ with the convention that $X^0 Y^0 = 1$, where the coefficients are determined via the evaluations in the proof of \lref{lemma:expression_of_factors}. Instead of explicitly evaluating $\sum_{j=1}^{d-1}|w_j|^k$ and $(\sum_{j=1}^{d-1}|w_j|^2)^2$ for each orbit explicitly, we replace them with the monomials $X^{N_1(w)} Y^{N_2(w)}$ according to each $w$. For instance, the vector $(2, -1, -1, 0, \dots, 0)$ gives rise to $X^2 Y$. Inheriting the multiplicity and counting results in the proof of \lref{lemma:expression_of_factors}, we obtain
    \begin{align*}
        \mathfrak{M}_{\nu_1}(X, Y) &= \frac{d-2}{d+1}X^2 + \frac{2}{d+1} X + \frac{1}{d+1}, \\
        \mathfrak{M}_{\nu_2}(X, Y) &= \frac{(d-2)(d-3)(d-4)}{d^2(d+3)} X^4 + \frac{4(d-2)(d-3)}{d^2(d+3)} X^3 + \frac{4(d-2)(d-3)}{d^2(d+3)} X^2Y \\ & \quad + \frac{4(d-2)}{d(d+3)} X^2 + \frac{8(d-2)}{d^2(d+3)} XY + \frac{4(d-2)}{d^2(d+3)} Y^2 + \frac{8(d-1)}{d^2(d+3)} X + \frac{8}{d^2(d+3)} Y + \frac{2}{d(d+3)}, \\
            \mathfrak{M}_{\nu_3}(X, Y) &= \frac{(d-3)(d-4)}{(d+1)(d+2)} X^4 + \frac{4(d-3)}{(d+1)(d+2)} X^3 + \frac{2(d-3)}{(d+1)(d+2)} X^2 Y \\
&\quad + \frac{2(2d-3)}{(d+1)(d+2)} X^2 + \frac{4}{(d+1)(d+2)} XY + \frac{8}{(d+1)(d+2)} X + \frac{2}{(d+1)(d+2)}, \\
    \mathfrak{M}_{\nu_4}(X, Y) &= \frac{(d-1)(d-2)(d-4)}{d^2(d+1)} X^4 + \frac{4(d-1)(d-2)}{d^2(d+1)} X^3  + \frac{4(d-1)(d-2)}{d^2(d+1)} X^2 + \frac{8(d-1)}{d^2(d+1)}{X} + \frac{2}{d(d+1)}.
    \end{align*}
    For simplicity, we keep the leading order of the coefficients in the large-$d$ regime. Since $0 < X, Y < 1$,
    \begin{align*}
        c_{\nu_1} = \mathfrak{M}_{\nu_1}(X, Y) &\simeq X^2 + \frac{1}{d} \left(X + 1\right), \\
        c_{\nu_2} =  \mathfrak{M}_{\nu_2}(X, Y) &\simeq X^4 +\frac{1}{d} \left(  4X^3 + 4X^2 Y + 4X^2 + \frac{8}{d} XY + \frac{4}{d} Y^2 + \frac{8}{d} X+ \frac{8}{d^2} Y + \frac{2}{d} \right) \\
        &\simeq X^4 +\frac{4}{d} \left(X^3 + X^2 Y + X^2 \right) + \bigo{\frac{1}{d^2}}, \\
                c_{\nu_3} = \mathfrak{M}_{\nu_3}(X, Y) &\simeq X^4 + \frac{1}{d}\left( 4X^3 + 2X^2 Y + 4X^2 + \frac{4}{d} XY + \frac{8}{d} X + \frac{2}{d} \right) \\
        &\simeq X^4 + \frac{2}{d}\left( 2X^3 + X^2 Y + 2X^2 \right) + \bigo{\frac{1}{d^2}} \\
        c_{\nu_4} = \mathfrak{M}_{\nu_4}(X, Y) &\simeq X^4 + \frac{1}{d}\left(4 X^3 + 4X^2 + \frac{8}{d} X + \frac{2}{d} \right) \simeq X^4 + \frac{4}{d} \left( X^3 + X^2 \right) + \bigo{\frac{1}{d^2}}.
    \end{align*}
    Plugging these into the coefficients in \eref{eqn:scaling_of_leading_term}, we can show their asymptotics in terms of $d$
    \begin{equation}
    \label{eqn:scaling_of_coefficient}
    \begin{aligned}
        \frac{1}{d^2 - 1} - \frac{2c_{\nu_1}}{d^2 - 4} + \frac{c_{\nu_2}}{2(d+1)(d+2)} + \frac{c_{\nu_4}}{2(d-1)(d-2)} &\simeq \frac{(1-X^2)^2}{d} + \bigo{\frac{1}{d^3}} = \bigo{\frac{1}{d^2}}, \\
        \frac{c_{\nu_1} - c_{\nu_2}}{d+2} + \frac{c_{\nu_1} - c_{\nu_4}}{d-2} &\simeq \frac{2(X^2 - X^4)}{d} + \bigo{\frac{1}{d^2}} = \bigo{\frac{1}{d}}, \\
        c_{\nu_2} - c_{\nu_1}^2 &\simeq \frac{2}{d} \left( X^3 + 2X^2 Y + X^2 \right) + \bigo{\frac{1}{d^2}} = \bigo{\frac{1}{d}},\\
        c_{\nu_2} + c_{\nu_4} - 2c_{\nu_1}^2 &\simeq \frac{4}{d} \left(X^3 + X^2 Y + X^2 \right) + \bigo{\frac{1}{d^2}} = \bigo{\frac{1}{d}}.
    \end{aligned}
    \end{equation}
    \eref{eqn:scaling_of_coefficient} characterizes the dependence of the leading coefficient on $d$, consistent with the result we reported in \eref{eqn:scaling_of_leading_term}. This confirms that the coefficients of the leading $1/N^k$ terms dominate those of the omitted residual terms, regardless of the dimension $d$.
\end{remark}

The following lemma allows us to reduce the worst-case maximization of 
$\Var[\hat X_j]$ over input states to pure states. 
The first statement of \lref{lemma:large_s_regime_estimator_X_variance} follows from this reduction and \eref{eqn:pure_state_variance_upper_bound}.

\begin{lemma}[Pure states suffice]\label{lemma:Var_reduce_pure}
Let $\mathcal H$ be a finite-dimensional Hilbert space, and $U\in\mathsf U(\mathcal H)$, $O$ being an observable on $\mathcal H$, and let the scalar $c>0$. Suppose $\mu$ is any normalized measure on $\mathsf U(\mathcal H)$, and let $p(\hat U)$ be a probability density with respect to $\mu$. Define the averaging operation on $O$ as
\begin{align}\label{eqn:def_Sigma_O}
    \Sigma_O
    :=
    \int_{\mathsf U(\mathcal H)}
    p(\hat U)\,
    \bigl(\hat U^\dagger O \hat U\bigr)
    \otimes
    \bigl(\hat U^\dagger O \hat U\bigr)
    \dd\mu(\hat U).
\end{align}
Then it holds that
\[
\max_{\rho\in\dens(\mathcal H)}
\left\{
    c\cdot\Tr \left[\Sigma_O(\rho\otimes\rho)\right]
    -
    \bigl[\Tr(OU\rho U^\dag)\bigr]^2
\right\}
=
\max_{\psi\in\mathfrak P(\mathcal H)}
\left\{
   c\cdot\Tr \left[\Sigma_O(\psi\otimes\psi)\right]
    -
    \bigl[\Tr(OU\psi U^\dag)\bigr]^2
\right\}.
\]
Equivalently, the maximum value over all requested density operators is always attained by a pure state.
\end{lemma}

Recall that, by \lref{lemma:SAR_channel_average_result} and \eref{eqn:linear_estimator}, for traceless observable $O\in\obs(\B)$, we have 
\[
\hat X_j
=
\frac{1}{\p_{\bfq}}
\Tr\left( O \hat U_j \rho_0 \hat U_j^\dagger \right)
=
\frac{1}{\p_{\bfq}}
\Tr\left( O \hat U_j \rho \hat U_j^\dagger \right),
\qquad
\E[\hat X_j]
=
\Tr(OU\rho U^\dagger).
\]
Thus,
\[
\Var[\hat X_j]
=
\E[\hat X_j^2]-\E[\hat X_j]^2    
=
\frac{1}{\p_{\bfq}^2}
\E\left[
    \left(
    \Tr\left( O\hat U_j\rho\hat U_j^\dagger \right)
    \right)^2
\right]
-
\bigl[\Tr(OU\rho U^\dagger)\bigr]^2 .
\]
Let $p(\hat U)$ be the probability density of the random outcome $\hat U_j$. 
Then the variance can be written as
\[
\begin{aligned}
\Var[\hat X_j]
&=
\frac{1}{\p_{\bfq}^2}
\int_{\mathsf U(\mathcal H)}
    p(\hat U)\,
    \left[
    \Tr\left( O\hat U\rho \hat U^\dagger \right)
    \right]^2
    \dd\mu(\hat U)
-
\bigl[\Tr(OU\rho U^\dagger)\bigr]^2        \\
&=
\frac{1}{\p_{\bfq}^2}
\Tr[\Sigma_O(\rho\otimes\rho)]
-
\bigl[\Tr(OU\rho U^\dagger)\bigr]^2,
\end{aligned}
\]
where $\Sigma_O$ is defined in \eref{eqn:def_Sigma_O}. 
Applying \lref{lemma:Var_reduce_pure} with $c=\p_{\bfq}^{-2}$, we see that the maximum of $\Var[\hat X_j]$ over all input states $\rho$ is attained by a pure state. 
Combining this reduction with \eref{eqn:pure_state_variance_upper_bound}, which bounds the variance for pure input states, we obtain
\[
\Var[\hat X_j]
=
\bigo{\frac{d^3}{s^2}}
\]
for arbitrary input state $\rho$, provide that $s\geq C d^2$ for a sufficiently large constant $C$. 
This therefore proves the first statement of \lref{lemma:large_s_regime_estimator_X_variance}.

\begin{claim} \label{claim:expectation_preserving_decomposition}
For every density operator $\rho$, there exists a finite pure-state
decomposition
\[
    \rho=\sum_j q_j\psi_j: \quad 
    q_j\geq 0,\quad \sum_j q_j=1,
\]
such that
$\Tr(\widetilde O\psi_j)=\Tr(\widetilde O\rho)$
for every $j$.
\end{claim}

\begin{proof}[Proof of Claim \ref{claim:expectation_preserving_decomposition}] Let $S=\operatorname{supp}(\rho)$, and let
$P_S$ be the projector onto $S$. When $\rho$ is pure, the claim is
trivial. Therefore, we assume that $\operatorname{rank}(\rho)\geq 2$ without loss of generality. 
Since $\rho$ is supported on $S$, the value of $\Tr(\widetilde O\rho)$ lies in the interval
\[
    \left[
        \lambda_{\min}(P_S\widetilde O P_S|_S),
        \lambda_{\max}(P_S\widetilde O P_S|_S)
    \right],
\]
where $\lambda_{\max}$ and $\lambda_{\min}$ denote the maximum and minimum eigenvalues, respectively, and $P_S\widetilde{O}P_S \big|_S$ denotes the restriction of the projected operator $P_S \widetilde{O} P_S$ to the subspace $S$.
Let $|v_{\min}\rangle$ and $|v_{\max}\rangle$ be eigenstates corresponding to $\lambda_{\min}$ and $\lambda_{\max}$. 
Consider varying a superposition of these two eigenstates continuously:
\[
|\phi_\theta\rangle
    =
    \cos\theta\,|v_{\min}\rangle
    +
    \sin\theta\,|v_{\max}\rangle, \quad \theta \in [0, 2\pi),
\]
the expectation value $\langle\phi_\theta|\widetilde O|\phi_\theta\rangle$ changes continuously from the minimum eigenvalue to the maximum eigenvalue. 
Therefore, there exists a unit vector
$|\phi\rangle\in S$ such that
$\langle\phi|\widetilde O|\phi\rangle=\Tr(\widetilde O\rho)$.
Then we choose
\[
    t:=\frac{1}{\langle \phi|\rho^{-1}_S|\phi\rangle},
\]
where $\rho^{-1}_S$ denotes the inverse of $\rho$ restricted to its support $S$. 
Then $0<t<1$, $\rho-t\phi\geq 0$, and $\rho-t\phi$ has rank strictly smaller than that of $\rho$. 
Define the density operator
$
    \rho'
    :=
    \frac{\rho-t\phi}{1-t}.
$ and
$\rho=t\phi+(1-t)\rho'$.
Moreover,
\[
\Tr(\widetilde O\rho')
=
\frac{\Tr(\widetilde O\rho)-t\,\Tr(\widetilde O\phi)}{1-t}
=
\Tr(\widetilde O\rho).
\]
Repeating this rank-reduction procedure finitely many times yields a pure-state
decomposition
$\rho=\sum_j q_j\psi_j$
such that
$\Tr(\widetilde O\psi_j)=\Tr(\widetilde O\rho)$
for all $j$. This proves the claim.
\end{proof}

\begin{proof}[Proof of \lref{lemma:Var_reduce_pure}]
For each unitary $\hat U$, we define the Heisenberg-evolved state
\[
    O_{\hat U}:=\hat U^\dagger O\hat U,
    \qquad
    \widetilde O:=U^\dagger O U .
\]
For any density operator $\rho$, define
\[
\begin{aligned}
    \mathscr{F}(\rho)
    &:=
    c\cdot\Tr[\Sigma_O(\rho\otimes\rho)]
    -
    \bigl[\Tr(OU\rho U^\dag)\bigr]^2 \\
   & =
    c\cdot\int_{\mathsf U(\mathcal H)}
    p(\hat U)\,
    \bigl[\Tr(O_{\hat U}\rho)\bigr]^2
    \dd\mu(\hat U)
    -
    \bigl[\Tr(\widetilde O\rho)\bigr]^2 .
\end{aligned}
\]
Now we apply Claim \ref{claim:expectation_preserving_decomposition} to an arbitrary $\rho\in\dens(\mathcal H)$. 
Let $\rho=\sum_j q_j\psi_j$ be the pure-state decomposition according to $O$ such that
$\Tr(\widetilde O\psi_j)=\Tr(\widetilde O\rho)$
for every $j$. 
For each $\hat U$, define
$x_j(\hat U):=\Tr(O_{\hat U}\psi_j)$.
Then
\[
    \Tr(O_{\hat U}\rho)
    =
    \sum_j q_j x_j(\hat U).
\]
By convexity of the map $x\mapsto x^2$, it holds that
$
    \bigl[\Tr(O_{\hat U}\rho)\bigr]^2
    =
    \left(\sum_j q_j x_j(\hat U)\right)^2
    \leq
    \sum_j q_j x_j(\hat U)^2
$. Integrating both sides of the inequality over $\hat U$, we obtain
$$
\begin{aligned}
\Tr[\Sigma_O(\rho\otimes\rho)]
    &=
    \int_{\mathsf U(\mathcal H)}
    p(\hat U)\,
    \bigl[\Tr(O_{\hat U}\rho)\bigr]^2
    \dd\mu(\hat U)     \\
    &\leq
    \sum_j q_j
    \int_{\mathsf U(\mathcal H)}
    p(\hat U)\,
    \bigl[\Tr(O_{\hat U}\psi_j)\bigr]^2
    \dd\mu(\hat U)                                      =
    \sum_j q_j
    \Tr[\Sigma_O(\psi_j\otimes\psi_j)] .
    \end{aligned}
$$
Since $\Tr(\widetilde O\psi_j)=\Tr(\widetilde O\rho)$ for every $j$, it holds that $
    \bigl[\Tr(\widetilde O\rho)\bigr]^2
    =
    \sum_j q_j \bigl[\Tr(\widetilde O\psi_j)\bigr]^2$, and therefore,
\[
\begin{aligned}
\mathscr{F}(\rho)
&=
c\cdot\Tr[\Sigma_O(\rho\otimes\rho)]
-
\bigl[\Tr(\widetilde O\rho)\bigr]^2   \\
&\leq
\sum_j q_j 
\left\{c\cdot
    \Tr[\Sigma_O(\psi_j\otimes\psi_j)]
    -
    \bigl[\Tr(\widetilde O\psi_j)\bigr]^2
\right\}                                      
=
\sum_j q_j \mathscr{F}(\psi_j).
\end{aligned}
\]
Hence, there exists at least one index $j_0$ such that
 $\mathscr{F}(\psi_{j_0})\geq \mathscr{F}(\rho)$.
Thus, every value attained by a mixed state is upper-bounded by a value attained by some pure state. Consequently,
\[
    \sup_{\rho\in\dens(\mathcal H)}\mathscr{F}(\rho)
    \leq
    \sup_{\psi\in\mathfrak P(\mathcal H)}\mathscr{F}(\psi).
\]
The reverse inequality follows immediately from
$\mathfrak P(\mathcal H)\subseteq \dens(\mathcal H)$.
Therefore,
\[
    \sup_{\rho\in\dens(\mathcal H)}\mathscr{F}(\rho)
    =
    \sup_{\psi\in\mathfrak P(\mathcal H)}\mathscr{F}(\psi).
\]

Finally, since $\dens(\mathcal H)$ is compact and $\mathscr{F}$ is continuous,
the supremum over $\dens(\mathcal H)$ is attained. 
The argument above then shows that a maximizer can be chosen to be pure. 
This completes the proof.
\end{proof}

\section{Proof of \tref{thm:cseu_lower_bound}} \label{sec:detailed_proof_of_cseu_lower_bound}

The skeleton of our lower-bound argument for the CSEU problem can be summarized as follows: First, we introduce the list-decoding problem, in which an algorithm queries an oracle with entries that depend affine-linearly on a hidden parameter and must output a list containing the associated label, and we will show that this problem is hard whenever the parameter space is large. Second, we construct a hard instance of it by a reflection-packing argument, and show that any algorithm for CSEU yields a list decoder for that instance, provided the list decoder may query a variant of the primitive oracle, which we call the fractional oracle. Third, adapting a standard argument from prior work \cite{HKOT23}, we could simulate each fractional query by primitive queries at the cost of losing a bounded success probability. Here we adopt the convention that every oracle, either primitive or fractional, is available together with its controlled version, so that the same simulation applies unchanged to controlled fractional queries. Finally, combining the simulation overhead with the hardness established in the first step yields the claimed query lower bound for CSEU.

    To start with, we derive the ingredient lemmas and set up the parameters in steps as standard technical tools that will be used somewhere in the subsequent proof.

\begin{definition}[Hamming distance and entropy functions] \label{def:hamming_distance_entropic_bound}
    For vectors on the lattice $\mathbf{x}, \mathbf{y} \in [N]^m$, the conventional Hamming distance $\ham(\mathbf{x}, \mathbf{y}) := |\{j \in [m]: x_j \neq y_j\}|$. For any $0 \leq u \leq 1$, we define $h(u) = -u \log u - (1-u) \log (1-u)$ and $\Phi(u) = u \log (e(1+u) / u)$ with the conventions that $0 \log 0 := 0$ and $\Phi(0) = 0$.
\end{definition}

\begin{lemma}[Entropy bounds, see, e.g., \cite{Roth2006}] \label{lemma:entropy_bound}
    For any choice of integers $m, N, k \in \N$ and $Q \geq 0$, $0 \leq \delta \leq 1/2$, it holds that
    $$
    \sum_{j=0}^{\lfloor \delta m \rfloor} \binom{m}{j} (N - 1)^j \leq e^{m h(\delta)} N^{\delta m}, \quad \log \binom{k + Q}{Q} \leq k \Phi\left( \frac{Q}{k} \right).
    $$
    where the functions $h(\cdot)$, $\Phi(\cdot)$ are defined as that in \dref{def:hamming_distance_entropic_bound}. Moreover, $\Phi$ monotonously increases on $[0, \infty)$ subject to $\lim_{u \to 0} \Phi(u) = 0$.
\end{lemma}

\begin{proof}[Proof of \lref{lemma:entropy_bound}]
    We first prove the second statement. Note that $\frac{\mathrm{d}}{\mathrm{d}u} \Phi(u) = \log(1 + 1/u) + \frac{u}{1+u} > 0$ for all $u > 0$. Moreover, via change of variable, $\lim_{u \to 0} u \log(e(1 + 1/u)) = \lim_{x \to \infty} \frac{\log(e(1+x))}{x} = 0$. As for the first inequality, assume without loss of generality that $\lfloor \delta m \rfloor > 0$, and that if we set $u := \lfloor \delta m \rfloor / m \leq \delta \leq 1 /2$, then for all $j \leq \lfloor \delta m \rfloor$, it holds that
    $$
    u^{j} (1 - u)^{m - j} \geq u^{\lfloor \delta m \rfloor} (1 - u)^{m - \lfloor \delta m \rfloor} = e^{-m h( u)},
    $$
    therefore, multiplying by $\binom{m}{j}(N-1)^j$ and summing over all possible $j$ yields
    $$
    \begin{aligned}
    \sum_{j=0}^{\lfloor \delta m \rfloor} \binom{m}{j} e^{-m h(\delta)} (N-1)^j &\leq 
    \sum_{j=0}^{\lfloor \delta m \rfloor} \binom{m}{j} e^{-m h(u)} (N-1)^j \leq \sum_{j=0}^{\lfloor \delta m \rfloor} \binom{m}{j} u^j (1 - u)^{m-j} (N-1)^j \\
    &\leq \sum_{j=0}^{\lfloor \delta m \rfloor} \binom{m}{j} u^j (1 - u)^{m-j} N^{\delta m} = N^{\delta m},
    \end{aligned}
    $$
    so multiplying both sides of the inequality by $e^{m h(\delta)}$ proves the first claim. The second inequality follows from the handy inequality that $\binom{k+Q}{Q} \leq ( \frac{e(k+Q)}{Q} )^Q$, and therefore
    $$
    \log \binom{k+Q}{Q} \leq Q \log \left( \frac{e(k+Q)}{Q}  \right) = k \Phi\left( \frac{Q}{k}\right),
    $$
    concluding the proof.
\end{proof}

\begin{lemma}[Haar moment identities \cite{CollinsSniady06, Mele2024introductiontohaar}] \label{lemma:haar_moments}
    Let the unit vector $\ket{u}$ be Haar random over the unit sphere of $\C^d$, and let $\F \in \mathbb{C}^{2n \times 2n}$ be the SWAP operator, then for any matrices $A, B \in \C^{d \times d}$ and vector $\ket{w} \in \C^d$,
    $$
    \begin{aligned}
    &\E_{u}\left[ \ket{u} \bra{u} \right] = \frac{\openone}{d}, \quad \E_{u} \left[ \ket{u} \bra{u}^{\otimes 2}  \right] = \frac{\openone \otimes \openone + \F}{d(d+1)}, \\
    &\E_{u} \left[ \braket{u | A | u} \braket{u | B | u} \right] = \frac{\Tr[A] \Tr[B] + \Tr[AB]}{d(d+1)}, \quad \E_{u}\left[ \left|\braket{u | w}\right|^2 \right] = \frac{\braket{w | w}}{d}, \quad \E_{u}\left[ \left|\braket{u | w}\right|^4 \right] = \frac{2 \braket{w | w}^2}{d(d+1)}.
    \end{aligned}
    $$
\end{lemma}

\begin{lemma}[Paley-Zygmund \cite{PaleyZygmund32}] \label{lemma:paley_zygmund}
    Let $X \geq 0$ be a positive r.v. with $0 < \E[X^2] < \infty$, then for every $0 < \theta < 1$, then
    $$
    \Pr\left[ X \geq \theta \E[X] \right] \geq (1 - \theta)^2 \frac{\E[X]^2}{\E[X^2]}.
    $$
\end{lemma}

\subsection{Establishing the general list decoding query lower bound} 
\label{sec:list_decoding_lower_bound}
The hard instance used in our lower bound argument is a list decoding problem: Given queries to an operator oracle whose entries are parameterized by a high-dimensional label, we are asked to output a prediction about the label from the candidate set. We first establish a general counting bound for the list decoding problem when the oracle depends affine-linearly on a bounded number of real parameters to show its hardness. Note that the affine-linearity is used here to take into account the access to the controlled version of the reflections in our hard family. Subsequent discussions will show how CSEU, with a certain number of oracular queries to the unknown reflections, solves the list decoding problem, thereby proving our lower-bound argument.

\begin{definition}[Affine list-decoding problem, see, e.g., \cite{list_decoding_2008}]
\label{def:affine_list_decoding_problem}
    Let $\mathcal{X}$ be a finite label set, let
    $\mathcal{Z} \subseteq \R^k$, and let
    $x \mapsto z(x) \in \mathcal{Z}$ be injective.
    Suppose $\{\mathscr{O}^{z}\}_{z \in \mathcal{Z}}$
    is a family of quantum oracles such that every matrix entry of
    $\mathscr{O}^{z}$ is affine linear in $z$. We say that a $Q$-query quantum algorithm solves the
    $(L,p)$-list-decoding problem for this oracle family if it always
    outputs a list $\mathsf{L} \subseteq \mathcal{X}$ satisfying
    $|\mathsf{L}| \leq L$, and its average list-membership probability
    satisfies
    $$
    \frac{1}{|\mathcal{X}|}
    \sum_{x \in \mathcal{X}}
    \Pr\left[x \in \mathsf{L} |  \mathscr{O}^{z(x)}\right]
    \geq p.
    $$
\end{definition}

To proceed, we first establish a general query lower bound for the task defined in Definition \ref{def:affine_list_decoding_problem}.

\begin{lemma}[Affine parameter listing counting bound, \cite{BealsEtAl01, she_et_al:LIPIcs.ITCS.2023.96}, adapted] \label{lemma:affine_list_counting_success_bound}
    Let a parameter space $\mathcal{Z} \subseteq \R^k$, and every matrix entry of an oracle $\mathscr{O}^z$ is affine linear in terms of $z \in \mathcal{Z}$, and $N$ labels are mapped injectively into $\mathcal{Z}$ via $x \mapsto z(x) = (z_j(x))_{j=1}^k$, then for any quantum algorithm that makes at most $Q$ queries\footnote{
    Specifically, with definite causal order, by the usual assumption. We are unsure about how indefinite causal order queries \cite{Chiribella_2013} would violate this bound.
    } to $\mathscr{O}^z$ and outputs a list $\mathsf{L} \subseteq [N]$ of size at most $L$ achieves the list membership of $x$ with average probability bounded from above by
    $$
    \frac{1}{N} \sum_{x=1}^N \Pr\left[ x \in \mathsf{L} | \mathscr{O}^{z(x)}\right] \leq \frac{L}{N} \binom{k + Q}{Q}.
    $$
\end{lemma}

\begin{proof}[Proof of \lref{lemma:affine_list_counting_success_bound}]
    Without loss of generality, we purify all the randomness, dilate all the quantum channels, and defer the quantum measurements. The intermediate state of the quantum circuit after having $j$ queries is given by
    $$
    \ket{\psi_z^{(j)}} = V_j \mathscr{O}^{z} V_{j-1} \mathscr{O}^{z}  \cdots V_{1} \mathscr{O}^{z} \ket{\psi_0} = \sum_{\nu: |\nu| \leq j} z^{\nu} \ket{v_{\nu}^{(j)}}, \quad \nu \in \N^{k}:~|\nu| := \sum_i v_i, ~ z^{v} := \prod_i z_i^{v_i},
    $$
    where $\ket{v_{\nu}^{(j)}}$ absorbs the superposition of states corresponding to making queries to $z$ of order $\nu$. Therefore, the final state $\psi_{z(x)}^Q$ lies in the subspace of dimension at most $\left| \{ \nu \in \N^k: |\nu| \leq Q \} \right| = \binom{k+Q}{Q}$, whose projector is given by $\Pi_{k, Q}$. Then for any POVM that identifies $x$, given by $\{E_{\widehat{x}}\}_{\widehat{x}}$, it holds that the exact recovery probability satisfies
    $$
    \frac{1}{N} \sum_{x=1}^N \braket{\psi_{z(x)}^{Q} | E_{x} | \psi_{z(x)}^Q} \leq \frac{1}{N} \sum_{x} \Tr[E_{x} \Pi_{k, Q}] = \frac{\Tr[\Pi_{k, Q}]}{N} \leq \frac{1}{N} \binom{k+Q}{Q}.
    $$
    Finally, since first predicting a size-$L$ list and then randomly picking one element as the exact prediction probability is always a bad strategy for exact recovery [cf. \cite[Proposition 4.2]{HKOT23}], it holds that
    $$
\frac{1}{NL} \sum_{x=1}^N \Pr\left[ x \in \mathsf{L} | \mathscr{O}^{z(x)} \right] \leq \frac{1}{N} \sum_{x=1}^N \Pr\left[ \widehat{x} = x | \mathscr{O}^{z(x)} \right] = \frac{1}{N} \sum_{x=1}^N \braket{\psi_{z(x)}^{Q} | E_{x} | \psi_{z(x)}^Q} \leq \frac{1}{N} \binom{k + Q}{Q}.
    $$
    Multiplying by $L$ on both sides of the inequality concludes the proof.
\end{proof}

\lref{lemma:affine_list_counting_success_bound} gives a general upper bound on the average success probability of any algorithm that attempts to solve the list-decoding problem. As is standard in lower-bound arguments, instantiating the problem with suitably chosen parameters yields the query lower bound stated in the following lemma:

\begin{lemma}[Affine list-decoding query lower bound]
\label{lemma:affine_list_decoding_query_lower_bound}
    Fix any constant $\beta > 0$.
    There exist constants $\zeta = \zeta(\beta) > 0$ and
    $k_0 = k_0(\beta) \in \N$ such that for an affine list-decoding problem as in \dref{def:affine_list_decoding_problem} with $k \geq k_0$ parameters,
    satisfying
    $
    |\mathcal{X}| \geq e^{\beta k} $ and $L \leq e^{\beta k/4}$, suppose a quantum algorithm makes at most $Q$ queries for some
    $s \geq 0$ that satisfies $Q \leq 50 + 100s$ solves the $(L, \frac{1}{2}e^{-\pi s})$-list-decoding problem, then we must have $s > \zeta k$.
\end{lemma}

\begin{proof}[Proof of \lref{lemma:affine_list_decoding_query_lower_bound}]
    Choose a sufficiently small constant $\zeta > 0$ such that $
    \Phi(101\zeta) \leq \frac{\beta}{4} $ and $\pi\zeta \leq \frac{\beta}{8}$, then we take $k_0$ sufficiently large such that, whenever $k \geq k_0$, it is guaranteed that $
    50 \leq \zeta k$ and $\log 2 \leq \frac{\beta k}{8}$. Suppose the contrary holds that $s \leq \zeta k$.
    Then $Q \leq 50 + 100s
    \leq 101\zeta k$. By the bounds presented in \lref{lemma:affine_list_counting_success_bound} and
    \lref{lemma:entropy_bound}, we have
    $$
    \begin{aligned}
    \log\left(
    \frac{1}{|\mathcal{X}|}
    \sum_{x \in \mathcal{X}}
    \Pr\left[x \in \mathsf{L} | \mathscr{O}^{z(x)}\right]
    \right)
    &\leq
    \log L
    + \log\binom{k+Q}{Q}
    - \log|\mathcal{X}| \\
    &\leq
    \frac{\beta k}{4}
    + k\Phi(Q/k)
    - \beta k  \leq
    \frac{\beta k}{4} + \frac{\beta k}{4} - \beta k = -\frac{\beta k}{2}.
    \end{aligned}
    $$
    Meanwhile, by the assumed lower bound on the list-membership probability,
    $$
    \begin{aligned}
    \log\left(
    \frac{1}{|\mathcal{X}|}
    \sum_{x \in \mathcal{X}}
    \Pr\left[x \in \mathsf{L} | \mathscr{O}^{z(x)}\right]
    \right) \geq
    -\log 2 - \pi s \geq -\frac{\beta k}{8} -\frac{\beta k}{8}=-\frac{\beta k}{4},
    \end{aligned}
    $$
    a contradiction. Therefore, we must have $s > \zeta k$, finishing the proof.
\end{proof}

Therefore, the above argument tells us that the list-decoding problem remains hard even when the required
success probability is exponentially small in terms of the effective query cost.

\subsection{Constructing the hard reflection list-decoding instance}  
To instantiate the high-level picture of the list-decoding problem defined in Appendix \ref{sec:list_decoding_lower_bound} so that we can correspondingly build a hard instance for CSEU lower bound, we start from a standard reflection packing argument \cite{HHJ17, HKOT23} as our building block, and embed them into higher-dimensional spaces. As reflections are Hermitian unitaries, they are natural candidates for arguing hardness of learning unitaries while saving us from considering queries to their conjugate transpose.

\begin{lemma}[{\cite[Lemma 8, reformulated]{HHJ17}}]
\label{lemma:existence_of_reflections}
    There exists a universal constant $k_0$ such that for every $k \geq k_0 \in \N$, there exists rank-$k$ projectors $\Pi_{1}, \Pi_{2}, \dots, \Pi_{N_0} \in \C^{2k \times 2k}$, such that $N_0 \geq e^{k^2/8}$, and for every $x \neq y \in [N_0]$, it holds that $\|\Pi_{x} - \Pi_{y}\|_1 \geq k/2$.
\end{lemma}

\begin{corollary}[{\cite[Proposition 4.1]{HKOT23}}] \label{cor:construction_of_HKOT_reflections}
    Let $k \geq k_0$, and a bit $b \in \{1, 2\}$, and dimension $D = 2k + b$. There exists a collection of $N_0 \geq e^{k^2/8}$ traceless Hermitian unitaries $\{V_x\}_{x \in [N_0]} \subseteq \U(2k)$ and reflections $R_x \in \U(D)$ such that
    $$
    \Pi_x = \frac{\openone_{2k} + V_x}{2}, \quad R_{x} = V_{x} \oplus \openone_b, \quad \forall x.
    $$
    where $\Pi_x$ satisfies the separations mentioned in \lref{lemma:existence_of_reflections}. Moreover, the projectors $Q_x = \frac{1}{2}(\openone_D - R_x)$ all has rank $k$.
\end{corollary}

\begin{proof}[Proof of \corref{cor:construction_of_HKOT_reflections}]
    Take $\{\Pi_x\}_x$ whose existence is guaranteed by \lref{lemma:existence_of_reflections}, then define $V_x = 2\Pi_x - \openone_{2k}$. Clearly $V_x$ is Hermitian, and $V_x^2 = 4\Pi_x^2 - 4\Pi_x + \openone_{2k} = \openone_{2k}$, and $\Tr[V_x] = 2k-2k = 0$. Moreover, $R_x = V_x \oplus \openone_b$ is a reflection, and that $\rank(Q_x) = \rank(\openone_D - R_x) = \rank((\openone_{2k} - V_x) \oplus 0_b) = \rank(\openone_{2k} - \Pi_x) = k$.
\end{proof}

\begin{lemma}[{\cite[Proposition 4.1, extended]{HKOT23}}] \label{lemma:construction_of_equal_rank_reflection_packing}
    There exists universal constants $D_0 \in \N$ and $\beta_1, \beta_2, \kappa > 0$, such that for every $D \geq D_0$, there exists a collection of reflections $\{R_x\}_{x \in [N]} \subseteq \U(D)$ and $Q_x := \frac{1}{2}(\openone - R_x)$, such that
    $$
    \beta_1 D^2 \leq \log N \leq \beta_2 D^2; \quad \forall x \neq y, \quad \|Q_x - Q_y\|_F \geq \kappa \sqrt{D},
    $$
    and that all $Q_x$ yield the same rank. Equivalently, one may write $N = e^{\Theta(D^2)}$ and $\|Q_x - Q_y\|_F = \Omega(\sqrt{D})$. Without loss of generality, for sufficiently large $D_0$, one may take $\beta_1 = 1/128$, $\beta_2 = 1/64$, and $\kappa = \sqrt{2}/8$.
\end{lemma}

\begin{proof}[Proof of \lref{lemma:construction_of_equal_rank_reflection_packing}]
    Still, we can write $D = 2k + b$ for each choice of $D$, where $b \in \{1, 2\}$. We inherit the reflection constructions from \corref{cor:construction_of_HKOT_reflections}. Then it holds that
    $$
   \forall x, y, \quad \left\| Q_x - Q_y \right\|_1 = \frac{1}{2}\left\|R_x - R_y \right\|_1 = \left\|\Pi_x - \Pi_y\right\|_1 \geq \frac{k}{2}, \quad \rank(Q_x - Q_y) \leq \rank(Q_x) + \rank(Q_y) = 2k.
    $$
    Therefore, by Cauchy-Schwarz [cf. \ref{def:schatten_norm_and_holder}], we have
    $$
    \left\|Q_x - Q_y \right\|_F \geq \frac{\left\|Q_x - Q_y \right\|_1}{\sqrt{\rank(Q_x - Q_y)}} \geq \frac{k/2}{\sqrt{2k}} = \frac{\sqrt{k}}{2\sqrt{2}} \geq \sqrt{\frac{D}{32}},
    $$
    where the last inequality follows from $k = \frac{1}{2}(D - b) \geq D/4$ given that $D_0 \geq 4$. Now it remains to make finer-grained choices of parameters: Take $N := \lceil e^{\beta_1 D^2} \rceil$, and $D$ sufficiently large, we have
    $$
    \beta_1 D^2 \leq \log N \leq \beta_1 D^2 + 1  \leq 2 \beta_1 D^2 := \beta_2 D^2,
    $$
    where taking, say, $\beta_1 = 1/128, \beta_2 = 1/64$ would suffice. Moreover, since $\beta_1 D^2 + 1 \leq (D-2)^2/32 \leq k^2/8$ holds universally, the reflections constructed in \corref{cor:construction_of_HKOT_reflections} is sufficiently large.
\end{proof}

We now use this reflection packing as the building block of an instance of list-decoding problem. Each block carries one reflection label, and the concatenation of the block labels forms the hidden word to be decoded via queries to the unknown reflection.

\begin{lemma}[Hard reflection family]
\label{lemma:hard_reflection_list_decoding_instance}
    There exist universal constants
    $D_0 \in \N$ and $\beta_1,\beta_2>0$ such that the following holds: For every $D_0 \leq D \leq d/4$, let
    $
    m := \left\lfloor \frac{d}{D} \right\rfloor$ and $k := mD^2$. 
    There exists a collection of reflections
    $\{R_x\}_{x \in [N]} \subseteq \U(D)$ and a label space $\mathcal{X} := [N]^m$, such that for every
    $\mathbf{x}=(x_1,\ldots,x_m)\in\mathcal{X}$, the concatenated
    reflection $
    R_{\mathbf{x}}
    :=
    \bigoplus_{j=1}^{m}R_{x_j}
    \oplus
    \openone_{d-mD}
    \in \U(d) $
    satisfies the following properties:
    \begin{enumerate}
        \item
        $
        \frac{1}{2}dD \leq k \leq dD$, $ \beta_1 k \leq \log |\mathcal{X}| = m\log N \leq \beta_2 k
        $.

        \item
        There exists an injective parameterization
        $ \mathbf{x}\mapsto z(\mathbf{x})\in\R^k $, 
        such that every matrix entry of $R_{\mathbf{x}}$ is affine linear in terms of $z(\mathbf{x})$.
        Moreover, any oracle obtained from $R_{\mathbf{x}}$ by a fixed
        combination of direct sums, tensoring with identities, signs,
        and controls remains affine linear in the same parameters.
    \end{enumerate}
\end{lemma}

\begin{proof}
    We take the reflection packing
    $\{R_x\}_{x\in[N]}$ constructed in
    \lref{lemma:construction_of_equal_rank_reflection_packing}.
    By construction, we have $\beta_1 D^2 \leq \log N \leq \beta_2 D^2$. Since $D\leq d/4$, we have
    $
    \frac{d}{2D}
    \leq
    m
    =
    \left\lfloor\frac{d}{D}\right\rfloor
    \leq
    \frac{d}{D},
    $, and therefore $ \frac{dD}{2} \leq k = mD^2 \leq dD
    $, this gives the first batch of properties.

    To obtain the affine parameterization, fix any real-linear coordinate isomorphism $g:\Herm(D)\rightarrow\R^{D^2},
    $
    so that $
    z(\mathbf{x})
    :=
    \left(
    g(R_{x_1}),
    \ldots,
    g(R_{x_m})
    \right)
    \in\R^{mD^2}
    $.
    The map is clearly injective, and every matrix entry of
    $R_{\mathbf{x}}$ is affine linear in these coordinates.
    The same remains true after operations including direct sums, signs, tensoring
    with identities, and controls. This completes the proof.
\end{proof}

By the statement of \lref{lemma:hard_reflection_list_decoding_instance}, the hidden word $\mathbf{x}\in[N]^m$ has already defined an exponentially large affine list-decoding instance. As per
\lref{lemma:affine_list_decoding_query_lower_bound}, it remains hard
whenever the output list is exponentially smaller than the full label space. 

Specifically, to make this witness appear at the $\eps$-accuracy scale of CSEU, we would pass the oracles to fractional reflections, as per the proof techniques of \cite{HKOT23}:

\begin{definition}[{Fractional reflections \cite[Definition 4.4]{HKOT23}, \cite{BCC17}}] \label{def:fractional_reflection}
    A reflection is a Hermitian unitary $R$ such that $R^2 = \openone$, and for any $\alpha \in [-1, 1]$, we define the fractional version of $R$ as
    $$
    R^{\alpha} = \frac{\openone + R}{2} + e^{-\I \pi \alpha} \frac{\openone - R}{2}.
    $$
\end{definition}

\begin{definition}[Instantiating hard reflection list-decoding problem] \label{def:hard_reflection_list_decoding_problem}
    Fix the hard reflection family $\{R_{\mathbf{x}}\}_{\mathbf{x}\in[N]^m}$ constructed in \lref{lemma:hard_reflection_list_decoding_instance}, with $k=mD^2$ and $\beta_1 k \leq m\log N$. We define the hard reflection list-decoding problem as follows: An unknown label $\mathbf{x}\in[N]^m$ specifies an oracle $\mathscr{O}_{\mathbf{x}}$ from an oracle family $\{\mathscr{O}_{\mathbf{x}}:\mathbf{x}\in[N]^m\}$ indexed by the same labels. Given black-box access to $\mathscr{O}_{\mathbf{x}}$, the task is to output a list $\mathsf{L}\subseteq[N]^m$ such that
    $
    |\mathsf{L}|\leq e^{\beta_1 k/4},
    $,
    and the ground-truth label $\mathbf{x}$ belongs to $\mathsf{L}$ with sufficiently large probability. We call $\mathscr{O}_{\mathbf{x}}=R_{\mathbf{x}}$ the primitive reflection realization of this problem, and will also consider the fractional-reflection realization $\mathscr{O}_{\mathbf{x}}=R_{\mathbf{x}}^{\alpha}$. By \lref{lemma:hard_reflection_list_decoding_instance}, the primitive reflection realization is affine linear in the $k$ real parameters.
\end{definition}

 The construction in Definition \ref{def:hard_reflection_list_decoding_problem} serves as an instantiation of the list decoding problem formalized in Definition \ref{def:affine_list_decoding_problem}, and allows us to separate the list-decoding task from the particular query model used to solve it. As the second part of the proof, in the subsequent sections, we will proceed to solving the same list-decoding problem under two different oracle realizations: CSEU will first give a solver using the fractional reflections $R_{\mathbf{x}}^\alpha$ in Appendix \ref{sec:cseu_solves_fractional_list_decoding}, which will then be converted into a solver using the primitive reflections $R_{\mathbf{x}}$ using the techniques presented in Appendix \ref{sec:fractional_generalized_query_normalization}, where the affine query lower bound applies.

\subsection{CSEU solves the hard list-decoding instance with fractional reflections}
\label{sec:cseu_solves_fractional_list_decoding}

In this step, we show that by an appropriate choice of a fixed collection of CSEU requests [cf. Problem \ref{prob:CSEU}], their predicted expectation values would unveil enough information to recover the unknown label of the queried fractional reflection up to a compact list. We will first show that CSEU can be utilized to separate distinct reflections from the primitive reflection family.

\begin{lemma}[Random detection via randomly rotated observables] \label{lemma:random_detection}
    Let $A \in \C^{D \times D}$ be any traceless Hermitian matrix, fix $\rho = \ket{1} \bra{1}$, $O = \ket{+} \bra{+}$ with $\ket{+} = \frac{1}{\sqrt{2}}(\ket{1} + \ket{2})$ for any pair of orthogonal states $\ket{1}, \ket{2} \in \C^D$, define rotated states $\rho(U) = U \rho U^\dagger$, $O(U) = U O U^\dagger$. For any $(\rho, O)$, we define the commutator Hermitian $H(\rho, O) := -\I [\rho, O]$, then it holds that
    $$
    \Pr_{U \sim \haar(\U(D))}\left[ \left| \Tr\left[ H(\rho(U), O(U)) A \right]  \right| \geq \frac{\|A\|_F}{2\sqrt{D^2 - 1}} \right] \geq \frac{1}{64}.
    $$
\end{lemma}

\begin{proof}[Proof of \lref{lemma:random_detection}]
    Clearly $\ket{u} := U \ket{1}$ and $\ket{v} := U \ket{2}$ are Haar random states. Then we have 
    $$
    H(\rho(U), O(U)) = \frac{\I}{2} \left( \ket{v} \bra{u} - \ket{u} \bra{v} \right), \quad \Tr\left[H(\rho(U), O(U)) A  \right] =   \Im \braket{v | A | u},
    $$
    and conditioned on $\ket{u}$, the vector $\ket{v}$ is Haar uniform over the unit sphere modulo the hyperplane $u^{\perp}$, and that using \lref{lemma:haar_moments} in the effective dimension $D-1$ yields
    $$
    \begin{aligned}
    \E\left[ \left| \braket{v | A | u} \right|^2 \right] &= \E_{u}\left[ \E_{v | u}\left[ \left| \braket{v | A | u} \right|^2 \right]  \right] =\E_u\left[  \frac{\braket{u | A^2 | u} - \left| \braket{u | A | u} \right|^2}{D - 1} \right] = \frac{\Tr[A^2]}{D(D-1)} - \frac{\Tr[A^2]}{D(D^2 - 1)} = \frac{\|A\|_F^2}{D^2 - 1}.
    \end{aligned}
    $$
    Since the target identity is invariant under the phased map $\ket{v} \mapsto e^{\I \theta} \ket{v}$ for any $\theta$, so that the identity $\braket{v | A | u}$ is circularly symmetric and that $\E[\Tr\left[H(\rho(U), O(U)) A  \right]^2] = \frac{\|A\|_F^2}{2(D^2 - 1)}$. Then we bound the fourth moment as follows: Take $\ket{w} = (\openone - \ket{u} \bra{u}) A \ket{u}$. Then $\braket{v | w} = \braket{v | A | u}$, so that $\E[\Tr[H(\rho(U), O(U)) A]^4] \leq \E[|\braket{v | A | u}|^4] = \E[|\braket{v|w}|^4]$.
    Applying \lref{lemma:haar_moments} iteratively gives
    $$
    \begin{aligned}
    \E \left[ \left|\braket{v | w}\right|^4 \right] &= \E_{u}\left[ \E_{v|u} \left[ \left|\braket{v | w}\right|^4 \right] \right] = \E_{u}\left[\frac{2\left|\braket{w | w} \right|^4}{D(D-1)} \right] \\
    &= \E_{u}\left[\frac{2 \left| \braket{u | A^2 | u} - \left| \braket{u | A| u} \right|^2 \right|^2 }{D(D-1)} \right] \leq \E_{u}\left[\frac{2\left| \braket{u | A^2 | u} \right|^2}{D(D-1)} \right] \\
    &= \frac{2\left( \Tr[A^2]^2 + \Tr[A^4] \right)}{D^2(D^2-1)} < \frac{4 \|A\|_F^4}{(D^2-1)^2} = 16 \E\left[\Tr\left[H(\rho(U), O(U)) A  \right]^2 \right]^2,
    \end{aligned}
    $$
    where the last inequality is due to Cauchy-Schwarz. Finally, we apply \ref{lemma:paley_zygmund} by taking $\theta = 1/2$:
    $$
    \begin{aligned}
        \Pr\left[   \Tr\left[ H(\rho(U), O(U)) A \right]^2   \geq \frac{1}{2} \E\left[  \Tr\left[ H(\rho(U), O(U)) A \right]^2 \right] \right] &=  \Pr\left[   \Tr\left[ H(\rho(U), O(U)) A \right]^2   \geq \frac{\|A\|_F^2}{4(D^2-1)} \right] \\
        &\geq \frac{1}{4} \frac{\E\left[ \Tr\left[ H(\rho(U), O(U)) A \right]^2 \right]^2}{\E\left[ \Tr\left[ H(\rho(U), O(U)) A \right]^4\right]} 
        \\
        &\geq \frac{1}{4} \frac{\E\left[ \Tr\left[ H(\rho(U), O(U)) A \right]^2 \right]^2}{\E\left[ \left|\braket{v | w} \right|^4\right]} \geq \frac{1}{64},
    \end{aligned}
    $$
    taking square root of both sides of the inequality in the event completes the proof.
\end{proof}

\begin{corollary} \label{cor:haar_random_reflection_detection}
    For the projections presented in \lref{lemma:construction_of_equal_rank_reflection_packing} with separation guarantees, the Haar random detection has the following performance guarantee:
    $$
  \forall x \neq y \in [N], \quad  \Pr_{U \sim \haar(\U(D))}\left[ \left| \Tr\left[ H(\rho(U), O(U))\left( Q_x - Q_y \right) \right]  \right| \geq \frac{\kappa}{2\sqrt{D}}\right] \geq \frac{1}{64}.
    $$
\end{corollary}

\begin{proof}[Proof of \corref{cor:haar_random_reflection_detection}]
    The statement follows directly from \lref{lemma:random_detection} by noting that
    $
    \frac{\left\| Q_x - Q_y\right\|_F}{2\sqrt{D^2 - 1}} \geq \frac{\kappa \sqrt{D}}{2\sqrt{D^2 - 1}} > \frac{\kappa}{2 \sqrt{D}}
    $ by construction. This completes the proof.
\end{proof}

The preceding argument gives a separating statistical guarantee for every pair of distinct reflections. When the queried oracles are fractional reflections [cf. Definition \ref{def:fractional_reflection}], it suffices to relate this statistic to the expectation values obtained via appropriate choices of CSEU requests. 

\begin{fact} \label{fact:exact_signal_extraction}
    Let $R$ be any reflection, still, define $Q = \frac{1}{2}(\openone - R)$, then it holds that
    $$
    \Tr\left[ O R^{\alpha} \rho {R^{\alpha}}^\dagger \right] - \Tr\left[ O {R^{\alpha}}^\dagger \rho R^{\alpha} \right] = 2 \sin(\pi \alpha) \Tr\left[ H(\rho, O) Q \right].
    $$
\end{fact}

\begin{proof}[Proof of Fact \ref{fact:exact_signal_extraction}]
    The proof follows from direct evaluation via expansion [cf. \dref{def:fractional_reflection}], note that $R^{\alpha} = (\openone - Q) + e^{-\I \pi \alpha} Q$, by the cyclicity of trace,
    $$
    \begin{aligned}
    \Tr\left[O\left(R^{\alpha} \rho {R^{\alpha}}^\dagger - {R^{\alpha}}^\dagger \rho R^{\alpha} \right) \right] &= \Tr\left[O\left(R^{\alpha} \rho {R^{\alpha}}^\dagger - {R^{\alpha}}^\dagger \rho R^{\alpha} \right) \right] \\
    &= 2 \I \sin(\pi \alpha) \Tr\left[O\left( (\openone - Q) \rho Q - Q \rho (\openone - Q) \right) \right] \\
    &=2 \I \sin(\pi \alpha) \Tr\left[O [\rho, Q] \right] =  2 \I \sin(\pi \alpha)\Tr\left[ [O, \rho] Q \right] = 2 \sin(\pi \alpha) \Tr\left[ H(\rho, O) Q \right],
    \end{aligned}
    $$
    completing the proof.
\end{proof}

We can now choose a proper fractional power so that the resulting expectation gap is sufficiently detectable by an $\eps$-accurate CSEU protocol, and fix a common finite collection of requests that works simultaneously for the entire reflections family.

\begin{lemma}[Indexing the hard family via CSEU] \label{lemma:CSEU_lower_bound_hard_family}
    There exists universal constants $a, \tau, C > 0$ such that suppose we set a sufficiently large $D := \lfloor \tau \min\{ d, \eps^{-2} \}  \rfloor$, and take the hard reflection list-decoding
    instance
    $
    \{R_{\mathbf{x}}:\mathbf{x}\in[N]^m\}
    $
    from \lref{lemma:hard_reflection_list_decoding_instance}.
    Set $
    \alpha:=a\eps\sqrt D$, $T:=\lceil CD^2\rceil$. Then there exist unitaries $W_1,\ldots,W_T\in\U(D)$ such that the fractional family
    $ \mathcal{F}
    := \{ R_{\mathbf{x}}^\alpha:
    \mathbf{x}\in[N]^m \}
    $
    satisfies each of the following properties:
    \begin{enumerate}
        \item $0 < \alpha \leq 1/2$, $\beta_1 k \leq \log |\mathcal{F}| \leq \beta_2 k$.
        \item For each $t \in [T]$, define the difference $g_t(x) := \Tr\left[ O(W_t) R_{x}^{\alpha} \rho(W_t) {R_{x}^{\alpha}}^\dagger \right] - \Tr\left[ \rho(W_t) R_{x}^{\alpha} O(W_t)  {R_{x}^{\alpha}}^\dagger \right]$, then it holds that
        $$
        \forall x \neq y \in [N], \quad \left| \left\{t \in [T]: |g_t(x) - g_t(y)| \geq 12 \eps \right\} \right| \geq \frac{1}{128} T.
        $$
        \item For every $j \in [m]$, $t \in [T]$, denote $\rho_{j,t}$ and $O_{j,t}$ as the operators that embed the rotated pure states $\rho(W_t)$ and $O(W_t)$ into the $j$-th subspace and set $0$ elsewhere. Then all $\rho_{j, t}$ and $O_{j, t}$ are rank-$1$ and $\Tr\left[ O_{j, t} R_{\mathbf{x}}^{\alpha} \rho_{j, t} {R_{\mathbf{x}}^{\alpha}}^\dagger \right] - \Tr\left[ \rho_{j, t} R_{\mathbf{x}}^{\alpha} O_{j, t} {R_{\mathbf{x}}^{\alpha}}^\dagger \right] = g_t(x_j)$. The fixed classically specified request family has cardinality $M := 2mT = \Theta(k)$.
    \end{enumerate}
\end{lemma}

\begin{proof}[Proof of \lref{lemma:CSEU_lower_bound_hard_family}]
    We take the hard reflection list-decoding instance constructed in \lref{lemma:hard_reflection_list_decoding_instance}, and fix the choices of $\rho$ and $O$ as those in \lref{lemma:random_detection}.
    By \corref{cor:haar_random_reflection_detection}, with the reformulation in Fact \ref{fact:exact_signal_extraction},
    $$
    \Tr\left[ O(U) R_{x}^{\alpha} \rho(U) {R_{x}^{\alpha}}^\dagger \right] - \Tr\left[ \rho(U) R_{x}^{\alpha} O(U)  {R_{x}^{\alpha}}^\dagger \right] = 2 \sin(\pi \alpha) \cdot \Tr\left[ H(\rho(U), O(U)) Q_x \right],
    $$
    then it holds that
    \begin{equation}
    \label{eqn:CSEU_detection_difference_cutoff}
    \begin{aligned}
    \Pr_{U \sim \haar(\U(D)) }&\bigg[ \bigg| \Tr\left[ O(U) R_{x}^{\alpha} \rho(U) {R_{x}^{\alpha}}^\dagger \right] - \Tr\left[ \rho(U) R_{x}^{\alpha} O(U)  {R_{x}^{\alpha}}^\dagger \right]  \\
    &\quad - \left(\Tr\left[ O(U) R_{y}^{\alpha} \rho(U) {R_{y}^{\alpha}}^\dagger \right] - \Tr\left[ \rho(U) R_{y}^{\alpha} O(U)  {R_{y}^{\alpha}}^\dagger \right] \right) \bigg| \geq \frac{\kappa \sin(\pi \alpha)}{\sqrt{D}}  \bigg] \geq \frac{1}{64}.
    \end{aligned}
    \end{equation}
    Then we take $a = 6 / \kappa$ and $\tau := \min\{1/4, 1/4a^2\}$, it holds that $D \leq \tau \eps^{-2}$ implies $\alpha \leq a \sqrt{\tau} \leq 1 / 2$, so that $\sin(\pi \alpha) \geq 2 \alpha$ under this condition. Therefore, $\frac{\kappa \sin(\pi \alpha)}{\sqrt{D}} \geq 2 \kappa \alpha / \sqrt{D} = 2 a \kappa \eps = 12 \eps$. Now we draw $W_1, W_2, \dots, W_T$ Haar randomly, for a fixed pair $x \neq y$, choose a sufficiently large coefficient $C > 4 \cdot 64^2 \beta_2$, applying Hoeffding bound with cutoff fraction $1/2$ and using \eref{eqn:CSEU_detection_difference_cutoff} yields
    $$
    \Pr_{W_1, \dots, W_T} \left[ \left|\left\{ t \in [T]: \left| g_t(x) - g_t(y) \right| \geq 12 \eps  \right\}  \right| < \frac{T}{128}  \right] \leq e^{- \frac{T}{2 \cdot 64^2}}.
    $$
    Note that there are at most $\binom{N}{2} \leq N^2 \leq e^{2 \beta_2 D^2}$ such pairs by \lref{lemma:construction_of_equal_rank_reflection_packing}, and $T \geq CD^2$. For this choice of $T$, taking the union bound over all unordered pairs $x \neq y$ yields a failure probability strictly smaller than $1$. By a probabilistic argument \cite{Alon2008}, there exists a collection $\{W_t\}_{t \in [T]}$ that satisfies the second batch of properties.

    Note that the map $\mathbf{x} \mapsto R_{\mathbf{x}}^{\alpha}$ is injective: the primitive map $\mathbf{x} \mapsto R_{\mathbf{x}}$ is injective by construction, while the map $R \mapsto R^{\alpha}$ is injective by noting that one can reformulate the definition in Definition \ref{def:fractional_reflection} so that
    $$
    R = \openone - \frac{2\left(R^{\alpha} - \openone \right)}{e^{-\I \pi \alpha} - 1}.
    $$
    Therefore, $|\mathcal{F}| = N^m$, and the first batch of properties follows together with the results in \lref{lemma:hard_reflection_list_decoding_instance} and the bound $0 < \alpha \leq 1/2$ established above. The third batch of properties follows directly from the block-diagonal construction. Moreover, $T = \Theta(D^2)$ and $k = mD^2$ again by \lref{lemma:hard_reflection_list_decoding_instance}, so that $M = 2mT = \Theta(mD^2) = \Theta(k)$, concluding the proof.
\end{proof}

We inherit the hard family and all the associated parameter choices throughout the remainder of the proof. It remains to turn these indexed expectation values into an actual list decoder. The redundancy of the request family would be useful here: the decoder only needs a (sufficiently large) constant fraction of the requests to be estimated accurately.

\begin{lemma}[From indexing to list decoding] \label{lemma:concatenating_CSEU_hard_family}
    Suppose the hard reflection list-decoding problem in \dref{def:hard_reflection_list_decoding_problem} and the CSEU request family from \lref{lemma:CSEU_lower_bound_hard_family} are fixed. There exists a universal integer $r \in \N$, such that for any CSEU protocol satisfying the pointwise guarantee stated in Problem \ref{prob:CSEU} and using $K(d, \eps)$ queries to the unknown unitary, it induces a protocol that uses at most $q = r K(d, \eps)$ queries for a sufficiently large constant $r$ to the unitary and outputs a list $\mathsf{L} \subseteq [N]^m$ such that
    $$
    \forall \mathbf{x} \in [N]^m, \quad \Pr[\mathbf{x} \in \mathsf{L}] \geq \frac{2}{3}, \quad |\mathsf{L}|  \leq e^{\frac{1}{4} \beta_1 k},
    $$
    provided that $D$ is sufficiently large.
\end{lemma}

\begin{proof}[Proof of \lref{lemma:concatenating_CSEU_hard_family}]
    We set a constant $\delta > 0$ and choose the number $r$ such that $e^{-r/18} \leq \frac{\delta}{24 \cdot 64} =: \eta$, and repeat the CSEU data acquisition stage independently $r = \bigo{1}$ times and apply the median trick among these $r$ repetitions for each of the $M$ fixed state-observable requests. Hoeffding's inequality yields a probability $\leq e^{-r/18} \leq \eta$ of the event that the median output is not $\eps$-accurate compared to the ground truth. Then we call $(j, t)$ ``good'' if both $\Tr[O_{j, t} R_{\mathbf{x} }^{\alpha}\rho_{j, t} {R_{\mathbf{x}}^{\alpha}}^{\dagger}]$ and $\Tr[\rho_{j, t} R_{\mathbf{x} }^{\alpha} O_{j, t} {R_{\mathbf{x}}^{\alpha}}^{\dagger}]$ are predicted $\eps$-accurate, and ``bad'' otherwise. Then a union bound with Markov inequality yields
    $
    \E\left[ \left| \{(j, t):(j, t) \text{ is bad} \} \right| \right] \leq M \eta = 2 \eta mT$, and $ \Pr\left[ \left| \{(j, t):(j, t) \text{ is bad} \} \right| \geq 6 \eta mT \right] \leq \frac{2 \eta mT}{6 \eta mT} = \frac{1}{3}
    $. Suppose for each $(\rho, O)$ the medians produce $\widehat{F}_{\mathbf{x}}(\rho, O)$ for the prediction of $\Tr[O R_{\mathbf{x}}^{\alpha} \rho {R_{\mathbf{x}}^{\alpha}}^\dagger]$ and define $\widehat{g}_t(x_j) = \widehat{F}_{\mathbf{x}}(\rho_{j, t}, O_{j, t}) - \widehat{F}_{\mathbf{x}}(O_{j, t}, \rho_{j, t})$, then clearly $| \widehat{g}_t(x_j) - {g}_t(x_j) | \leq 2 \eps$. Therefore, for each $y \in [N]$, define the index counting $S_j(y) := \left| \{ t \in [T]: | \widehat{g}_t(x_j) - {g}_t(y) | \leq 2 \eps \} \right|$ and let $\widehat{x}_j := \mathrm{argmax}_{x \in [N] } S_j(x)$. Then in block $j$, let $b_j := | \{t \in [T]: (j, t) \text{ bad} \} |$, then the ground truth has score $\geq T - b_j$. For any violating $y \neq x_j$, using \lref{lemma:CSEU_lower_bound_hard_family} property 2 guarantees that indices unconditionally not counted into $S_j(y)$ reads $\geq  T / 128 - b_j$ since $12 \eps > 2\eps$, therefore $S_j(y) \leq (1  - 1/128) T + b_j$. 
    \par Whenever $b_j < T / 256$, we have $T - b_j > (1 - 1/128) T + b_j$, so that we must have $\widehat{x}_j = x_j$. When $b_j \geq T/256$, on the event that $\left| \{(j, t):(j, t) \text{ is bad} \} \right| \leq 6 \eta mT$, a maximization argument yields that the number of bad blocks $J$ should satisfy
    $$
    J \cdot \frac{T}{256} \leq 6 \eta mT \quad \implies \quad J \leq \frac{6 \cdot 256 \delta m}{24 \cdot 64} = \delta m.
    $$
    Then suppose we output the Hamming ball $\mathsf{L} \subseteq [N]^m$ centered around the estimate $\widehat{\mathbf{x}}$ such that
    $$
    \mathsf{L} = \left\{ \mathbf{y} \in [N]^m: \ham(\mathbf{y}, \widehat{\mathbf{x}}) \leq \lfloor \delta m \rfloor \right\},
    $$
    then by our previous argument, we have $\Pr[\mathbf{x} \in \mathsf{L}] \geq \Pr\left[ \left| \{(j, t):(j, t) \text{ is bad} \} \right| \leq 6 \eta mT \right] \geq \frac{2}{3}$. As $\log N \leq \beta_2 D^2$, by the counting argument over the ball $\mathsf{L}$, since each satisfying $\mathbf{y}$ differs with $\widehat{\mathbf{x}}$ in $j \leq \lfloor \delta m \rfloor$ coordinates and has $N-1$ free choices each, the upper bound in \lref{lemma:entropy_bound} gives
    $$
    |\mathsf{L}| = \sum_{j=0}^{\lfloor \delta m \rfloor}\binom{m}{j} (N-1)^j \leq e^{m h(\delta)} N^{\delta m} \quad \implies \quad \frac{1}{k} \log |\mathsf{L}| \leq \frac{1}{k} \left( m h(\delta) + \delta m \log N \right) \leq \frac{h(\delta)}{D^2} + \delta \beta_2.
    $$
    Then for sufficiently small $\delta$ [cf. \lref{lemma:entropy_bound}], say, $\delta = 1/16$, and sufficiently large $D$ such that $h(\delta)/D^2 \leq \frac{1}{8} \beta_1$, and therefore $\log |\mathsf{L}| \leq k(\frac{1}{8} \beta_1 + \frac{1}{16} \beta_2) = \frac{1}{4} \beta_1 k$, this completes the proof. Hence, a $K(d,\eps)$-query CSEU protocol helps construct, up to a universal constant-factor amplification, a solver for the hard list-decoding instance using queries to the unknown fractional reflection
$R_{\mathbf{x}}^\alpha$.
\end{proof}

Thus, approximate recovery of the indexed statistics confines the hidden label to a sufficiently small Hamming neighborhood, and gives the desired list-decoding formulation.

\begin{remark}
    Note that the list decoding problem in \lref{lemma:concatenating_CSEU_hard_family} only requires that the decoding succeeds in estimating a sufficiently large constant fraction of the $M$ requests we propose, not all of them, so that the fraction of ``bad'' requests is below a certain constant. This prevents a loose bound introduced via the $\log M$ factor that ensures pointwise success probability.
\end{remark}

\subsection{Converting fractional-reflection queries of all types to reflection queries}

\label{sec:fractional_generalized_query_normalization}

The decoder obtained above still queries the fractional reflection $R_{\mathbf{x}}^\alpha$, whereas the hardness statement in Appendix \ref{sec:list_decoding_lower_bound} applies after reducing back to an affine reflection oracle. We therefore first pad the different query variants into a single fractional-reflection oracle, and then convert the fractional-reflection query decoder into a reflection-query decoder, in the same vein as the lower-bound argument of \cite{HKOT23}. Specifically, the resulting success-probability loss is tolerated by our hardness bound.

\begin{lemma}[Padding queries of different types; see, e.g., \cite{Belovs2024tamingquantumtime}]
\label{lemma:padding_reflections_of_different_types}
    Let $\{R_x\}_{x \in \mathcal{X}}$ be a family of $d$-dimensional reflections, for any algorithms making at most $q$ (adaptive) queries to $R_{x}^{\alpha}, {R_{x}^{\alpha}}^\dagger, \mathsf{c}R_{x}^{\alpha}, \mathsf{c}{R_{x}^{\alpha}}^\dagger$, there exists a family of reflections $\{\widehat{R}_{x}\}_{x \in \mathcal{X}}$ such that after necessary purification and padding, the same pre-measurement state produced by the algorithm can be produced by an alternative algorithm with exactly $q$ queries to only $\widehat{R}_{x}^{\alpha}$s. Moreover, every matrix entry of $\mathsf{c}\widehat{R}_{x}$ is affine linear in the matrix entries of $R_x$.
\end{lemma}

\begin{proof}[Proof of \lref{lemma:padding_reflections_of_different_types}]
    Enlarge the query registers to adapt the control bits of the controlled-version queries. Since reflections are Hermitian, it suffices to consider (fractional) queries to $R_{\alpha}$ and $\mathsf{c}R_{x}$. Then we define the enlarged unitary
    $$
    \widehat{R}_{x} := \openone_{2d} \oplus (\openone_2 \otimes R_{x}) \oplus [-\openone_2 \otimes R_{x}] \oplus \mathsf{c} R_x \oplus [-\mathsf{c}R_x],
    $$
    which is clearly a reflection, and the negative operators are to satisfy fractional queries, as for every reflection $R$, $R^{-\alpha} = e^{\I \pi \alpha} (-R)^{\alpha}$. Notably, $(\openone_2 \otimes R_x)^{\alpha} = \openone_2 \otimes R_{x}^{\alpha}$ and $(\mathsf{c} R_x)^{\alpha} = \mathsf{c} R_{x}^{\alpha}$. For each query, an $x$-independent routing unitary moves the selector and query registers into the corresponding block, and adding phase $e^{\I \pi \alpha}$ helps implement the fractional queries to the inverse oracle. This does not change the total number of queries, finishing the proof.
\end{proof}

With this padding convention, it suffices to consider queries to a single fractional-reflection oracle.

\begin{lemma}[{\cite[Lemma 4.5]{HKOT23}}] \label{lemma:HKOT_fraction_query_conversion}
    Let $R$ be a reflection, and $0 \leq \alpha \leq 1$, set $Z = R^{\alpha}$. Suppose a quantum circuit $\mathcal{C} \in \U(d)$ making $q$ queries to either $Z$ or $Z^\dagger$, then there exists a circuit $\mathcal{C}_{R} \in \U(2^q d)$ with $q$ extra ancillae that makes at most $50 + 100 \alpha q$ queries to $\mathsf{c}R$ such that for every input state $\ket{\psi}$, it holds that
    $$
    \mathcal{C}_R \ket{0^q} \ket{\psi} = \nu \ket{0^q} \ket{\varphi} + \ket{\Phi^{\perp}}
    $$
    such that $\ket{\varphi}$ is a normalized state, $(\bra{0^q} \otimes \openone) \ket{\Phi^{\perp}} = 0$, and it holds that
    $$
    |\nu| \geq 0.99 e^{-\pi \alpha q / 2}, \quad \left\| \ket{\varphi} - \mathcal{C} \ket{\psi} \right\|_2 \leq e^{-99 \pi \alpha q - 20}.
    $$
\end{lemma}

\begin{corollary}
\label{cor:list_output_probability}
    Let $\{R_x\}_{x \in \mathcal{X}}$ be a family of $d$-dimensional reflections, for any algorithms making at most $q$ (adaptive) queries to $R_{x}^{\alpha}, {R_{x}^{\alpha}}^\dagger, \mathsf{c}R_{x}^{\alpha}, \mathsf{c}{R_{x}^{\alpha}}^\dagger$, and always outputs a list $\mathsf{L}$ of size at most $L$, and that $\Pr[x \in \mathsf{L}] \geq 2/3$ for all $x \in \mathcal{X}$, then there exists a family of reflections $\{ \widehat{R}_x \}_{x \in \mathcal{X}}$ and a fixed algorithm making at most $Q \leq 50 + 100 \alpha q$ queries to $\mathsf{c} \widehat{R}_x$, outputting a list $\mathsf{L}$ of size at most $L$, such that $\Pr[x \in \mathsf{L}] \geq \frac{1}{2} e^{-\pi \alpha q}$.
\end{corollary}

\begin{proof}[Proof of \corref{cor:list_output_probability}]
    The conversion presented in \lref{lemma:padding_reflections_of_different_types} directly gives us a quantum circuit $\mathcal{C}$ making $q$ queries to the extended reflection $\{\widehat{R}_x^{\alpha}\}_{x \in \mathcal{X}}$ and produces the same pre-measurement quantum state $\ket{\psi}$. Then we apply \lref{lemma:HKOT_fraction_query_conversion} to convert it into a circuit that uses at most $Q \leq 50 + 100 \alpha q$ queries to $\mathsf{c}\widehat{R}_x$, such that the postselected state $\ket{\varphi}$ satisfies $\|\ket{\varphi} - \ket{\psi} \|_2 \leq e^{-99 \pi \alpha q - 20}$ with probability $\geq 0.99 e^{-\pi \alpha q/2}$, and we simply output $\mathsf{L} = \varnothing$ on the failure of postselection. Finally, let the POVM $\{E_x\}_x$ predicts whether the output list contains $x$, then by definition $\left|\Tr\left[ E_x(\varphi - \psi)\right] \right| \leq \|\varphi - \psi\|_1 \leq 2 \|\ket{\varphi} - \ket{\psi}\|_2 \leq 2 e^{-99 \pi \alpha q - 20}$. Then the list-success probability of the converted circuit should satisfy
    $$
    \Pr\left[ x \in \mathsf{L} \right] \geq \left( \frac{2}{3} - 2e^{-99 \pi \alpha q - 20} \right) \cdot 0.99^2 e^{-\pi \alpha q} \geq \frac{1}{2} e^{-\pi \alpha q},
    $$
    finishing the proof.
\end{proof}

In short, a list decoder making $q$ queries to $R_{\mathbf{x}}^\alpha$ induces a decoder for the same hidden label making $Q = \mathcal{O}(\alpha q)$ queries to an affine reflection oracle, while properly retaining list-membership probability $e^{-\mathcal{O}(\alpha q)}$. This completes the last component needed for the proof of the CSEU lower bound.

\subsection{Putting everything together}
We are finally ready to prove \tref{thm:cseu_lower_bound}. Suppose the minimal dependence of the CSEU protocol that satisfies the pointwise requirement stated in Problem \ref{prob:CSEU} is $K(d,\eps)$. We inherit all the choices of the parameters and the hard reflection list-decoding instance from \lref{lemma:CSEU_lower_bound_hard_family}, and take constants $d_0,\eps_0$ appropriately so that all the thresholds in the preceding arguments are satisfied. Applying \lref{lemma:concatenating_CSEU_hard_family} yields a list-decoding algorithm that uses $q=rK(d,\eps)$ queries to the unknown reflection drawn from the hard family for the universal integer $r$, such that for every hard label $\mathbf{x}\in[N]^m$, $\Pr[\mathbf{x}\in\mathsf{L}]\geq 2/3$, and we have $|\mathsf{L}|\leq L \leq e^{\beta_1k/4}$.

Applying \corref{cor:list_output_probability} to this list decoder further gives an algorithm making at most $Q\leq 50+100\alpha q $
queries to the concatenated affine reflection oracle $\mathsf{c}\widehat{R}_{\mathbf{x}}$, while its average list-membership probability satisfies
$$
\frac{1}{N^m}\sum_{\mathbf{x}\in[N]^m}\Pr[\mathbf{x}\in\mathsf{L}]\geq \frac{1}{2}e^{-\pi\alpha q}.
$$
On the other hand, by \lref{lemma:hard_reflection_list_decoding_instance}, the hard label space satisfies $N^m\geq e^{\beta_1k}$, and the converted oracle is affine linear in the same $k=mD^2$ real parameters. Therefore, applying the lower bound in \lref{lemma:affine_list_decoding_query_lower_bound} with $s:=\alpha q$ and $\beta:=\beta_1$ yields $\alpha q > \zeta k$ for some universal constant $\zeta>0$.

It remains to plug in all the parameters. Recall that by our construction in \lref{lemma:CSEU_lower_bound_hard_family} that $\alpha=a\eps\sqrt{D}$, while \lref{lemma:hard_reflection_list_decoding_instance} gives $k\geq dD/2$. Combining these relations with $\alpha q>\zeta k$ yields
$$
K(d,\eps)>\frac{\zeta}{ar}\frac{k}{\eps\sqrt{D}}>\frac{\zeta}{2ar}\frac{d\sqrt{D}}{\eps}.
$$
For sufficiently large $d_0$ and sufficiently small $\eps_0$, $\tau\min\{d,\eps^{-2}\}\geq 2$, and by $\lfloor x \rfloor \geq \frac{x}{2}$ for $x \geq 2$, it holds that
$$
D=\left\lfloor\tau\min\{d,\eps^{-2}\}\right\rfloor\geq \frac{\tau}{2}\min\left\{d,\frac{1}{\eps^2}\right\}.
$$
Combining these inequalities, we finally obtain
$$
K(d,\eps)>\frac{\zeta}{2ar}\sqrt{\frac{\tau}{2}}\cdot\frac{d}{\eps}\sqrt{\min\left\{d,\frac{1}{\eps^2}\right\}}=c\cdot\min\left\{\frac{d^{3/2}}{\eps},\frac{d}{\eps^2}\right\},
$$
where $c := \frac{\zeta}{2ar} \sqrt{\frac{\tau}{2}} = \frac{\zeta}{4608 \sqrt{2} r} = \frac{\Phi^{-1}(1/512)}{101 \cdot 4608 \sqrt{2} \cdot \lceil 18 \log(24576)\rceil} > 0$ is a universal constant by our choice of auxiliary parameters. This concludes the lower bound scaling claimed in \tref{thm:cseu_lower_bound}.

\section{Further details on Hamiltonian learning}\label{app:Hamiltonian_learning}

In this section, we provide the full technical details underlying our results for learning Hamiltonians stated in the main text. We first present the explicit learning protocol and prove the corresponding upper bounds, then establish corresponding lower bounds, and finally compare our performance guarantees with existing protocols.

\subsection{Our Hamiltonian learning protocol and its performance}
\label{appendix:ham_learning}

We first give an explicit construction of our Hamiltonian learning protocol and prove the upper bounds for the total evolution time stated in Theorem~\ref{thm:hamiltonian_learning_informal} of the main text.
The construction follows the polynomial-interpolation approach of \cite{caro2022learning,stilck2024efficient,gu2024practical}, combined with our CSEU protocol for simultaneously estimating the required linear properties of short-time unitary channels.

Let $H=\sum_{\mathbf p\in\{0,1,2,3\}^n}\mu_H(\mathbf p)\sigma_{\mathbf p}$ be the Pauli expansion of the unknown traceless $n$-qubit Hamiltonian, where $d=2^n$, $\sigma_{\mathbf 0}=\openone$, and $\mu_H(\mathbf 0)=0$. Since the Pauli operators satisfy the orthogonality $\Tr(\sigma_{\mathbf p}\sigma_{\mathbf p'}) = d\,\delta_{\mathbf p,\mathbf p'}$, we have $\mu_H(\mathbf p)=\Tr(H\sigma_{\mathbf p})/d$.

We now explain how to convert the estimation of a Pauli coefficient $\mu_H(\mathbf p)$ into the estimation of a short-time dynamical expectation value. Fix any nontrivial Pauli label $\mathbf p\in\{0,1,2,3\}^n\setminus\{\mathbf 0\}$. Choose an index $k=k(\mathbf p)\in[n]$ such that $p_k\neq 0$. Then choose another Pauli label $\mathbf p'\in\{0,1,2,3\}^n$ such that $p_k'\in\{1,2,3\}\setminus\{p_k\}$ and $p_j'=0$ for all $j\neq k$. By construction, $\sigma_{\mathbf p'}$ is a single-qubit Pauli operator that anticommutes with $\sigma_{\mathbf p}$.

Define the Hermitian operator $\tau_{\mathbf p}:=\frac{\mathrm{i}}{2}[\sigma_{\mathbf p},\sigma_{\mathbf p'}]$, which is again a Pauli operator, and we set up the request
\begin{equation}
\label{eq:rho_O_def}
\rho_{\mathbf p}:=\frac{1}{d}(\openone+\tau_{\mathbf p}), \qquad O_{\mathbf p}:=\sigma_{\mathbf p'}.
\end{equation}
Because $\tau_{\mathbf p}$ is Hermitian and satisfies $\tau_{\mathbf p}^2=\openone$, the operator $\rho_{\mathbf p}$ is a valid quantum state.

The crucial property of this construction is that the first derivative of the short-time expectation value of $O_{\mathbf p}$ in the evolved state $U_t \rho_{\mathbf p} U_t^\dag$ encodes the target coefficient $\mu_H(\mathbf p)$. 
Indeed, define
\[
f_{\mathbf p}(t)
:=
\Tr \bigl(O_{\mathbf p}\,U_t \rho_{\mathbf p} U_t^\dag\bigr)
=
\Tr \bigl(\sigma_{\mathbf p'}\,e^{-\mathrm{i} Ht}\rho_{\mathbf p}e^{\mathrm{i} Ht}\bigr).
\]
Then we have 
\[
\left. \frac{{\rm d} f_{\mathbf p}(t)}{{\rm d} t} \right|_{t=0}
=
\Tr \bigl(\mathrm{i} [H,\sigma_{\mathbf p'}]\rho_{\mathbf p}\bigr)
=
\Tr \bigl(H\,\mathrm{i}[\sigma_{\mathbf p'},\rho_{\mathbf p}]\bigr)
=
\frac{2}{d}\Tr(H\sigma_{\mathbf p})
=
2\mu_H(\bfp), 
\]
where the third equality holds because 
\[
\mathrm{i}[\sigma_{\mathbf p'},\rho_{\mathbf p}]
=
\frac{\mathrm{i}}{d}[\sigma_{\mathbf p'},\tau_{\mathbf p}]
=
\frac{2}{d}\sigma_{\mathbf p}. 
\]
Thus, recovering the first derivative of $f_{\mathbf p}(t)$ at $t=0$ is equivalent to recovering $\mu_H(\mathbf p)$.
Fortunately, this derivative can be extracted by performing polynomial interpolation
on $f_{\mathbf p}(t)$ \cite{caro2022learning,stilck2024efficient,gu2024practical}. Building on this technique, our Hamiltonian learning protocol proceeds as follows: 

\paragraph{Step 1: Choose target accuracy and the interpolation times.}
Following the protocol in \cite{caro2022learning}, to learn the Pauli coefficients each within additive error $\eps \|H\|_\infty$, we choose parameters 
\begin{align}\label{eqn:widetilde_varepsilon}
T:= \frac{1}{\|H\|_\infty}, 
\qquad
\widetilde\varepsilon
:=
\frac{3T\varepsilon\|H\|_\infty}{4(L-1)L(2L-1)}, 
\qquad
L:=
\left\lceil
2\log \left(
\frac{8 }{\sqrt{2\pi\ln 2}\,\varepsilon}
\right)
\right\rceil, 
\end{align}
and interpolation time instances
\[
t_j:=\frac{T}{2}(1+z_j),
\qquad
z_j:=-\cos \left(\frac{2j-1}{2L}\pi\right), 
\qquad 
j \in [L]. 
\]

\paragraph{Step 2: Estimate all required short-time expectation values.}
For each $j\in[L]$, we treat
$U_{t_j}=e^{-\mathrm{i} Ht_j}$
as an unknown unitary, and
apply our CSEU protocol to simultaneously estimate, for all $\mathbf p\neq \mathbf 0$, the quantities
$\Tr \bigl(O_{\mathbf p}U_{t_j}\rho_{\mathbf p} U_{t_j}^\dagger\bigr)$
up to additive error $\widetilde\varepsilon$, with overall success probability at least $1-\frac{1}{3L}$.
Denote the resulting estimates by
$2\widehat\alpha_{t_j}^{(\mathbf p)}\approx \Tr \bigl(O_{\mathbf p}U_{t_j}\rho_{\mathbf p} U_{t_j}^\dag\bigr)$.

\paragraph{Step 3: Compute the Chebyshev coefficients.}
For each nontrivial $\mathbf p$, 
define the approximate Chebyshev coefficients by
\[
\widehat b_0^{(\mathbf p)}
=
\frac{1}{L}\sum_{j=1}^L \widehat\alpha_{t_j}^{(\mathbf p)},
\qquad \text{and} \qquad
\widehat b_\ell^{(\mathbf p)}
=
\frac{2}{L}\sum_{j=1}^L \widehat\alpha_{t_j}^{(\mathbf p)}\,T_\ell(z_j)
\quad 
(1\leq  \ell\leq  L-1), 
\]
where $T_\ell$ is the $\ell$-th Chebyshev polynomial.

\paragraph{Step 4: Reconstruct the Pauli coefficients.}
Finally, reconstruct the Pauli coefficients
\begin{align}\label{eq:hatmu}
\widehat\mu_H(\mathbf 0)=0, 
\qquad 
\widehat\mu_H(\mathbf p)
:=
-\frac{2}{T}
\sum_{\ell=0}^{L-1}
(-1)^\ell \ell^2 \, \widehat b_\ell^{(\mathbf p)},
\quad \forall 
\mathbf p\in\{0,1,2,3\}^n\setminus\{\mathbf 0\}.
\end{align}
and the Hamiltonian 
\begin{align}\label{eq:hatH}
\widehat H
=
\sum_{\mathbf p\in\{0,1,2,3\}^n\setminus\{\mathbf 0\}}
\widehat\mu_H(\mathbf p)\sigma_{\mathbf p}. 
\end{align}

The performance of our Hamiltonian learning protocol
is stated in the following Proposition, which confirms the first statement of
\tref{thm:hamiltonian_learning_informal} in the main text.

\begin{proposition}
\label{prop:hamiltonian_learning_formal}
For every $0 < \eps < 1$, our Hamiltonian learning protocol described above
uses
\begin{align}\label{eq:Hlearning_Queries}
K_{\mathrm{Ham}} = \begin{cases}
\widetilde{\bigO}\left(\dfrac{d}{\eps}\right),
&
0<\eps\leq \dfrac{1}{d^{3/2}},
\\[2ex]
\widetilde{\bigO}\left( \dfrac{d^{7/4}}{\sqrt{\eps}}\right),
&
\dfrac{1}{d^{3/2}}<\eps\leq \dfrac{1}{d^{3/4}},
\\[2ex]
\widetilde{\bigO}\left( \dfrac{d }{\eps^{3/2}}\right),
&
\dfrac{1}{d^{3/4}} <\eps<1,
\end{cases}
\end{align}
parallel queries to the real time-evolution unitaries $e^{-\mathrm{i}Ht}$,
where each query evolves for time $\bigo{\|H\|_\infty^{-1}}$.
Its total evolution time is thus given by
\begin{align} 
T_{\mathrm{tot}}=
\begin{cases}
\widetilde{\bigO}\left(\dfrac{d}{\|H\|_\infty\eps}\right),
&
0< \eps \leq \dfrac{1}{d^{3/2}},
\\[2ex]
\widetilde{\bigO}\left(\dfrac{d^{7/4}}{\|H\|_\infty\sqrt{\eps}}\right),
&
\dfrac{1}{d^{3/2}}<\eps \leq \dfrac{1}{d^{3/4}},
\\[2ex]
\widetilde{\bigO}\left(\dfrac{d }{\|H\|_\infty\eps^{3/2}}\right),
&
\dfrac{1}{d^{3/4}}<\eps  < 1.
\end{cases}
\end{align}
In addition, with probability at least $2/3$, the reconstructed Pauli coefficients
$\widehat\mu_H(\mathbf p)$ in \eref{eq:hatmu} satisfy
\begin{equation}
\label{eq:pauli_coe_eps0}
\left|
\widehat\mu_H(\mathbf p)-\mu_H(\mathbf p)
\right|
\leq \varepsilon \|H\|_\infty
\quad
\forall \, 
\mathbf p\in\{0,1,2,3\}^n.
\end{equation}
\end{proposition}

\begin{lemma}[{\cite[Appendix D]{caro2022learning}}] 
\label{lem:cheb_interp}
For any fixed nontrivial $\mathbf p\in\{0,1,2,3\}^n\setminus\{\mathbf 0\}$, if
\begin{equation}
\label{eq:interp_assumption}
\left|
2\widehat\alpha_{t_j}^{(\mathbf p)}
-
\Tr \bigl(O_{\mathbf p}U_{t_j}\rho_{\mathbf p} U_{t_j}^\dagger\bigr)
\right|
\leq
\widetilde\varepsilon
\qquad
\forall \,  j\in[L],
\end{equation}
where $\widetilde\varepsilon$ is defined in \eref{eqn:widetilde_varepsilon}, then $\widehat\mu_H(\mathbf p)$ given in \eref{eq:hatmu} satisfies
$|\widehat\mu_H(\mathbf p)-\mu_H(\mathbf p)|\leq \varepsilon\|H\|_\infty$.
\end{lemma}

\begin{proof}[Proof of Proposition~\ref{prop:hamiltonian_learning_formal}]
For each $j\in[L]$, Step 2 of our protocol applies CSEU to the unknown
unitary $U_{t_j}=e^{-\mathrm{i}Ht_j}$ and simultaneously estimates,
for all $\mathbf p\neq\mathbf 0$, the quantities
$\Tr (
O_{\mathbf p}U_{t_j}\rho_{\mathbf p}U_{t_j}^\dagger
)$
up to additive error $\widetilde\varepsilon$, with overall success
probability at least $1-\delta$, where $\delta=1/(3L)$.
Thus, by the union bound, with probability at least
$1-\sum_{j=1}^L\delta=2/3$,
the produced estimates satisfy
\begin{equation}
\label{eq:all_good_event}
\left|
2\widehat\alpha_{t_j}^{(\mathbf p)}
-
\Tr \left(
O_{\mathbf p}U_{t_j}\rho_{\mathbf p}U_{t_j}^\dagger
\right)
\right|
\leq
\widetilde\varepsilon
\qquad
\forall \, 
\mathbf p\neq\mathbf 0,\ j\in[L].
\end{equation}
Condition on this event.
Then for every $\mathbf p\neq\mathbf 0$,
Lemma~\ref{lem:cheb_interp} implies
$\left|
\widehat\mu_H(\mathbf p)-\mu_H(\mathbf p)
\right|
\leq\varepsilon \|H\|_\infty$.
For $\mathbf p=\mathbf 0$, we have
$\widehat\mu_H(\mathbf 0)=\mu_H(\mathbf 0)=0$
by construction and the tracelessness of $H$.
Therefore, \eref{eq:pauli_coe_eps0} holds for all
$\mathbf p\in\{0,1,2,3\}^n$. It remains to bound the query complexity and the total evolution time.
For each fixed $j\in[L]$, we need to simultaneously estimate the
collection of $d^2-1$ quantities
$
\{
\Tr (
O_{\mathbf p}U_{t_j}\rho_{\mathbf p}U_{t_j}^\dagger):
\mathbf p\in
\{0,1,2,3\}^n\setminus\{\mathbf 0\}
\}
$
up to additive error $\widetilde\varepsilon$, with overall success
probability at least $1-\delta$.

By \eref{eq:rho_O_def}, for every nontrivial $\mathbf p$,
$\Tr(\rho_{\mathbf p}^2)=2/d$ and
$\Tr(O_{\mathbf p}^2)=d$. 
Applying \tref{thm:cseu_upper_bound} with accuracy $\widetilde\varepsilon$, and using the performance guarantee in Remark~\ref{remark:Mproperties} to estimate all the aforementioned $d^2-1$ quantities
simultaneously, we obtain the query complexity for a single choice of evolution time:
\[
N_j
=
\bigO\left(
K_{\rm CSEU}(d,\widetilde\varepsilon) \cdot \log\left(\frac{d^2-1}{\delta}\right)
\right),
\]
where
\[
K_{\rm CSEU}(d,\widetilde\varepsilon) =
\begin{cases}
\bigO\left(\dfrac{d}{\widetilde\varepsilon}\right),
&
0<\widetilde\varepsilon\leq \dfrac{1}{d^{3/2}},
\\[2ex]
\bigO\left(\dfrac{d^{7/4}}{\sqrt{\widetilde\varepsilon}}\right),
&
\dfrac{1}{d^{3/2}}<\widetilde\varepsilon\leq \dfrac{1}{d^{3/4}},
\\[2ex]
\bigO\left(\dfrac{d}{\widetilde\varepsilon^{3/2}}\right),
&
\dfrac{1}{d^{3/4}} <\widetilde\varepsilon.
\end{cases}
\]
Recall that in \eref{eqn:widetilde_varepsilon} we have chosen
\[
\widetilde\varepsilon
=
\Theta(T\varepsilon \|H\|_\infty L^{-3})
=
\Theta\left(
\frac{\varepsilon}{L^3}
\right),
\qquad
L=
\Theta\left(
\log\left(\eps^{-1}\right)
\right),
\]
and $\delta=\Theta(L^{-1})$.
Therefore, suppressing polylogarithmic factors in
$d$ and $\varepsilon^{-1}$, for each $j\in[L]$, we have
\begin{align}
N_j=\begin{cases}
\widetilde{\bigO}\left(\dfrac{d}{\eps}\right),
&
0<\eps\leq \dfrac{1}{d^{3/2}},
\\[2ex]
\widetilde{\bigO}\left( \dfrac{d^{7/4}}{\sqrt{\eps}}\right),
&
\dfrac{1}{d^{3/2}}<\eps\leq \dfrac{1}{d^{3/4}},
\\[2ex]
\widetilde{\bigO}\left( \dfrac{d }{\eps^{3/2}}\right),
&
\dfrac{1}{d^{3/4}} <\eps<1.
\end{cases}
\end{align}
Since the number of interpolation times is
$L=\Theta(\log(\varepsilon^{-1}))$,
summing over all $j\in[L]$ only introduces an additional
polylogarithmic factor. Hence the total number of queries $\sum_{j=1}^L N_j$ scales as that in \eref{eq:Hlearning_Queries}. Finally, each query to $e^{-\mathrm{i}Ht_j}$ uses evolution time
$t_j\leq T=\|H\|_\infty^{-1}$, sthe total evolution time is therefore bounded from above by
\[T_{\mathrm{tot}}=
\sum_{j=1}^L t_jN_j
\leq
\frac{1}{\|H\|_\infty} \sum_{j=1}^L N_j
=
\begin{cases}
\widetilde{\bigO}\left(\dfrac{d}{\|H\|_\infty\eps}\right),
&
0< \eps \leq \dfrac{1}{d^{3/2}},
\\[2ex]
\widetilde{\bigO}\left(\dfrac{d^{7/4}}{\|H\|_\infty\sqrt{\eps}}\right),
&
\dfrac{1}{d^{3/2}}<\eps \leq \dfrac{1}{d^{3/4}},
\\[2ex]
\widetilde{\bigO}\left(\dfrac{d }{\|H\|_\infty\eps^{3/2}}\right),
&
\dfrac{1}{d^{3/4}}<\eps  < 1.
\end{cases}
\]
This completes the proof.
\end{proof}

\subsection{Proof of the lower bounds for learning general Hamiltonians}
\label{appendix:ham_learning_lower_bound}

In this section, we prove the two lower bounds stated in \tref{thm:hamiltonian_learning_informal} via an information-theoretic approach. For notational and technical convenience, we first introduce the necessary information-theoretic notions and establish several auxiliary lemmas.

\begin{definition}[\cite{Nielsen2012}]
For two random variables $X,Y$ supported on $\mathcal X,\mathcal Y$, respectively, with joint distribution $p(x,y)=\Pr[X=x,Y=y]$, their joint entropy is defined as
\[
    S(X,Y)=-\sum_{x\in\mathcal X,y\in\mathcal Y}p(x,y)\log p(x,y).
\]
The marginal entropies are given by
\[
    S(X)=-\sum_{x\in\mathcal X}p(x)\log p(x),
    \qquad
    S(Y)=-\sum_{y\in\mathcal Y}p(y)\log p(y),
\]
where $p(x)=\sum_{y\in\mathcal Y}p(x,y)$ and $p(y)=\sum_{x\in\mathcal X}p(x,y)$. 
The mutual information of $X$ and $Y$ is
\[
    I(X:Y)=S(X)+S(Y)-S(X,Y).
\]
\end{definition}

\begin{definition}[\cite{Nielsen2012, Watrous_2018}]
\label{def:quantum_information_quantities}
    For a quantum state $\sigma_A$ on a
finite-dimensional system $\H_A$, we write
\[
    S(A)_\sigma
    :=
    - \tr(\sigma_A\log\sigma_A)
\]
for its von Neumann entropy. If $\sigma_{AB}$ is a bipartite state on $\H_A \otimes \H_B$, we define the marginal state entropy
\[
    S(A)_\sigma
    :=
    S(\sigma_A),
    \qquad
    \sigma_A:= \tr_B(\sigma_{AB}).
\]
For a bipartite state $\sigma_{AB}$, define the quantum mutual information between $\H_A$ and $\H_B$ via state $\sigma$
\[
    I(A:B)_\sigma
    :=
    S(A)_\sigma+S(B)_\sigma-S(AB)_\sigma.
\]
For a classical probability distribution $\{p_x\}_{x \in \X}$, the Holevo information of the ensemble $\{ (p_x, \rho_x) \}_{x \in \X}$ is the mutual information on the classical-quantum state $\sigma(\{ (p_x, \rho_x) \}_{x \in \X}) = \sum_{x \in \X} p_x \ket{x} \bra{x}_{A} \otimes \rho_x^B$, denoted as
$$
\chi\left(\{ (p_x, \rho_x) \}_{x \in \X}\right) = I(A:B)_{\sigma(\{ (p_x, \rho_x) \}_{x \in \X})} = S\left( \sum_{x \in \X} p_x \rho_x \right) - \sum_{x \in \X} p_x S(\rho_x).
$$
The Holevo information obeys the data processing inequality: Under the mapping of any quantum channel $\mathcal{E}$, the Holevo information of an ensemble is non-increasing: $\chi\left(\{ (p_x, \mathcal{E}(\rho_x)) \}_{x \in \X}\right) \leq \chi\left(\{ (p_x, \rho_x) \}_{x \in \X}\right)$. 
\end{definition}
Note that these information-theoretic quantities have several useful properties.

\begin{lemma}[Holevo's bound, {\cite[Theorem 12.1]{Nielsen2012}}]
\label{lemma:holevo_bound}
Assume that two local parties, say, Alice and Bob, reside on $\H_A$ and $\H_B$. Alice prepares the state $\rho_x$ with probability $p_x$, sends it to Bob, and Bob retrieves information by applying the measurements $\{ \Pi_y \}_{y \in \mathcal{Y}}$. Let $X$ be the r.v. that represents Alice's state preparation, and $Y$ be the r.v. that represents the measurement outcome subject to $\Pr[Y = y | X = x] = \Tr[\Pi_y \rho_x]$ for any $y \in \mathcal{Y}$. Then the mutual information of $X$ and $Y$ is bounded by
$$
I(X:Y) \leq \chi\left( \left\{p_x, \rho_x \right\}_{x \in \X} \right).
$$
\end{lemma}

\begin{lemma}[Fano's inequality, {\cite[Box 12.2]{Nielsen2012}}]
\label{lemma:fano}
    Suppose we are inferring the value of an r.v. $X$ taking values in $\X$ based on an observation of another r.v. $Y$ via $\widehat X = f(Y)$ for some measurable $f$. Let $p_e = \Pr[X \neq \widehat X]$ be the probability of getting an erroneous estimate; then it holds that
    $$
    - p_e \log(p_e) - (1-p_e) \log(1-p_e) + p_e \log\left(\left|\X\right| -1 \right) \geq S(X) - I(X:Y).
    $$
\end{lemma}

Our proof builds upon the following small incremental entangling theorem assisted by ancillae.

\begin{lemma}[Small incremental entangling, \cite{Bravyi2007, VanAcoleyenMarienVerstraete2013, MarienAudenaertVanAcoleyenVerstraete2016}]
\label{lem:SIE}
Let $\H_A,\H_a,\H_B,\H_b$ be finite-dimensional Hilbert spaces, and let
$K$ be a Hermitian operator supported on $\H_a\otimes \H_b$. For an arbitrary
pure state $|\Psi\rangle_{AaBb} \in \H_A \otimes \H_a \otimes \H_B \otimes \H_b$, define the evolved state
\[
    |\Psi(t)\rangle
    :=
    e^{-\mathrm{i}tK_{ab}}|\Psi\rangle .
\]
Then there exists a constant $C_{\rm SIE}>0$ such that for
all $t\geq 0$,
\[
    \left|
    S(Aa)_{\Psi(t)}-S(Aa)_{\Psi(0)}
    \right|
    \leq 
    C_{\rm SIE}t\|K\|_\infty
    \log \bigl(\min\{\dim \H_a,\dim \H_b\}\bigr).
\]
\end{lemma}

\begin{lemma}[Small incremental Holevo information]
\label{lem:small-incremental-holevo}
Let $\H_X$ be a Hilbert space (classical register) with an orthonormal basis
$\{|x\rangle\}_{x}$. Let $\H_Q$ be a $d$-dimensional
quantum system, and let $\H_A$ be an arbitrary auxiliary quantum system.
Let $\{p_x\}_x$ be a probability distribution, let $\rho_x^{QA}$ be
states on $\H_Q \otimes \H_A$, and let $H_x$ be Hermitian operators on $\H_Q$. For
$t\geq 0$, we define
\[
    \rho_x^{QA}(t)
    :=
    e^{-\mathrm{i}tH_x\otimes \openone_A}
    \rho_x^{QA}
    e^{\mathrm{i}tH_x\otimes \openone_A},
\]
the classical-quantum state
\[
    \omega_{XQA}(t)
    :=
    \sum_x p_x |x\rangle\langle x|_X\otimes \rho_x^{QA}(t),
\]
and the Holevo information for the time-evolved ensemble $\{p_x,\rho_x^{QA}(t)\}_x$
\[
    \chi(t):=I(X:QA)_{\omega(t)} = \chi\left(\{p_x,\rho_x^{QA}(t)\}_x\right).
\]
Then it holds that
\[
    |\chi(t)-\chi(0)|
    \leq 
    C_{\rm SIE}t\,\max_x\|H_x\|_\infty\log d .
\]
\end{lemma}

\begin{proof}[Proof of \lref{lem:small-incremental-holevo}]
If $\max_x\|H_x\|_\infty=0$, then every $H_x=0$, hence $\rho_x^{QA}(t)=\rho_x^{QA}$
for all $x$, and the claim is immediate. We therefore assume that $\max_x\|H_x\|_\infty>0$.

Firstly, the evolution is driven by a controlled Hamiltonian
\[
    K_{XQ}
    :=
    \sum_x |x\rangle\langle x|_X\otimes H_x^Q .
\]
The evolved quantum state is therefore
\[
    \omega_{XQA}(t)
    =
    e^{-\mathrm{i}tK_{XQ}}
    \omega_{XQA}(0)
    e^{\mathrm{i}tK_{XQ}}.
\]
Since $K_{XQ}$ is already block diagonal in the register $X$,
\[
    \|K_{XQ}\|_\infty
    =
    \max_x \|H_x\|_\infty. 
\]
We can reformulate the difference of Holevo information $\chi(t) - \chi(0)$ using the invariance of von Neumann entropy under unitary evolution, which yields $S(\rho_x^{QA}(t)) = S(\rho_x^{QA})$ for all $x$:
\begin{equation}
\label{eq:chit_chi0}
\begin{aligned}
    \chi(t) - \chi(0)
    &=
    I(X:QA)_{\omega(t)} - I(X:QA)_{\omega(0)} \\
    &=
    S(QA)_{\omega(t)}-S(QA|X)_{\omega(t)} - \left(S(QA)_{\omega(0)}-S(QA|X)_{\omega(0)}  \right) \\
    &=
    S \left(\sum_x p_x\rho_x^{QA}(t)\right)
    -
    \sum_x p_x S(\rho_x^{QA}(t)) - \left(S \left(\sum_x p_x\rho_x^{QA}\right)
    -
    \sum_x p_x S(\rho_x^{QA})\right) \\
    &= S \left(\sum_x p_x\rho_x^{QA}(t)\right)
    -
    S \left(\sum_x p_x\rho_x^{QA}\right) \\
    &=
    S(QA)_{\omega(t)}-S(QA)_{\omega(0)}.
\end{aligned}
\end{equation}

Now let $|\Psi(0)\rangle_{XQAR}$ be any purification of
$\omega_{XQA}(0)$, where $\H_R$ is a reference system. Define
\[
    |\Psi(t)\rangle_{XQAR}
    :=
    e^{-\mathrm{i}tK_{XQ}}|\Psi(0)\rangle_{XQAR}.
\]
Tracing out $\H_R$ gives $\omega_{XQA}(t)$, and hence the $QA$-marginal
of $|\Psi(t)\rangle$ is exactly $\omega_{QA}(t)$. Therefore,
\[
    S(QA)_{\omega(t)}
    =
    S(QA)_{\Psi(t)}.
\]
Apply Lemma~\ref{lem:SIE} to the space bipartition $QA:XR$, we have
\[
\begin{aligned}
    \left|
    S(QA)_{\omega(t)}-S(QA)_{\omega(0)}
    \right|
    &=
    \left|
    S(QA)_{\Psi(t)}-S(QA)_{\Psi(0)}
    \right| \\
    &\leq 
    C_{\rm SIE}t\,\|K_{XQ}\|_\infty\log \left(\min\{\dim \H_Q,\dim \H_X\}\right)  
    \\
    &\leq 
    C_{\rm SIE}t\,\max_x\|H_x\|_\infty \log d .
\end{aligned}
\]
Combining this with \eref{eq:chit_chi0}, we have
\[
    |\chi(t)-\chi(0)|
    \leq 
    C_{\rm SIE}t\,\max_x\|H_x\|_\infty\log d .
\]
This completes the proof. 
\end{proof}

\begin{lemma}[Packing of random reflections, {\cite[Section 5]{Bluhm_2026}}]
\label{lem:reflection-packing}
There exist constants $d_0\in\mathbb N$ and $c_0>0$
such that the following holds for every even integer $d\geq d_0$.
Let
\[
    O:=\operatorname{diag}\left\{\openone_{d/2},-\openone_{d/2}\right\}.
\]
Then there exist unitaries $U_1,\ldots,U_M\in\mathsf U(d)$ with
\[
    M\geq \exp(c_0d^2),
\]
such that for all $x\neq  y \in [M]$ and $R_x:=U_xOU_x^\dagger$,
\[
    \frac{1}{\sqrt d}\|R_x-R_y\|_F\geq 1.
\]
\end{lemma}

The following theorem establishes the lower bound stated in the second item of \tref{thm:hamiltonian_learning_informal}.

\begin{theorem}[Lower bound for learning Hamiltonians in NFN]
\label{thm:coherent-total-time-lower-bound}
There exists a constant $d_0\in\mathbb N$ 
such that the following holds for every even integer $d\geq d_0$ and
every $0<\eps<1/4$. Consider any coherent Hamiltonian learning protocol that is given access only to the forward real-time evolution unitaries
$U_t:=e^{-\mathrm{i}Ht}$ for tunable times $t\ge0$, where $H$ is an unknown traceless $d$-dimensional Hamiltonian $H$ with bounded operator norm $\|H\|_{\infty}$. If the protocol successfully outputs $\widehat H$ with high probability:
$$
    \Pr \left[
    \frac{1}{\sqrt d}\left\|\widehat H-H\right\|_F\leq  \eps \|H\|_\infty
    \right]\geq \frac{2}{3}
$$
for any such $H$, then it must use total evolution time
\[
    T=\Omega \left(\frac{d^2}{\eps \|H\|_\infty\log d}\right).
\]
\end{theorem}

\begin{proof}[Proof of \tref{thm:coherent-total-time-lower-bound}]

\begin{figure}[tbp!]
    \centering
    \includegraphics[width=\linewidth]{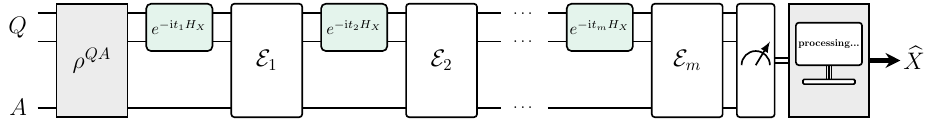}
    \caption{An arbitrary protocol that learns a Hamiltonian through coherent queries to its time evolution.}
    \label{fig:hamiltonian_learning_from_time_evolution}
\end{figure}

We will work with a general coherent query model: the protocol deterministically prepares an initial state $\rho = \rho^{QA}$, applies a finite sequence of time evolutions $\{e^{-\mathrm{i}t_jH}\}_{j=1}^m$ through oracle queries to the unknown Hamiltonian $H$, interspersed with quantum channels $\{\mathcal E_j\}_{j=1}^m$ that are independent of $H$, and finally performs a measurement followed by classical post-processing. 
An illustration is provided in \autoref{fig:hamiltonian_learning_from_time_evolution}. 
We define $T=\sum_{j=1}^m t_j$ as the total duration of all oracle uses. 
If several uses of the time evolution are performed in parallel, their durations are counted separately. 
Such parallel protocols are included in this model by unfolding a parallel layer into sequential oracle uses acting on different query registers; the required routing can be implemented by $H$-independent channels, such as SWAP operations between the query register $Q$ and ancillary register $A$. 

We remark that this generic learning framework automatically includes protocols where the initial state is chosen randomly. 
If the protocol samples an index $\ell \in \mathcal{L}$ with probability $p_\ell$, prepares the initial state $\rho_\ell$ upon obtaining $\ell$, and keeps $\ell$ as part of its classical output, the same procedure can be represented within the above framework as follows. 
We take the fixed initial state to be the quantum-classical state
$$
    \rho^{QAA_0}=\sum_{\ell \in \mathcal{L}} p_\ell  \rho_\ell^{QA} \otimes |\ell\rangle\langle \ell|_{A_0},
$$
where $A_0$ is the additional ancillary register that stores the classical randomness. The register $A_0$ stays unchanged throughout the protocol and is measured at the end. The initial sampling outcome $\ell$ is thus revealed via the measurement outcome.
For an agnostic learner, the randomness of the input state is independent of the unknown Hamiltonian\footnote{The same observation applies to our later proof of the lower bound for learning Pauli transfer matrix in Appendix~\ref{appendix:ptm_learning_lower_bound}.
} and thus contains no information about which Hamiltonian is being queried.
Therefore, without loss of generality, we can incorporate $A_0$ into $A$ and assume that the initial state is fixed.

Let
$O:=\operatorname{diag}\{\openone_{d/2},-\openone_{d/2}\}$.
Note that $ \tr(O)=0$, $O^2=\openone$, and $\|O\|_\infty=1$.
By Lemma~\ref{lem:reflection-packing}, there exist constant
$c_0>0$ and unitaries $U_1,\ldots,U_M\in\mathsf U(d)$, with
$M\geq \exp(c_0d^2)$,
such that 
\[
    \frac{1}{\sqrt d}\|R_x-R_y\|_F\geq 1
\]
for all $x\neq  y$, where $R_x:=U_xOU_x^\dagger$. Each $R_x$ is traceless, subject to $R_x^2=\openone$ and $\|R_x\|_\infty=1$.

Then, we set
$H_x:=4\eps \|H\|_\infty R_x$. 
For every $0<\eps<1/4$, $H_x$ is a valid Hamiltonian in the promised class:
$$
\tr(H_x)=0, \qquad \|H_x\|_\infty = 4\eps \|H\|_\infty\|R_x\|_\infty = 4\eps \|H\|_\infty \leq  \|H\|_\infty.
$$
Moreover, for all $x\neq  y$,
\[
\begin{aligned}
    \frac{1}{\sqrt d}\|H_x-H_y\|_F
    =
    4\eps \|H\|_\infty \cdot \frac{1}{\sqrt d}\|R_x-R_y\|_F  
    \ge
    4\eps \|H\|_\infty.
\end{aligned}
\]
That is, $\{H_x\}_{x=1}^M$ is a $4\eps \|H\|_{\infty}$-separated family in
normalized Frobenius norm.

Let $X$ be uniformly distributed on $[M]$. Conditioned on
$X=x$, the unknown Hamiltonian given to the learning protocol is
$H_x$. Let $\widehat H$ be the output of the protocol. From $\widehat H$,
define the nearest-neighbor decoder
\[
    \widehat X
    :=
    \mathop{\mathrm{argmin}}_{x\in[M]}
    \frac{1}{\sqrt d}\left\|\widehat H-H_x\right\|_F ,
\]
with ties broken arbitrarily.
We claim that whenever
\[
    \frac{1}{\sqrt d}\left\|\widehat H-H_X\right\|_F\leq  \eps \|H\|_\infty,
\]
we must have $\widehat X=X$. Indeed, for any $y\neq  X$, the triangle
inequality gives
\[
\begin{aligned}
    \frac{1}{\sqrt d}\left\|\widehat H-H_y\right\|_F
    &\ge
    \frac{1}{\sqrt d}\left\|H_X-H_y\right\|_F
    -
    \frac{1}{\sqrt d}\left\|\widehat H-H_X\right\|_F 
    \ge
    4\eps \|H\|_\infty-\eps \|H\|_\infty
    =
    3\eps \|H\|_\infty
    >
    \eps \|H\|_\infty.
\end{aligned}
\]
On the other hand, the distance from $\widehat H$ to $H_X$ is at
most $\eps$. Thus $H_X$ is strictly closer to $\widehat H$ than
any other $H_y$, ensuring that $\widehat X=X$.

By the performance guarantee of the learning protocol, the error probability
$p_e=\Pr[\widehat X\neq  X]$ is at most $1/3$. 
Applying \lref{lemma:fano} with $Y=\widehat X$ and $f$ being the identity map, we obtain
\[
\begin{aligned}
I(X:\widehat X)
&\ge
S(X)
+
p_e\log p_e
+
(1-p_e)\log(1-p_e)
-
p_e\log(M-1) \\
&=
\log M
+
p_e\log p_e
+
(1-p_e)\log(1-p_e)
-
p_e\log(M-1) \\
&\ge
\log M-\log 2-\frac13\log(M-1) 
\ge
\frac23\log M-\log 2
=
\Omega(d^2),
\end{aligned}
\]
where we used $M\geq \exp(c_0d^2)$.

It remains to bound $I(X:\widehat X)$ from above in terms of the total
evolution time. We track the Holevo information of the ensemble of
all possible algorithm states during the learning procedure. The learning protocol takes the generic form in \autoref{fig:hamiltonian_learning_from_time_evolution}. The quantum system is given by $\H_Q \otimes \H_A$, where $\H_Q$ is the $d$-dimensional query system on which the
unknown Hamiltonian acts, and $\H_A$ denotes the space of all auxiliary registers and
quantum memory.

Consider the $j$-th oracle-evolution segment, and let its duration be $t_j$. 
For $x\in[M]$, define the evolution channel generated by $H_x$ as
$\mathcal U_{x,t}(\cdot):=e^{-\mathrm{i}tH_x\otimes\openone_A}(\cdot)e^{\mathrm{i}tH_x\otimes\openone_A}$.
We define the state at the beginning of the $j$-th segment recursively via
\[
    \rho_x^{(1)}(0):=\rho,
    \qquad
    \rho_x^{(j)}(0)
    =
    \mathcal E_{j-1}\circ \mathcal U_{x,t_{j-1}}\circ
    \cdots \circ
    \mathcal E_1\circ \mathcal U_{x,t_1}(\rho),
    \quad \forall \, j\geq 2.
\]
During the $j$-th segment, for any $0\leq  s\leq  t_j$, the state evolves as
\[
    \rho_x^{(j)}(s)
    =
    e^{-\mathrm{i}sH_x\otimes \openone_A}
    \rho_x^{(j)}(0)
    e^{\mathrm{i}sH_x\otimes \openone_A}.
\]
We denote the Holevo information of the ensemble at time $s$ in this segment by
\[
    \chi_j(s)
    :=
    \chi\left(
    \left\{
    p_x,\rho_x^{(j)}(s)
    \right\}_{x\in[M]}
    \right),
\]
where $p_x=1/M$ for all $x\in[M]$.

Next, we bound the increase of $\chi_j(s)$ during this segment from above. By Lemma~\ref{lem:small-incremental-holevo}, we have
\[
    |\chi_j(t_j)-\chi_j(0)|
    \leq 
    C_{\rm SIE}t_j\,\max_x\|H_x\|_\infty\log d .
\]
Since
$\max_x\|H_x\|_\infty=4\eps$,
there exists a constant $C'>0$ such that
\begin{equation}
\label{eqn:holevo_difference_upper_bound}
    \chi_j(t_j)-\chi_j(0)
    \leq 
    C'\eps \|H\|_\infty\,t_j\log d .
\end{equation}

Between oracle-evolution segments, the protocol applies channels $\{ \mathcal{E}_j \}_{j=1}^{m-1}$ that
are independent of $x$. By the data-processing inequality [cf. Definition \ref{def:quantum_information_quantities}], such channels cannot increase the Holevo
information.
Before the first oracle use, the protocol's initial state $\rho$ is independent of
$X$, and hence the initial Holevo information $\chi_1(0) = 0$. Therefore, summing over all oracle-evolution segments gives
$$
\begin{aligned}
\sum_{j=1}^{m} \left( \chi_{j}(t_j) - \chi_{j}(0) \right) &= \sum_{j=2}^{m} \left[ \chi\left( \left\{ p_x, \rho_x^{(j)}(t_j) \right\}_{x \in [M] } \right) - \chi\left( \left\{ p_x, \rho_x^{(j)}(0) \right\}_{x \in [M] } \right) \right] + \chi_1(t_1) \\
&=\sum_{j=2}^{m} \left[ \chi\left( \left\{ p_x, \rho_x^{(j)}(t_j) \right\}_{x \in [M] } \right) - \chi\left( \left\{ p_x, \mathcal{E}_{j-1}\left( \rho_x^{(j-1)}(t_{j-1}) \right) \right\}_{x \in [M] } \right) \right] + \chi_1(t_1) \\
&\geq \sum_{j=2}^{m} \left[ \chi\left( \left\{ p_x, \rho_x^{(j)}(t_j) \right\}_{x \in [M] } \right) - \chi\left( \left\{ p_x,  \rho_x^{(j-1)}(t_{j-1})  \right\}_{x \in [M] } \right) \right] + \chi_1(t_1) \\
&= \chi\left( \left\{ p_x, \rho_x^{(m)}(t_m) \right\}_{x \in [M] } \right) = \chi_m(t_m).
\end{aligned}
$$
Combining this with \eref{eqn:holevo_difference_upper_bound}, if we denote the Holevo information of the ensemble
of the final quantum states before the final measurement as $\chi_{\rm final}$, as specified in \autoref{fig:hamiltonian_learning_from_time_evolution}, it holds that
$$
\begin{aligned}
    \chi_{\rm final} = \chi\left( \left\{ p_x, \mathcal{E}_m\left(\rho_x^{(m)}(t_m) \right) \right\}_{x \in [M] } \right) \leq \chi_m(t_m) &\leq \sum_{j=1}^m \left(\chi_j(t_j) - \chi_j(0)\right)
    \\
    &\leq 
    C'\eps \|H\|_\infty\log d \cdot \sum_{j=1}^m t_j
    =
    C'\eps \|H\|_\infty T\log d .
    \end{aligned}
$$

Finally, $\widehat X$ is obtained from the final quantum state by a
measurement followed by classical post-processing. Using \lref{lemma:holevo_bound},
the mutual information between $X$ and any classical variable obtained
by measuring the final state and conducting arbitrary post-processing can not exceed $\chi_{\rm final}$. Therefore, we have
\[
    I(X:\widehat X)
    \leq 
    \chi_{\rm final}
    \leq 
    C'\eps \|H\|_\infty T\log d .
\]
Combining this upper bound with the Fano lower bound
$I(X:\widehat X)\geq \Omega(d^2)$, we obtain
$C'\eps \|H\|_\infty T\log d \geq \Omega(d^2)$, 
or equivalently,
\[
    T\ge
    \Omega \left(\frac{d^2}{\eps \|H\|_\infty \log d}\right).
\]
This completes the proof.
\end{proof}

Finally, the following corollary establishes the lower bound stated in the first item of \tref{thm:hamiltonian_learning_informal}.

\begin{corollary}[Lower bound for estimating Pauli coefficients of Hamiltonians]
\label{cor:pauli-coeff-lower-bound}
Let $d=2^n$. There exist a constant $n_0\in\mathbb N$ such that the following holds for every integer $n\geq n_0$ and
every $0<\eps<\frac{1}{4d}$. Consider any coherent Hamiltonian learning
protocol that is granted access only to the real-time evolution unitaries
$U_t:=e^{-\mathrm{i}Ht}$ $(t\geq 0)$, 
of an unknown traceless $n$-qubit Hamiltonian $H$ with bounded operator norm $\|H\|_{\infty}$. If the protocol outputs estimates
$\{\widehat\mu_H(\mathbf p)\}_{\mathbf p\in\{0,1,2,3\}^n}$ such that
\[
    \Pr  \Big[\forall\,\bfp\in\{0,1,2,3\}^n, \
    |\widehat\mu_H(\bfp)-\mu_H(\bfp)|\leq  \eps \|H\|_\infty
    \Big]\geq  \frac{2}{3}
\]
for any such $H$, then it must use total evolution time
\[
    T=\Omega \left(\frac{d}{\eps \|H\|_\infty\log d}\right).
\]
\end{corollary}

\begin{proof}[Proof of Corollary~\ref{cor:pauli-coeff-lower-bound}]
Let $\varepsilon_{\rm F}$ denote the target normalized error in NFN in \tref{thm:coherent-total-time-lower-bound}. We apply that theorem with $\varepsilon_{\rm F}:=d\eps$. 
Then $0<\eps<\frac{1}{4d}$ implies $0<\varepsilon_{\rm F}=d\eps<1/4$, so the theorem applies.

Suppose that a protocol outputs estimates
$\{\widehat\mu_H(\bfp)\}_{\bfp\in\{0,1,2,3\}^n}$ such that
\[
    |\widehat\mu_H(\bfp)-\mu_H(\bfp)|\leq  \eps \|H\|_\infty,
    \qquad
    \forall\, \bfp \in  \{0,1,2,3\}^n
\]
with probability at least $2/3$. Then we can reconstruct the Hamiltonian via
\[
    \widehat H
    :=
    \sum_{\bfp\in\{0,1,2,3\}^n}
    \widehat\mu_H(\bfp)\sigma_{\bfp}.
\]
Using the Hilbert-Schmidt orthogonality
$ \frac{1}{d} \Tr(\sigma_{\bfp}\sigma_{\bfq})=\delta_{\bfp,\bfq}$, we have
\[
\begin{aligned}
\frac1{\sqrt d}\left\|\widehat H-H\right\|_F=
\sqrt{\frac1d\left\|\widehat H-H\right\|_F^2}
    =
\sqrt{\sum_{\bfp\in\{0,1,2,3\}^n}
    |\widehat\mu_H(\bfp)-\mu_H(\bfp)|^2 }
    \leq 
\sqrt{4^n\eps^2 \|H\|_\infty^2}
    =
    d\eps \|H\|_\infty.
\end{aligned}
\]
Therefore, estimating all Pauli coefficients to normalized error
$\eps$ with probability at least $2/3$ yields a NFN estimate within normalized error $\varepsilon_{\rm F}=d\eps$ and the
same success probability.

By \tref{thm:coherent-total-time-lower-bound}, such a protocol must consume total evolution time
\[
    T = \Omega \left(\frac{d^2}{\varepsilon_{\rm F} \|H\|_\infty\log d} \right) = \Omega \left(\frac{d^2}{d\eps \|H\|_\infty\log d} \right) = \Omega \left(\frac{d}{\eps \|H\|_\infty\log d} \right),
\]
concluding the proof.
\end{proof}

\section{Proof of \tref{thm:unitary_ptm_learning}} \label{appendix:ptm_learning_lower_bound}

In this section, we prove \tref{thm:unitary_ptm_learning} in the main text. 

\begin{proof}[Proof of upper bound in \tref{thm:unitary_ptm_learning}]
We show how to estimate all PTM entries using our CSEU protocol. First, we consider the nontrivial entries with $\bfp, \bfp' \neq  \mathbf 0$. 
For each such $\bfp'$, define the state and observable
\[
    \rho_{\bfp' }:=\frac{\openone+\sigma_{\bfp' }}{d},
    \qquad
    O_{\bfp}:=\sigma_{\bfp}.
\]
Then $\rho_{\bfp' }$ is a valid quantum state, since $\sigma_{\bfp' }$ has eigenvalues $\pm1$, and
\[
    \Tr(\rho_{\bfp' }^2)
    =
    \frac1{d^2}\Tr\bigl((\openone+\sigma_{\bfp' })^2\bigr)
    =
    \frac{2}{d}.
\]
Moreover, $O_{\bfp}$ satisfies $\|O_{\bfp}\|_\infty=1$ and
$\Tr(\sigma_{\bfp}^2)=d$.
Thus these requests belong to the specific CSEU class (see Problem~\ref{prob:restricted_CSEU}) with parameters
$\P=2/d$ and $\B=d$. 
For every $\bfp,\bfp' \neq  \mathbf 0$, we have
\[
    \Tr\left(O_{\bfp}U\rho_{\bfp' }U^\dagger\right)
    =
    \frac1d
    \Tr\left(\sigma_{\bfp}U(\openone+\sigma_{\bfp' })U^\dagger\right)        
    =
    \frac1d
    \Tr\left(\sigma_{\bfp}U\sigma_{\bfp' }U^\dagger\right)
    =
    R_U(\bfp,\bfp' ).
\]
We can now apply \tref{thm:cseu_upper_bound} and distinguish three regimes of the error $\eps$: whenever $0< \eps\leq d^{-3/4}$, the condition of the second statement of \tref{thm:cseu_upper_bound} is satisfied, so we apply it with $\P=\Tr(\rho_{\bfp'}^2)=2/d$. Then $d^{3/2} \sqrt{\P}=\sqrt{2} d$ and $(1-2/d)^{1/4}\leq 1$, the query complexity reads
$$
  K=\mathcal{O}\left(\frac{d}{\eps}+\frac{d^{7/4}}{\sqrt{\eps}} \right).
$$
The first term dominates precisely when $\eps \leq d^{-3/2}$, in which case $K = \mathcal{O}(d/\eps)$. Otherwise, for $d^{-3/2} < \eps \leq d^{-3/4}$, the second term dominates and we have $K = \mathcal{O}(d^{7/4}/\sqrt{\eps})$.

It remains to track the regime $d^{-3/4}<\eps< 1$, where we would use the first statement of \tref{thm:cseu_upper_bound}.
Since $\B\P=2$, the variational bound gives
\[
  K=\mathcal{O}\biggl(\min_{1\leq s\leq d}\biggl\{\frac{1}{\eps^2}
  \biggl(s+\frac{2d^2}{s^3}+2\biggr)+\frac{\sqrt{2}\,d^2}{s\eps}\biggr\}\biggr).
\]
Then we take $s=\lceil d\sqrt{\eps}\,\rceil$, so that $1\leq s\leq d$ and $s=\Theta(d\sqrt{\eps})$, so both $s/\eps^2$ and $d^2/(s\eps)$ are of order $d/\eps^{3/2}$. The constant term is absorbed since $2/\eps^2\leq 2s/\eps^2$, and $d^2/(s^3\eps^2)=\mathcal{O}(d^{-1}\eps^{-7/2})$ also scales as $\mathcal{O}(d/\eps^{3/2})$ because $\eps>d^{-3/4}\geq d^{-1}$. Therefore, the total number of queries is $K = \mathcal{O}(d/\eps^{3/2})$ throughout this regime.

Collecting the above arguments, each nontrivial PTM entry of $\mathcal{U}$ can be estimated to additive error $\eps$ using
$$
K =
\begin{cases}
\mathcal{O}\left(\dfrac{d}{\eps}\right),
&
0<\eps\leq \dfrac{1}{d^{3/2}}, \\[2ex]
\mathcal{O}\left(\dfrac{d^{7/4}}{\sqrt{\eps}}\right),
&
\dfrac{1}{d^{3/2}}<\eps\leq \dfrac{1}{d^{3/4}}, \\[2ex]
\mathcal{O}\left(\dfrac{d}{\eps^{3/2}}\right),
&
\dfrac{1}{d^{3/4}} <\eps<1,
\end{cases}
$$
parallel and non-adaptive queries to $U$. It remains to estimate all entries simultaneously. Note that there are at most $\bigo{d^4}$ PTM entries, so by the guarantee in Remark~\ref{remark:Mproperties}, repeating the CSEU protocol $\mathcal{O}(\log d)$ times suffices, then the total query complexity reads
\[
K_{\mathrm{PTM}} = K \cdot \mathcal{O}(\log d)
=
\begin{cases}
\widetilde{\mathcal{O}}\left(\dfrac{d}{\eps}\right),
&
0<\eps\leq \dfrac{1}{d^{3/2}},
\\[2ex]
\widetilde{\mathcal{O}}\left(\dfrac{d^{7/4}}{\sqrt{\eps}}\right),
&
\dfrac{1}{d^{3/2}} <\eps\leq \dfrac{1}{d^{3/4}},
\\[2ex]
\widetilde{\mathcal{O}}\left(\dfrac{d}{\eps^{3/2}}\right),
&
\dfrac{1}{d^{3/4}} <\eps<1.
\end{cases}
\]
Finally, the remaining entries with $\bfp$ or $\bfp'=\mathbf 0$ are known exactly: since $\Tr(\sigma_{\bfp})=0$ for each $\bfp\neq \mathbf{0}$, we have
\[
R_U(\mathbf 0,\mathbf 0)=1,
\quad
R_U(\bfp,\mathbf 0)=0:
\quad \forall \bfp \neq \mathbf{0},
\quad
R_U(\mathbf 0,\bfp')=0:
\quad \forall \bfp' \neq \mathbf{0}.
\]
This completes the proof of the claimed upper bound.
\end{proof}

We next prove the lower bound. 
The argument is information-theoretic and proceeds by reducing PTM learning to the problem of identifying a hidden element from a large packing of unitary channels. 
We first collect the auxiliary lemmas needed for the proof, including the existence of the packing as the hard instance.

\begin{lemma}[Pauli Parseval identity]
\label{lemma:Pauli_orthogonality}
Let $\{\sigma_{\bfp}\}_{\bfp\in\{0,1,2,3\}^n}$ be the $n$-qubit Pauli operators on a $d=2^n$ dimensional Hilbert space.
Then for any operator $B\in\mathbb C^{d\times d}$,
\[
    \sum_{\bfp\in\{0,1,2,3\}^n}
    \left|
    \frac1d\Tr(\sigma_{\bfp}B)
    \right|^2
    =
    \frac1d\|B\|_{\mathrm F}^2 .
\]
\end{lemma}

\begin{proof}
Since the normalized Pauli operators
$\left\{\sigma_{\bfp}/\sqrt d\right\}_{\bfp\in\{0,1,2,3\}^n}$
form an orthonormal basis of $\mathbb C^{d\times d}$ with respect to the Hilbert-Schmidt inner product, by Parseval's identity,
\[
\sum_{\bfp\in\{0,1,2,3\}^n}
    \left|
    \frac1d\Tr(\sigma_{\bfp}B)
    \right|^2
=
\frac1d
\sum_{\bfp}
    \left|
    \Tr\left(\frac{\sigma_{\bfp}}{\sqrt d}B\right)
    \right|^2        
=
\frac1d
\sum_{\bfp}
    \left|
    \left\langle
    \frac{\sigma_{\bfp}}{\sqrt d},B
    \right\rangle_{\mathrm{HS}}
    \right|^2        
=
\frac1d
\left\langle B,B\right\rangle_{\mathrm{HS}}
=
\frac1d\|B\|_{\mathrm F}^2 .
\]
This proves the lemma. 
\end{proof}

\begin{lemma}[PTM separation for small rotations]
\label{lem:ptm-separation-packed-reflections}
Let $d=2^n$ and 
$O:=\operatorname{diag}(\openone_{d/2},-\openone_{d/2})$.
Suppose $R_x=V_xOV_x^\dagger$, $x\in[M]$, satisfy
$\frac1{\sqrt d}\|R_x-R_y\|_{\mathrm F}\geq 1$
for all $x\neq  y$.
Then there exist universal constants $\lambda_0>0$ and $b_0>0$, such that for every $0<\lambda<\lambda_0$,
if $U_x:=e^{-\mathrm{i}\lambda R_x}$,
then 
\[
    \max_{\bfp,\bfp' \in\{0,1,2,3\}^n}
    \left|
    R_{U_x}(\bfp,\bfp' )-R_{U_y}(\bfp,\bfp' )
    \right|
    \ge
    b_0\frac{\lambda}{d} \quad \forall x\neq  y.
\]
\end{lemma}

\begin{proof}[Proof of \lref{lem:ptm-separation-packed-reflections}]
Fix $x\neq  y$. For any Hermitian $R$ with $\|R\|_\infty\leq 1$, define
$U_R:=e^{-\mathrm{i}\lambda R}$.
We first expand $R_{U_R}(\bfp,\bfp' )$ to the first order in $\lambda$. For $0<\lambda< 1$, Taylor expanding the conjugation map
$X\mapsto e^{-\mathrm{i}\lambda R}Xe^{\mathrm{i}\lambda R}$ gives
\[
    e^{-\mathrm{i}\lambda R}\sigma_{\bfp' }e^{\mathrm{i}\lambda R}
    =
    \sigma_{\bfp' }
    -
    \mathrm{i}\lambda [R,\sigma_{\bfp' }]
    +
    E_{R,\bfp' },
\]
where the remainder second-order term $E_{R,\bfp' }$ satisfies
\[
    \|E_{R,\bfp' }\|_{\mathrm F}
    \leq
    C\lambda^2\|[R,[R,\sigma_{\bfp' }]]\|_{\mathrm F}        
    \leq
    2C\lambda^2\|R\|_\infty\|[R,\sigma_{\bfp' }]\|_{\mathrm F}   
    \leq
    4C\lambda^2\|R\|_\infty^2\|\sigma_{\bfp' }\|_{\mathrm F}     
    \leq
    4C\lambda^2\|\sigma_{\bfp' }\|_{\mathrm F}
    =
    4C\lambda^2\sqrt d .
\]
for some constant $C>0$. 
Here we used $\|R\|_\infty\leq 1$ and the standard inequality
\[
    \|[R,X]\|_{\mathrm F}
    \leq
    \|RX\|_{\mathrm F}+\|XR\|_{\mathrm F}
    \leq
    2\|R\|_\infty\|X\|_{\mathrm F}.
\]
For simplicity, we absorb the coefficient $4$ into the constant $C$ below. 

Then we have 
\[
R_{U_R}(\bfp,\bfp' )
=\frac1d
\Tr \left(\sigma_{\bfp}U_R\sigma_{\bfp' }U_R^\dagger\right)
=
\delta_{\bfp,\bfp' }
-
\frac{\mathrm{i}\lambda}{d}
\Tr \left(\sigma_{\bfp}[R,\sigma_{\bfp' }]\right)
+
\Delta_{R,\bfp,\bfp' },
\]
where
\[
\Delta_{R,\bfp,\bfp' }
    :=
    \frac1d\Tr(\sigma_{\bfp}E_{R,\bfp' }).
\]
Using \lref{lemma:Pauli_orthogonality}, the collection of remainders satisfies
\begin{align}\label{eqn:Delta_bound1}
\sum_{\bfp,\bfp' }|\Delta_{R,\bfp,\bfp' }|^2
=
\sum_{\bfp' }
\sum_{\bfp}\,
\left|
\frac1d\Tr(\sigma_{\bfp}E_{R,\bfp' })
\right|^2        
=
\sum_{\bfp' }
\frac1d\|E_{R,\bfp' }\|_{\mathrm F}^2        
\leq
C^2 d^2\lambda^4 .
\end{align}

Now compare $R=R_x$ and $R=R_y$. Let
$A:=R_x-R_y$.
Since both $R_x$ and $R_y$ are traceless, $A$ is traceless. From the
first-order expansion above,
\[
\begin{aligned}
    R_{U_x}(\bfp,\bfp' )-R_{U_y}(\bfp,\bfp' )
    &=
    -\frac{\mathrm{i}\lambda}{d}
    \Tr \left(\sigma_{\bfp}[A,\sigma_{\bfp' }]\right)
    +
    \left(\Delta_{R_x,\bfp,\bfp' }-\Delta_{R_y,\bfp,\bfp' } \right).
\end{aligned}
\]
Taking the $\ell_2$-norm over all PTM entries and using the reverse triangle inequality, we obtain
\begin{align}\label{eqn:sum_pq_Rx-Ry_expan}
    \left(
    \sum_{\bfp,\bfp' }
    |R_{U_x}(\bfp,\bfp' )-R_{U_y}(\bfp,\bfp' )|^2
    \right)^{1/2}                                      
    &\ge
    \lambda
    \left(
    \sum_{\bfp,\bfp' }\,
    \left|
    \frac1d
    \Tr \left(\sigma_{\bfp}[A,\sigma_{\bfp' }]\right)
    \right|^2
    \right)^{1/2}   
    -
    \left(
    \sum_{\bfp,\bfp' }
    |\Delta_{R_x,\bfp,\bfp' }-\Delta_{R_y,\bfp,\bfp' }|^2
    \right)^{1/2} .
\end{align}
Here, the second term can be bounded using  \eref{eqn:Delta_bound1} and the triangle inequality:
\begin{align}\label{eqn:UB_second_term}
    \left(
    \sum_{\bfp,\bfp' }
    |\Delta_{R_x,\bfp,\bfp' }-\Delta_{R_y,\bfp,\bfp' }|^2
    \right)^{1/2}
    \leq
    \left(
    \sum_{\bfp,\bfp' }
    |\Delta_{R_x,\bfp,\bfp' }|^2
    \right)^{1/2}
    +
    \left(
    \sum_{\bfp,\bfp' }
    |\Delta_{R_y,\bfp,\bfp' }|^2
    \right)^{1/2}        
    \leq
    2Cd\lambda^2 .
\end{align}

We now lower bound the first-order term. 
By \lref{lemma:Pauli_orthogonality} we have
\begin{align}\label{eqn:sum_pq_expan}
    \sum_{\bfp,\bfp' }\,
    \left|
    \frac1d
    \Tr \left(\sigma_{\bfp}[A,\sigma_{\bfp' }]\right)
    \right|^2
    =
    \frac1d
    \sum_{\bfp' }
    \|[A,\sigma_{\bfp' }]\|_{\mathrm F}^2 . 
\end{align}
Note that 
\[
\begin{aligned}
    \sum_{\bfp' }\|[A,\sigma_{\bfp' }]\|_{\mathrm F}^2
    =
    \sum_{\bfp' }
    \Tr \left((A\sigma_{\bfp' }-\sigma_{\bfp' }A)^\dagger
    (A\sigma_{\bfp' }-\sigma_{\bfp' }A)\right)       
    =
    2d^2\Tr(A^2)
    -
    2\sum_{\bfp' }\Tr(A\sigma_{\bfp' }A\sigma_{\bfp' }).
\end{aligned}
\]
Using the Pauli twirling identity
$\sum_{\bfp' }\sigma_{\bfp' }A\sigma_{\bfp' }=d\,\Tr(A)\openone$
and $\Tr(A)=0$, the second term vanishes. Hence
\[
    \sum_{\bfp' }\|[A,\sigma_{\bfp' }]\|_{\mathrm F}^2
    =
    2d^2\|A\|_{\mathrm F}^2 .
\]
Combining this relation with \eref{eqn:sum_pq_expan}, we obtain
\begin{align}\label{eqn:UB_first_term}
    \left(
    \sum_{\bfp,\bfp' }\,
    \left|
    \frac1d
    \Tr \left(\sigma_{\bfp}[A,\sigma_{\bfp' }]\right)
    \right|^2
    \right)^{1/2}
    &=
    \sqrt{2d}\,\|A\|_{\mathrm F} \ge
    \sqrt2\, d, 
\end{align}
where we have used the packing assumption
$\|A\|_{\mathrm F}=\|R_x-R_y\|_{\mathrm F}\geq \sqrt d$.

Equations~\eqref{eqn:sum_pq_Rx-Ry_expan}, \eqref{eqn:UB_second_term}, and \eqref{eqn:UB_first_term} together imply that 
\[
\begin{aligned}
    \left(
    \sum_{\bfp,\bfp' }
    |R_{U_x}(\bfp,\bfp' )-R_{U_y}(\bfp,\bfp' )|^2
    \right)^{1/2}
    \ge
    \sqrt2\,\lambda d
    -
    2Cd\lambda^2 .
\end{aligned}
\]
Then we choose $
    \lambda_0:= \min\left\{1,\frac{\sqrt2}{4C}\right\} $, so that for every $0<\lambda<\lambda_0$,
\[
    \left(
    \sum_{\bfp,\bfp' }
    |R_{U_x}(\bfp,\bfp' )-R_{U_y}(\bfp,\bfp' )|^2
    \right)^{1/2}
    \ge
    b_0\lambda d
\]
for some constant $b_0>0$.

Finally, there are in total $d^4$ PTM entries. Therefore,
\[
    \max_{\bfp,\bfp' }
    |R_{U_x}(\bfp,\bfp' )-R_{U_y}(\bfp,\bfp' )|
    \ge
    \frac{1}{d^2}
    \left(
    \sum_{\bfp,\bfp' }
    |R_{U_x}(\bfp,\bfp' )-R_{U_y}(\bfp,\bfp' )|^2
    \right)^{1/2}       
    \ge
    b_0\frac{\lambda}{d}.
\]
This proves the lemma.
\end{proof}

\begin{proof}[Proof of the lower bound in \tref{thm:unitary_ptm_learning}]
Let $d_0,c_0>0$ be the universal constants from \lref{lem:reflection-packing}. Thus, for every even $d\geq d_0$, there exists a collection of $M\geq \exp(c_0d^2)$ reflections
\[
    R_x=V_xOV_x^\dagger,
    \qquad
    x\in[M],
\]
such that
\[
    \frac1{\sqrt d}\|R_x-R_y\|_{\mathrm F}\geq 1
    \qquad
    \forall\, x\neq  y .
\]

Let $b_0,\lambda_0>0$ be the constants from
\lref{lem:ptm-separation-packed-reflections}. We set
$\lambda:=4d\eps/b_0$, and choose the constant $c>0$ in the theorem small enough so that $4c/b_0<\lambda_0$.
Then the assumption $0<\eps<c/d$ implies $0<\lambda<\lambda_0$.

For each $x\in[M]$, define
$U_x:=e^{-\mathrm{i}\lambda R_x}$ and $\caU_x(\cdot)=U_x(\cdot)U_x^\dag$. 
By Lemma~\ref{lem:ptm-separation-packed-reflections}, for every
$x\neq  y$,
\[
    \max_{\bfp,\bfp' }\,
    \left|
    R_{U_x}(\bfp,\bfp' )-R_{U_y}(\bfp,\bfp' )
    \right|
    \ge
    b_0 \cdot \frac{\lambda}{d}
    =
    4\eps .
\]
Thus the PTM vectors of the candidate unitary channels are
$4\eps$-separated in entrywise $\ell_\infty$ distance.

Let $X$ be uniformly distributed on $[M]$, and suppose that the
unknown channel given to the learner is $\mathcal U_X$. Let
$\widehat R_{\bfp,\bfp' }$ be the estimates output by the learner.
Define the nearest-neighbor decoder
\[
    \widehat X
    :=
    \arg\min_{x\in[M]}
    \max_{\bfp,\bfp' }\,
    \left|
    \widehat R_{\bfp,\bfp' }
    -
    R_{U_x}(\bfp,\bfp' )
    \right|,
\]
with ties broken arbitrarily. On the event that all PTM entries of $\mathcal U_X$ are estimated within additive error $\eps$, 
\[
    \max_{\bfp,\bfp' }\,
    \left|
    \widehat R_{\bfp,\bfp' }
    -
    R_{U_X}(\bfp,\bfp' )
    \right|
    \leq \eps .
\]
For any $y\neq  X$, the $4\eps$-separation gives
\[
\begin{aligned}
    \max_{\bfp,\bfp' }\,
    \left|
    \widehat R_{\bfp,\bfp' }
    -
    R_{U_y}(\bfp,\bfp' )
    \right|
    \ge
    \max_{\bfp,\bfp' }\,
    \left|
    R_{U_X}(\bfp,\bfp' )
    -
    R_{U_y}(\bfp,\bfp' )
    \right|  
    -
    \max_{\bfp,\bfp' }\,
    \left|
    \widehat R_{\bfp,\bfp' }
    -
    R_{U_X}(\bfp,\bfp' )
    \right|  
    \ge
    4\eps-\eps
    >
    \eps .
\end{aligned}
\]
Hence $\widehat X=X$ on the success event. Since the assumed PTM learning guarantee holds for every unitary channel, it holds in
particular for every $\mathcal U_x$, and therefore, the probability of an erroneous guess 
$p_e=\Pr[\widehat X\neq  X]\leq 1/3$.

Applying \lref{lemma:fano} with $Y=\widehat X$ and $f$ being the identity map, we obtain
\[
\begin{aligned}
    I(X:\widehat X)
    &\ge
    S(X)
    +
    p_e\log p_e
    +
    (1-p_e)\log(1-p_e)
    -
    p_e\log(M-1) \\
    &=
    \log M
    +
    p_e\log p_e
    +
    (1-p_e)\log(1-p_e)
    -
    p_e\log(M-1) \\
    &\ge
    \log M-\log 2-\frac13\log(M-1) 
    \ge
    \frac23\log M-\log 2
    =
    \Omega(d^2).
\end{aligned}
\]

It remains to upper bound the information that $K$ queries can reveal
about $X$. We use the standard coherent query model: the learner starts
from an $X$-independent state, applies $X$-independent quantum
channels between oracle uses, and each oracle use applies
$\mathcal U_x$ to a $d$-dimensional query register $Q$. If several
uses of $\mathcal U_x$ are made in parallel, we count them separately;
equivalently, an $r$-fold parallel use $\mathcal U_x^{\otimes r}$ can
be viewed as $r$ consecutive oracle uses acting on different query
registers, since these uses commute. Intermediate measurements and
classical randomness can be deferred coherently, so this model is without
loss of generality for protocols with at most $K$ oracle uses.

One query to $\mathcal U_x$ is the time-one evolution generated by
$\lambda R_x$ on the query register. Consider the $j$-th query. Right before this query, conditioned on $X=x$, let the learner's state be $\rho_x^{(j)}$ on $QA$, where $A$ denotes all auxiliary registers and quantum memory. During the query, define the continuous interpolation
\[
    \rho_x^{(j)}(s)
    =
    e^{-\mathrm{i}s\lambda R_x\otimes I_A}
    \rho_x^{(j)}
    e^{\mathrm{i}s\lambda R_x\otimes I_A}: \quad  0\leq s\le1 .
\]
Let
\[
    \omega_{XQA}^{(j)}(s)
    :=
    \frac1M\sum_{x=1}^M
    |x\rangle\langle x|_X\otimes \rho_x^{(j)}(s),
\]
and define
\[
    \chi_j(s):=I(X:QA)_{\omega^{(j)}(s)} .
\]
By Lemma~\ref{lem:small-incremental-holevo}, applied with
$H_x=\lambda R_x$, we have
\[
    |\chi_j(1)-\chi_j(0)|
    \leq
    C_{\rm SIE}\lambda\log d ,
\]
because $\|R_x\|_\infty=1$ for all $x$. Between oracle calls, the
learner applies channels independent of $x$, which cannot increase the
mutual information by the data-processing inequality. Since the initial
state is independent of $X$, summing over all $K$ queries gives
\[
    \chi_{\rm final}
    \leq
    C_{\rm SIE}K\lambda\log d .
\]

Finally, $\widehat X$ is obtained from the final quantum state by a
measurement followed by classical post-processing. Therefore, by Holevo's
theorem and data processing,
\[
    I(X:\widehat X)
    \leq
    \chi_{\rm final}
    \leq
    C_{\rm SIE}K\lambda\log d .
\]
Combining this with $I(X:\widehat X)\geq c_1d^2$, we get
\[
    c_1d^2
    \leq
    C_{\rm SIE}K\lambda\log d .
\]
Using $\lambda=4d\eps/b_0$, this implies
\[
    K
    \ge
    \frac{c_1b_0}{4C_{\rm SIE}}
    \frac{d}{\eps\log d}
    =
    \Omega \left(\frac{d}{\eps\log d}\right).
\]
This completes the proof.
\end{proof}

\section{Proof of \pref{prop:TomoEfficiency}}\label{appendix:proof_LB}

Our proof of \pref{prop:TomoEfficiency} relies on the existence and properties of covering nets of unit vectors on $\C^d$. These nets provide a finite-size approximation to the set of $d$-dimensional quantum states, serving as fingerprints that allow us to reconstruct the unknown unitary from shadow estimates via a brute-force search over $\U(d)$.

\begin{definition}[\cite{Nielsen2012, Aubrun2017}]
    Let $\mathcal H$ be a $d$-dimensional Hilbert space and $\covering=\{\ket{\psi_j}\}_{j=1}^{S}$ be a set of pure states on $\mathcal H$ such that, for every pure state $\ket{\phi}$ on $\mathcal H$, there exists some $j \in [S]$ satisfying $\frac{1}{2}\|\phi-\psi_j\|_1\leq  \eps$. We call such a set $\covering = \covering(\mathfrak{P}(\mathbb{C}^{d}), \|\cdot\|_1, \eps)$ a pure-state $\eps$-covering net in trace distance on $\mathcal H$. In the latter context, we would omit the distance metric, defaulting to the trace distance.
\end{definition}

\begin{lemma}[{\cite[Theorem 5.11]{Aubrun2017}}]
\label{lemma:epsnet}
For any $0<\eps<1$, there exists a pure-state $\eps$-covering net $\covering$ on a $d$-dimensional Hilbert space satisfying
\[
\left(\frac{1}{\eps}\right)^{c_1 d}
\leq  \left|\covering\right| \leq 
\left(\frac{1}{\eps}\right)^{c_2 d}
\]
for some universal constants $c_1,c_2>0$.
\end{lemma}

We now describe the canonical reduction underlying our proof of \pref{prop:TomoEfficiency}. 
Suppose $\mathcal A$ is a protocol that solves the CSEU task in Problem~\ref{prob:CSEU}. 
Using $\mathcal A$ as a subroutine, we construct a protocol $\mathcal A_{\rm tomo}$ for process tomography of an arbitrary unknown unitary $U\in\U(d)$.

\paragraph{Step 1: Estimate linear properties of $U$.}
Let $\covering =\{\ket{\psi_a}\}_a$ be a pure-state $1/40$-covering net on $\mathbb C^d$ with $|\covering|=\exp(\mathcal{O}(d))$. 
Apply the CSEU protocol $\mathcal A$ to simultaneously estimate, for all pairs of net states $(\psi_a,\psi_b)\in \covering\times \covering$, the quantities
$
\Tr(\psi_a\,U\psi_b U^\dagger)
$
up to additive error $\eps$. We denote the resulting estimates as
$
\widehat t_{a,b}\approx \Tr(\psi_a\,U\psi_b U^\dagger).
$
Note that each $\psi_a$ is an observable in $\mathrm{Obs}(1)$, since
$\|\psi_a\|_\infty=1$ and
$\Tr\left[\left(\psi_a-\openone/d\right)^2\right]\le1$.
\paragraph{Step 2: Reconstruct a consistent unitary.}
Search over $\U(d)$ and find one unitary $\hat U$ satisfying
\begin{equation}
\label{eqn:fit_estimates}
\left|
\widehat t_{a,b}-\Tr\bigl(\psi_a\,\hat U\psi_b \hat U^\dagger\bigr)
\right|
\leq  \eps
\qquad
\forall \psi_a,\psi_b\in \covering.
\end{equation}
Finally, output $\hat U$ as the reconstructed unitary.

\bigskip

We now prove that this construction has the claimed query complexity and reconstruction guarantee stated in \pref{prop:TomoEfficiency}.
We restate the proposition in a more rigorous form:

\begin{proposition}[Formal version of \pref{prop:TomoEfficiency}]
\label{prop:TomoEfficiency_formal}
Suppose $\mathcal A$ is a protocol that solves Problem~\ref{prob:CSEU} with parameter $\B=1$ and accuracy $0<\eps< 1/5$, using $K(d,\eps)$ oracle queries. 
Then, by running $\mathcal A$ as a subroutine, one can construct a unitary channel tomography protocol $\mathcal A_{\rm tomo}$ that uses $\bigo{d}\,K(d,\eps)$ oracle queries in total and outputs a classical description of a unitary $\widehat U$ satisfying
\[
\Pr \left[
\bigl\|\widehat{\mathcal U}-\mathcal U\bigr\|_\diamond \leq  5\eps
\right]
\geq \frac23 .
\]
Moreover, this reduction preserves the query architecture of $\mathcal A$: It repeats $\mathcal A$ independently and then performs classical post-processing on the resulting outputs. 
In particular, if $\mathcal A$ uses only parallel and non-adaptive oracle queries to $U$, then so does $\mathcal A_{\rm tomo}$. 
\end{proposition}

\begin{proof}[Proof of \pref{prop:TomoEfficiency_formal}]
By the choice of $\covering$, we have $|\covering|=\exp(\mathcal{O}(d))$. In particular, the family of quantities
\[
\left\{
\Tr(\psi_a\,U\psi_b U^\dagger)
:
\psi_a,\psi_b\in\covering
\right\}
\]
has size $|\covering|^2=\exp(\bigo{d})$. By Remark~\ref{remark:Mproperties}, to simultaneously estimate this many quantities within additive error $\eps$, it suffices to run $\mathcal A$ independently for
\[
T=\mathcal{O}\bigl(\log |\covering|^2\bigr)=\mathcal{O}(d)
\]
times and then apply the standard median trick in post-processing. Thus Step~1 of $\mathcal A_{\rm tomo}$ uses $K(d,\eps)\cdot \bigo{d}$ queries in total, and with probability at least $2/3$ all estimates are correct up to additive error $\eps$, namely,
\begin{equation}
\label{eq:good_estimation_event}
\left|
\widehat t_{a,b}-\Tr(\psi_a\,U\psi_b U^\dagger)
\right|
\leq  \eps
\qquad
\forall \psi_a,\psi_b\in\covering.
\end{equation}

Under \eref{eq:good_estimation_event}, the true unitary $U$ itself satisfies the consistency condition \eref{eqn:fit_estimates}. Therefore, the exhaustive search in Step~2 succeeds in finding at least one unitary $\hat U$ satisfying \eref{eqn:fit_estimates}.

We claim that any unitary $\hat U$ satisfying \eref{eqn:fit_estimates} must obey
$\|\hat{\mathcal U}-\mathcal U\|_\diamond \leq  5\eps$.
Otherwise, if $\|\hat{\mathcal U}-\mathcal U\|_\diamond > 5\eps$, then \lref{lemma:covering_net_saturation} below with $\eta=5\eps$ yields $\psi_{\rm in},\psi_{\rm out}\in\covering$ such that
\[
\left|
\Tr \left(\psi_{\rm out}\,\hat U\psi_{\rm in}\hat U^\dagger\right)
-
\Tr \left(\psi_{\rm out}\,U\psi_{\rm in}U^\dagger\right)
\right|
>
2\eps.
\]
But setting $\psi_a=\psi_{\rm out}$ and $\psi_b=\psi_{\rm in}$, the triangle inequality together with Eqs.~\eqref{eqn:fit_estimates} and \eqref{eq:good_estimation_event} gives
\[
\left|
\Tr\bigl(\psi_a\,\hat U\psi_b\hat U^\dagger\bigr)
-
\Tr\bigl(\psi_a\,U\psi_bU^\dagger\bigr)
\right|
\leq 
\left|
\Tr\bigl(\psi_a\,\hat U\psi_b\hat U^\dagger\bigr)-\widehat t_{a,b}
\right|
+
\left|
\widehat t_{a,b}-\Tr\bigl(\psi_a\,U\psi_bU^\dagger\bigr)
\right|
\leq  2\eps,
\]
contradicting the previous inequality.
Therefore, every unitary $\hat U$ satisfying \eref{eqn:fit_estimates} must satisfy
$\bigl\|\hat{\mathcal U}-\mathcal U\bigr\|_\diamond \leq  5\eps$.
Since the event \eref{eq:good_estimation_event} occurs with probability at least $2/3$, the same lower bound holds for the success probability of $\mathcal A_{\rm tomo}$. 

Finally, the reduction preserves the oracle-query architecture of
$\mathcal A$. The $T=\mathcal{O}(d)$ independent invocations of $\mathcal A$ can
be executed in parallel, while the median amplification and the
reconstruction of $\widehat U$ involve only classical post-processing.
Hence, if the oracle-query stage of $\mathcal A$ is parallel and
non-adaptive, then the oracle-query stage of
$\mathcal A_{\rm tomo}$ is also parallel and non-adaptive.
This completes the proof. 
\end{proof}

\begin{lemma}
\label{lemma:covering_net_saturation}
Suppose $0<\eta<1$, $U,V\in \U(d)$, $\left\|\mathcal U-\mathcal V \right\|_\diamond>\eta$, and
$\covering=\{\ket{\psi_a}\}_a$ is a pure-state $1/40$-covering net  on $\mathbb{C}^d$ with $|\covering|=\exp(\mathcal{O}(d))$.  
Then there exist two pure states $\psi_{\rm in}, \psi_{\rm out}\in \covering$ such that 
$$
\left| \Tr\left( \psi_{\rm out} U \psi_{\rm in} U^\dagger \right) - \Tr\left( \psi_{\rm out}  V\psi_{\rm in} V^\dagger \right) \right| 
> \frac{2}{5}\eta. 
$$
\end{lemma}

\begin{proof}[Proof of \lref{lemma:covering_net_saturation}]
Define the linear map
$\Upsilon(\cdot) := \mathcal U(\cdot)-\mathcal V(\cdot)$. 
Since $\Upsilon$ is the difference of the actions of two unitary channels, the diamond norm is attained on a pure-state input according to \lref{lemma:saturability_diamond_norm_Krank_1}.
Hence, there exists a pure state $\rho$ that saturates the trace norm of $\Upsilon(\rho)$ and thus the diamond norm of $\Upsilon$, i.e., 
$\|\Upsilon(\rho)\|_1=\|\Upsilon\|_\diamond$. 
Moreover, according to \lref{lemma:Helstrom2}, there exists a rank-one projector $O$
onto the positive eigenspace of $\Upsilon(\rho)$ that satisfies
\[
2\left|\Tr \bigl(O\,\Upsilon(\rho)\bigr)\right|
=
\|\Upsilon(\rho)\|_1
=
\|\Upsilon\|_\diamond.
\]

Since $\covering$ is a pure-state $1/40$-covering net, there exist pure states
$\psi_{\rm in},\psi_{\rm out}\in\covering$ such that
\[
\|\psi_{\rm in}-\rho\|_1\leq  \frac{1}{20},
\qquad
\|\psi_{\rm out}-O\|_1\leq  \frac{1}{20}.
\]
This pair satisfies the desired conclusion, since
\begin{align*}
&\left| \Tr\left( \psi_{\rm out} U \psi_{\rm in} U^\dagger \right) - \Tr\left( \psi_{\rm out}  V\psi_{\rm in} V^\dagger \right) \right|  =
\left|\Tr\bigl(\psi_{\rm out}\,\Upsilon(\psi_{\rm in})\bigr)\right| \\
&\quad \ge
\left|\Tr\bigl(O\,\Upsilon(\rho)\bigr)\right|
-
\left|\Tr\bigl((O-\psi_{\rm out})\,\Upsilon(\rho)\bigr)\right|
-
\left|\Tr\bigl(\psi_{\rm out}\,\Upsilon(\rho-\psi_{\rm in})\bigr)\right|.
\end{align*}
The first term $\left|\Tr \bigl(O\,\Upsilon(\rho)\bigr)\right|
=\frac{1}{2}\|\Upsilon\|_\diamond$. 
For the second term, using $\|O-\psi_{\rm out}\|_{\infty} \leq  \|O-\psi_{\rm out}\|_1\leq  1/20$ and Hölder's inequality, we have
\[
\left|\Tr\bigl((O-\psi_{\rm out})\,\Upsilon(\rho)\bigr)\right|
\leq 
\|O-\psi_{\rm out}\|_{\infty} \,\|\Upsilon(\rho)\|_1
\leq 
\frac{1}{20}\|\Upsilon\|_\diamond.
\]
For the third term, using $\|\psi_{\rm out}\|_{\infty} =1$ and
\lref{fact:diamond_norm_contractivity}, 
we get
\[
\left|\Tr\bigl(\psi_{\rm out}\,\Upsilon(\rho-\psi_{\rm in})\bigr)\right|
\leq 
\|\psi_{\rm out}\|_{\infty} \|\Upsilon(\rho-\psi_{\rm in})\|_1
=
\|\Upsilon(\rho-\psi_{\rm in})\|_1
\leq 
\|\Upsilon\|_\diamond\,\|\rho-\psi_{\rm in}\|_1
\leq 
\frac{1}{20}\|\Upsilon\|_\diamond.
\]
Therefore,
\[
\left| \Tr\left( \psi_{\rm out} U \psi_{\rm in} U^\dagger \right) - \Tr\left( \psi_{\rm out}  V\psi_{\rm in} V^\dagger \right) \right| 
\ge
\left(\frac{1}{2}-\frac{1}{20}-\frac{1}{20}\right)\|\Upsilon\|_\diamond
=
\frac{2}{5}\|\Upsilon\|_\diamond 
>
\frac{2}{5}\eta,
\]
where we have used the assumption $\|\Upsilon\|_\diamond> \eta$.
This completes the proof.
\end{proof}

Using \tref{thm:cseu_upper_bound_informal} in the main text with parameter $\B=1$ and accuracy
$\eps/5$, we obtain a CSEU protocol that uses
$\mathcal{O}(\min\{d\eps^{-2},d^{3/2}\eps^{-1}\})$ parallel and non-adaptive queries.
Applying \pref{prop:TomoEfficiency_formal}, the resulting unitary
tomography protocol uses
$\mathcal{O}(d)\cdot \mathcal{O}(\min\{d\eps^{-2},d^{3/2}\eps^{-1}\})
=
\mathcal{O}(\min\{d^2\eps^{-2},d^{5/2}\eps^{-1}\})$
queries and achieves diamond-norm error at most $\eps$.
Since the reduction preserves the query architecture, the tomography
protocol is also parallel and non-adaptive.
This recovers the
best-known query-complexity upper bound for unitary channel tomography with parallel
queries \cite{HKOT23,Yang_2020}.

\begin{remark}[A canonical tomography approach from parallel queries via classical shadow estimation] \label{remark:canonical_tomography_from_shadow}
    The Haah-Kothari-O'Donnell-Tang unitary channel tomography protocol \cite{HKOT23} provides a canonical approach to bootstrap a coarse, constant-error estimate of a unitary channel to a Heisenberg-limited one. It utilizes sequential, interleaved queries to evaluate the power $(U V_j^\dagger)^{p_j}$, where $V_j$ is the intermediate estimate of $U$ at the $j$-th iteration, similar to the phase estimation protocol. The protocol adaptively updates $V_j$ to $V_{j+1}$ and adjusts $p_j$ via quantum measurement feedback and classical post-processing. This bootstrap relies on the continuous geometry of the candidate unitary class and may not directly extend to certain settings, for example, when the candidate class is a non-closed subgroup of $\U(d)$ \cite[Remark~3.4]{HKOT23}.
Nevertheless, the framework has been widely adopted in various quantum process tomography tasks \cite{Zhao_2024, christensen2026learningfermioniclinearoptics, grewal2026efficientlearningstructuredquantum}.

    In comparison, our unitary channel tomography saves us from the demand for adaptivity and the group structure of candidate unitaries, thanks to the nature of classical shadows. Data collection is independent of the observables (and states), while query complexity scales logarithmically with respect to the number of expectations to be estimated. To learn unitaries that are known a priori in a finite subset $\mathsf{S} \subset \U(d)$, a canonical approach is to choose any pair of unitaries in $\mathsf{S}$ that are far in diamond distance, compute their identifier states [cf. \pref{prop:TomoEfficiency_formal}], and store them in the database. Whenever we are to identify $U \in \mathsf{S}$ via classical shadow estimation, we are to estimate $\Tr[\psi_a U \psi_b U^\dagger]$ for all pairs $(\psi_a, \psi_b)$ to a proper precision, which is destined to be small for the ground truth. This canonical approach readily generalizes to any subset of $\U(d)$ possessing a covering net of sufficient precision, provided the identifier states are precomputed accordingly. Notably, a constant-covering net for the set of $d$-dimensional pure states is always a valid source for building the identifier states [cf. \pref{prop:TomoEfficiency_formal}], while it is sometimes unnecessary for subsets with a much smaller covering net. This hypothesis-testing-based tomography guarantees that, upon collection, the classical shadows already contain the information necessary to reconstruct the unknown unitary.
\end{remark}

\section{Proof of Corollary
\ref{corollary:learn_shallow_circuit}}
\label{sec:proof_for_learning_shallow_circuit}

\begin{lemma}[Sewing lemma, {\cite[Lemma 9]{Huang_2024}}] \label{lemma:sewing_lemma}
    For an $n$-qubit unitary $U$, let $\sigma_i := \sigma_{\mathrm{p}, i} \otimes \openone_{[n] \setminus \{i\} }$ where $\mathrm{p} \in \{1, 2, 3\}$ be the local Pauli operator only non-trivial on qubit $i$, and $\widehat{O}_{i, \sigma}$ be the approximation of the Heisenberg-evolved operator $O_{i, \sigma} = U^\dagger \sigma_i U$ subject to $\|\widehat{O}_{i, \sigma}  - U^\dagger \sigma_{i} U\|_{\infty} \leq \eps_{i, \sigma}$, then there exists a sewing supermap $\mathcal{S}_U = \mathcal{S}(\{\widehat{O}_{i, \sigma}\}_{i, \sigma})$ from these estimates that satisfies
$$
\left\| \mathcal{S}_U - \mathcal{U} \otimes \mathcal{U}^\dagger \right\|_{\diamond} \leq 2 \sum_{i=1}^n \sum_{\mathrm{p} \in \{1, 2, 3\} } \eps_{i, \sigma_{\mathrm{p}}}.
$$
\end{lemma}

\begin{proof}[Proof of Corollary \ref{corollary:learn_shallow_circuit}]
Via the data processing inequality of diamond norm \cite{Nielsen2012}, the channel $\hat{\mathcal{S}}_U = \Tr_{\mathrm{anc}} \circ \mathcal{S}_U$ obtained by tracing out the ancilla for the sewing supermap in \lref{lemma:sewing_lemma} gives a $(2 \sum_{i=1}^n \sum_{\mathrm{p} \in \{1, 2, 3\} } \eps_{i, \sigma_{\mathrm{p}}} )$-approximation of $\mathcal{U}$ in diamond distance. To ensure that $\| \hat{\mathcal{S}}_U - \mathcal{U} \|_{\diamond} \leq \eps$, it suffices to learn each evolved operator $U^\dagger \sigma_i U$ to error $\eps' = \eps/(6n)$. By a standard lightcone argument \cite[Fact 5]{Huang_2024}, the evolution of the local operator $\sigma_i$ under a depth-$D$ $\qnc^0$ unitary $U^\dagger$ is confined to a lightcone $\mathfrak{L}_{D}(i) \subseteq [n]$ whose size is bounded by $|\mathfrak{L}_{D}(i)| \leq 2^D$, where $D = \bigo{1}$ even when the circuit is all-to-all connected. Suppose
$$
U^\dagger \sigma_i U = \sum_{\bfp \in \{0, 1, 2, 3\}^{|\mathfrak{L}_{D}(i)|}} \mu_{O_{i, \sigma}}(\bfp) \sigma_{\bfp} \otimes \openone_{\overline{\mathfrak{L}_{D}(i)}}
$$
under the Pauli basis. Since $U^\dagger\sigma_iU$ is traceless, we have $\mu_{O_{i,\sigma}}(\mathbf 0)=0$ and for the identity component we simply set $\widehat{\mu}_{O_{i,\sigma}}(\mathbf 0)=0$. For every nontrivial $\bfp$, taking $O = \sigma_i$ and $\rho = 2^{-n} (\openone + \sigma_{\bfp} \otimes \openone_{\overline{\mathfrak{L}_{D}(i)}})$ as CSEU requests produces an estimate of the corresponding Pauli coefficient by the following equality
$$
\Tr\left[ O U \rho U^\dagger \right] = \Tr\left[U^\dagger \sigma_i U\cdot \frac{\openone+ \sigma_{\bfp} \otimes \openone_{\overline{\mathfrak{L}_{D}(i)}}}{2^n} \right]= \mu_{O_{i, \sigma}}(\bfp).
$$
Since $|\mathfrak{L}_{D}(i)|\leq 2^D$, it suffices to set $\widetilde{\eps} = \eps' / (2\sqrt{2})^{2^D} = \eps / 6n(2\sqrt{2})^{2^D}$. Constructing $\widehat{O}_{i,\sigma}$ from estimates of the Pauli coefficients with additive error at most $\widetilde{\eps}$ then suffices to fulfill the precision requirement of the sewing procedure in Lemma \ref{lemma:sewing_lemma}: 
\begin{align*}
\left\| \widehat{O}_{i, \sigma} - U^\dagger \sigma_i U \right\|_{\infty} &= \left\|
\sum_{\bfp \in \{0, 1, 2, 3\}^{|\mathfrak{L}_{D}(i)|}} \widehat{\mu}_{O_{i, \sigma}}(\bfp) \sigma_{\bfp} \otimes \openone_{\overline{\mathfrak{L}_{D}(i)}} - \sum_{\bfp \in \{0, 1, 2, 3\}^{|\mathfrak{L}_{D}(i)|}} \mu_{O_{i, \sigma}}(\bfp) \sigma_{\bfp} \otimes
\openone_{\overline{\mathfrak{L}_{D}(i)}} \right\|_{\infty} \\
&\leq \left\| \sum_{\bfp \in \{0, 1, 2, 3\}^{|\mathfrak{L}_{D}(i)|}} \left( \widehat{\mu}_{O_{i, \sigma}}(\bfp) - \mu_{O_{i, \sigma}}(\bfp) \right) \sigma_{\bfp} \right\|_{\infty}
\\
&\leq \sqrt{ \sum_{\bfp \in \{0, 1, 2, 3\}^{|\mathfrak{L}_{D}(i)|}}\left| \widehat{\mu}_{O_{i, \sigma}}(\bfp) - \mu_{O_{i, \sigma}}(\bfp) \right|^2 \|\sigma_{\bfp}\|_F^2}
\\
&\leq
(2\sqrt{2})^{|\mathfrak{L}_{D}(i)|}
\max_{\bfp \in \{0, 1, 2, 3\}^{|\mathfrak{L}_{D}(i)|}}
\left|\widehat{\mu}_{O_{i, \sigma}}(\bfp)- \mu_{O_{i, \sigma}}(\bfp)\right| \leq (2\sqrt{2})^{2^D}\widetilde{\eps} = \eps'.
\end{align*}

Correspondingly, we will need $M=3\sum_{i=1}^n \left(4^{|\mathfrak{L}_D(i)|}-1\right) \leq 3n \cdot 4^{2^D} = \bigo{n} $ such pairs of observables and states. For every such CSEU request, we have $\Tr(\rho^2)=2/d$ and $\Tr(O^2)=d$. Therefore, by an analogous argument to the proof in Appendix \ref{appendix:ptm_learning_lower_bound}, the required number of queries to $U$ scales as
$$
K
=
\begin{cases}
\bigo{\dfrac{d}{\widetilde{\eps}}\log M},
& 0<\widetilde{\eps}\leq \dfrac{1}{d^{3/2}},
\\[2ex]
\bigo{\dfrac{d^{7/4}}{\sqrt{\widetilde{\eps}}}\log M},
&\dfrac{1}{d^{3/2}} <\widetilde{\eps}\leq \dfrac{1}{d^{3/4}}, \\[2ex]
\bigo{\dfrac{d}{\widetilde{\eps}^{3/2}}\log M},
& \dfrac{1}{d^{3/4}} <\widetilde{\eps}<1.
\end{cases}
$$
Since $D=\bigo{1}$, we have
$\widetilde{\eps}=\Theta (\eps/n )$
and $\log M=\bigo{\log n}$.
Thus, up to some constant factors depending solely on the quantum circuit depth $D$, the query complexity reads
$$
K_{\qnc^0} = K \cdot \bigo{\log n} = \begin{cases} \bigo{\dfrac{dn\log n}{\eps}}, & 0<\eps \leq \dfrac{n}{d^{3/2}},
\\[2ex]
\bigo{ \dfrac{d^{7/4} \sqrt{n} \log n }{\sqrt{\eps}} }, & \dfrac{n}{d^{3/2}} <\eps \leq \dfrac{n}{d^{3/4}}, \\[2.1ex]
\bigo{\dfrac{d  n^{3/2} \log n }{\eps^{3/2}}}, & \dfrac{n}{d^{3/4}} <\eps <1.
\end{cases}
$$
This concludes the proof.
\end{proof}

\end{document}